\documentclass{article}
\usepackage[margin=1in]{geometry}

\usepackage{amssymb}
\usepackage{amsmath}
\usepackage{amsthm}
\usepackage{bbm}
\usepackage{mathtools}
\usepackage{enumitem}
\usepackage{todonotes}
\usepackage{algorithm}
\usepackage{algpseudocode}
\usepackage{hyperref}
\usepackage{cleveref}
\usepackage{comment}
\usepackage{titling}
\usepackage{multicol}
\usepackage{slashed}
\usepackage{tikz}
\usepackage{adjustbox} 
\usepackage{tikz-3dplot}
\usetikzlibrary{backgrounds}
\usetikzlibrary{plotmarks}
\usetikzlibrary{shapes, calc}
\usepackage{latexsym}
\usepackage{thmtools}
\usepackage{marginnote}
\usepackage{todonotes}

\makeatletter
\providecommand*{\theHALG@line}{}
\renewcommand*{\theHALG@line}{\thealgorithm.\arabic{ALG@line}}
\makeatother

\usepackage{caption}
\numberwithin{equation}{section}

\newtheorem{thm}{Theorem}[section]
\newtheorem{definition}[thm]{Definition}
\newtheorem{lemma}[thm]{Lemma}
\newtheorem{claim}[thm]{Claim}
\newtheorem{fact}[thm]{Fact}
\newtheorem{remark}[thm]{Remark}

\newtheorem*{term*}{Terminology}
\newtheorem*{model*}{Input model (IM)}
\newtheorem*{goal*}{Algorithmic goal (AG)}
\newtheorem*{sett*}{Parameter setting (PS)}
\newtheorem*{question*}{Question (Q)}
\newtheorem*{lemma*}{Lemma}
\newcommand{\imlabel}{\hyperref[im]{\upshape (IM)}}
\newcommand{\aglabel}{\hyperref[ag]{\upshape (AG)}}
\newcommand{\pslabel}{\hyperref[ps]{\upshape (PS)}}

\newtheorem*{nota}{Notation}

\newcounter{todocounter}

\newcommand{\wt}{\widetilde}

\newcommand{\N}{\ensuremath{\mathbb{N}}}
\newcommand{\R}{\ensuremath{\mathbb{R}}}
\DeclareMathOperator*{\E}{\ensuremath{\mathbb{E}}}

\newcommand{\1}{\ensuremath{\mathbbm{1}}}
\newcommand{\e}{\ensuremath{\epsilon}}
\newcommand{\calL}{\ensuremath{\mathcal{L}}}
\newcommand{\tmu}{\ensuremath{\wt{\mu}}}
\newcommand{\D}{\ensuremath{P}}

\newcommand{\op}{\ensuremath{\mathrm{op}}}
\newcommand{\diag}{\ensuremath{\text{diag}}}
\newcommand{\spn}{\ensuremath{\mathrm{span}}}
\newcommand{\unif}{\ensuremath{\mathrm{unif}}}
\newcommand{\supp}{\ensuremath{\mathrm{supp}}}

\newcommand{\cross}{\ensuremath{\mathrm{cross}}}

\newcommand{\inedge}{\ensuremath{\mathrm{in}}}
\newcommand{\Lin}{\ensuremath{\mathcal{L}^{\inedge}}}
\newcommand{\Lcross}{\ensuremath{\mathcal{L}^{\cross}}}

\newcommand{\poly}{\ensuremath{\mathrm{poly}}}
\newcommand{\polylog}{\ensuremath{\mathrm{polylog}}}
\newcommand{\opt}{\ensuremath{\mathrm{OPT}}}

\newcommand{\vol}{\ensuremath{\mathrm{vol}}}

\newcommand{\nei}{\ensuremath{\mathrm{N}}}

\newcommand{\crossvstarapx}{\ensuremath{V^*}}
\newcommand{\crossgapx}{\ensuremath{G}}
\newcommand{\crossvapx}{\ensuremath{V}}

\newcommand{\spec}{\ensuremath{\mathrm{SpecCluster}}}
\newcommand{\apx}{\ensuremath{\mathrm{apx}}}
\newcommand{\specapx}{\ensuremath{\mathrm{SpecCluster}}}
\newcommand{\M}{\ensuremath{\mathcal{X}}}
\newcommand{\Mapx}{\ensuremath{\mathcal{X}}}

\newcommand{\NM}{\ensuremath{\nei(\mathcal{X})^*}}
\newcommand{\NMapx}{\ensuremath{\nei(\mathcal{X})^*}}

\newcommand{\im}{\ensuremath{\Gamma}}
\newcommand{\imapx}{\ensuremath{\Gamma}}

\newcommand{\speccore}{\ensuremath{\mathrm{SpecCore}}}
\newcommand{\mislabeled}{\ensuremath{\mathrm{\Lambda}}}
\newcommand{\lab}{\ensuremath{\mathrm{LabelCluster}}}

\newcommand{\robcon}{\ensuremath{\textsc{RobustConn}}}
\newcommand{\walk}{\text{\ensuremath{\bf{w}}}}
\newcommand{\fn}{\ensuremath{\textsc{FN}}}
\newcommand{\fp}{\ensuremath{\textsc{FP}}}
\newcommand{\specdp}{\ensuremath{\textsc{DP}}}

\title{Spectral Clustering with Side Information} 

\author{
\centering
\begin{tabular}{ccc}
\begin{tabular}{c}
Hendrik Fichtenberger\\
Google Research
\end{tabular}
&
\begin{tabular}{c}
Michael Kapralov\\
EPFL
\end{tabular}
&
\begin{tabular}{c}
Ekaterina Kochetkova\\
EPFL
\end{tabular}
\\[1.4em]
\begin{tabular}{c}
Silvio Lattanzi\\
Google Research
\end{tabular}
&
\begin{tabular}{c}
Davide Mazzali\\
EPFL
\end{tabular}
&
\begin{tabular}{c}
Weronika Wrzos-Kaminska\\
EPFL
\end{tabular}
\end{tabular}
}

\date{}

\begin{document}
\begin{titlingpage}
\maketitle
\abstract{
\noindent
In the graph clustering problem with a planted solution, the input is a graph on $n$ vertices partitioned into $k$ clusters, and the task is to infer the clusters from graph structure. A standard assumption is that clusters induce well-connected subgraphs (i.e. $\Omega(1)$-expanders), and connections between clusters are sparse (i.e. clusters form $\epsilon$-sparse cuts). Such a graph defines the clustering uniquely up to $\approx \epsilon$ misclassification rate, and efficient algorithms for achieving this rate are known. While this vanilla version of graph clustering is extremely well studied,  in the practice of graph analysis vertices of the graph are typically equipped with labels, or features, that provide additional information on cluster ids of the vertices. For example, each vertex could be equipped with a cluster label that is corrupted independently with probability $\delta$. Using either of the two sources of information separately leads to misclassification rate $\min\{\epsilon, \delta\}$, but can one combine the two to achieve misclassification rate $\approx \epsilon \delta$?

In this paper, we give an affirmative answer to this question, and furthermore show that such  a misclassification rate can be achieved in sublinear time in the number of vertices $n$. Our key algorithmic insight is a new observation on ``spectrally ambiguous'' vertices in a well-clusterable graph. 

While our sublinear-time classifier achieves the nearly optimal $\approx \widetilde O(\epsilon \delta)$ misclassification rate, the approximate clusters that it outputs do not necessarily induce expanders in the graph $G$. In our second result, we give a polynomial-time algorithm for reweighting edges of the input graph from the original $(k, \epsilon, \Omega(1))$-clusterable instance to a $(k, \widetilde O(\epsilon \delta), \Omega(1))$-clusterable instance (for constant $k$), improving sparsity of cuts nearly optimally and preserving expansion properties of the communities -- an algorithm for {\em refining community structure} of the input graph.
}

\end{titlingpage}

\setcounter{page}{0}
\tableofcontents
\thispagestyle{empty}
\newpage

\section{Introduction}
\label{sec:intro}

In the clustering problem one is given $n$ data items and the task is to partition them into similar groups, or clusters. The clustering is sometimes based on individual feature embeddings of the data items, such as in the $k$-means problem and its variants, or pairwise similarity information, such as in graph clustering (or its signed version, namely correlation clustering). In the latter case, the algorithm is given a graph $G=(V, E)$ and is tasked with partitioning vertices of $G$ into $k$ densely connected subgraphs, or ``communities'', with few edges crossing from one community to another. This vanilla graph clustering problem has received considerable attention in the literature
\cite{kvv04,von07, mmv12, LGT12,LRTV12,GT14, KLL17, GKLMS21, CAdOM24}, but in the practice of graph analysis one very often has access not just to the graph, but also to feature vectors associated with the vertices. One example is the {\em node property prediction} problem: see \cite{hu2020open} for the problem formulation and \cite{zhang2021scr,li2024long} for some of the best clustering approaches. In the node property prediction problem one would like to cluster vertices based on their features as well as the structure of a graph that they belong to. Examples include predicting the conference that a paper appeared at given its abstract and the citation graph, predicting category of an Amazon product given the product description and the co-purchasing graph, or predict the presence of protein functions (i.e. labels on individual proteins) given a set of edges representing interactions between proteins, each equipped with a feature vector. Thus, an exciting broad direction is the following.

\vspace{0.1in}

\fbox{
\parbox{0.9\textwidth}{
\begin{center}
How does one optimally cluster a graph equipped with vertex features, \\ i.e. perform clustering with side information?
\end{center}
}
}

\vspace{0.1in}

\noindent
This question has been nicely answered in the literature on the stochastic block model \cite{DeshpandeSMM18,MosselXuITCS, JMLR:v24:20-1419}, where the graph with planted cluster structure is sampled and presented to the algorithm together with labels on vertices that represent cluster ids, but are independently corrupted with some small probability.  Recently, also local clustering algorithms have been proposed for clustering problems with side information~\cite{yang2023weighted} based on a general contextual random graph model, which generalize both the stochastic block model and classic random models for planted cluster model. In contrast, our work considers a worst-case model of graphs,  where the notion of ``communities'' is based on conductance.

\begin{definition}[Conductance]
\label{def:conductance}
    Let $G=(V,E)$ be a graph, and let $\emptyset \neq S \subseteq V$. We define the conductance of $S$ in $G$, denoted $\Phi_G(S)$, as the ratio between the number of edges $|E(S,V\setminus S)|$ across $S$ and the volume of the smaller side $\min\{\vol(S),\vol(V\setminus S)\}$\footnote{The notation $\vol(S)$ denotes the sum of degree of vertices in $S$.}. Furthermore, we define the conductance of $G$, denoted $\Phi(G)$, as the minimum $\Phi_G(S)$ over all such cuts $S$.
\end{definition}

\noindent
Specifically, we work with worst-case graphs with a planted clustering as per the following definition.

\begin{definition}[$(k,\epsilon,\phi)$-clustering]
\label{def:clustering}
    Let $G=(V,E)$ be a graph, let $k \ge 2$ be an integer, and let $\epsilon, \phi \in (0,1)$. We say that a partitioning $C_1,\dots,C_k$ of $V$ is a $(k,\epsilon,\phi)$-clustering of $G$ if for every $i \in [k]$ one has $\Phi_G(C_i) \le \epsilon$ and \footnote{The notation $G\{C_i\}$ refers to the subgraph of $G$ induces $C_i$ where we measure volumes with respect to the degrees in $G$.}$\Phi(G\{C_i\}) \ge \phi$. Furthermore, we say that $G$ is $(k,\epsilon,\phi)$-clusterable, or that it admits a $(k,\epsilon,\phi)$-clustering, if there exists such a partitioning $C_1,\dots,C_k$ of $V$.
\end{definition}

\noindent
The problem of clustering such graphs has received a lot of attention in the literature~\cite{CzumajPS15, chiplunkar2018testing, GKLMS21}. A given $(k, \epsilon, \phi)$-clusterable graph can admit multiple $(k,\epsilon,\phi)$-clusterings, even though all such clusterings of a given graph must be $\epsilon/\poly(\phi)$-close to each other (see \Cref{fig:uninformative} and discussion at the end of \Cref{sec:tech_overview} for mode details). Thus, existing work has focused on efficiently recovering a clustering that is $\epsilon/\poly(\phi)$-accurate~\cite{CzumajPS15, chiplunkar2018testing, GKLMS21}. In this paper we ask a different question.

\vspace{0.1in}

\fbox{
\parbox{0.9\textwidth}{
\begin{center}
Can one recover a much better than $\epsilon/\poly(\phi)$-approximation to a {\em target clustering} $C_1, C_2, \ldots, C_k$ of an input $(k, \epsilon, \phi)$-clusterable graph given {\em side information} in the form of vertex labels perturbed with probability $\delta\ll 1$?
\end{center}
}
}

\vspace{0.1in}

\noindent
Specifically, we consider a setting where we are given a $d$-regular $(k, \epsilon, \phi)$-clusterable graph together with a labeling of vertices that provides a hint as to which cluster each vertex belongs to in the target clustering $C_1, C_2,\ldots, C_k$, where these hints are wrong with probability $\delta>0$ for every given vertex, independently. More formally, we work in the following model:

\begin{model*}\label{im}
    Let $V$ be a set of $n$ elements, let $k \ge 2$ be an integer, let $\epsilon,\phi \in (0,1)$ be the conductance parameters, let $\eta \ge 1$ be the balance parameter, and let $\delta \in (0,1)$ be the label perturbation parameter. An adversary performs the following steps:
    \begin{enumerate}
        \item it picks a partitioning $C_1,\dots,C_k$ of $V$ such that $|C_i|/|C_j| \le \eta$ for all $i,j \in [k]$;
        \item it picks an integer $d\ge 3$ and a $d$-regular graph $G=(V,E)$ such that $C_1,\dots,C_k$ is a $(k,\epsilon,\phi)$-clustering of $G$ as per \Cref{def:clustering};
        \item for every $u \in V$, it picks a distribution $p(u)$ over $[k]$ so that the probability of $i \sim p(u)$ being equal to the id of the cluster $u$ belongs to, i.e. $u \in C_{i}$, is at least $1 - \delta$;
        \item independently for every vertex $u \in V$, it samples a label $\sigma(u)$ from the distribution $p(u)$. 
    \end{enumerate}
    Then, an algorithm is given  $G,\sigma$, and the parameters $n, d, k, \epsilon, \phi, \eta$ as input.
\end{model*}

\begin{remark}
    For simplicity of the presentation, we only consider $d$-regular graphs in our input model. The algorithms which we develop in this paper, however, can be applied to $d$-bounded degree graphs, with the notion of conductance (\Cref{def:conductance}) normalised by $d|S|$ rather than $\vol(S)$, via a standard reduction. More specifically, we convert a $d$-bounded degree graph into a $d$-regular graph by adding $d - \deg(x)$ self-loops to each vertex $x \in V$ and apply the algorithms to the resulting graph.
\end{remark}

\noindent
The goal is to design an algorithm that, given a clusterable graph and labeling of its vertices as above, determines which cluster every vertex belongs to as accurately as possible. As mentioned above, the graph alone already determines the clustering up to misclassification rate\footnote{The misclassification rate is the fraction of vertices for which the inferred cluster id does not correspond to the target clustering (see the algorithmic goal \aglabel{}).} of $\epsilon/\poly(\phi)$, i.e. incorrectly assigns at most an $\epsilon/\poly(\phi)$ fraction of vertices. The labels themselves can be used to guess the clustering up to an $O(\delta)$ misclassification rate, and it is not hard to show that a misclassification rate smaller than the product of the two, i.e. $\epsilon \delta/\poly(\phi)$, is not possible -- see \Cref{sec:tech_overview}. Thus, our goal is to obtain an algorithm with misclassification rate $\gamma$ as close to the product $\epsilon\delta/\poly(\phi)$ as possible. Formally,  we seek to do the following:

\begin{goal*}\label{ag}
    Given $G$, $\sigma$, and the parameters $k, \phi, \eta$ as per the input model above, output a labeling $\alpha: V \to [k]$ that with good probability has misclassification rate at most $\gamma \in (0,1)$, i.e. the partitioning $\widehat{C}_1 = \{u \in V: \, \alpha(u)=1\}, \dots, \widehat{C}_k = \{u \in V: \, \alpha(u)=k\}$ satisfies
\begin{equation*}
    \sum_{i=1}^k\left|\widehat{C}_i \triangle C_i\right| \le \gamma n \, .
\end{equation*}
\end{goal*}

\noindent
The algorithm has access to two noisy observations derived from the target partitioning $C_1,\dots,C_k$: the graph $G$ and the labels $\sigma$. The idea is that combining these two distinct observations of the same signal, the algorithm is able to approximately recover the target partitioning. For this intuition to make sense at all, we restrict ourselves to work in a parameter setting where each of these signals is somewhat useful in isolation.

\begin{sett*}\label{ps}
    In Sections \ref{sec:intro} and \ref{sec:tech_overview}, we restrict the input model above to $k,\eta = O(1)$, $\phi = \Omega(1)$, and $\epsilon ,\delta \ll 1$ for simplicity of presentation.
\end{sett*}

\noindent
Under the above parameter setting, our central question is the following.

\vspace{0.1in}

\fbox{
\parbox{0.9\textwidth}{
\begin{center}
    Can an efficient algorithm that sees both $G$ and $\sigma$ achieve misclassification rate $\gamma \approx {\epsilon \delta}$,
    hence doing better than what can be done by looking at either graph structure or labels in isolation?
\end{center}
}
}

\vspace{0.1in}

\noindent
Our main result is an affirmative answer to the above question. In fact, we even show that such misclassification rate can be achieved by a sublinear time clustering oracle, i.e., a small space data structure that allows sublinear (in the number of vertices $n$) time access to the $\approx \epsilon\delta$-approximate clustering. The clusters that our algorithm returns, however, do not necessarily induce expanders, like in the original input instance. We then also ask the following.

\vspace{0.1in}

\fbox{
\parbox{0.9\textwidth}{
\begin{center}
    Can an efficient algorithm that sees both $G$ and $\sigma$  {\em refine} the original $(k, \epsilon, \phi)$-clusterable graph to a $(k, \epsilon', \phi)$-clusterable graph for $\epsilon'\approx \epsilon \delta$?
\end{center}
}
}

\subsection{Our results}
\label{sec:ourresults}

\subsubsection*{Spectral clustering oracles with side information}

\noindent
One might wonder why we are asking specifically about misclassification rate $\epsilon \delta$. The reason is that, even when we are given both signals, we cannot beat the product of what each of them can provide in isolation. In the next paragraph, we show an example supporting this claim.

\paragraph{Why $\epsilon\delta$ is the best we can hope for.} To show this, let us fix $V$ and $k=2$. We show a partitioning $C_1,C_2$ of $V$ that is a $(2,\epsilon,\phi)$-clustering of a graph $G=(V,E)$, but such graph $G$ is completely uninformative about $\epsilon n$ vertices. Since the labeling will be wrong on a $ \delta$-fraction of these $\epsilon n$ vertices, we are bound to misclassify an $ \epsilon \delta$-fraction of vertices. We construct the graph $G$ as follows (see \Cref{fig:uninformative}): we partition $V$ into $A,B$ and a set of middle vertices $M$, where $|A|,|B| \approx n/2$ and $|M|\approx \epsilon n$; we define $G[A]$ and $G[B]$ to be $\phi$-expanders; we symmetrically connect every middle vertex in $M$ to $d/2$ neighbors in $A$ and in $B$. One can see that any partitioning of $V$ into $C_1,C_2$ such that $A\subseteq C_1$, $B\subseteq C_2$, and where $M$ is equally split between $C_1,C_2$, is a valid $(2,\epsilon,\phi)$-clustering of $G$. In this sense, $G$ is uninformative about the $ \epsilon n$ middle vertices~$M$.

\begin{figure}[h]
	\centering
	\begin{tikzpicture}[scale=1]

\def\Xradius{2}
\def\Yradius{1.3}
\definecolor{mydarkgreen}{RGB}{30,100,40}

\coordinate (Lcenter) at (-3,0);
 \node[anchor=west, scale = 1] at (-1, -1) {$A$};
 
\coordinate (Rcenter) at (3,0);
 \node[anchor=west, scale = 1] at (5, -1) {$B$}; 
 
\coordinate (Ccenter) at (0,3.3);
\coordinate (Rcenter) at (3,0);
 \node[anchor=west, scale = 1] at (-2, 3.0) {$M$};

\draw[fill=blue!15, draw=blue] (Lcenter) ellipse[x radius=\Xradius, y radius=\Yradius];
\draw[fill=blue!15, draw=blue] (Rcenter) ellipse[x radius=\Xradius, y radius=\Yradius];
\draw[fill=red!20, draw=red] (Ccenter) ellipse[x radius=1.2, y radius=0.8];

\newcommand{\placeEvenNodes}[5]{%
  
  \foreach \i in {0,...,\numexpr#2-1} {
    \pgfmathsetmacro{\angle}{#5 + 360 * \i / #2}
    \pgfmathsetmacro{\xshift}{\Xradius * cos(\angle)}
    \pgfmathsetmacro{\yshift}{\Yradius * sin(\angle)}
    \path let \p1 = #3 in
      node[circle, fill=black, draw=black, minimum size=#4, inner sep=0pt]
      (#1\i) at (\x1+10*\xshift,\y1+10*\yshift) {};
  }
}

\newcommand{\placeEvenNodesLarger}[5]{%
  \pgfmathsetseed{123} 
  \foreach \i in {0,...,\numexpr#2-1} {
    \pgfmathsetmacro{\angle}{#5 + 360 * \i / #2}
    \pgfmathsetmacro{\xshift}{\Xradius * cos(\angle)}
    \pgfmathsetmacro{\yshift}{\Yradius * sin(\angle)}
    \pgfmathsetmacro{\rnum}{20 + rand*5}
    \path let \p1 = #3 in
      node[circle, fill=black, draw=black, minimum size=#4, inner sep=0pt]
      (#1\i) at (\x1+\rnum*\xshift,\y1+\rnum*\yshift) {};
   
  }
}

\newcommand{\placeEvenNodesSmaler}[5]{%
  \pgfmathsetseed{123} 
  \foreach \i in {0,...,\numexpr#2-1} {
    \pgfmathsetmacro{\angle}{#5 + 360 * \i / #2}
    \pgfmathsetmacro{\xshift}{\Xradius * cos(\angle)}
    \pgfmathsetmacro{\yshift}{\Yradius * sin(\angle)}
    \pgfmathsetmacro{\rnum}{15 + rand*10}
    \path let \p1 = #3 in
      node[circle, fill=black, draw=black, minimum size=#4, inner sep=0pt]
      (#1\i) at (\x1+\rnum*\xshift,\y1+\rnum*\yshift) {};
   
  }
}

\placeEvenNodesLarger{Lbig}{6}{(Lcenter)}{3pt}{0}
\placeEvenNodesSmaler{Lsmall}{25}{(Lcenter)}{2pt}{9}

\placeEvenNodesLarger{Rbig}{6}{(Rcenter)}{3pt}{0}
\placeEvenNodesSmaler{Rsmall}{25}{(Rcenter)}{2pt}{9}

\placeEvenNodes{Cmid}{3}{(Ccenter)}{3pt}{0}

\foreach \i in {0,...,2} {
  \pgfmathtruncatemacro{\j}{2*\i}
  \pgfmathtruncatemacro{\k}{2*\i + 1}
  \ifnum\j<12
    \draw[-, >=stealth, black, thin]  (Cmid\i) -- (Lbig\j);
  \fi
  \ifnum\k<12
    \draw[-, >=stealth, black, thin]  (Cmid\i) -- (Lbig\k);
  \fi
}

\foreach \i in {0,...,2} {
  \pgfmathtruncatemacro{\j}{2*\i}
  \pgfmathtruncatemacro{\k}{2*\i + 1}
  \ifnum\j<12
    \draw[-, >=stealth, black, thin]  (Cmid\i) -- (Rbig\j);
  \fi
  \ifnum\k<12
    \draw[-, >=stealth, black, thin]  (Cmid\i) -- (Rbig\k);
  \fi
}

\draw[rotate = 135][mydarkgreen, dashed, thick] (2.5, 0.3) ellipse[x radius=2.5, y radius=4.0];

 \node[anchor=west, scale = 1, color = mydarkgreen] at (-3, 5) {$C_1$};

\draw[mydarkgreen, dashed, thick] (3,0.1) ellipse[x radius=2.1, y radius=1.7];
 \node[anchor=west, scale = 1, color = mydarkgreen] at (5.5, 0) {$C_2$};
\draw[rotate = 45][red, dotted, thick] (2.5, -0.3) ellipse[x radius=2.5, y radius=4.0];

 \node[anchor=west, scale = 1, color = red] at (3, 5) {$\widetilde{C}_1$};

\draw[red, dotted, thick] (-3,0.1) ellipse[x radius=2.1, y radius=1.7];

 \node[anchor=west, scale = 1, color = red] at (-6, 0) {$\widetilde{C}_2$};

\draw[->, thick] (-4,4.5) -- (-0.8,3.5);
\node[anchor=west, align=left] at (-7.5,4.5) {vertices in $M$\\may be in any cluster};

\end{tikzpicture}
	  \caption{In the clustering $C_1,C_2$ (dashed circles), the vertices in $M$ belong to $C_1$; in the clustering  $\widetilde{C}_1,\widetilde{C}_2$ (dotted circles), the vertices in $M$ belong to $\widetilde{C}_2$. Since both are valid $(2,\epsilon,\phi)$-clustering of this graph, it is uninformative as to which cluster the vertices in $M$ belong to.}
\label{fig:uninformative}
\end{figure}
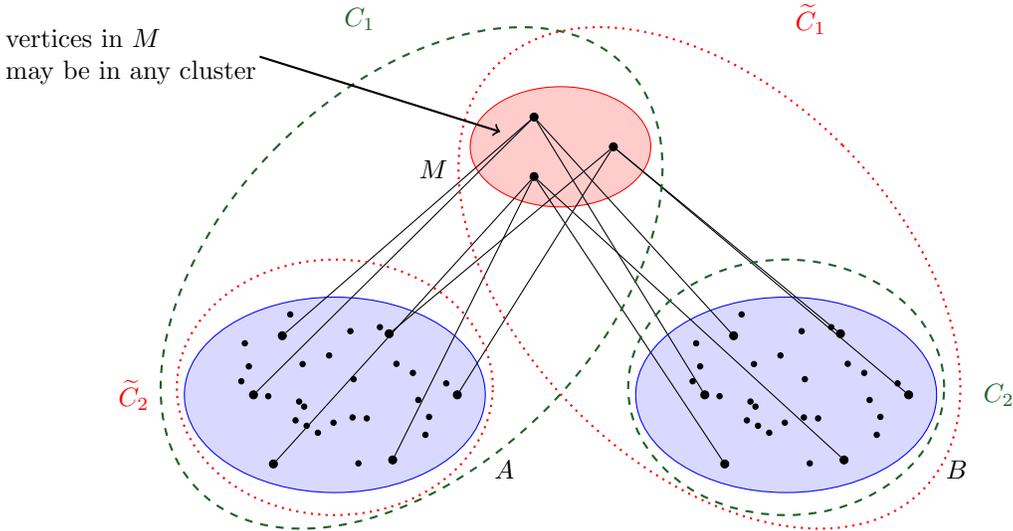

\vspace{1em}
\noindent
Our main result is an algorithm that meets the optimal misclassification rate of $\epsilon \delta$ up to a logarithmic factor in $1/\delta$, thus positively answering our question above. Moreover, this algorithm classifies any given vertex in sublinear time.

\begin{thm}[Classifier, see \Cref{thm:sublinear}]
\label{thm:informal_classsifier}
    There is a \smash{$d\cdot n^{1/2+O(\epsilon)}\poly(\log(n/\delta))$}-time\footnote{For the stated preprocessing time, we additionally require $\epsilon \geq 1/\poly(\log(n))$. For more details, see \Cref{thm:sublinear} and \Cref{rem:additional_bounds_eps}.} algorithm that, given $G$ and $\sigma$ as per \imlabel{} and respecting the parameter setting \pslabel{}, prepares a data structure that features query time \smash{$d\cdot n^{1/2+O(\epsilon)}\poly(\log(n/\delta))$} such that with probability $0.9$ over the internal randomness of the algorithm and the draw of $\sigma$, all but an \smash{$\widetilde{O}(\epsilon \delta)$}\footnote{Here, the notation \smash{$\widetilde{O}(\cdot)$} hides a $\log (1/\delta)$ factor.} fraction of queries $u \in V$ are answered with a cluster id $i \in [k]$ such that $u \in C_i$, i.e. it has misclassification rate \smash{$\widetilde{O}(\epsilon \delta)$}.
\end{thm}

\begin{remark}
    Theorems \ref{thm:informal_classsifier}, \ref{thm:informal_reweight_applied} and \ref{thm:informal_reweight} are informal versions of our main results which, in particular, apply to a more general setting of parameters.
\end{remark}

\subsubsection*{Refining community structure with side information}
Our \Cref{thm:informal_classsifier} demonstrates that using the additional signal from $\sigma$ can significantly enhance the recovery of the target clustering, far beyond what is achievable by analyzing $G$ alone. Our second result answers our last question, showing that we can embed the signal from $\sigma$ in $G$ by reweighting its edges to refine its community structure. In other words, we use $\sigma$ to morph $G$ into a new graph that carries a much better signal of the target clustering compared to $G$.

\begin{thm}[Refining communities -- see \Cref{thm:round_random_sigma}]
\label{thm:informal_reweight_applied}
    There is a polynomial-time algorithm that, given $G$ and $\sigma$ as per \imlabel{} and respecting the parameter setting \pslabel{}, outputs a reweighting of $G$ that with probability $0.9$ over $\sigma$ admits a \smash{$(k,\widetilde{O}(\epsilon \delta),\Omega(1))$}-clustering $C_1',\dots,C_k'$ that is \smash{$\widetilde{O}(\epsilon \delta)$}-close to the target clustering i.e. $\sum_{i =1}^k |C_i'\triangle C_i| = \widetilde{O}(\epsilon \delta ) \cdot n$.
\end{thm}

\noindent
 In fact, our result is more general: given any labeling which is correct on all but $\gamma \ll \epsilon$ vertices, without any assumption on how it is generated, we can reweight $G$ to boost the sparsity of cuts between clusters from $\epsilon$ to $\approx \gamma$.

 \begin{thm}[Refining communities, general -- see \Cref{thm:round_to_clustering}]
\label{thm:informal_reweight}
    Let $d\ge 3$, $k=O(1)$, $\gamma \ll\epsilon \ll 1$, let $G=(V,E)$ be a $d$-regular graph which admits a $(k,\epsilon,\Omega(1))$-clustering $C_1,\dots,C_k$ such that $|C_i|/|C_j| = O(1)$ for all $i,j \in [k]$, and let $\alpha:V\rightarrow [k]$ be a labeling such that $u \notin C_{\alpha(u)}$ for all but a $\gamma$-fraction of $u \in V$. Then, there is a polynomial-time algorithm that, given $G$ and $\alpha$, outputs a reweighting of $G$ that admits a $(k,O(\gamma),\Omega(1))$-clustering $C_1',\dots,C_k'$ that is $O(\gamma)$-close to the target clustering, i.e. $\sum_{i =1}^k |C_i'\triangle C_i| = O(\gamma) \cdot n$.
\end{thm}

\noindent
One can see that \Cref{thm:informal_reweight} and \Cref{thm:informal_classsifier} together readily give \Cref{thm:informal_reweight_applied}: we can compute a label $\alpha(u)$ for every vertex with the classifier from \Cref{thm:informal_classsifier}, which is correct on all but $ \widetilde{O}(\epsilon \delta )$ vertices, and then run the algorithm from \Cref{thm:informal_reweight} on $G$ and $\alpha$.

\subsection{Related work}
As we mentioned before, the problem of recovering the planted clustering in clusterable graphs has been well-studied in the literature. One of the first examples is the work of Goldreich and Ron~\cite{gr11} on expansion testing, which essentially considered the setting $k=1$. Follow-up works~\cite{KaleS08,NachmiasS10,gr11,DBLP:journals/cpc/CzumajS10,KalePS13} resulted in clustering oracles that work for $k=2$ and any small constant $\epsilon>0$, using similar collision counting primitives for approximating the spectral embeddings. Furthermore, the work of~\cite{CzumajPS15} was the first to consider the case of general $k$, but only for $\epsilon\ll 1/({\text{poly}(k) \log n})$. The case of general $\epsilon$ and general $k$ required new techniques. To extract more structure from random walk distributions, \cite{chiplunkar2018testing,GKLMS21} introduced new linear algebraic post-processing techniques.

A relatively newer body of work examines stochastic block models (SBM) augmented by noisy vertex labels. For example, \cite{DeshpandeSMM18, MosselXuITCS, JMLR:v24:20-1419} focused on a version where nodes receive noisy labels (or feature vectors), and proposed computationally efficient algorithms achieving optimal recovery. Several works build on the idea of using the information coming from the noisy labels jointly with the information about the graph structure, such as \cite{newman2016structure, jiang2015stochastic}, which allows for accurate recovery algorithms even if the vertex labels are misleading. 

\section{Technical overview}
\label{sec:tech_overview}

In this section, we discuss the main ideas behind our results. In \Cref{sec:techoverview_attemps}, we begin by describing a few natural attempts to design an algorithm that combines $G$ and $\sigma$ to achieve a misclassification rate of $\epsilon \delta$. In \Cref{sec:techoverview_classifier}, we present the main technical tools used to achieve this rate, thereby proving \Cref{thm:informal_classsifier}. Finally, in \Cref{sec:techoverview_reweight}, we show how to reweight $G$ to refine its community structure, as stated in \Cref{thm:informal_reweight}.

\subsection{Some natural approaches and why they fail}
\label{sec:techoverview_attemps}

\noindent
We are interested in the setting where we are given a $d$-regular graph $G=(V,E)$ and a labeling $\sigma : V \rightarrow [k]$ with the following properties: $G$ admits an unknown $(k,\epsilon,\phi)$-clustering $C_1,\dots,C_k$, $\sigma$ assigns the correct cluster id with probability $1-\delta$ or a wrong cluster id with probability $\delta$ independently for each vertex;  $k = O(1)$, $\phi = \Omega(1)$, $\epsilon ,\delta \ll 1$, and $|C_i|/|C_j| = O(1)$ for all $i,j \in [k]$. The goal is to design an algorithm that given $G$ and $\sigma$ is able to recover $C_1,\dots,C_k$ as accurately as possible. As mentioned above, the best one can hope for is a misclassification rate of $\epsilon \delta$. We begin our discussion by considering a few natural attempts at obtaining such an algorithm.

\paragraph{Majority-voting.} One might hope that a simple majority vote among a vertex's neighbors would suffice to classify it. This fails even for $k=2$, for the following reason: a $(2,\epsilon,\phi)$-clustering of $G$ allows for $ \epsilon n$ problematic vertices in, say, $C_1$ to have $1-\phi$ fraction of their neighbors in $C_2$. Hence, this approach misclassifies at least an $ \epsilon$-fraction of vertices. Even doing a majority vote among a vertex's $t$-hop neighborhood suffers from the same barrier: most walks starting from a problematic vertex leak to $C_2$ in just one step, so also their $t$-hop neighborhood mainly consists of vertices from $C_2$.

\paragraph{Majority-voting++.}
Still in the simplified setting of $k=2$, there is more nuanced (but still very simple) variant of the vanilla majority-voting scheme discussed above that achieves misclassification rate $\approx \epsilon \delta$ if $\delta \le \epsilon$. This can be done by moving the critical threshold from $1/2$ (i.e. the usual threshold for majority-voting) to, say, $2/3 \cdot \phi $. However, it is not clear how this algorithm can be generalized to the general regime of $\delta$ and $d$. We discuss in \Cref{sec:majvoting++} the correctness of this simple classifier and the obstacles in extending it to solve the problem in general. 
\paragraph{Naive spectral clustering.}
It is known that the spectrum of clusterable graphs is quite informative of the target partitioning. It is then natural to try to combine the information provided by the spectral properties of the graph with the information given by the labeling $\sigma$. Before discussing how one might go about this, we want to be more precise regarding the kind of signal given by the spectrum of $G$: consider the eigenvectors $x_1,\dots,x_n \in \R^V$ of its normalized Laplacian, and let us define the $k$-dimensional embedding of a vertex $u$ to be the vector $f_u \in \R^k$ whose $i$-th entry equals the value of $x_i$ in $u$'s coordinate. Then, for every $i$ we define the cluster mean of $C_i$ as
\begin{align}
\label{eq:clustermean}
    \mu_i = \frac{1}{|C_i|}\sum_{u \in C_i} f_u \, .
\end{align}
It is known that the cluster means are far from each other and that the spectral embedding of most vertices tends to be close to their cluster mean. Formally, this is captured in the following lemma~\cite{CzumajPS15,GKLMS21} (sometimes referred to as the \textit{variance bound} as it bounds the directional variance of points in the spectral embedding around their cluster means; these cluster means $\mu_i$ are nearly orthogonal):
\begin{lemma}[Variance bound-- see \Cref{lemma:variancebound}, \ref{lemma:clustermeans}]
\label{informal:varbound}
    For any $k$-dimensional unit-norm vector $\alpha$, one has
    \begin{equation*} \sum_{i =1}^k \sum_{u \in C_i} \langle f_u-\mu_i,\alpha\rangle^2 \lesssim \epsilon   \, ,
    \end{equation*}
    and moreover
    \begin{equation*}
           \forall \, i\neq j \in [k], \,\,\, \left|\langle \mu_i , \mu_j \rangle\right| \lesssim \sqrt{{\epsilon}}  \cdot \frac{k}{n} \quad   \quad \text{and}\quad   \quad \forall \, i \in [k], \,\,\, \|\mu_i\|^2_2 \approx  \left(1\pm \sqrt{\epsilon}\right) \cdot\frac{k}{n} \, .
    \end{equation*}
\end{lemma}
\noindent
It is then natural to distinguish vertices that are close to some cluster mean from vertices that are far from each of them. We say that the vertices of the first kind form \textit{spectral clusters}, while the second kind are \textit{cross vertices}, defined below.
\begin{definition}[Spectral clusters and cross vertices -- see \Cref{def:spec}]
\label{informal:spec}
    For $i \in [k]$, we define the $i$-th spectral cluster as 
    $$
    \spec(i) = \{u \in V: \, \|f_u-\mu_i\| \lesssim \phi \|\mu_i\|\},
    $$
     which are pairwise disjoint by design, since \Cref{informal:varbound} implies that the $\mu_i$'s are far from each other.
    Moreover, let 
    $$\M = V \setminus \cup_{i \in [k]} \spec(i)$$ denote the set of cross vertices. The notion of a spectral cluster induces a natural {\emph spectral labeling}  $\tau: V \rightarrow [k] \cup \{*\}$ that maps a cross vertex to $*$ (undefined) and other vertices to the spectral clusters they belong to:
    \begin{equation*}
        \tau(u) = \begin{cases}
            i\, , & \quad \text{\upshape if } u \in \spec(i) \\
            *\, , & \quad \text{\upshape if } u \in \M
        \end{cases} \, .
    \end{equation*}
\end{definition}
\noindent

\begin{remark}\label{rm:spec}
    A vertex $u\in V$ that belongs to $\spec(i)$, i.e. satisfies $\tau(u)=i$, {\bf does not} necessarily belong to cluster $C_i$. Instead, $\tau(u)=i$ simply means that spectral information suggests that the cluster of $u$ is $i$. In fact, up to $\approx \epsilon/\phi$ fraction of vertices can be misclassified by this rule. Our algorithm will combine this prediction with the label $\sigma(u)$ carefully to achieve $\approx \epsilon \delta$ misclassification rate.
\end{remark}

\noindent
As per \Cref{rm:spec}, we would ideally like the vertices in the $i$-th spectral cluster to correspond to the vertices in $C_i$. However, this is not necessarily true, which motivates the notion of \textit{impostors}: vertices that spectrally look like they belong to $C_i$ but are actually from $C_j$.

\begin{definition}[Impostors -- see \Cref{def:impostor}]
    \label{informal:impostor}
    For $i \neq j\in[k]$, we define the set of $(j\to i)$-impostors, denoted $\im(i,j)$, to be the subset of $C_j$ that lies in the $i$-th spectral cluster, i.e. $\im(i,j) = \spec(i) \cap C_j$. Furthermore, let  for convenience $\im = \cup_{i \neq j \in [k]} \im(i,j)$ be the set of all impostor vertices.
\end{definition}
\noindent
With this notation in place, we can compactly state the recovery guarantee provided by the variance bound: all but $\epsilon n$ vertices are close to their cluster mean.

\begin{lemma}[Few cross and impostor vertices -- see \Cref{lemma:impostor_n_cross_size}]
    \label{informal:fewbad}
    At most $\approx \epsilon n$ vertices are cross or impostor vertices, i.e.
    $|\M \cup \im| \lesssim \epsilon n$.
    \end{lemma}

\noindent
In order to better appreciate the guarantee of this lemma, we introduce some notation to refer to the true cluster ids.

\begin{nota}
    We use $\iota: V \rightarrow [k]$ to denote the true cluster ids, such that $v \in C_{\iota(v)}$ for every $v \in V$.
\end{nota}

\noindent
Then, \Cref{informal:fewbad} can be restated as follows.

\begin{remark}
    \Cref{informal:fewbad} guarantees that $\|f_u - \mu_{\iota(u)}\|_2 \lesssim \phi \|\mu_{\iota(u)}\|_2^2$ for all but $\approx \epsilon n$ vertices. In other words, one has $\tau(u)=\iota(u)$ for all but $\approx \epsilon n$ vertices.
\end{remark}

\noindent
In light of this, one might try to combine the spectral information of $G$ with the randomly perturbed labeling $\sigma$. A naive attempt to classify a vertex $u \in V$ can be the following.

\begin{algorithm}
\caption*{Naive spectral clustering}\label{alg:approach2}
\begin{algorithmic}[0]
\State \textbf{Input:} $G$, $\sigma$, and a vertex $u \in V$
\State \textbf{Output:} a cluster id in $[k]$
\If{$\tau(u) \neq * $}  \Comment{$u$ is in a spectral cluster} 
\State \Return $\tau(u)$
\Else  \Comment{$u$ is a cross vertex, i.e. $u \in \M$}
\State \Return $\sigma(u)$
\EndIf
\end{algorithmic}
\end{algorithm}

\noindent
The above algorithm would give misclassification rate $\epsilon \delta$, if we knew that any vertex that is far from its cluster mean is a cross vertex. While this is true for certain classes of random graphs~\cite{afwz17,DLS21,BJKMMW24}, it is not known for worst-case clusterable graphs: according to the variance bound and \Cref{informal:fewbad}, it is plausible that $\epsilon n$ vertices in $C_i$ are not only far from $\mu_i$ but also close to $\mu_j$. In other words, the naive spectral clustering algorithm is easily fooled by the (possibly $\epsilon n$ many) impostor vertices (see \Cref{informal:impostor} and \Cref{informal:fewbad}). While we do not know of an explicit example graph that realizes this possibility, we do not know how to rule it out even when $k=2$. Therefore, all we can say about the algorithm outlined above is that it achieves a misclassification rate of $\epsilon$.

\subsection{Our approach}
\label{sec:techoverview_classifier}

Our goal is to leverage the spectral structure of clusterable graphs together with the perturbed labeling $\sigma$, as in the naive spectral clustering algorithm. That approach attempts to exploit the fact that one can easily identify vertices that lie far from any cluster mean~-- these are the cross vertices, which belong to no spectral cluster (see \Cref{def:spec})~-- and assigns labels to them using  $\sigma$. The key obstacle to achieving a misclassification rate better than $\epsilon$ is that we cannot easily identify the impostor vertices -- those that appear close to a wrong cluster (see \Cref{def:impostor}) -- whose existence we cannot rule out or bound more tightly than  $\epsilon n$, as previously discussed. However, as it turns out, there is a crucial property of impostors that we can exploit to identify them: vaguely, most of them are able to leave the set of vertices that have the same label $\sigma$ and are in the same spectral cluster. To make this more precise, we introduce convenient notation to refer to the set of vertices that have a given label $i$ in $\sigma$. These sets, called \emph{label clusters}, can be used by an algorithm as proxies to true clusters $C_i$'s.

\begin{definition}[Label clusters -- see \Cref{def:labelcluster}]
\label{informal:labcluster}
For $i \in [k]$, we define the $i$-th label cluster to be the set $\lab(i) = \{u \in V: \sigma(u)=i\}$.
\end{definition}

\noindent
Then, consider a vertex $u$ in the $i$-th spectral cluster with label $\sigma(u)=j$. We want to test if $u$ is a $(j\to i)$-impostor by checking if it can leave the set $\spec(i) \cap \lab(j)$. Specifically, we check if it can reach the cross vertices $\M$ (which are disjoint from $\spec(i) \cap \lab(j)$ by \Cref{informal:spec}) by walking in the subgraph induced by $(\spec(i) \cap \lab(j)) \cup \M$. We call this subgraph the $(i,j)$-cross graph, defined below and illustrated in \Cref{fig:cross}.

\begin{definition}[Cross graphs -- see \Cref{def:crossgraph}]
\label{informal:crossgraph}
     For $i \neq j\in[k]$, we define the cross graph $G_{i,j}=(V,E_{i,j})$ to be the subgraph induced by the vertices in the $i$-th spectral cluster with label $j$, together with the cross vertices. In other words, $G_{i,j}$ is the subgraph induced by the edges on $(\spec(i) \cap \lab(j)) \cup \M$.
\end{definition}

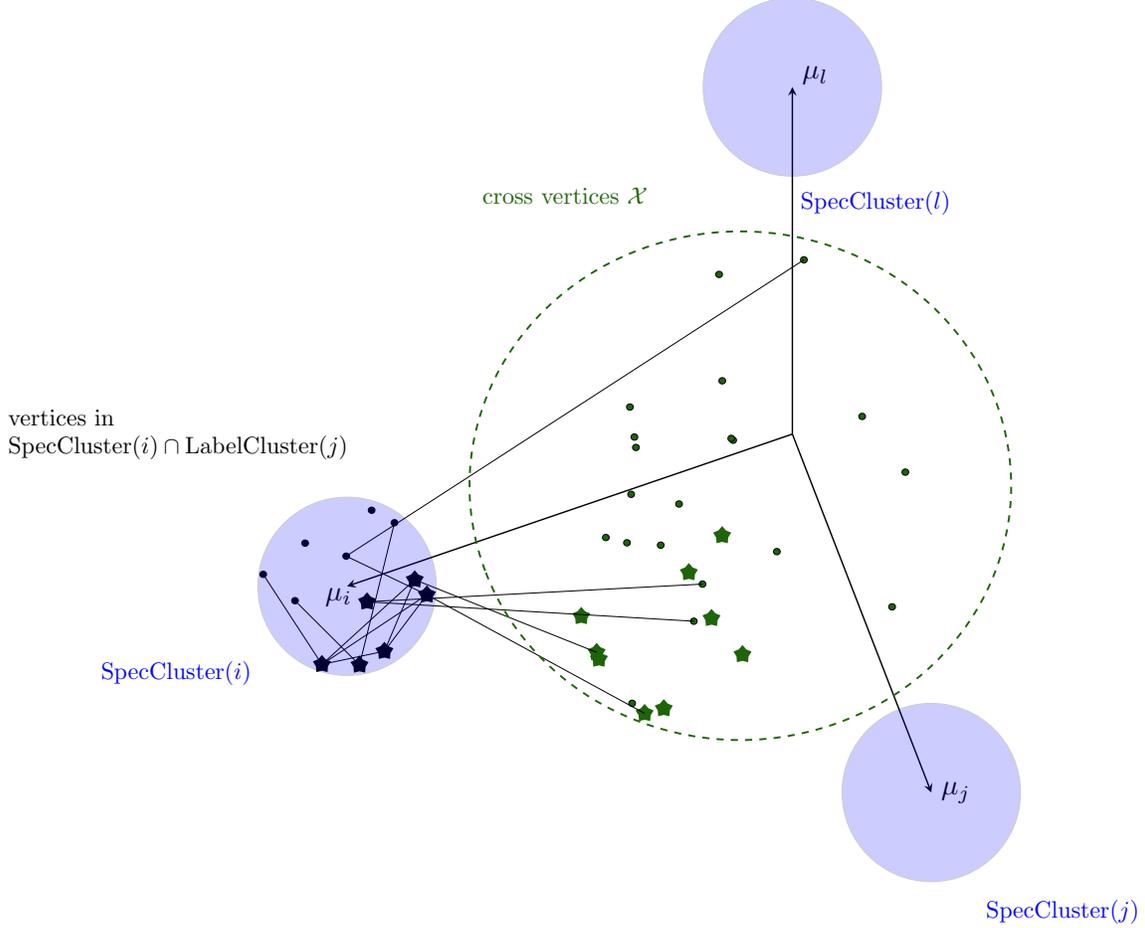
\begin{figure}[ht]
	\centering
    
	\scalebox{0.9}{\tdplotsetmaincoords{20}{160}
\begin{adjustbox}{trim=0 0 0 0,clip}
\begin{tikzpicture}[tdplot_main_coords, scale=5]
\definecolor{mydarkgreen}{RGB}{30,100,10}


\pgfdeclarelayer{background}
\pgfdeclarelayer{foreground}
\pgfsetlayers{background,main,foreground}

\begin{pgfonlayer}{background}
  \draw[->, >=stealth, semithick] (0,0,0) -- (1.4, 0, 0);
  \draw[->, >=stealth, semithick] (0,0,0) -- (0, 1.2, 0);
  \draw[->, >=stealth, semithick] (0,0,0) -- (0, 0, 3.0);
\end{pgfonlayer}

\begin{pgfonlayer}{foreground}
  \draw plot [mark=*, mark size=7.5, mark options={fill=blue, opacity=0.2}] coordinates {(1.4,0,0)};
  \draw plot [mark=*, mark size=7.5, mark options={fill=blue, opacity=0.2}] coordinates {(0,1.2,0)};
  \draw plot [mark=*, mark size=7.5, mark options={fill=blue, opacity=0.2}] coordinates {(0,0,3.0)};
\end{pgfonlayer}

  \coordinate (A) at (1.4, 0, 0);
  \coordinate (B) at (0, 1.2, 0);
  \coordinate (C) at (0, 0, 3.0);


  \node[anchor=west, scale = 1.0, color = blue] at (2.2, 0, 0) {$\spec(i)$};
  \node[anchor=west, scale = 1.0, color = blue] at (0, 1.6, 0) {$\spec(j)$};
  \node[anchor=west, scale = 1.0, color = blue] at (0, 0, 2) {$\spec(l)$};


  \pgfmathsetseed{2025}

\foreach \i in {0,...,14} {
  \pgfmathsetmacro{\rx}{0.3 + 1.0*(rnd-0.5)}
  \pgfmathsetmacro{\ry}{0.4 + 1.0*(rnd-0.5)}
  \pgfmathsetmacro{\rz}{1.0 + 2.0*(rnd-0.5)}
  \coordinate (g\i) at (\rx,\ry,\rz);
  \draw[tdplot_main_coords, fill=mydarkgreen] (g\i) circle[radius=0.01];
}

\foreach \i in {15,...,19} {
  \pgfmathsetmacro{\rx}{0.3 + 1.0*(rnd-0.5)}
  \pgfmathsetmacro{\ry}{0.4 + 1.0*(rnd-0.5)}
  \pgfmathsetmacro{\rz}{1.0 + 2.0*(rnd-0.5)}
  \draw[tdplot_main_coords, fill=mydarkgreen] (\rx,\ry,\rz) circle[radius=0.01];
}

\foreach \i/\name in {1/s0, 2/s1, 3/s2} {
\pgfmathsetmacro{\sx}{0.55 + 0.5*(rnd-0.5)}
\pgfmathsetmacro{\sy}{0.45 + 0.5*(rnd-0.5)}
\coordinate (s0) at (\sx, \sy, 0);
\node[tdplot_main_coords, star, star points=5, fill=mydarkgreen, minimum size=8pt, inner sep=0pt] at (s0) {};

\pgfmathsetmacro{\sx}{0.55 + 0.5*(rnd-0.5)}
\pgfmathsetmacro{\sy}{0.45 + 0.5*(rnd-0.5)}
\coordinate (s1) at (\sx, \sy, 0);
\node[tdplot_main_coords, star, star points=5, fill=mydarkgreen, minimum size=8pt, inner sep=0pt] at (s1) {};

\pgfmathsetmacro{\sx}{0.55 + 0.5*(rnd-0.5)}
\pgfmathsetmacro{\sy}{0.45 + 0.5*(rnd-0.5)}
\coordinate (s2) at (\sx, \sy, 0);
\node[tdplot_main_coords, star, star points=5, fill=mydarkgreen, minimum size=8pt, inner sep=0pt] at (s2) {};

}

  \coordinate (v0) at (1.37, 0.09, 0.07);
  \coordinate (v1) at (1.57, 0.25, 0.13);
  \coordinate (v2) at (1.37, 0.24, 0.03);
  \coordinate (v3) at (1.47, 0.30, 0.16);
  \coordinate (v4) at (1.2, 0.14, 0.10);
  \coordinate (v5) at (1.22, 0.09, 0.12);

\draw[black, thin] (v0) -- (g0);
\draw[black, thin] (v4) -- (s0);
\draw[black, thin] (v5) -- (s1);
\draw[black, thin] (v0) -- (g7);

\foreach \v in {v0,v1,v2,v3,v4,v5}
  \node[tdplot_main_coords, star, star points=5, fill=black, minimum size=8pt, inner sep=0pt] at (\v) {};





  \draw[black, thin] (v1) -- (v2);
  \draw[black, thin] (v1) -- (v4);

  \draw[black, thin] (v2) -- (v4);
  \draw[black, thin] (v1) -- (v5);
  \draw[black, thin] (v2) -- (v5);

  
  \coordinate (f0) at (1.37, -0.09, 0);
  \coordinate (f1) at (1.62, -0.12, 0);
  \coordinate (f2) at (1.47, -0.17, 0);
  \coordinate (f3) at (1.56, -0.01, 0);
  \coordinate (f4) at (1.2, -0.14, 0);
  \coordinate (f5) at (1.25, -0.2, 0);
  
  \foreach \v in {f0,f1,f2,f3,f4, f5}
    \draw[tdplot_main_coords, fill=black] (\v) circle[radius=0.01];
    
  \draw[black, thin] (v3) -- (f4);
  \draw[black, thin] (v1) -- (f1);
  \draw[black, thin] (v3) -- (f3);
  \draw[black, thin] (v4) -- (f0);

  \draw[black, thin] (f0) -- (g1);

  \node[anchor=west, align=left, color = black] at (2.2, -0.8, 0) {vertices in  \\ $\spec(i) \cap \lab(j)$};
  

\draw[mydarkgreen, dashed, thick] (0.2, 0.1, 0) ellipse[x radius=0.8, y radius=0.8];

\node[anchor=west, scale = 1.0, color = mydarkgreen] at (1, 0, 3) {cross vertices $\M$};

\node[anchor=west, scale = 1.2, color = black] at (1.5, 0, 0) {$\mu_i$};
\node[anchor=west, scale = 1.2, color = black] at (0, 1.2, 0) {$\mu_j$};
\node[anchor=west, scale = 1.2, color = black] at (0, 0, 3.1) {$\mu_l$};

\end{tikzpicture}
\end{adjustbox}}
	  \caption{Illustration of the the cross graph $G_{i, j}$ (see \Cref{def:crossgraph}). The axes represent three cluster means $\mu_i, \mu_j$ and $\mu_l$ for distinct $i,j,l \in [k]$.  The vertex set of $G_{i, j}$ can be partitioned into $\spec(i)\cap \lab(j)$ (vertices inside the shaded circle labeled \( \spec(i) \)) and $\M$ (vertices inside the dashed circle labeled ``cross vertices $\M$"). The vertices marked by stars illustrate vertices that truly belong to $C_j$. 
      The black lines illustrate the edges with at least one endpoint in $\spec(i)\cap \lab(j)$, while edges with both endpoints in $\M$ are omitted. 
      Note that the vertices in \( \spec(i) \cap \lab(j) \) that truly belong to \( C_j \) are typically connected to $\M$, while most of the vertices that do not belong to $C_j$ are not. }
\label{fig:cross}
\end{figure}
\noindent
With this notation, we can state our main technical insight: all but a small fraction of $(j \to i)$-impostors can leave the set $\spec(i) \cap \lab(j)$ in the cross graph; in particular, they can reach the cross vertices $\M$. Note that this test is amenable to an algorithm, since it can compute $\spec(i)$, $\lab(j)$, $\M$, and therefore test reachability.

\begin{lemma}[Impostors can reach the cross vertices-- see \Cref{claim:reachability_p}, \ref{claim:proxy}]
    \label{informal:conn_to_X}
    For any $i\neq j \in [k]$, consider the $(j\to i)$-impostors $u \in \im(i, j)$ with label $\sigma(u)=j$, i.e. vertices from $C_j$ with a correct label $\sigma$ that are wrongly placed in the $i$-th spectral cluster. In expectation over $\sigma$, all but $\approx \epsilon \delta n$ of these vertices are connected to the cross vertices $\M$ in the cross graph $G_{i,j}$ induced by  $(\spec(i) \cap \lab(j)) \cup \M$.
\end{lemma}

\noindent
To prove \Cref{informal:conn_to_X}, we use the following fact: vertices in a spectral cluster that are not directed neighbors of $\M$ have few neighbors in other spectral clusters. In other words, spectral clusters are sparsely connected to each other.

\begin{lemma}[Few neighbors across spectral clusters-- see \Cref{lemma:cores_sparsely_connected}]
    \label{informal:neighbors}
    Let $i \in [k]$ and let $u$ be a vertex in the $i$-th spectral cluster that is not a direct neighbor of $\M$, i.e. $u \in \spec(i) \setminus \nei_G(\M)$. Then, the number of neighbors of $u$ in spectral clusters other than $\spec(i)$ is at most $\phi/2 \cdot d$, i.e.
    \begin{equation*}
        \left|\nei_G(u) \cap \left( \bigcup_{j \neq i} \spec(j)\right)\right| \le \frac{\phi}{2} \cdot d \, .
    \end{equation*}
\end{lemma}
\begin{proof}[Proof sketch]
Since the eigenvalues associated with the first $k$ eigenvectors of $G$ are bounded by $\approx \epsilon$, the embedding $f_u$ is very close to the average of the embeddings of the neighbors of $u$, i.e.
\begin{equation*}
    \left\| d\cdot f_u -  \sum_{v \in \nei_G(u)}f_v \right\|_2 \lesssim d\epsilon \cdot \|f_u\|_2 \approx d \epsilon \cdot \|\mu_i\|_2 \, ,
\end{equation*}
where the last step follows from $u \in \spec(i)$. Now, note that the average \smash{$\frac{1}{d} \sum_{v \in \nei_G(u)}f_v$} is contributed to by $v$'s in $\spec(i)$ and $v$'s in other spectral clusters, since $u \notin \nei_G(\M)$. But every $v$ in $\spec(i)$ must have an embedding very close to $f_u$, which gives
\begin{equation*}
   \left\|  \sum_{\substack{v \in \nei_G(u):\\v\in \cup_{j\neq i} \spec(j)}}(f_u -f_v) \right\|_2  \approx  \left\|  \sum_{v \in \nei_G(u)}(f_u -f_v) \right\|_2   = \left\| d\cdot f_u -  \sum_{v \in \nei_G(u)}f_v \right\|_2 \lesssim d \epsilon \cdot \|\mu_i\|_2\, .
\end{equation*}
Recall again that $f_u \approx \mu_i$. Moreover, for any $v$ that is considered in the sum on the left-hand side, we have $f_v \approx \mu_j$ for some $j \neq i$ (by the definition of spectral clusters, see \Cref{def:spec}). Then, using the fact that all the cluster means 
have $\ell_2$-norm roughly $\approx \sqrt{k/n} $ and are almost orthogonal (by \Cref{informal:varbound}), we have
\begin{equation*}
      \Delta \cdot \sqrt{\frac{k}{n}}  \lesssim  \left\|  \sum_{\substack{v \in \nei_G(u):\\v\in \cup_{j\neq i} \spec(j)}}(f_u -f_v) \right\|_2  \lesssim d \epsilon \cdot \sqrt{\frac{k}{n}}\, ,
\end{equation*}
where $\Delta$ is the number of terms in the sum on the left-hand side, i.e. $\Delta = |\nei_G(u) \cap \cup_{j \neq i} \spec(j))|$. Rearranging, we get the claim. The reason why we obtain a bound of $\approx \epsilon d$ instead of $\approx \phi d$ is that we did not keep careful track of the errors approximating the embeddings by their cluster mean, as well as the errors incurred by assuming the cluster means to be all of equal length.
\end{proof}

\noindent
Equipped with this observation, we can show \Cref{informal:conn_to_X}.

\begin{proof}[Proof sketch of \Cref{informal:conn_to_X}]
The first crucial observation is that we can combine \Cref{informal:neighbors} with the assumption that $C_j$ induces an expander to argue that the $(j\to i)$-impostors $\im(i,j)$ have nice expansion properties.
    \begin{claim}[Impostors expand-- see \Cref{lem:imposters_expand}]
\label{informal:impostors_expand}
    Let $i \neq j \in [k]$, and let $S$ be a subset of $(j \to i)$-impostors that does not have direct neighbors in $\M$, i.e. $S \subseteq \im(i,j) \setminus \nei_G(\M)$. Then, one has $|E(S,\im(i,j) \setminus S)| \ge \frac{\phi}{2} d|S|$.
\end{claim}
\begin{proof}[Proof sketch]
    By \Cref{informal:neighbors}, every $u \in S$ has $(1-\phi/2)d$ of its neighbors in $\im(i,j)$. Now, recall that the $(j\to i)$-impostors $\im(i,j)$ are vertices from $C_j$, so the cut $S$ must $\phi$-expand in $G\{C_j\}$. Combining these two observations, we get the claim.
\end{proof}

\noindent
Now, let $Q$ be the set of $(j \to i)$-impostors $u \in \im(i,j)$ with label $\sigma(u)=j$ that cannot reach $\M$ in the cross graph $G_{i,j}$. Also, let $R$ be the set of $(j \to i)$-impostors $v \in \im(i,j)$ with a wrong label $\sigma(v)\neq j$. One can note that $Q$ does not overlap $\nei_G(\M)$, because if it did it would be able to reach $\M$ in $G_{i,j}$. Therefore, \Cref{informal:impostors_expand} applies, and it gives
$$|E(Q,\im(i,j) \setminus Q)| \ge \frac{\phi}{2} d|Q| \, .$$
Next, we upper-bound
$|E(Q,\im(i,j) \setminus Q)|$ in terms of $R$ to prove the lemma. By definition of $Q$, it cannot have edges to $(j \to i)$-impostors with label $\sigma(u)=j$ that can reach $\M$. Hence, $|E(Q,\im(i,j) \setminus Q)|$ is only contributed to by edges from $Q$ to $R$, which we defined to be the set of $(j \to i)$-impostors with a wrong label $\sigma(v)\neq j$. Thus, we have
$$ d|R| \ge |E(Q,R)| \ge |E(Q,\im(i,j) \setminus Q)| \ge \frac{\phi}{2} d|Q| \, ,$$
where the leftmost inequality follows by $d$-regularity of $G$. Rearranging gives $|Q| \lesssim |R|$. Since $R$ consists of mislabeled impostors, we can combine \Cref{informal:fewbad} with linearity of expectation to bound its size by $\delta |\im(i,j)| \lesssim \epsilon\delta n$, which proves the result.
\end{proof}

\noindent
This insight provided by \Cref{informal:conn_to_X} suggests a modification of the naive spectral approach, where we identify both cross vertices and most impostor vertices: the first kind is detected by checking $\tau(u)=*$, while the second kind can be detected by checking connectivity in the cross graphs. There is a small technical obstacle in this approach, which actually also affects naive spectral approach: we cannot exactly compute the spectral clusters or the cross vertices, since these definitions as we have stated them (see \Cref{informal:spec}) depend on the exact spectral embedding and the exact cluster means, which in turn depend on the clusters $C_1, \dots,C_k$. Therefore, our algorithms will rely on an oracle providing approximate access to these quantities.
\begin{definition}[$\ell_2$-oracle, see \Cref{def:apx}, \ref{def:approxmeans}]
    \label{informal:oracle}
    We denote by $\textsc{apx-norm}$ an oracle that provides approximate access to the quantities needed to define spectral clusters $\spec(i)$, cross vertices $\M$, and the spectral labeling $\tau$. In particular, this oracle provides approximate access to the $\ell_2$ distances and $\ell_2$ norms of spectral embeddings $f_u$'s and cluster means $\mu_i$'s.
\end{definition}

\noindent
\Cref{alg:to} outlines our approach: it uses the provided oracle $\textsc{apx-norm}$ to compute $\tau(u)$, which determines whether $u$ is in a spectral cluster or if it is a cross vertex (see \Cref{informal:spec}). Then, if the spectral information provided by $\tau(u)$ matches the information provided by $\sigma(u)$, i.e. there is $i \in [k]$ such that $\tau(u)=\sigma(u)=i$, then we output $i$ as the cluster id of $u$ (line~\eqref{enum:honest_return}). If instead $\tau(u)=*$, i.e.  $u$ is a cross vertex, the spectral information is not useful and we output the cluster id given by $\sigma(u)$ (line~\eqref{enum:cross_return}). Finally, if the spectral information and the label disagree, i.e. $\tau(u)=i$ and $\sigma(u)=j$ for some $i\neq j \in [k]$, we want to test if $u$ is an impostor or not as suggested by \Cref{informal:conn_to_X}. Line~\eqref{enum:conntest} uses the following notation.

\begin{definition}[Reachability test]
    \label{informal:reach}
    For a vertex $u \in V$ and a subgraph $H$, we denote by $\textsc{reach}(u,H,\M)$ a procedure that checks if the vertex $u$ can reach the set of cross vertices $\M$ in $H$. In our context, $H$ will be the $(i,j)$-cross graph $G_{i,j}$ induced by the edges on $\spec(i) \cap \lab(j)$. The procedure $\textsc{reach}(u,H,\M)$ can be implemented as a simple BFS (as in our polynomial-time algorithm \Cref{alg:polytime}) or as random-walk-based primitive (as in our sublinear-time algorithm \Cref{alg:sketch+labels}).
\end{definition}

\begin{algorithm}
\caption{Meta-algorithm generalizing \Cref{alg:polytime} and \Cref{alg:sketch+labels}}\label{alg:to}
\begin{algorithmic}[1]
\State \textbf{Input:} $G$, $\sigma$, an oracle $\textsc{apx-norm}$, a vertex $u \in V$
\State \textbf{Output:} a cluster id in $[k]$
\If{$\tau(u) = \sigma(u)$} \Return $\sigma(u)$ \Comment{label agrees with spectral information} \label{enum:honest_return} 
\ElsIf{$\tau(u) = *$} \Return $\sigma(u)$ \Comment{spectral information is ambiguous}  \label{enum:cross_return}
\Else
    \Comment{label disagrees with spectral information}
    \State $H \gets $ cross graph induced by $\left(\spec(\tau(u)) \cap \lab(\sigma(u))\right) \cup \M$ \Comment{see \Cref{informal:crossgraph}}

    \If{$\textsc{reach}(u,H,\M)$ returns ``yes''} \label{enum:conntest} \Comment{$u$ can reach $\M$ in the cross graph $H$, see \Cref{informal:reach}}
        \State \Return $\sigma(u)$ \Comment{trust the label, $u$ is likely an impostor} \label{enum:impostor_return}
    \Else
        \State \Return $\tau(u)$ \Comment{trust the spectral information, the label is likely wrong} \label{enum:no_impostor_return}
    \EndIf
\EndIf

\end{algorithmic}
\end{algorithm}

\noindent
In \Cref{sec:setting}, we explicitly define the central spectral objects (spectral cluster, cross vertices, etc.) in terms of these approximate norms, and in \Cref{subsec:apx_spec_oracles} and \Cref{subsec:apxmeans} we discuss how to implement such an oracle. In \Cref{sec:robcon}, we detail how to implement the reachability test (see line~\eqref{enum:conntest}) in sublinear time. In the remainder of this section, for sake of simplicity, we assume that we have access to exact spectral embeddings and cluster means, and we assume that the reachability test is always correct.
\\~\\
We now show how to bound the number of vertices that can possibly be misclassified in lines~\eqref{enum:honest_return},~\eqref{enum:cross_return},~\eqref{enum:impostor_return},~\eqref{enum:no_impostor_return}.

\begin{lemma}[Informal version of \Cref{lemma:spec_or_cross}]
\label{informal:2cases}
The number of vertices that can be misclassified in lines~\eqref{enum:honest_return} and~\eqref{enum:cross_return} is at most $\approx \delta \epsilon n$ in expectation over $\sigma$.
\end{lemma}
\begin{proof}[Proof sketch]
We rely on a simple observation: every vertex $u$ that is misclassified in  line~\eqref{enum:honest_return} must be an impostor whose label $\sigma(u)$ is wrong (because $\tau(u)=\sigma(u)\neq \iota(u)$), and every vertex $u$ that is misclassified in line~\eqref{enum:cross_return} must be a cross vertex whose label $\sigma(u)$ is wrong (because $\tau(u)=*$ and $\sigma(u) \neq \iota(u)$). By \Cref{informal:fewbad} and linearity of expectation, we get the claimed bound.
\end{proof}

\noindent
Let us now consider the vertices that can be misclassified in lines~\eqref{enum:impostor_return},~\eqref{enum:no_impostor_return}. The ones misclassified in lines~\eqref{enum:impostor_return} can be thought of as the \textit{false positives} for the reachability test, while the ones misclassified in lines~\eqref{enum:no_impostor_return} can be thought of as \textit{false negatives}. \Cref{informal:conn_to_X} already suggests that there should be few false negatives, as shown below.

\begin{lemma}[Few false negatives-- see \Cref{lem:false_neg_p}]
\label{informal:falseneg}
    The number of vertices that can be misclassified in line~\eqref{enum:no_impostor_return} is at most $\approx \delta \epsilon n$ in expectation over $\sigma$.
\end{lemma}
\begin{proof}[Proof sketch]
Let $u$ be any vertex that is misclassified in line~\eqref{enum:no_impostor_return}, and let $i,j \in [k]$ such that $\tau(u)=i$ (i.e. $u$ is in the $i$-th spectral cluster) and $\iota(u)=j$ (i.e. $u \in C_j$). One can observe that, because $u$ has been misclassified, we must have $i \neq j$, which means that $u$ is a $(j \to i)$-impostor vertex, i.e. $u \in \im(i,j)$. Looking at $\sigma$, there are two types of vertices that can be misclassified in line~\eqref{enum:no_impostor_return}: those with the correct label $\sigma(u)=j$ and those with a wrong label $\sigma(u) \neq j$. The latter vertices are then impostors with the wrong label, so following the same proof as for \Cref{informal:2cases} we bound the number of such vertices by $\approx \delta \epsilon n$. For the vertices with the correct label $\sigma(u)=\iota(u)=j$, they must have failed the reachability test in line~\eqref{enum:conntest} in order to be misclassified in line~\eqref{enum:no_impostor_return}. The number of vertices in $\im(i,j)$ with the correct label $\sigma(u)=j$ that cannot reach $\M$ in the cross graph $G_{i,j}$ is at most $\approx \epsilon\delta n$ by \Cref{informal:conn_to_X}. Summing up over all $i \neq j \in [k]$ gives the claimed bound (recall that $k=O(1)$).
\end{proof}

\noindent
We are left with the task of showing that there are few false positives. Intuitively, they are contained in connected components induced by vertices with a wrong label, and these components can be shown to be small by a Galton-Watson type of argument.

\begin{lemma}[Few false positives-- see \Cref{lem:false_pos_p}]
\label{informal:falsepos}
    The number of vertices that can be misclassified in line~\eqref{enum:impostor_return} is at most $\approx d \delta \epsilon n$ in expectation over $\sigma$.
\end{lemma}
\begin{proof}[Proof sketch]
 Let $u$ be a vertex that is misclassified in line~\eqref{enum:impostor_return}, and let $i\neq j \in [k]$ such that $u$ belongs to the $i$-th spectral cluster and to the $j$-th label cluster, i.e. $u \in \spec(i) \cap \lab(j)$. For $u$ to be misclassified in line~\eqref{enum:impostor_return}, it must be the case that the label $\sigma(u)$ is wrong, i.e. $\iota(u) \neq j$. Moreover, it must be the case that $u$ can reach $\M$ in the $(i,j)$-cross graph $G_{i,j}$. Recall that $G_{i,j}$ is the subgraph induced by the edges on $(\spec(i) \cap \lab(j))\cup \M$. Note that these vertices can be of three types: the cross vertices $\M$, the impostors $\im(i,j)$ with the correct label, and the vertices in $\spec(i)$ with a wrong label. Note that our problematic vertex $u$ must belong to the latter set. Therefore, the number of such problematic vertices $u$ corresponds to the number of vertices that are reachable from a mislabeled vertex in $\nei_G(\im(i,j))$ or in $\nei_G(\M)$. This can be thought of as starting a Galton-Watson process from every $\nei_G(\im(i,j)) \cup \nei_G(\M)$, where the offspring distribution is $d$ with probability $\delta$ and $0$ with probability $1-\delta$. Hence, the number of such vertices can be bounded by $\approx \delta |\nei_G(\im(i,j)) \cup \nei_G(\M)|$. By $d$-regularity of $G$ and \Cref{informal:fewbad}, we get the claim by summing over all pairs $i\neq j\in [k]$ (again, recall that $k=O(1)$). 
\end{proof}

\noindent
Combining \Cref{informal:2cases}, \Cref{informal:falseneg}, and \Cref{informal:falsepos}, we conclude the following result.

\begin{lemma}[$d\epsilon\delta$-rate in polynomial time-- see \Cref{thm:polytime}]
\label{informal:polyalgo}
    A naive implementation of \Cref{alg:to} gives a polynomial-time algorithm with an expected misclassification rate of
\begin{equation*}
    \epsilon\delta n + \epsilon\delta n + d \cdot \epsilon\delta n \approx d \cdot \epsilon \delta n \, . 
\end{equation*}
\end{lemma}
\noindent
This essentially corresponds to \Cref{alg:polytime}. The analysis of this algorithm, in \Cref{sec:poly}, is essentially a formal version of the arguments described above. This result is summarized in \Cref{thm:polytime}.

Finally, we remark that while the proof sketches of \Cref{informal:falseneg} and \Cref{informal:falsepos} assumed $k = O(1)$, our more careful analysis in \Cref{sec:poly} gives the misclassification rate $O(\epsilon \delta d)$ also for general $k \lesssim 1/\epsilon$.

\paragraph{Removing the factor of $d$.} In \Cref{sec:sub_alg}, we consider a more refined implementation of \Cref{alg:to}, namely, \Cref{alg:sketch+labels}. There, we carry out an inherently more intricate analysis than the one of \Cref{thm:polytime}, which allows to get rid of the factor $d$ in the misclassification rate, at the cost of an extra \smash{$\log(1/\delta)$} factor, thereby nearly matching the best possible rate of $\epsilon\delta$. Note that the factor-$d$ loss in \Cref{informal:polyalgo} comes from the false positives (recall \Cref{informal:falsepos}), i.e. vertices with a wrong label that are able to reach $\M$ in their cross graph (see line~\eqref{enum:conntest} in \Cref{alg:to}). Hence, we should make the reachability test harder to pass for vertices with a wrong label, while not making it much harder to pass for the false negatives, i.e. impostors with the correct label. Recall that the way we proved that only few such vertices fail to pass the reachability test in \Cref{informal:conn_to_X} is by exploiting their expansion properties guaranteed by \Cref{informal:impostors_expand}. Therefore, imagine implementing the reachability test $\textsc{reach}(u,H,\M)$ in line~\eqref{enum:conntest} so as to check whether $u$ has many paths to $\M$: the number of false negatives, i.e. the number of impostors with the correct label, that fail to pass the test remains roughly $\epsilon\delta n$, by virtue of the expansion property we already proved in \Cref{informal:impostors_expand}; the number of false positives should intuitively decrease, since a vertex with a wrong label that are able to reach $\M$ through multiple paths for it to be misclassified.

\paragraph{Sublinear-time implementation.}  \Cref{alg:to} presented above is actually amenable to being made sublinear time: given the $\ell_2$-oracle $\textsc{apx-norm}$, we can check whether a vertex has multiple paths to a target set by running a few random walks of bounded length. Moreover, it is known how to implement $\ell_2$-oracle $\textsc{apx-norm}$ in sublinear time \cite{GKLMS21}. We present the sublinear version of \Cref{alg:to}
(\Cref{alg:sketch+labels}) and prove its guarantees in \Cref{sec:sub_alg}. 

\noindent

\subsection{Refining communities}
\label{sec:techoverview_reweight}
Given a graph $G$ that admits a $(k,\epsilon,\Omega(1))$-clustering $C_1,\dots,C_k$, and given some additional information about the cluster ids of its vertices, we would like to reweight the edges of $G$ to obtain a graph $G'$ with better clusterability in the following sense: for some $\gamma \ll \epsilon $, we demand that $G'$ admits a $(k,O(\gamma),\Omega(1))$-clustering $C'_1,\dots,C'_k$ that is $\gamma$-close to $C_1,\dots,C_k$. Of course, this task would be easy to achieve if we were given the clustering $C_1,\dots,C_k$: we can just down-weight every crossing edge in
$$E_\cross = \bigcup_{i\neq j \in [k]} E(C_i,C_j)$$
to have a weight of $0$. The reweighted graph $G'$ will then have a perfect community structure, as it would consists of $k$ connected components, each of which induces an expander.

In reality, we do not have access to $C_1,\dots,C_k$. However, we can have approximate access to $C_1,\dots,C_k$: with the classifier from the previous section (e.g. \Cref{informal:polyalgo} or \Cref{thm:informal_classsifier}), we can compute a labeling $\alpha:V\rightarrow [k]$ that assigns the correct cluster id to all but an $\epsilon \delta$ fraction of vertices. In particular, $\alpha$ gives approximate access to the set of crossing edges $\cup_{i\neq j} E(C_i,C_j)$.

\begin{definition}[Flagged edges -- see \Cref{def:flagged}]
\label{informal:flag}
We define the set of edges flagged by $\alpha$, denoted $F$, as the set of edges $(u,v) \in E$ whose endpoints are assigned to different clusters by $\alpha$, i.e. $F=\{(u,v) \in E: \, \alpha(u)\neq \alpha(v)\}$.
\end{definition}

\begin{lemma}[Flagged edges approximate crossing edges -- \Cref{claim:FcapE_cross_k}]
\label{informal:flagcross}
If $\alpha$ has a misclassification rate of $\gamma$, then the set of flagged and crossing edges differ in at most $\gamma dn$ edges, i.e. $|F \triangle E_\cross| \le \gamma dn$.
\end{lemma}
\begin{proof}[Proof sketch]
Every vertex misclassified by $\alpha$ contributes at most $d$ edges to $|F \triangle E_\cross|$.
\end{proof}

\noindent
In light of this, one might hope that down-weighting every edge in $F$ to $0$ results in a graph $G'$ that consists of $k$ disjoint expanders. However, this hope is readily dashed, as we now illustrate. Consider a vertex $u \in C_1$ with, say, $3$ neighbors $a,b,c$ in $C_1$; then, if $\alpha$ incorrectly flags $(u,a),(u,b),(u,c)$ (which is admissible by \Cref{informal:flagcross}), the vertex $u$ will be isolated in the resulting graph $G'$. The presence of an isolated vertex implies that $G'$ does not admit a $(k,\xi,\Omega(1))$-clustering for any $\xi > 0$.

This observation prompts us with the idea of down-weighting the flagged edges so as to minimize the weight of the flagged edges $F$ in $G'$ subject to the subgraphs $G'\{C_i\}$ remaining $\Omega(1)$-expanders. This idea cannot be implemented directly, as we do not have exact access to $G'\{C_i\}$. Instead, we resort to the notion of \textit{multi-way conductance}, defined below.

\begin{definition}[Multi-way conductance \cite{LGT12} -- see \Cref{def:rho}]
    The $(k+1)$-th way conductance of $G'$ is the value $\rho_{k+1}(G')$ such that no matter how we pick $k+1$ disjoint cuts $Q_1,\dots,Q_{k+1}$ in $G'$, one of them will be $\rho_{k+1}(G')$-expanding, i.e.
    \begin{equation*}
        \rho_{k+1}(G') = \min_{\substack{Q_1,\dots,Q_{k+1} \subseteq V :\\ \forall \, i\neq j, \, Q_i \cap Q_j = \emptyset}} \, \, \,  \max_{i \in [k+1]}\Phi_{G'}(Q_i) \, .
    \end{equation*}
\end{definition}

\noindent
We will use this quantity as a proxy for the conductance of the $G'\{C_i\}$'s. In particular, we consider the following optimization problem that solves for a weight function $x \in [0,1]^E$ so as to minimize the weight on the flagged edges subject to $(k+1)$-th way conductance of the reweighted graph $G_x=(V,E,x)$ remaining large enough.

\begin{equation}
\label{prog:exptime}
    \begin{aligned}
    \text{minimize} & \quad  \sum_{e \in F}x_e \\
    \text{subject to} & \quad x_e =1 \quad \forall \, e \in E\setminus F\\
    & \quad \rho_{k+1}(G_x) \ge \Omega(1) \\
    & \quad x \in [0,1]^E
\end{aligned}
\end{equation}

\noindent
Let us ignore for a moment the fact that we cannot solve such an optimization problem as is. We can show that a solution $x$ yields a graph $G_x$ that admits a clustering with quality proportional to the objective value of $x$.

\begin{lemma}[Clustering $G_x$ -- \Cref{lemma:good_clustering}]
\label{informal:reweight_clust}
    If $x$ is a feasible solution to the program in~\eqref{prog:exptime} with objective value $\text{\upshape OBJ}= \nu \cdot dn$, then the resulting graph $G=(V,E,x)$ admits a $(k,O(\nu),\Omega(1))$-clustering $C'_1,\dots,C'_k$ such that $\sum_i |C_i \triangle C'_i| \lesssim \nu \cdot n$.
\end{lemma}

\noindent
One can also observe that there is a solution to  the program in~\eqref{prog:exptime} with objective value $\approx \gamma$.

\begin{lemma}[Low-objective solution -- \Cref{lemma:OPT_SDP_k}]
\label{informal:lowobj}
    There is a feasible solution $x^*$ to the program in~\eqref{prog:exptime} with objective value $\approx \gamma dn$. Specifically, this solution gives a weight of $0$ to every flagged edge that is also crossing, and leaves  every other edge untouched, i.e. $x^* = \1 - \1_{F \cap E_\cross}$.
\end{lemma}
\begin{proof}[Proof sketch]
    Since we use $x^* = \1 - \1_{F \cap E_\cross}$, the objective value is given by the flagged edges that are not crossing, i.e. $|F \setminus E_\cross|$. By \Cref{informal:flagcross}, we now that this is at most $\gamma dn$. To check feasibility of $x^*$, we observe that we are not decreasing the weight of any edge inside the clusters $G\{C_i\}$'s, which therefore leaves their conductance unchanged. Hence, however we pick $k+1$ disjoint sets $Q_1,\dots,Q_{k+1}$, one of them must cross one of  $G_x\{C_1\}, \dots, G_x\{C_k\}$. Therefore, one of $Q_1,\dots,Q_{k+1}$ has conductance at least $\Omega(1)$ in $G_x$, showing that $\rho_{k+1}(G_x) \ge \Omega(1)$.
\end{proof}

\noindent
\Cref{informal:lowobj} together with \Cref{informal:reweight_clust} imply that solving the program in~\eqref{prog:exptime} and using the weights it outputs, we obtain a reweighted version of $G$ that admits a $(k,O(\gamma),\Omega(1))$-clustering $C'_1,\dots,C'_k$ such that $\sum_i |C_i \triangle C'_i| \lesssim \gamma \cdot n$. To turn the program in~\eqref{prog:exptime} into a polynomial-time algorithm, we use higher-order Cheeger inequalities~\cite{LGT12} to rewrite the constraint $\rho_{k+1}(G_x) \ge \Omega(1)$ as an eigenvalue constraint. We exploit the clusterability of the input graph $G$ to further rewrite the eigenvalue constraint as a semidefinite constraint, which gives \Cref{thm:informal_reweight}.

\section{Preliminaries and notation}
\label{sec:prelims}

\paragraph{Graph notation.} Given a graph $G= (V,E)$ and a set $S \subseteq V$, we let $E(S)$ denote the set of edges with both endpoints in $S$, and we let $G[S]$ denote the induced subgraph on $S$, that is $ G[S] \coloneqq (S, E(S)). $
We use the notation $G\{S\}$ to denote the graph $G[S]$ to which self-loops have been added so that every vertex in $G\{S\}$ has the same degree as in $G$. The volume of $S$, denoted $\vol(S)$, is defined as the sum of degrees of vertices in $S$. If $G$ has self-loops, we write it as a triple $G=(V,E,\ell)$, where $(V,E)$ is a simple graph (so the self-loops are not represented as edges) and $\ell$ maps each $u \in V$ to the number of self-loops on $u$. The degree of a vertex in $G=(V,E,\ell)$ is the number of its incident edges plus the number of its self-loops, and the volume of $S \subseteq V$ uses this notion of degree. The adjacency matrix of  $G=(V,E,\ell)$ is the adjacency matrix of $(V,E)$ plus $\diag(\ell)$. The set of neighbors of a vertex $u \in V$ is the set $\nei_G(u)=\{v \in V: (u,v) \in E\}$, and the set of neighbors of $S \subseteq V$ is the set $\nei_G(S) = \cup_{u \in S} \nei_G(u)$.

\paragraph{Linear algebra notation.} For a vertex $v \in V$, we write $\1_v \in \R^V$ for the indicator vector of $v$, that is the vector that is $1$ at index $v$ and $0$ elsewhere. Similarly, for a set $S \subseteq V$, we write $\1_S$ for the indicator of $S$. For a matrix $M$, we write $\|M\|_{\mathrm{op}}$ for its operator (spectral) norm, $\|M\|_{\mathrm{op}} = \max_{\|x\|_2 = 1} \|Mx\|_2.$

\paragraph{Graph access model.} In this paper, we work with the \emph{bounded-degree graph model}. This means that, for a graph $G=(V,E)$ with degrees bounded by $d$, we can specify a vertex $v \in V$ as well as a number $i: 1\leq i \leq d$, and access the $i$-th neighbor of the vertex $v$ in constant time.

\subsection{Facts from spectral graph theory}
Throughout this section, let us fix an integer $d \ge 3$ and $d$-regular graph $G=(V,E)$, and also fix an integer $k \ge 2$ and conductance parameters $\epsilon,\phi \in (0,1)$.
Let $A$ be the adjacency matrix of $G$, and let $\mathcal{L} = I -\frac{1}{d}A$ denote its normalized Laplacian. Also let $\calL = U \Sigma U^T$ denote its eigendecomposition.

\begin{definition}[Spectral embedding]
    \label{def:emb} Let $x_1,\dots,x_n$ be the eigenvectors of $\mathcal L$ sorted in non-decreasing order of eigenvalue. We define the spectral embedding of $V$, denoted $(f_u)_{u \in V}$, as \[f_u \coloneqq \sum_{i =1}^k \langle x_i, \1_u\rangle e_i = U_{[k]}^\top  \1_u\] for each $u \in V$, where $e_1,\dots, e_k$ is the standard basis of $\R^k$ and $U_{[k]} \in \mathbb{R}^{n\times k}$ denotes the matrix with columns $x_1, \ldots, x_k$.
\end{definition}

\begin{lemma}[Lemma 3 from \cite{GKLMS21}]
\label{lem:bnd-lambda}
Let $\lambda_k$ and $\lambda_{k+1}$ be the $k$-th and $(k+1)$-th smallest eigenvalue of $\calL$, respectively.
If $G$ is $(k,\epsilon,\phi)$-clusterable as per \Cref{def:clustering}, then $\lambda_k\leq 2\epsilon$ and $\lambda_{k+1}\geq {\phi^2}/{2}$. 
\end{lemma}

\begin{remark}\label{remark:spectral_unique}
If $G$ is $(k, \epsilon, \phi)$-clusterable with $\epsilon/\phi^2$ small enough, then
	it follows from Lemma~\ref{lem:bnd-lambda} that the space spanned by the bottom $k$ eigenvectors of $\mathcal{L}$ is uniquely defined, i.e. the choice of $U_{[k]}$ is unique up to multiplication by an orthonormal matrix \smash{$R\in \mathbb{R}^{k \times k}$} on the right. We note 
	that while the choice of $f_u$ for $u \in V$ is not unique,  the dot product 
	between the spectral embedding of $u\in V$ and $v\in V$ is well defined, since for every orthonormal 
	$R\in \mathbb{R}^{k\times k}$ one has  
	\[\langle Rf_u, Rf_v\rangle=(Rf_u)^\top(Rf_v)=\left(f_u\right)^\top (R^\top R) \left(f_v\right)=\left(f_u\right)^\top \left(f_v\right) \, .\]
\end{remark}

\noindent
Since we are ultimately interested in recovering the clustering by using (also) the spectral embedding, it is useful to know what the ``typical'' embedding is for a cluster $C_i$. Hence, we next define the cluster means.

\begin{definition}[Cluster means]
    \label{def:clusmeans}
    Let $C_1,\dots,C_k $ be a partitioning of $V$ and let $(f_u)_{u \in V}$ be its spectral embedding. We define the cluster means (or cluster centers), denoted $\mu_1,\dots,\mu_k$, as $$\mu_i=\frac{1}{|C_i|}\sum_{u \in C_i} f_u \quad \text{ for each $i \in [k]$.}$$
\end{definition}

\begin{lemma}[Spectral embedding is close to the average of neighbors] \label{lemma:fx_close_to_avg_neighbors} If $G$ is $(k, \epsilon, \phi)$-clusterable, then for every $v \in V$ one has 
\[\left\|f_v - \frac{1}{d}\sum_{w \in \nei_G(v)}f_{w}\right\|_2 \leq 2 \epsilon \|f_v\|_2 \, .\]
\end{lemma}
\begin{proof}

Let $0 \leq \lambda_1 , \ldots , \lambda_n$ denote the eigenvalues of $\mathcal{L}$, and let $x_1, \ldots, x_{n}$ be the corresponding eigenvectors. 
By \Cref{def:emb} we have 
\begin{equation}\label{eq:fx_close_to_avg_neighbors1}
    U_{[k]}^\top \mathcal{L} \1_v = \frac{1}{d} U_{[k]}^\top\sum_{w \in \nei_G(v)}(\1_v - \1_w) = f_v - \frac{1}{d}\sum_{w \in \nei_G(v)} f_w. 
\end{equation}
On the other hand, expanding $\1_v$ in the eigenbasis, we obtain 
\begin{equation}\label{eq:fx_close_to_avg_neighbors2}
    U_{[k]}^\top \mathcal{L} \1_v =    U_{[k]}^\top \mathcal{L}\left(\sum_{i = 1}^{n} \langle \1_v, x_i\rangle x_i \right) =   U_{[k]}^\top\left( \sum_{i = 1}^n \lambda_i \langle \1_v, x_i\rangle x_i \right) = \sum_{i = 1}^k \lambda_i \langle \1_v, x_i\rangle e_i,
\end{equation}
where $\{e_i\}_{i=1}^k$ form an orthonormal basis of $\R^k$. 
Combining~\eqref{eq:fx_close_to_avg_neighbors1} and~\eqref{eq:fx_close_to_avg_neighbors2}, we obtain 
\begin{equation*}
    \left \|f_v - \frac{1}{d}\sum_{w \in \nei_G(v)}f_{w}\right\|_2 = \left\| \sum_{i = 1}^k \lambda_i \langle \1_v, x_i\rangle e_i\right\|_2 \leq 2\e \left\|\sum_{i = 1}^k \langle \1_v, x_i\rangle e_i  \right\|_2 = 2\epsilon  \|f_v\|_2, 
\end{equation*}
where the first inequality uses the fact that $0 \leq \lambda_i \leq 2\epsilon$ for $i = 1, \ldots , k$ by \Cref{lem:bnd-lambda}. 
\end{proof}

 \begin{lemma}[Lemma 7 from \cite{GKLMS21}]\label{lemma:clustermeans} If $G$ admits a $(k,\epsilon,\phi)$-clustering $C_1,\dots,C_k$, then the following hold:

\begin{enumerate}[label=\textbf{\upshape (\arabic*)}]
    \item for all $i \in [k]$, one has
    \[
    \left| \|\mu_i\|_2^2 - \frac{1}{|C_i|} \right| \leq \frac{4 \sqrt{\epsilon}}{\phi} \frac{1}{|C_i|}  \, ;\label{bulletpt:mu_norm}
    \]
    \item for all $i \neq j \in [k]$, one has
    \[
    |\langle \mu_i, \mu_j \rangle| \leq \frac{8 \sqrt{\epsilon}}{\phi} \frac{1}{\sqrt{|C_i||C_j|}}\label{bulletpt:mu_i_dot_mu_j} \, .
    \]
\end{enumerate}
\end{lemma}
\begin{remark}\label{rem:mu_i_norm}
Assuming that ${|C_i|}/{|C_j|} \leq \eta$ for all $i,j \in [k]$, one has
\smash{$\frac{1}{\eta}\cdot \frac{n}{k}\le |C_i| \le \eta \cdot \frac{n}{k}$} for all $i \in [k]$. Combining this with {\bf \ref{bulletpt:mu_norm}} and further assuming $\epsilon \le \phi^2/64$, for all $i \in [k] $ one has $$ \frac{1}{2 \eta}\cdot \frac{k}{n} \leq \|\mu_i\|^2_2 \leq 2 \eta \cdot \frac{k}{n} \, .$$
\end{remark}

\begin{lemma}[Lemma 6 from \cite{GKLMS21}; ``variance bounds"]\label{lemma:variancebound}
If $G$ admits a $(k,\epsilon,\phi)$-clustering $C_1,\dots,C_k$, then for all $\alpha \in \mathbb{R}^k$ with $\|\alpha\|_2=1$ we have 
\[ \sum_{i =1}^k \sum_{v \in C_i}\left \langle f_v - \mu_i, \alpha\right \rangle^2 \leq \frac{4\epsilon}{\phi^2}\, .\]
\end{lemma}

\begin{lemma}\label{claim:sum_of_distances}
If $G$ admits a $(k,\epsilon,\phi)$-clustering $C_1,\dots,C_k$, then 
$$\sum_{i =1}^k \sum_{v \in C_i} \|f_v - \mu_i\|^2_2 \leq \frac{4\epsilon k}{\phi^2} \, .$$
\end{lemma}
\begin{proof}
    Let  $\alpha_1, \ldots, \alpha_k \in \mathbb{R}^k$ be an orthonormal basis of $\mathbb{R}^k$. Applying Lemma \ref{lemma:variancebound} to $\alpha_1, \dots ,\alpha_k$ and summing, we obtain 
\begin{equation*}\label{eq:Bdelta1}
\sum_{i = 1}^k\sum_{v \in C_i}\|f_v - \mu_i\|^2_2 = \sum_{j=1}^k \sum_{i=1}^k \sum_{v \in C_i}\left \langle f_v - \mu_i, \alpha_j\right \rangle^2 \leq \frac{4\epsilon k}{\phi^2}.
\end{equation*}
\end{proof}
\begin{lemma}[Lemma 9 from \cite{GKLMS21}]\label{lemma:gklmsL9}
     If $G$ admits a $(k,\epsilon,\phi)$-clustering $C_1,\dots,C_k$, then for all $\alpha \in \mathbb{R}^k$ with $\|\alpha\|_2=1$ we have
     $$\left|\alpha^\top \left(I_k -  \sum_{i =1}^k |C_i| \mu_i \mu_i ^\top \right)\alpha \right| \leq \frac{4\sqrt{\epsilon}}{\phi}\|\alpha\|^2_2 \, .$$
\end{lemma}

\section{Model setting and basic concepts}
\label{sec:setting}

Throughout this paper, we fix an $n$-vertex $d$-regular graph $G=(V,E)$ that admits a $(k,\epsilon,\phi)$-clustering $C_1,\dots,C_k$ as per our input model \imlabel{}. A complete list of the objects and parameters that we use and assumptions we make on them is presented in \Cref{fig:setting} and \Cref{fig:labels}. In the remainder of this section we introduce these objects and prove some of their basic properties.

\begin{table}[!h]
 \fbox{\begin{minipage}{\textwidth}
\begin{multicols}{2}
\begin{flushleft}
	{\small
\begin{itemize}
	\item $G=(V,E)$ --- regular graph
    \item $d \ge 3$ --- degree of $G$
    \item $n =|V|$ --- number of vertices in $G$
    \item $\eta = O(1)$ --- upper bound on $\max_{i,j \in [k]} \frac{|C_i|}{|C_j|}$
    \item $\epsilon,\phi \in (0, 1)$ --- conductance parameters
    \item $k \ge 2$  --- number of communities
    \item $\{C_i\}_{i \in [k]}  $ --- $(k,\epsilon,\phi)$-clustering of $G$  (Def. \ref{def:clustering})
    \item $\mathcal{L}$ --- normalized Laplacian of $G$
    \item $(f_u)_{u \in V}$ --- spectral embedding of $V$  (Def. \ref{def:emb})
    \item $(\mu_i)_{i \in [k]}$ --- cluster means (Def. \ref{def:clusmeans})
    \item $(\tmu_i)_{i \in [k]}$ --- approximate cluster means (Def. \ref{def:approxmeans})
    \item $\delta \in (0,1)$ --- perturbation parameter 
    \item ${\xi} \in (0,1)$ -- approximate inner products quality
    \item $\langle \cdot, \cdot\rangle_{\apx}$ -- approximate inner products (Def. \ref{def:apx})
    \item $\| \cdot \|_{\apx}$ --- approximate distance (Rmk. \ref{def:dist_apx})
    \item $L = \lceil \frac{150}{\phi^{2}} \log(1/\delta) \rceil$ --- random walk length (Sec. \ref{sec:sub_alg})
    \item $\epsilon/\phi^6 \le 1/(10^{5}{\eta^4})$, $\phi^2\eta < 1/10^3$
    \item $k\log k \le {\phi^6 }/({10^9\eta^4 \epsilon })$
    \item $\delta d \le 1/100$
    \item ${\xi}/{n} \le {\phi^2}/{(20^4\eta)} \cdot \min_i\{\|\mu_i\|^2_2\}$
\end{itemize}}
\end{flushleft}
\end{multicols}	
 \end{minipage}}	
\caption{Parameters and objects used in this section (\Cref{sec:setting}) and \Cref{sec:poly,sec:sub_alg}.}
\label{fig:setting}
\end{table}

\begin{table}[!h]
 \fbox{\begin{minipage}{\textwidth}
\begin{flushleft}
	{\small
\begin{itemize}
    \item $\iota: V\rightarrow [k]$ --- target labeling: $u \in C_{\iota(u)}$
    \item $\sigma: V\rightarrow [k]$ --- perturbed labeling as per our input model \imlabel{}
    \item $\tau: V\rightarrow [k] \cup \{*\}$ --- spectral labeling: $u \in \M$ if $\tau(u)=*$, $u \in \spec(\tau(u))$ otherwise  (Def. \ref{def:spec}, Lem. \ref{lemma:disjoint_balls})
\end{itemize}}
\end{flushleft}
 \end{minipage}}	
\caption{Different labelings of the vertices available to us.}
\label{fig:labels}
\end{table}

\noindent
Our algorithm will work with the spectral embeddings of the vertices of $G$. More specifically, the algorithms which we build in this paper need to compute inner products between the spectral embeddings of the vertices of $G$. It suffices for our purposes to compute the inner products approximately. In \Cref{def:apx} we define the requirement which a function should satisfy in order to be a useful approximation.

\begin{definition}[Approximate spectral inner product function]\label{def:apx} Let $\xi$ be as per \Cref{fig:setting}. We say that a function $\langle \cdot, \cdot\rangle_{\apx}: V \times V \to \R$ is an \emph{approximate inner product function} if for any $u, v \in V$ we have 
\[\left|\langle f_u, f_v\rangle_{\apx} - \langle f_u, f_v\rangle \right| \leq \frac{\xi}{n}.\]
We sometimes concisely refer to the inner product and approximate inner product of spectral embeddings as \emph{spectral inner product} and \emph{approximate spectral inner product} respectively.
\end{definition}
\begin{remark}[Notation $\langle f_u, f_v \rangle_{\apx}$] Technically, the function $\langle \cdot, \cdot \rangle_{\apx}$ takes as input two vertices $u$ and $v$ and not their spectral embeddings. We will slightly abuse notation and write $\langle f_u, f_v \rangle_{\apx}$ instead of $\langle u, v\rangle_{\apx}$ to highlight that $\langle f_u, f_v \rangle_{\apx}$ is an approximation of a spectral object $\langle f_u, f_v \rangle$.
    
\end{remark}

\begin{remark}[Approximate spectral distance function]\label{def:dist_apx}
    We can use the the approximate inner product function as per \Cref{def:apx} to approximately calculate the $\ell_2$-norms of spectral embeddings $f_u$ as well as the distances $f_u - f_v$ as follows:
    \[\|f_u\|^2_{\apx} \coloneqq \langle f_u, f_u\rangle_{\apx}\, ,\]
    \[\|f_u - f_v\|^2_{\apx} \coloneqq \|f_u\|^2_{\apx} + \|f_v\|^2_{\apx} - 2\langle f_u, f_v \rangle_{\apx}\, .\] 
By \Cref{def:apx}, one has
\[\left|\|f_u\|^2_{\apx} - \|f_u\|^2_2\right| = \left|\langle f_u, f_u\rangle_{\apx} - \langle f_u, f_u\rangle\right| \leq \frac{\xi}{n}\,;\]
\[\left|\|f_u\|_{\apx} - \|f_u\|_2\right| \leq \sqrt{\left|\|f_u\|^2_{\apx} - \|f_u\|^2_2\right|}  \leq \sqrt{\frac{\xi}{n}}\,;\]
\[\left|\|f_u - f_v\|^2_{\apx} - \|f_u - f_v\|^2_{2}\right| \leq \left|\|f_u\|^2_{\apx} - \|f_u\|^2_{2}\right| + \left|\|f_v\|^2_{\apx} - \|f_v\|^2_2\right| + 2\left|\langle f_u, f_v\rangle_{\apx} - \langle f_u, f_v\rangle \right| \leq \frac{4\xi}{n}\,;\]
\[\left|\|f_u - f_v\|_{\apx} - \|f_u - f_v\|_{2}\right| \leq \sqrt{\left|\|f_u - f_v\|^2_{\apx} - \|f_u - f_v\|^2_{2}\right|} \leq 2\sqrt{\frac{\xi}{n}}\, .\]
\end{remark}

\noindent
We will not fix a particular approximate inner product function in the definitions and results~-- we state them in terms of $\langle \cdot, \cdot \rangle_{\apx}$. We specify the choice of a function whenever we state an algorithm which computes approximate spectral inner products as a subroutine. We describe the set of approximate spectral inner product functions which we use in \Cref{subsec:apx_spec_oracles}.

Apart from spectral  inner products, our algorithm will also require access to the cluster means, defined in \Cref{def:clusmeans}. Note that, as our algorithm never explicitly computes spectral embeddings, finding the cluster means as exact vectors is both not feasible nor useful. Just as with spectral inner products, we can opt for approximations to the cluster means. In \Cref{def:approxmeans} we formalize which vectors are both sufficiently close to the cluster means and are accessible to our algorithms.

\begin{definition}[Approximate cluster means]
    \label{def:approxmeans}
    Let $\phi$ and $\mu_1,\dots,\mu_k$ as per \Cref{fig:setting}, $\|\cdot\|_{\apx}$ as per \Cref{def:dist_apx}. We say that the vectors $\tmu_1,\dots\tmu_k \in \R^k$ are \emph{approximate cluster means} if \begin{itemize}
        \item $\|\mu_i-\tmu_i\|^2_2 \le \frac{\phi^2}{1600\eta} \|\mu_i\|^2_2$ for all $i \in [k]$;
        \item  there exists a set of vertices $\{u_i\}_{i \in [k]} \subseteq V$ such that $\tmu_i = f_{u_i}$ for all $i \in [k]$.
    \end{itemize}
\end{definition}

\noindent
Spectral embeddings together with the approximate cluster means give rise to natural vertex classes: the vertices whose embedding is close to an approximate cluster mean (called spectral clusters), and the vertices whose embedding is far from every approximate cluster mean (called cross vertices).

\begin{definition}[Spectral clusters and cross vertices]\label{def:spec}
   Consider the setting of \Cref{fig:setting}. For $i \in [k]$, we define the $i$-th spectral cluster, denoted $\spec(i)$, as
   $$ \spec(i) \coloneqq \left \{v \in V \colon  \|f_v - \tmu_i\|^2_{\apx} < \frac{\phi^2}{400 \eta}\|\tmu_i\|_{\apx}^2 \right \} \, .$$
    We also define the set of cross vertices, denoted $\M$, as
   $$ \M \coloneqq V  \setminus \left(\bigcup_{i \in [k]}\spec(i)\right)  \, .$$ 
\end{definition}

\noindent
If follows from Definitions \ref{def:approxmeans} and \ref{def:spec} that vertices in a spectral cluster $\spec(i)$ are close to the true cluster mean $\mu_i$. 

\begin{lemma}\label{claim:spec_exact_mu}
Consider the setting of \Cref{fig:setting}. For all  $i \in [k]$ and all $u \in \spec(i)$, one has 
$$ \|f_u - \mu_i\|^2_{2} \le \frac{\phi^2}{100\eta}\|\mu_i\|^2_{2} \, .$$
\end{lemma}
\begin{proof}
We have
    \begin{align*}
      \|f_u - \mu_i\|_{2}   &\leq \|f_u - \tmu_i\|_{2} + \|\tmu_i - \mu_i\|_2 \\
      & \leq \|f_u - \tmu_i\|_{\apx} + \frac{\phi}{40\sqrt{\eta}}\|\mu_i\|_2 + 2\sqrt{\frac{\xi}{n}} && \text{by \Cref{def:dist_apx} and \Cref{def:approxmeans}}\\
      & \leq \frac{\phi}{20\sqrt{\eta}}\|\tmu_i\|_{\apx} + \frac{\phi}{40\sqrt{\eta}}\|\mu_i\|_2 + 2\sqrt{\frac{\xi}{n}} && \text{by \Cref{def:spec}}\\
 & \leq \frac{\phi}{20\sqrt \eta}\|\tmu_i\|_2 + \frac{\phi}{40\sqrt{\eta}}\|\mu_i\|_2 +  2\sqrt{\frac{\xi}{n}} + \frac{\phi}{20\sqrt{\eta}}\cdot\sqrt{\frac{\xi}{n}} &&  \text{by \Cref{def:dist_apx} } \\
 & \leq  \frac{\phi}{20\sqrt \eta}(\|\mu_i\|_2 + \|\tmu_i - \mu_i\|_2) +  \left(\frac{\phi}{40\sqrt{\eta}} +\frac{\phi}{100\sqrt{\eta}}\right)\|\mu_i\|_2&& \text{since $\sqrt{\frac{\xi}{n}} \leq \frac{\phi}{400\sqrt{\eta}}\|\mu_i\|_2$, $\frac{\phi}{\sqrt{\eta}} \leq 1$}  \\
 & \leq \frac{\phi}{20\sqrt{\eta}}\left(1+\frac{\phi}{40\sqrt{\eta}}\right)\|\mu_i\|_2 + \left(\frac{\phi}{40\sqrt{\eta}} +\frac{\phi}{100\sqrt{\eta}}\right)\|\mu_i\|_2 &&  \text{by Definition \ref{def:approxmeans}}\\
 & \leq \frac{69}{800} \frac{\phi}{\sqrt{\eta}}\|\mu_i\|_2, && \text{using the fact that $\frac{\phi}{\sqrt{\eta}} \leq 1$. } \\
 & <\frac{\phi}{10\sqrt \eta}\|\mu_i\|_2.&& 
 \end{align*}
\end{proof}

\noindent
It is clear from \Cref{def:spec} that the cross vertices are disjoint from any spectral cluster, but one can in fact show that any two spectral clusters are also disjoint. This means that the spectral clusters and the cross vertices form a partitioning of the vertices. In particular, the spectral labeling $\tau$ as per \Cref{fig:labels} is well defined.

\begin{lemma}[Spectral clusters are disjoint]\label{lemma:disjoint_balls}
Consider the setting of \Cref{fig:setting}. For all  $i , j \in [k]$ with $i \neq j$, one has
    $$ \spec(i) \cap \spec(j) = \emptyset.$$
\end{lemma}
\begin{proof}
    Suppose that there exists $u \in \spec(i) \cap \spec(j)$ for some  $i , j \in [k]$ such that $i \neq j$. Then, by triangle inequality, we get
    \begin{align*}
        \|\tmu_i - \tmu_j\|_{\apx}  & \leq  \|\tmu_i - \tmu_j\|_{2} + 2\sqrt{\frac{\xi}{n}} && \text{by \Cref{def:dist_apx}}\\
        & \leq \|f_u - \tmu_i\|_2  + \|f_u - \tmu_j\|_2 + 2\sqrt{\frac{\xi}{n}} && \\
        & \leq \frac{\phi}{20\sqrt{\eta}}(\|\tmu_i\|_2 + \|\tmu_j\|_2) + 8\sqrt{\frac{\xi}{n}}&& \text{by Def. \ref{def:spec} and \Cref{def:dist_apx}}\\
        & \leq\frac{\phi}{20\sqrt{\eta}} \left( \| \tmu_i - \mu_i\|_2 + \|\mu_i\|_2 +\| \tmu_j - \mu_j\|_2 + \|\mu_j\|_2  \right) + 8\sqrt{\frac{\xi}{n}} \\
        & \leq \frac{\phi}{10}(\|\mu_i\|_2 + \|\mu_j\|_2) + 8\sqrt{\frac{\xi}{n}}&& \text{by Definition \ref{def:approxmeans}} \, .
    \end{align*}
    On the other hand, we have 
    \begin{align*}
         \|\tmu_i - \tmu_j\|_{\apx} & \geq  \|\tmu_i - \tmu_j\|_{2} - 2\sqrt{\frac{\xi}{n}} &&  \text{by \Cref{def:dist_apx}}\\
         & \geq \| \mu_i - \mu_j\|_2 - \| \tmu_i - \mu_i\|_2 - \| \tmu_j - \mu_j\|_2 - 2\sqrt{\frac{\xi}{n}} &&  \\ 
         & \geq \sqrt{\|\mu_i\|^2_2 + \|\mu_j\|^2_2 - 2 \langle \mu_i,\mu_j\rangle }  - \frac{\phi}{40 \sqrt{\eta}}(\|\mu_i\|_2 + \|\mu_j\|_2)- 2\sqrt{\frac{\xi}{n}}  && \text{by Definition \ref{def:approxmeans}}\\
         & \geq  \sqrt{ \|\mu_i\|^2_2 + \|\mu_j\|^2_2  - \frac{32\sqrt{\epsilon}}{\phi} \|\mu_i\|_2 \|\mu_j\|_2} - \frac{\phi}{10}(\|\mu_i\|_2 + \|\mu_j\|_2) - 2\sqrt{\frac{\xi}{n}} &&  \text{by Lemma \ref{lemma:clustermeans}} \\ 
         & \geq \frac{1}{2}(\|\mu_i\|_2 + \|\mu_j\|_2) -  \frac{\phi}{10}(\|\mu_i\|_2 + \|\mu_j\|_2) - 2\sqrt{\frac{\xi}{n}}  && \text{using $\frac{\epsilon}{\phi^2} \eta^2 \leq 10^{-5}$}\\
         & >  \phi\cdot ( \|\mu_i\|_2 + \|\mu_j\|_2) - 2\sqrt{\frac{\xi}{n}} &&  \text{using $\phi \leq 0.1$,}
         \end{align*}
where $\frac{\epsilon}{\phi^2} \eta^2 \leq 10^{-5}$ and $\phi \leq 0.1$ are assumptions stated in \Cref{fig:setting}. It remains to note that
$$\phi\cdot ( \|\mu_i\|_2 + \|\mu_j\|_2) - 2\sqrt{\frac{\xi}{n}} > \frac{\phi}{10}(\|\mu_i\|_2 + \|\mu_j\|_2) + 8\sqrt{\frac{\xi}{n}} \, , $$
because ${\phi} \leq 1/\sqrt{\eta}$ and $\sqrt{\frac{\xi}{n}} \leq \frac{\phi}{100\sqrt{\eta}}\min\{\|\mu_i\|_2, \|\mu_j\|_2\}$, which leads to a contradiction.  
\end{proof}

\noindent
While we would like the spectral clusters to reflect the target communities $C_1,\dots,C_k$, not all vertices in a spectral cluster are in the corresponding community. It is then useful to sometimes restrict our attention to the vertices in a spectral cluster that are indeed in the corresponding community.

\begin{definition}[Spectral cores]\label{def:speccore}
    Consider the setting of \Cref{fig:setting}. For $i \in [k]$, we define the $i$-th spectral core, denoted $\speccore(i)$, as
    $$\speccore(i) \coloneqq \spec(i) \cap C_i \, .$$
\end{definition}

\noindent
We will also want to restrict our attention to vertices that spectrally look like they are from the $i$-th cluster but in fact belong to $C_j$. We call these vertices \emph{impostors}, because they can easily fool a spectral algorithm. This is not the case for the cross vertices, since they do not pretend to be spectrally close to the wrong cluster, but rather their embedding is uninformative.

\begin{definition}[Spectral impostors]\label{def:impostor}
   Consider the setting of \Cref{fig:setting}. For  $i , j \in [k]$ with $i \neq j$, we define the $(j \to i)$-spectral impostors, denoted $\im(i,j)$, as
  $$\im(i,j) \coloneqq \spec(i) \cap C_j \, .$$
  We define the set of all spectral impostors, denoted $\im$, as
  $$ \im \coloneqq \bigcup_{\substack{i,j \in[k]: \\ i\neq j}} \im(i,j) \, .$$
We will often refer to spectral impostors simply as \emph{impostors}.
\end{definition}
\begin{remark}[Spectral impostors are disjoint]\label{rem:disjoint_im} For all $i_1, i_2, j_1, j_2 \in [k]$ with $(i_1, j_1)\neq (i_2, j_2)$, one has
\[\im(i_1, j_1) \cap \im(i_2, j_2) = \emptyset\]
since all of the spectral clusters are disjoint by \Cref{lemma:disjoint_balls} and since all of the clusters $C_1, \ldots, C_k$ are disjoint.
\end{remark}

\noindent
The perturbed labeling, which the algorithm is given as input, also gives rise to natural vertex classes: we can partition the vertices according to their label $\sigma(u)$, and call these classes \emph{label clusters}.

\begin{definition}[Label clusters]
\label{def:labelcluster}
Let $\sigma$, as per \Cref{fig:labels}. For $i \in [k]$, we define the $i$-th label cluster, denoted $\lab(i)$,  as
\begin{equation*}
    \lab(i) \coloneqq \{u \in V: \sigma(u)=i\} \, .
\end{equation*}
\end{definition}

\begin{definition}[Mislabeled vertices]\label{def:mislabeled} Let $\sigma, \iota$  as per \Cref{fig:labels}. We define the set of all vertices mislabeled by $\sigma$, denoted $\mislabeled$, as
\[\mislabeled \coloneqq \left\{v \in V: \sigma(v) \neq \iota(v)\right\}\]
    
\end{definition}

\begin{nota}
    Let $\sigma, \tau$ as per \Cref{fig:labels}. For a vertex $u \in V$, we abuse notation and write $\lab(u)$ and $\spec(u)$ to refer to $\lab(\sigma(u))$ and $\spec(\tau(u))$ respectively. For an index $i \in [k]$, we also write $\spec(\neg i)$ to refer to $\cup_{j \in [k]\setminus\{i\}} \spec(j)$.
\end{nota}

\noindent
As briefly mentioned before, our algorithm will combine spectral information with label information. It will in particular try to get a sense for which vertices are impostors by checking if vertices part of both the $i$-th spectral cluster and the $j$-th label cluster can reach the cross vertices.  This gives rise to the following notion.

\begin{definition}[Cross graphs]\label{def:crossgraph}
Let $G$ as per \Cref{fig:setting} and $\sigma$ as per \Cref{fig:labels}. For $i,j \in [k]$ with $i \neq j$, we let $V_{i,j} \coloneqq (\spec(i) \cap \lab(j))\cup \M$ and define  the $(i, j)$-cross graph to be the graph $G_{i,j}=(V, E_{i,j},\ell)$ consisting of all the edges induced by $V_{i,j}$ together with a number of self-loops to preserve regularity, i.e.
$$E_{i,j} \coloneqq \{(u,v) \in E: \, u,v \in V_{i,j}\} \quad \text{and} \quad \ell(u) \coloneqq d-|\{e \in E_{i,j}:\,u\in e\}| \text{ for all } u \in V \, . $$
\end{definition}

\begin{remark}
    Cross graphs are defined over the whole vertex set $V$ because we will want to conveniently take unions of them. For our purposes, it does not matter if there are isolated vertices, since we will be interested in the reachability of vertices in $(\spec(i) \cap \lab(j)) \cup \M$.
\end{remark}

\noindent
We now proceed to prove properties of spectral clusters that will be crucial for our analysis.

\subsection{Properties of spectral clusters}
\label{subsec:specinfo}
A key property of the spectral clusters is that any vertex that belongs to $\spec(i)$ and is not connected to $\M$ can have at most a $\frac{\phi}{2}$ fraction of its edges to other spectral clusters.

\begin{lemma}[Few neighbors across spectral clusters]\label{lemma:cores_sparsely_connected}
Consider the setting of \Cref{fig:setting}. For every $i \in [k]$ and every $u \in \spec(i) \setminus \nei_G(\M)$ one has
$$\left|\nei_G(u) \cap \spec(\neg i) \right| < d \cdot \frac{\phi}{2} \, .$$ 
\end{lemma}

\begin{proof}
The high-level idea is as follows: If $u \in \spec(i)$, then by \Cref{claim:spec_exact_mu}, $f_u$ is close to the cluster mean $\mu_i$. On the other hand, by Lemma \ref{lemma:fx_close_to_avg_neighbors}, we know that $f_u$ is close to the average of its neighbours' embeddings. In particular, if $u$ has many neighbours in other spectral clusters, then $f_u$ has a large component coming from other cluster means. But since the cluster means are almost orthogonal (by Lemma \ref{lemma:clustermeans}), this would mean that $f_u$ has a large component orthogonal to $\mu_i$, contradicting the fact that $f_u$ is close to $\mu_i$. We now prove this more formally.
\\~\\
Suppose for contradiction that   $|\nei_G(u) \cap \spec(\neg i) | \geq  d \cdot {\phi}/{2}$ for some $u \in \spec(i) \setminus \nei_G(u)$. By Lemma \ref{lemma:fx_close_to_avg_neighbors},  we have 
\begin{equation}\label{eq:nbhrs}
   \left  \|f_u - \frac{1}{d}\sum_{v \in N_G(u)} f_v \right \|_2 \leq 2\epsilon \|f_u\|_2. 
\end{equation}
\noindent
By the assumption that $u \in \spec (i)$ (as per Definition \ref{def:spec}), we get
\begin{align*}
    \frac{\phi}{10\sqrt{\eta}}\|\mu_i\|_2 & \geq \|f_u - \mu_i\|_2 &&  \text{by \Cref{claim:spec_exact_mu}} \\
    & =  \left \|f_u - \left(\frac{1}{d}\sum_{v \in \nei_G(u)}f_v \right) + \left(\frac{1}{d}\sum_{v \in \nei_G(u)}f_v\right)  - \mu_i \right\|_2
   &&  \\ 
   & \geq  \left \|  \mu_i - \frac{1}{d}\sum_{v \in \nei_G(u)}f_v \right \|_2 -\left \| f_u - \frac{1}{d}\sum_{v \in \nei_G(u)}f_v \right \|_2 &&  \\
    & \geq  \left \|  \mu_i - \frac{1}{d}\sum_{v \in \nei_G(u)}f_v \right \|_2 - 2\epsilon \|f_u\|_2  &&  \text{by \eqref{eq:nbhrs}}. 
\end{align*}

\noindent
Recall the definition of the spectral labeling $\tau$ from \Cref{fig:labels}: $\tau(v) = *$ for a vertex $v \in \M$ and otherwise $\tau(v)$ returns the spectral cluster which $v$ belongs to. Since $ u \notin \nei_G(\M)$, we have that $\tau(v) \neq *$ for all  $v \in \nei_G(u)$. In particular, $ v\in \spec(\tau(v))$ for all $v \in \nei_G(u)$. Hence,  by \Cref{claim:spec_exact_mu}, for all  $v \in \nei_G(u)$, it holds that

\begin{equation}\label{eq:4.16}
    \|f_v - \mu_{\tau(v)}\|_2 \leq \frac{1}{10} \frac{\phi}{\sqrt{\eta}}\|\mu_{\tau(v)}\|_2 \leq \frac{\phi}{5}\|\mu_i\|_2, 
\end{equation}
where the last inequality follows from \Cref{lemma:clustermeans} {\bf \ref{bulletpt:mu_norm}} together with the assumption that $\eta \geq  \frac{|C_j|}{|C_l|}$ for all $j,l \in [k]$. Hence, we get 
\begin{align*}
\frac{1}{10} \frac{\phi}{\sqrt{\eta}}\|\mu_i\|_2 
& \geq \left \|  \mu_i - \frac{1}{d}\sum_{v \in \nei_G(u)}f_v \right \|_2 - 2\epsilon \|f_u\|_2  \\
& = \left\| \mu_i- \frac{1}{d}\sum_{v \in \nei_G(u)}\mu_{\tau(v)} - \frac{1}{d} \sum_{v \in \nei_G(u)}\left(f_v - \mu_{\tau(v)}\right) \right\|_2 - 2 \epsilon \|f_u\|_2 \\
& \geq \left\| \mu_i - \frac{1}{d}\sum_{v \in \nei_G(u)}\mu_{\tau(v)}\right\|_2 - \frac{1}{d}\sum_{v \in \nei_G(u)} \| f_v - \mu_{\tau(v)}\|_2 - 2\epsilon \|f_u\|_2 \\
&  \geq  \left\| \mu_i - \frac{1}{d}\sum_{v \in \nei_G(u)}\mu_{\tau(v)}\right\|_2 - \frac{\phi}{5} \|\mu_i\|_2 - 4 \epsilon \|\mu_i\|_2 \qquad\qquad\qquad \text{by \eqref{eq:4.16} and \Cref{claim:spec_exact_mu}}. 
\end{align*}
\noindent
Now let $l \coloneqq |\{ v \in \nei_G(u) : \tau(v) \neq i\}|$ be the number of neighbors that belong to a different spectral cluster. By assumption, it holds that $l \geq \frac{\phi}{2}d$. We have 
$$ \mu_i - \frac{1}{d}\sum_{v \in \nei_G(u)}\mu_{\tau(v)} = \frac{1}{d} \sum_{v \in  \nei_G(u)}(\mu_i - \mu_{\tau(v)} ) = \frac{1}{d}\sum_{v \in  \nei_G(u) : \tau(v) \neq i}(\mu_i - \mu_{\tau(v)}) = \frac{l}{d}\mu_i - \sum_{v \in  \nei_G(u) : \tau(v) \neq i}\mu_{\tau(v)} \, , $$
which gives 
\begin{align*}
    \left\| \mu_i - \frac{1}{d}\sum_{v \in  \nei_G(u)}\mu_{\tau(v)}\right\|_2^2 & = \left\|\frac{l}{d}\mu_i - \frac{1}{d}\sum_{v \in  \nei_G(u) : \tau(v) \neq i}\mu_{\tau(v)} \right\|_2^2 \\
    & = \left\|\frac{l}{d}\mu_i \right\|_2^2 + \frac{1}{d^2}\left\| \sum_{v \in  \nei_G(u)  :\tau(v) \neq i}\mu_{\tau(v)}\right\|_2^2 - 2\frac{l}{d}\left| \left\langle\mu_i, \frac{1}{d}\sum_{v \in  \nei_G(u)  :\tau(v) \neq i}\mu_{\tau(v)} \right \rangle\right| \\
 & \geq \frac{l^2}{d^2} \|\mu_i\|^2_2 - 2 \frac{l^2}{d^2}\frac{16 \sqrt{\epsilon}}{\phi^2}\sqrt{\eta}\|\mu_i\|_2^2 \qquad \text{by Lemma \ref{lemma:clustermeans} and using $\frac{l}{d} \geq \frac{\phi}{2}$} \\
 & = \frac{l^2}{d^2}\|\mu_i\|^2_2\left(1- \frac{32\sqrt{\epsilon} \sqrt{\eta}}{\phi^2} \right) \\
& \geq \left(\frac{l}{d} \|\mu_i\|_2 \left(1- \frac{32\sqrt{\epsilon} \sqrt{\eta}}{\phi^2} \right) \right)^2 \\
& \geq \left( \frac{\phi}{3}\|\mu_i\|_2 \right)^2 \qquad \text{by assumption that $l \geq \frac{\phi}{2}d$, and $\frac{\epsilon}{\phi^6} \leq \frac{10^{-5}}{\eta^4}$ as per \Cref{fig:setting}. } 
\end{align*}
\noindent
Hence, we get
\begin{align*}
   \frac{1}{10} \frac{\phi}{\sqrt{\eta}}\|\mu_i\|_2 &  \geq \left\| \mu_i - \frac{1}{d}\sum_{v \in  \nei_G(u)}\mu_{\tau(v)}\right\|_2 - \frac{\phi}{5} \|\mu_i\|_2 - 4 \epsilon \|\mu_i\|_2  \\
& \geq  \frac{\phi}{3}\|\mu_i\|_2-  \frac{\phi}{5}\|\mu_i\|_2 - 4 \epsilon \|\mu_i\|_2    \\
& =\left(  \frac{2}{15}\phi - 4\epsilon \right)\|\mu_i\|_2 \\
& > \frac{\phi}{10}\|\mu_i\|_2 \qquad \qquad\qquad\text{by assumption $\frac{\epsilon}{\phi^2} \leq 10^{-5}$},
\end{align*}
\noindent
which is a contradiction. 
\end{proof}

\noindent
The impostors together with the cross vertices are all the vertices for which the spectral embeddings assign a wrong cluster, in the sense that it does not give the correct classification. By known results, one can show that they account for at most an $\epsilon\cdot\poly(1/\phi)$ fraction of vertices.

\begin{lemma}[$\im \cup \M$ is small]\label{lemma:impostor_n_cross_size}
Consider the setting of \Cref{fig:setting}. Then,
    $$ |\im \cup \M| \leq 2 \cdot 10^4 \cdot \eta^2 \cdot \frac{\epsilon}{\phi^4 }n \, .$$
\end{lemma}
\begin{proof}
The main idea is to combine \Cref{claim:sum_of_distances} together with the observation that for all $u \in \im \cup \M$, the distance $\|f_u - \mu_{\iota(u)}\|_2$ is large. We first show that for all $u \in \im \cup \M,$ it holds that 
 
 \begin{equation}\label{eq:imposter_size1}
 \|f_u- \tmu_{\iota(u)}\|_{\apx}^2 \geq   \frac{\phi^2}{400 \eta} \|\tmu_{\iota(u)} \|^2_{\apx}. 
 \end{equation}
 \noindent
 Indeed, if $u \in \M$ is a cross vertex, then~\eqref{eq:imposter_size1} holds by definition of $\M$ (Definition  \ref{def:spec}). 
 If instead $u \in \im$ is an impostor, then $u \in \spec(i) \cap C_j$ for some $j \neq i$ (by Definition \ref{def:impostor}). In particular, since the spectral clusters are disjoint (by Lemma \ref{lemma:disjoint_balls}), we have $u \notin \spec(j)$, and so 
$\|f_u - \tmu_{j}\|_{\apx}^2 \geq  {\phi^2}/({400 \eta})\cdot \|\tmu_{j}\|^2_{\apx}$ (by Definition \ref{def:spec}), which gives~\eqref{eq:imposter_size1} since $\iota(u)=j$. Thus, we get

 \begin{align*}
     \|f_u-  \mu_{\iota(u)}\|_2 & \geq \|f_u-  \tmu_{\iota(u)}\|_2  -  \|\mu_{\iota(u)} - \tmu_{\iota(u)}\|_2 &&  \\
     & \geq  \|f_u-  \tmu_{\iota(u)}\|_{\apx} - 2\sqrt{\frac{\xi}{n}} -  \|\mu_{\iota(u)} - \tmu_{\iota(u)}\|_2  &&  \text{by \Cref{def:dist_apx}}\\ 
     & \geq \frac{\phi}{20\sqrt{\eta}}\|\tmu_{\iota(u)}\|_{\apx} - \frac{\phi}{40\sqrt{\eta}}\|\mu_{\iota(u)}\|_2 - 2\sqrt{\frac{\xi}{n}} && \text{by \eqref{eq:imposter_size1} and Def. \ref{def:approxmeans}}\\
      & \geq \frac{\phi}{20\sqrt{\eta}}\|\tmu_{\iota(u)}\|_{2} - \frac{\phi}{40\sqrt{\eta}}\|\mu_{\iota(u)}\|_2 - 2\sqrt{\frac{\xi}{n}} - \frac{\phi}{20\sqrt{\eta}}\cdot\sqrt{\frac{\xi}{n}} && \text{by \Cref{def:dist_apx}}\\
     & \geq \frac{\phi}{20\sqrt{\eta}}\left( \|\mu_{\iota(u)}\|_2 - \|\mu_{\iota(u)} - \tmu_{\iota(u)}\|_2\right)  - \left(\frac{\phi}{40\sqrt{\eta}} + \frac{\phi}{100\sqrt{\eta}}\right)\|\mu_{\iota(u)}\|_2 && \text{by \Cref{fig:setting}} \\ 
     & \geq \frac{\phi}{20\sqrt{\eta}}\left(1-\frac{\phi}{40\sqrt{\eta}} \right)\|\mu_{\iota(u)}\|_2 - \left(\frac{\phi}{40\sqrt{\eta}} +\frac{\phi}{100\sqrt{\eta}} \right)\|\mu_{\iota(u)}\|_2  && \text{by Definition \ref{def:approxmeans}}\\
     & \geq \frac{119}{8000}\frac{\phi}{\sqrt{\eta}}\|\mu_{\iota(u)}\|_2, &&  \text{using $\frac{\phi}{\sqrt{\eta}} \leq 0.1$.}\\
 \end{align*}
 \noindent
 so 
\begin{equation}\label{eq:imposter_size4}
\|f_u-  \mu_{\iota(u)}\|_2^2 \geq \frac{14161}{64 \cdot 10^6} \frac{ \phi^2}{ \eta} \|\mu_{\iota(u)}\|^2_2.  
\end{equation}
\noindent
Summing~\eqref{eq:imposter_size4} over all $u \in \im \cup \M$, we get 
\begin{equation}\label{eq:imposter_size2}
 \sum_{u \in \im \cup \M}\|f_u-  \mu_{\iota(u)}\|_2^2 \geq \frac{14161}{64 \cdot 10^6} \frac{\phi^2}{\eta}  \sum_{u \in \im \cup \M} \|\mu_{\iota(u)}\|^2_2 \geq  |\im \cup \M| \cdot \frac{\phi^2}{5000\eta^2} \cdot \frac{k}{n}. 
\end{equation}
Here, the last inequality uses that $\| \mu_i\|^2_2 \geq  (1-4/10^5)\frac{k}{n}\frac{1}{\eta}$ for all $i$, by \Cref{rem:mu_i_norm}. 
On the other hand, by \Cref{claim:sum_of_distances}, we have 
\begin{equation}\label{eq:imposter_size3}
\sum_{u \in \im \cup \M}\|f_u-  \mu_{\iota(u)}\|_2^2 \leq \sum_{u \in V} \|f_u - \mu_{\iota(u)}\|_2^2 \leq \frac{4 \e k}{\phi^2}.
\end{equation}
Combining Equations \eqref{eq:imposter_size2} and \eqref{eq:imposter_size3} gives the result.   
\end{proof}

\noindent
An important property of impostors is that they essentially induce an expander. This will help the algorithm detect the impostors, hence averting being fooled by their spectral embedding.

\begin{lemma}[Impostors expand]\label{lem:imposters_expand}
Consider the setting of \Cref{fig:setting}.
For every $i ,j \in [k]$ with $i \neq j$ and every $S \subseteq \im(i,j) \setminus \nei_G(\M)$, it holds that
$$|E(S, \im(i,j) \setminus S)|\geq \frac{\phi}{2}d|S| \, .$$ 
\end{lemma}

\begin{proof}
Note that \smash{$|S| \leq |\im| \leq 2 \cdot 10^4 \cdot \eta^2  \frac{\epsilon}{\phi^4}n\leq \frac{|C_j|}{2}$}, where the second inequality follows by \Cref{lemma:impostor_n_cross_size}, and the last inequality follows by the setting of $\epsilon, \phi$ as per \Cref{fig:setting}. Since $S \subseteq C_j$, and since $G\{C_j\}$ is a $\phi$-expander, we have 
\begin{equation}\label{eq:S_expands}
|E(S, C_j \setminus S)| \geq \phi\cdot  d \cdot |S|.
\end{equation}
\noindent
We will now show that a large fraction of the edges in $E(S, C_j \setminus S)$ go to the $(j \to i)$-impostors $\im(i,j)$. 
We have
\begin{align*}
 C_j 
 & \subseteq \M \cup \spec(\neg i) \cup  \im(i,j) \, .
 \end{align*}
Thus,
 \begin{equation}\label{eq:edgecounting}
 \begin{aligned}
 |E(S, C_j \setminus S)| &  \leq \left|E(S, \M)\right| +\left|E(S,  \spec(\neg i))  \right| + |E(S, \im(i,j) \setminus S)| \\
 & = \left|E(S, \spec(\neg i))\right| + |E(S, \im(i,j) \setminus S)|  \qquad \text{by the assumption $S \cap \nei_G(\M) = \emptyset$} \\
 & \leq \sum_{u \in S}\left|\nei_G(u) \cap  \spec(\neg i))  \right|  + |E(S, \im(i,j) \setminus S)| \\
 & \leq |S| \cdot d \cdot \frac{\phi}{2} + |E(S, \im(i,j) \setminus S)| \qquad \qquad \qquad \qquad \qquad \qquad \text{by \Cref{lemma:cores_sparsely_connected}.}
 \end{aligned}
 \end{equation}

\noindent
 Here the last inequality uses the assumption that $S \subseteq \im(i,j) \setminus \nei_G(\M) \subseteq \spec(i) \setminus  \nei_G(\M)$, (since $\im(i,j) \subseteq \spec(i) $ by Definitions \ref{def:spec} and \ref{def:impostor}), so every vertex $u \in S$ satisfies the assumption in \Cref{lemma:cores_sparsely_connected}. Combining Equations \eqref{eq:S_expands} and \eqref{eq:edgecounting}, and rearranging gives the result.
\end{proof}
\begin{remark}\label{remark:NM_forward_ref}
\Cref{lem:imposters_expand} and \Cref{lemma:cores_sparsely_connected} only hold outside of 
 $\nei_G(\M)$. Since $|\nei_G(\M)| \approx \epsilon dn $, this is one of the reasons why our algorithm in Section \ref{sec:poly} (Algorithm \ref{alg:polytime}) will incur a $d$-dependence in the misclassification rate. Later, in Section \ref{sec:sub_alg}, we improve this by defining a subset $\NM \subseteq \nei_G(\M)$ of size $|\NM| \approx \epsilon n$ (see Definition \ref{def:NM}), and proving 
 that \Cref{lem:imposters_expand} and \Cref{lemma:cores_sparsely_connected} can also be extended to $\nei_G(\M) \setminus \NM$ (See \Cref{lem:imposters_expand_stronger} and \Cref{claim:NM_neighbors}, respectively). 

\end{remark}

\subsection{Approximate spectral inner product oracles and approximate cluster means}\label{subsec:apx_spec_oracles}

In this section, we present two algorithms for approximating spectral inner products: a simpler polynomial time algorithm and a more sophisticated sublinear time and space algorithm. As it is essential for the polynomial and sublinear time classifiers to have access to approximate cluster means, we also present a subroutine for finding approximate cluster means in \Cref{lemma:pi_computation}, which may be used by both.  We present the analysis of \Cref{lemma:pi_computation} in \Cref{subsec:apxmeans}.

\noindent
We use the polynomial time algorithm in \Cref{rem:euclidian_oracle} as a subroutine in the polynomial time classifier, introduced in \Cref{sec:poly}.

\begin{restatable}[Polynomial time algorithm for approximating $U_{[k]}U^T_{[k]}$]{remark}{projmatrix} \label{rem:proj_matrix_apx} There exists an $O(n^2\cdot k^2\cdot \frac{\log n}{\phi^2})$ time algorithm which recovers a matrix $\widehat{Q}$ such that
\[\|\widehat{Q}\widehat{Q}^T - U_{[k]}U^T_{[k]}\|_2 \leq 1/n^{100}.\]
\end{restatable} We prove the above remark in \Cref{subsec:apxmeans}.

\begin{remark}[Euclidean inner product oracle]\label{rem:euclidian_oracle}
    From \Cref{def:emb}, we have $f_u = U^T_{[k]}\1_u$. From \Cref{rem:proj_matrix_apx}, there exists a polynomial time algorithm which recovers a matrix $\widehat{Q} \in \R^{n \times k}$ such that $\|U_{[k]}U^T_{[k]} - \widehat{Q}\widehat{Q}^T\|_2 \leq 1/\poly(n)$. From there, we have
    \[\left|\langle f_u, f_v\rangle  - \1_u\widehat{Q}\widehat{Q}^T\1_v\right|= \left|\1^T_uU_{[k]}U^T_{[k]}\1_v - \1_u\widehat{Q}\widehat{Q}^T\1_v\right| \leq \|U_{[k]}U^T_{[k]} - \widehat{Q}\widehat{Q}^T\|_2 \leq 1/\poly(n),\]
    so we can in polynomial time approximate the inner products up to precision $1/\poly(n)$. 
\end{remark}

In \Cref{sec:sub_alg}, we use a sublinear time and space algorithm, introduced in \cite{GKLMS21}. We refer to this algorithm as Spectral Dot Oracle and we provide its performance guarantees below. 

\begin{thm}[Spectral Dot Product Oracle from \cite{GKLMS21}]\label{thm:spec_dot_prod_oracle} Let $\epsilon, \phi \in (0, 1)$ with $\epsilon \leq \frac{\phi^2}{10^5}$. Let $G = (V, E)$ be a $d$-regular graph that admits a $(k, \epsilon, \phi)$-clustering $C_1, \ldots, C_k$. Let $1 > \xi > 1/n^5$ . There exists a preprocessing algorithm which computes in time $O\left(2^{O\left(\frac{\phi^2}{\epsilon}k^4\log(k)\right)}\cdot n^{1/2+O(\epsilon/\phi^2)}\cdot\poly(\log (n))\cdot\frac{1}{\phi^2}\cdot \poly(1/\xi)\right)$ a sublinear space data structure
$\mathcal{D}$ of size $O\left(\poly(k)\cdot n^{1/2+O(\epsilon/\phi^2)}\cdot\poly(\log (n))\cdot\poly(1/\xi)\right)$ such that with probability at least $1 - n^{97}$ the following
property is satisfied:

There exists an algorithm, referred to as \emph{Spectral Dot Product Oracle} which has access to $\mathcal{D}$ and for every pair of vertices $u, v \in V$ computes an output value $\specdp(u, v)$
such that 
\[ \left|\specdp(u, v) - \langle f_u, f_v\rangle\right| \leq \frac{\xi}{n}.\]
The running time of Spectral Dot Product Oracle for every query $u, v \in V$ is 
\[O\left(\poly(k)\cdot n^{1/2 + O(\epsilon/\phi^2)}\cdot\poly(\log(n))\cdot \frac{1}{\phi^2}\cdot\poly(1/\xi)\right).\]
\end{thm}
\begin{remark} We note that \Cref{thm:spec_dot_prod_oracle} is stated slightly differently in \cite{GKLMS21}. There, both the preprocessing algorithm and Spectral Dot Product Oracle are probabilistic algorithms; the probability with which the preproccessing algorithm is successful is $1 - n^{100}$ and the probability with which the Spectral Dot Product Oracle is successful, conditioned on the success of the preprocessing algorithm, is also $1 - n^{100}$ for any $u, v \in V$. We note, however, that the preproccessing algorithm and Spectral Dot Product Oracle use different random bits, and together they require $O\left(\poly(k)\cdot n^{1/2+O(\epsilon/\phi^2)}\cdot\poly(\log (n))\cdot\poly(1/\phi)\right)$ many random bits in total to process all pairs of vertices in $V$. Therefore, without loss of generality, we may assume that all of these random bits are sampled in the preproccesing stage. The preprocessing stage is then successful if it both returns a valid data structure $\mathcal{D}$ and if the Spectral Dot Product Oracle returns valid approximations to the spectral inner products for all pairs of vertices in $V$ which, by the union bound arguments, has probability at least
\[1 - n^{100} - n^2\cdot n^{100} > 1 - n^{97}.\]
\end{remark}

\begin{remark}\label{rem:specdp_runtime} We fix $\xi = \frac{\phi^2}{(20)^4\eta}\min_i\{\|\mu_i\|^2_2\}n \geq \Omega(\phi^2k)$ so that $\specdp$ is a valid approximate inner product function, as per \Cref{def:apx}. The space complexity of the data structure $\mathcal{D}$ is then bounded by
\[O\left(\poly(k)\cdot n^{1/2+O(\epsilon/\phi^2)}\cdot\poly(\log (n))\cdot\poly(1/\phi)\right),\]
and one call to the Spectral Dot Product Oracle takes time
\[O\left(\poly(k)\cdot n^{1/2+O(\epsilon/\phi^2)}\cdot\poly(\log (n))\cdot\poly(1/\phi)\right).\]
\end{remark}

Finally, we present \Cref{lemma:pi_computation}, which shows that we can compute the approximate clusters means. 
\begin{restatable}{lemma}{approxmu}\label{lemma:pi_computation}
There exists an algorithm that returns $k$ vertices $\{u_1, \dots, u_k\}$ (\Cref{alg:app_centers}) and an algorithm that returns a permutation $\pi: [k] \rightarrow [k]$ (\Cref{alg:permutation}) such that with probability 0.97 over the internal randomness of \Cref{alg:app_centers} and \Cref{alg:permutation} the following holds:
$$ \Pr_{\sigma}\left[ \forall i \in [k]: \|f_{u_i} - \mu_{\pi(i)}\|^2_2 \leq \frac{\phi^2}{1600\eta}\|\mu_{\pi(i)}\|^2_2 \right] \geq1 - 10^{-3}.$$ 
\end{restatable}

Our algorithm for computing the cluster means (\Cref{alg:app_centers}) is similar to Algorithm 1 of \cite{SP23}. We present the algorithm and its analysis in Appendix \ref{subsec:apxmeans}. We also present the algorithm for finding the right permutation (\Cref{alg:permutation}) and its analysis, and prove \Cref{lemma:pi_computation}, in Appendix \ref{subsec:apxmeans}. As we see later, the algorithm for computing the cluster means depends on the choice of the approximate spectral inner product oracle and, if $Q$ is the time required to make one call to the chosen oracle, then the runtime of the algorithm is $O(k^2\cdot\log^2(k)\cdot Q)$.
\section{Warm-up: polynomial-time classifier with suboptimal rate}
\label{sec:poly}
\noindent
In this section, we outline and analyze \Cref{alg:polytime}, a polynomial-time algorithm with misclassification rate $\approx d\epsilon \delta$. The algorithm takes as input the graph, the approximate cluster means, and the perturbed label. It outputs a cluster id for each vertex. We will show that the number of vertices such that the output id is different from $\iota(u)$ is at most a $O({d\delta\epsilon}/{\phi^5})$-fraction.

In the previous sections, we defined several objects, such as the cross graphs $\crossgapx_{i, j}$, the spectral clusters $\spec(i)$, and the spectral labeling $\tau$, which we would like to use in our algorithm. The definitions of those objects depend on the choice of an approximate inner product function. In line~\eqref{line:fix_apx_p} we fix the choice of an approximate inner product function as the Euclidean inner product, introduced in \Cref{subsec:apx_spec_oracles}.

We now describe the main idea of \Cref{alg:polytime}. In the previous sections, we showed that the cluster means are roughly orthogonal (see \Cref{lemma:clustermeans}) and  that the spectral embeddings of all except for $\approx \epsilon n$ vertices are located very close to their cluster means (see \Cref{lemma:impostor_n_cross_size}). If we knew that for each of those $\approx \epsilon n$ vertices the spectral embedding is separated from each of the cluster means (recall that we refer to this set as $\M$, \Cref{def:spec}), then (given access to the cluster means), the following simple algorithm would suffice: (1) for a vertex $u$, verify if $f_u$ is close to some cluster mean; (2) if yes, return the number of this cluster mean; (3) if no, return $\sigma(u)$.
Indeed, in that case only vertices which are far from any cluster mean could potentially be misclassified, which, by Markov inequality, would yield that the total number of vertices misclassified by the above algorithm is $\approx \epsilon\delta n$. 
However, we are not able to establish that a vertex is either in $\M$ or it is spectrally close to its own cluster mean. There could potentially be $\approx \epsilon n$ many vertices that actually belong to some cluster $C_j$ but whose spectral embeddings are located close to the cluster mean $\mu_i$ for $i \neq j$ (recall that we refer to such vertices as $(j \to i)$-impostors \Cref{def:impostor}). If we settled for the above algorithm, we could misclassify up to $\approx \epsilon n$ many vertices.

Instead, in \Cref{lem:imposters_expand} we show that for any pair $i \neq j \in [k]$ the vertices in $\im(i, j)$ are well-connected to $\M$ and induce an expander-like subgraph. We argue that, due to the expander-like properties of $\im(i, j)$, even if we remove all of the vertices misclassified by $\sigma$ from $\im(i, j)$, the remaining set of vertices is still well-connected to $\M$. This allows us to modify the above tentative algorithm to detect not only vertices from $\M$, but also vertices from $\im$. This idea results in \Cref{alg:polytime}, which operates as follows: for a vertex $u$, if $f_u$ is close to some cluster mean and its predicted cluster $\sigma(u)$ agrees with the spectral information, then return $\sigma(u)$ (line~\eqref{line:core_p}); if $f_u$ is far from all cluster means, then return $\sigma(u)$ (line~\eqref{line:far_p}); if $f_u$ is close to some cluster mean but its predicted cluster $\sigma(u)$ is different from the one suggested by the spectral information (that is, if a vertex looks like it is in $\im$), then check if $u$ can reach $\M$ (see \Cref{def:reachability}) through only vertices in the same predicted cluster and whose embeddings are close to the same cluster mean (line~\eqref{line:conntest}). In other words, the last step checks reachability in the $(\tau(u),\sigma(u))$-cross graph (see \Cref{def:crossgraph}).

\begin{definition}[Reachability]\label{def:reachability}
    For a graph $H = (V(H), E(H))$ we say that a subset $S \subseteq V(H)$ is reachable from a vertex $u \in V(H)$ if the connected component of $u$ contains at least one vertex from $S$. We write $u \sim_H S$ if $S$ is reachable from $u$ in the graph $H$ and $u \slashed{\sim}_H S$ otherwise.
\end{definition}

\begin{algorithm}
\caption{Polynomial-time classifier}\label{alg:polytime}
\begin{algorithmic}[1]
\State \textbf{Input:} $G, \phi, \eta, (\widetilde{\mu}_i)_{i \in [k]}$ as per \Cref{fig:setting}, $\sigma$ as per \Cref{fig:labels}, and a vertex $u \in V$
\State \textbf{Output:} a cluster id in $[k]$
\State $\langle \cdot, \cdot \rangle_{\text{apx}} \gets \langle \cdot, \cdot \rangle$ \Comment{Euclidean inner products, see \Cref{rem:euclidian_oracle}} \label{line:fix_apx_p}

\If{$\tau(u) = \sigma(u)$} \Return $\sigma(u)$ \Comment{label agrees with spectral information} \label{line:core_p}
\ElsIf{$\tau(u) = *$} \Return $\sigma(u)$ \Comment{spectral information is ambiguous} \label{line:far_p}
\Else
    \Comment{label disagrees with spectral information}
    \State $H \gets G_{\tau(u), \sigma(u)}$ \Comment{see \Cref{def:crossgraph}}

    \If{$u \sim_H \mathcal{X}$} \label{line:conntest}\Comment{see \Cref{def:reachability}}
        \State \Return $\sigma(u)$ \Comment{trust the label, $u$ is likely a spectral impostor} \label{line:impostor_p}
    \Else
        \State \Return $\tau(u)$ \Comment{trust spectral information, the label is likely wrong} \label{line:non-impostor_p}
    \EndIf
\EndIf

\end{algorithmic}
\end{algorithm}

\noindent
For \Cref{alg:polytime} to run, we need to feed it $(\widetilde{\mu}_i)_{i \in [k]}$. In \Cref{subsec:apxmeans}, we describe and analyze \Cref{alg:app_centers} and \Cref{alg:permutation}, and show that with good probability they produce valid approximate cluster means $(\widetilde{\mu}_i)_{i \in [k]}$. Together, the three algorithms give the following result.

\begin{thm}[Polynomial time classifier]
\label{thm:polytime}
    There is a polynomial time algorithm that given as input $G,\phi,k,\eta$ as per \Cref{fig:setting} and $\sigma$ as per \Cref{fig:labels} outputs a labeling $\alpha: V\rightarrow [k]$ with misclassification rate $O({d\delta\epsilon}/{\phi^5})$ with probability $0.9$ over the internal randomness of the algorithm and the draw of $\sigma$.
\end{thm}
\noindent
To prove the theorem, we analyze the number of vertices misclassified by \Cref{alg:polytime}. We do so by bounding the number of vertices misclassified in each return point of \Cref{alg:polytime}, namely lines~\eqref{line:core_p},~\eqref{line:far_p},~\eqref{line:impostor_p},~\eqref{line:non-impostor_p}. We begin by bounding the expected number of vertices that can be misclassified in line~\eqref{line:core_p} or~\eqref{line:far_p}.

\begin{lemma}[Misclassification in lines~\eqref{line:core_p} and~\eqref{line:far_p}]\label{lemma:spec_or_cross}
Consider the setting of \Cref{fig:setting}, and let $\iota,\sigma,\tau$ as per \Cref{fig:labels}. The expected fraction of vertices $u \in V$ such that $\tau(u)$ is either $\sigma(u)$ or $*$ and $\sigma(u) \neq \iota(u)$, is at most $O({\delta\epsilon}/{\phi^4})$, i.e.

     \begin{equation*}
        \E_\sigma\left[\left|\{u \in V: \, \tau(u) \in \{\sigma(u),*\} \text{\upshape and } \sigma(u) \neq \iota(u) \}\right| \right] \le 2 \cdot 10^4 \cdot \eta^2 \cdot \frac{\epsilon \delta}{\phi^4 }n \, .
    \end{equation*}

    \end{lemma} 
\begin{proof}
For any $\sigma$, the set of vertices $u \in V$ such that $\tau(u) \in \{\sigma(u),*\} \text{ and } \sigma(u) \neq \iota(u)$ is necessarily contained in $ \{v \in \im \cup \M: \,   \sigma(v) \neq \iota(v)\}$. By linearity of expectation and \Cref{lemma:impostor_n_cross_size}, we get the claim.
\end{proof}

\noindent
Next, we bound the expected number of vertices that can be misclassified in line~\eqref{line:impostor_p}, which can be thought of as the ``false positives'' of the connectivity test in line~\eqref{line:conntest}. Proving this bound relies on the randomness of $\sigma$ to relate the percolation of the connectivity test in line~\eqref{line:conntest} to a Galton-Watson type of process.

\begin{lemma}[Misclassification in line~\eqref{line:impostor_p}]
\label{lem:false_pos_p}
    Consider the setting of \Cref{fig:setting}, and let $\iota,\sigma,\tau$ as per \Cref{fig:labels}. The expected fraction of vertices $u \in V$ such that $\sigma(u)$ is wrong and $u$ is reachable from $\M$ in the $(\tau(u),\sigma(u))$-cross graph is at most $O({d\epsilon\delta}/{\phi^4})$, i.e.
    \begin{equation*}
        \E_\sigma\left[\left|\{u \in V: \,\sigma(u) \notin \{\iota(u), \tau(u)\}, \tau(u) \neq *\text{ \upshape and } u \sim
    _{G_{\tau(u),\sigma(u)}}\M\}\right| \right] \le 6 \cdot 10^4 \cdot \eta^2 \cdot \frac{d \epsilon \delta }{\phi^4 }n\, .
    \end{equation*}
\end{lemma}

\begin{proof} 
    The proof idea can be summarized as follows: (1) consider a breadth-first search starting from the set of cross and impostor vertices $\M\cup \im$ that only traverses an edge if the other endpoint is mislabeled; (2) observe that any mislabeled vertex that can reach $\M$ in its cross graph is visited by this breadth-first search at distance $1$ or higher; (3) think of the breadth-first search as a Galton-Watson process whose offspring distribution has mean $\delta d < 1$, and conclude that we only visit $O(\delta d)|\M\cup \im|$ vertices on top of the starting ones. To make this formal, 
    let $S_1$ be the set of neighbors of cross and impostor vertices, i.e.

    \begin{equation*}
        S_1 = \nei_G(\M \cup \im) \, .
    \end{equation*}
    For a fixed labeling $\alpha$ (a possible realization of the random labeling $\sigma$),  we define the following objects.
    For all $i \in [n]$ we let $S_{i+1}(\alpha)$ be the set of neighbors in $G$ of the mislabeled vertices in $S_{i}(\alpha)$ that are not in any previously defined sets, i.e.
    \begin{equation*}
        S_{i+1}(\alpha) = \nei_G\left(\{u \in S_{i}(\alpha) : \alpha(u)\neq \iota(u)\}\right)\setminus \left(\bigcup_{j =1}^{i} S_j(\alpha)\right) \, .
    \end{equation*}
    Note that the definition of the set $S_1$ does not depend on the labeling $\alpha$ but, for the sake of uniformity of notation, we will still write $S_1(\alpha)$. For convenience, we let $\widehat{S}_i(\alpha)$ be the set of  vertices in $S_i(\alpha)$ with a wrong label, i.e.
    \begin{equation*}
        \widehat{S}_i(\alpha)= \{u \in S_i(\alpha) : \, \alpha(u)\neq \iota(u)\}\, ,
    \end{equation*}
    and let ${S}_{\le i}(\alpha)$ be the set of vertices in any of $S_1(\alpha),\dots,S_i(\alpha)$, i.e.
    \begin{equation*}
        S_{\le i}(\alpha) = \bigcup_{j =1}^i S_j(\alpha) \, ,
    \end{equation*}
    so we can rewrite
    \begin{equation*}
        S_{i+1}(\alpha)=\nei_G(\widehat{S}_i(\alpha))\setminus S_{\le i}(\alpha) \, .
    \end{equation*}
    Similarly, we also let $\widehat{S}_{\le i}(\alpha)$ be the set of vertices in any of $\widehat{S}_1(\alpha),\dots,\widehat{S}_i(\alpha)$, i.e.
    \begin{equation*}
        \widehat{S}_{\le i}(\alpha) = \bigcup_{j =1}^i \widehat{S}_j(\alpha) \, ,
    \end{equation*}
    Next, we relate these level sets to the quantity we want to ultimately bound.
    \begin{claim}\label{claim:sets_S}For any fixed $\alpha: V \rightarrow [k]$, one has
        \begin{equation*}
            \{u \in V: \, \alpha(u) \notin \{\iota(u), \tau(u)\}, \tau(u) \neq * \text{ and } u \sim_{G_{\tau(u),\sigma(u)}} \M\} \subseteq \widehat{S}_{\le n}(\alpha) \cup \{u \in \M: \, \alpha(u) \neq \iota(u)\}  \, .
        \end{equation*}
    \end{claim}
    \begin{proof} We let $G'=(V,E')$ be the subgraph of $G$ obtained by the union of all the $(i,j)$-cross graphs, i.e. we define
    \begin{equation*}
        E' =  \bigcup_{i,j \in [k]: \, i\neq j} E_{ij} \, .
    \end{equation*}
        Clearly, for any vertex $u \in V$ such that $\alpha(u) \notin \{\iota(u), \tau(u)\}, \tau(u) \neq * \text{ and } u \sim_{G_{\tau(u),\alpha(u)}} \M$ one has that $\alpha(u) \neq \iota(u) \text{ and } u \sim_{G'} \M$. Therefore, we prove
        \begin{equation}
        \label{eq:goal}
            \{u \in V: \,\alpha(u) \neq \iota(u) \text{ and } u \sim_{G'} \M\} \subseteq \widehat{S}_{\le n}(\alpha) \cup \{u \in \M: \, \alpha(u) \neq \iota(u)\}  \, .
        \end{equation}
        For any $u \in V$ such that $\alpha(u) \neq \iota(u)$ and $ u \sim_{G'} \M$, there exists a simple path $w_0,\dots,w_t$ for $0\le t \le n-1$ where $w_0 \in \M$, $w_t=u$, and $(w_{r-1},w_{r}) \in E'$ for all $r \in [t]$. Then, for every $t \in \{0,\dots,n-1\}$, we define the set $Z_t$ to be the set of vertices $u \in V$ with $\alpha(u) \neq \iota(u)$ such that there exists a simple path $w_0,\dots,w_t$ where $w_0 \in \M$, $w_t=u$, and $(w_{r-1},w_{r}) \in E'$ for all $r \in [t]$. Also, we define  \smash{$ \widehat{S}_{\le 0}(\alpha) \coloneqq \emptyset$}. We prove~\eqref{eq:goal} by showing by induction that for all $t \in \{0,\dots,n-1\}$ one has
        \begin{equation*}
            Z_t \subseteq  \widehat{S}_{\le t}(\alpha) \cup \{u \in \M: \, \alpha(u) \neq \iota(u)\} \, .
        \end{equation*}
        For the base case, let $u \in V$ be any vertex with  $\alpha(u)\neq \iota(u)$ which is also connected to $\M$ in $G'$ by a path of length $0$. Trivially, we have \smash{$u \in \{v \in \M: \, \alpha(v) \neq \iota(v)\}$}. For the inductive step, let $0\le t \le n-2$ and assume that \smash{$Z_t \subseteq  \widehat{S}_{\le t}(\alpha) \cup \{u \in \M: \, \alpha(u) \neq \iota(u)\}$}. Let \smash{$w \in Z_{t+1}\setminus Z_t$} be a vertex with \smash{$\alpha(w)\neq \iota(w)$} that is connected to \smash{$\M$} in \smash{$G'$} by a path of length $t+1$. Then, there is a vertex $u \in V$ connected to \smash{$\M$} in \smash{$G'$} by a path of length at most $t$ such that $w$ is a neighbor of \smash{$u$ in $G'$}. In particular, $w$ is a neighbor of $u$ in $G$. We then consider two cases.
        \begin{itemize}
            \item If $\alpha(u)\neq \iota(u)$, then $u \in Z_t$ and hence \smash{$u \in \widehat{S}_{\le t}(\alpha) \cup \{v \in \M: \, \alpha(v) \neq \iota(v)\}$} by the inductive hypothesis. Thus, \smash{$w \in S_{\le t+1}(\alpha)$}. Since $\alpha(w)\neq i(w)$ by assumption, we have that \smash{$w \in \widehat{S}_{\le t+1}(\alpha)$}.
            \item Suppose $\alpha(u)= \iota(u) = j$ for some $j \in [k]$. Note that since $u$ is not isolated in $G'$, $u \notin \spec(j)$. Hence, either $u \in \M$ or there exists $i \in [k]$ such that $i\neq j$ and $u \in \lab(j) \cap \spec(i)$. Hence, $u \in \im \cup \M$, which implies that $w \in S_1(\alpha)$. Recalling that $\alpha(w)\neq i(w)$, we conclude  $w \in \widehat{S}_1(\alpha)$.
        \end{itemize}
    \end{proof}
    \noindent
    By virtue of the above claim, it suffices to bound \smash{$|\widehat{S}_{\le n}(\alpha) \cup \{v \in \M: \, \alpha(v) \neq \iota(v)\}|$}. By linearity of expectation, the expected number of vertices in $\{u \in \M: \alpha(u)\neq \iota(u)\}$ is easily upper-bounded by $\delta |\M|$. Then, we only need to bound \smash{$|\widehat{S}_{\le n}(\alpha)|$}. To do so, we observe that for each $i \in [n]$, the expected size of \smash{$\widehat{S}_{i}(\alpha)$} drops by a factor of $\delta d$ at every level.
    \begin{restatable}{claim}{expdrop}
    \label{claim:exp_drop}
    For all $i =2,\dots,n-1$, one has $\E_\sigma[|\widehat{S}_{i+1}(\sigma)|] \le \delta d \cdot \E_\sigma[|\widehat{S}_{i}(\sigma)|]$.
    \end{restatable}
    \noindent
    We defer the proof of \Cref{claim:exp_drop} to \Cref{subsec:expdrop}, but the idea is simple: \smash{$\widehat{S}_{i+1}$} essentially consists of the mislabeled neighbors of \smash{$\widehat{S}_{i}$}, and since the graph is $d$-regular and each neighbor is mislabeled with probability $\delta$, we get a factor $\delta d$ drop. Then,
by virtue of \Cref{claim:exp_drop} and by definition of $S_1$, we have
    \begin{equation*}
        \E_\sigma\left[|\widehat{S}_i(\sigma)|\right] \le (\delta d)^i \cdot |\M\cup\im|
    \end{equation*}
    for every $i \in [n]$. Then, by linearity of expectation, we get
    \begin{equation*}
        \E_\sigma\left[|\widehat{S}_{\le n}(\sigma)|\right] \le |\M\cup\im| \sum_{i =1}^n (\delta d)^i \le 2 \delta  d|\M \cup \im| \, .
    \end{equation*}
    By \Cref{lemma:impostor_n_cross_size}, we conclude the proof.
\end{proof}

\noindent
Finally, we bound the number of vertices that can be misclassified in line~\eqref{line:non-impostor_p}. These vertices can be thought of as the ``false negatives'' of the connectivity test in line~\eqref{line:conntest}. To show this last bound, we rely on the randomness of $\sigma$ in order to apply Markov's inequality to the number of mislabeled impostors.

\begin{lemma}[Misclassification in line~\eqref{line:non-impostor_p}]
\label{lem:false_neg_p}
    Consider the setting of \Cref{fig:setting}, and let $\iota,\sigma,\tau$ as per \Cref{fig:labels}. The expected fraction of vertices over the draw of $\sigma$ such that $\tau(u) \in [k]$ is wrong but $\sigma(u)$ is correct and $u$ is not reachable from $\M$ in the $(\tau(u),\sigma(u))$-cross graph is at most $O(\epsilon\delta/\phi^5)$, i.e.
    \begin{equation*}
        \E_\sigma\left[\left|\{u \in V: \, \tau(u) \notin \{*,\sigma(u)\} \text{ and }  \sigma(u) = \iota(u)  \text{ and } u \slashed{\sim}_{G_{\tau(u),\sigma(u)}} \M\}\right| \right] \le 4 \cdot 10^4 \cdot \eta^2 \cdot \frac{\epsilon \delta}{\phi^5 }n \, .
    \end{equation*}
    In other words, at most $4 \cdot 10^4 \cdot \eta^2 \cdot {\epsilon \delta}/{\phi^5 } \cdot n$ correctly labeled impostor vertices $u$ are not connected to the set $\M$ in the cross graph $G_{\tau(u), \sigma(u)}$, in expectation over $\sigma$.
\end{lemma}

\begin{proof} The main idea of this proof is as follows. Recall that in \Cref{lem:imposters_expand} we showed that for all $i \neq j$ the set $\im(i, j)\setminus \nei_G(\M)$ has certain expander properties: every subset $S$ of $\im(i, j)\setminus \nei_G(\M)$ has many edges going to $\im(i, j)\setminus S$. From here, we get that all vertices in $\im(i, j)$ are connected to $\M$ in the graph induced by $\im(i, j)\cup \M$. Moreover, as we show in \Cref{claim:reachability_p}, removing an arbitrary set of vertices $R$ from $\im(i, j)$ prevents only a few vertices in $\im(i, j)\setminus R$ from reaching $\M$ in the graph induced by $(\im(i, j)\setminus R)\cup \M$. If we choose $R \coloneqq \im(i, j) \cap \mislabeled$ (recall from  \Cref{def:mislabeled} that $\Lambda$ denotes the set of misclassified vertices), then the above implies the following: all the correctly labeled vertices in $\im(i, j)$, except for a small subset, are connected to the set $\M$ only through correctly labeled vertices in $\im(i, j)$. This is even stronger than the statement of \Cref{lem:false_neg_p}, as for every correctly labeled $u \in \im(i, j)$ we have that the cross graph $G_{i, j}$ contains both correctly labeled vertices from $\im(i, j)$ and some mislabeled vertices from $\spec(i)$.

We now formalize the above argument. 
First, we introduce notation for the induced graphs which we discussed above. Let $E^*_{i,j}=\{(u,v) \in E: \, u,v \in \im(i,j) \cup \M\}$ and let $G^*_{i, j} = (\im(i, j) \cup \M, E^*_{i, j})$ be the subgraph of $G$ induced by $\im(i, j) \cup \M$. Additionally, for a set of vertices $R \subseteq V$, we denote by $E^*_{i,j}\setminus R $ the subset of edges of $E^*_{i,j}$ except all of the edges incident on $R$ (i.e. $E^*_{i,j}\setminus R = E^*_{i,j}\setminus \{(u,v) \in E: u \in R \text{ or } v \in R\} $), and define $G^*_{i,j}\setminus R=((\im(i, j) \cup \M)\setminus R ,E^*_{i,j}\setminus R)$~-- the subgraph of $G$ induced by $(\im(i, j) \cup \M)\setminus R$. With this notation, we first argue that removing a few vertices from $G^*_{i,j}$ only causes a few more to become disconnected from $\M$.
\begin{claim}\label{claim:reachability_p} Let $i,j \in [k] $ with $i \neq j$ and let $R \subseteq \im(i,j) $. Then, number of vertices in $\im(i,j)$ that are not reachable from $\M$  in $G^*_{i,j}\setminus R$ is at most ${2|R|}/{\phi}$.
\end{claim}
\begin{proof}
Let $S \subseteq \im(i,j)\setminus R$ denote the set of vertices that are isolated from $\M$ in $G^*_{i,j}\setminus R$, i.e.
\begin{equation*}
    S = \left\{ u \in \im(i,j)\setminus R: \, u \slashed{\sim}_{G^*_{i,j}\setminus R} \M\right\} \, .
\end{equation*}
First, we argue that $S \cap \nei_G(\M) = \emptyset$. Indeed,  every vertex in $\nei_G(\M) \setminus R$ has a neighbor in $\M$. But from the definition of $S$, no vertex in $S$ can have a neighbor in $\M$. Therefore, $S \cap \nei_G(\M) = \emptyset$.  Hence, we have $S \subseteq \im(i,j) \setminus \nei_G(\M)$, which allows us to apply \Cref{lem:imposters_expand} and get 

$$|E(S, \im(i,j) \setminus S)|\geq \frac{\phi}{2}d|S| \, .$$
\noindent
Note that $\im(i,j) \setminus S $ consists of two types of vertices: the vertices reachable from $\M$ in $G^*_{i,j}\setminus R$, and the vertices in $R$.  Observe that $S$ cannot have any edges to the vertices reachable from $\M$ in $G^*_{i,j}\setminus R$, or else there would be at least one vertex in $S$ which is reachable from $\M$ which contradicts the definition of $S$. So $ E(S, \im(i,j)  \setminus S) \subseteq  E(S, R).$ Therefore, we have
\[\frac{\phi}{2}|S|d \leq |E(S, \im(i,j)  \setminus S)| \le |E(S, R)| \leq |R|d \, ,\]
which rearranges to $|S| \leq {2|R|}/{\phi}.$

\end{proof}
\noindent
We define the set $R_{i,j} = \im(i, j)\cap\mislabeled$ of mislabeled $(j \to i)$-impostors for each $i,j \in [k]$ with $i \neq j$. Just as discussed in the proof sketch, we would like to apply \Cref{claim:reachability_p} with $R \coloneqq R_{i,j}$ to each of the graphs $G^*_{i,j}$. We do this in \Cref{claim:proxy} and derive from that the statement of \Cref{lem:false_neg_p}. 

\begin{claim}
\label{claim:proxy}
   Suppose $u \in V$ is such that $\tau(u) \notin \{*, \sigma(u)\}$, $ \sigma(u) = \iota(u)$, and \smash{$ u \slashed{\sim}_{G_{\tau(u), \sigma(u)}} \M$}.  For $i = \tau(u)$ and $j = \iota(u)$ we have that \smash{$u \slashed{\sim}_{G^*_{i,j}\setminus R_{i,j}} \M$} and $u \in \im(i,j)$.
\end{claim}

\begin{proof}
    By the statement of the claim, $u$ is a correctly labeled impostor, and so, by the choice of $i$ and $j$, $u \in \im(i,j)$. Since $G^*_{i, j}\setminus R_{i, j}$ is the subgraph of $G$ induced by $(\im(i, j) \cup \M)\setminus R_{i,j}$ we have that $E_{\tau(u), \sigma(u)} \supseteq E^*_{i,j}\setminus R_{i,j}$, as $G_{\tau(u), \sigma(u)}$ contains $G^*_{i, j}\setminus R_{i, j}$ as a subgraph. Hence, the claim follows. 
\end{proof}

\noindent
By virtue of \Cref{claim:proxy} and \Cref{claim:reachability_p}, for any fixed outcome of $\sigma$ we have
\begin{align*}
    \left|\{u \in V: \, \tau(u) \notin \{*,\sigma(u)\},\, \sigma(u) = \iota(u)\, , u \slashed{\sim}_{G_{\tau(u), \sigma(u)}} \M\}\right| & \le \sum_{\substack{i,j \in [k]:\\i\neq j}} \left|\left\{ u \in \im(i,j): u \slashed{\sim}_{G^*_{i,j}\setminus R_{i,j}} \M\right\}\right|  && \text{by \Cref{claim:proxy}} \\
    & \le \sum_{\substack{i,j \in [k]:\\i\neq j}} \frac{2}{\phi }|R_{i,j}| && \text{by \Cref{claim:reachability_p}} \\
    & = \frac{2}{\phi }|\im \cap \mislabeled| \, .
\end{align*}
Finally, taking the expectation over the random labeling $\sigma$, we get
\begin{align*}
    \E_{\sigma}\left[  \left|\{u \in V: \, \tau(u) \notin \{*,\sigma(u)\},\, \sigma(u) = \iota(u)\, , u \slashed{\sim}_{G_{\tau(u), \sigma(u)}} \M\}\right|\right] & \le \frac{2}{\phi}\E_{\sigma}\left[|\im\cap \mislabeled|\right] \\
    & \le \frac{2\delta}{\phi}\left|\im\right| \qquad && \text{by \Cref{fig:labels}} \\
    & \le 4 \cdot 10^4 \cdot \eta^2 \cdot \frac{\epsilon \delta}{\phi^5 }n && \text{by \Cref{lemma:impostor_n_cross_size}.}
\end{align*}

\end{proof}

\noindent
Combining the above three lemmas, we can conclude a proof of the main result of the section.
\begin{proof}[Proof of \Cref{thm:polytime}]

The polynomial time algorithm which we refer to in the statement of \Cref{thm:polytime} consists of two blocks.
    \begin{enumerate}
        \item First, compute our approximate inner product function via \Cref{rem:euclidian_oracle} with $\xi = 1/n^{10}$ so as to meet the requirement of \Cref{fig:setting}.
        \item Then, we run \Cref{alg:app_centers} and \Cref{alg:permutation} (see Appendix \ref{subsec:apxmeans}) to obtain $\tmu_1,\dots,\tmu_k$\footnote{In reality,  \Cref{alg:app_centers} and \Cref{alg:permutation} do not return $k$ vectors but rather $k$ vertices $u_1,\dots,u_k$ that indicate the embedding that should be used for the corresponding approximate cluster mean. This is because we do not have explicit access to vectors, but only inner product access. Hence, \Cref{alg:polytime} will query $\langle f_{u_i}, \cdot\rangle_\apx$ whenever it wants to access $\langle \tmu_i, \cdot\rangle_\apx$. We write $\mu_1,\dots,\mu_k$ for readability.};
        \item Finally, we launch \Cref{alg:polytime} with $\tmu_1,\dots,\tmu_k$ as input to classify each of the vertices $v \in V$. Let $\alpha : V\rightarrow [k]$ be the mapping of a vertex $u \in V$ to the cluster id output by \Cref{alg:polytime} on input $u$.
    \end{enumerate}

\paragraph{Probability of success and mislassification rate.}

Since both \Cref{alg:polytime} and \Cref{alg:permutation} use the randomness of $\sigma$, we should be careful while analyzing the probability of success  of the above algorithm.  \Cref{alg:app_centers} and \Cref{alg:permutation} use additional random bits (for sampling) different from the random bits of $\sigma$, and therefore this randomness is independent from the randomness of $\sigma$. \Cref{lemma:pi_computation} shows that with probability $0.97$ over the randomness of \Cref{alg:app_centers} and \Cref{alg:permutation} one has that $\tmu_1,\dots,\tmu_k$ meet the condition of \Cref{def:approxmeans} with probability $0.999$ over $\sigma$.

Let us then fix the internal randomness \Cref{alg:app_centers} and \Cref{alg:permutation} such that $\tmu_1,\dots,\tmu_k$ meet the condition of \Cref{def:approxmeans} with probability $0.999$ over $\sigma$. One can note that $\sigma$ only influences the ordering of the $k$ vectors $\tmu_1,\dots,\tmu_k$: this is because the set of $k$ vectors is computed by \Cref{alg:app_centers} and \Cref{alg:permutation} computes the permutation $\pi$ that gives their ordering, but \Cref{alg:app_centers} does not use $\sigma$ at all. Hence, when the randomness of \Cref{alg:app_centers} and \Cref{alg:permutation} is fixed, the labeling $\alpha$ constructed with \Cref{alg:polytime} only depends on $\sigma$ and on the permutation $\pi$ produced by \Cref{alg:permutation}. Let $\textsc{Fail}(\sigma, \pi(\sigma))$ denote the event that the labeling $\alpha$ is wrong on more than $10^9 \eta^2 d\epsilon\delta/\phi^5 \cdot n$ vertices when \Cref{alg:polytime} is given as input the labels and $\tmu_1,\dots,\tmu_k$ which are indexed using the permutation $\pi$ (which is an output of \Cref{alg:permutation} and therefore it depends on $\sigma$). Observe that \Cref{alg:polytime} itself is deterministic, and therefore $\textsc{Fail}$ is a deterministic function of a labeling and a permutation. Let $\pi^*$ denote the correct permutation for the set of vectors produced by \Cref{alg:app_centers}. 
\begin{equation}\label{eq:pr_justification_poly}
    \begin{aligned}
    & \hphantom{=~} \Pr_{\sigma}[\textsc{Fail}(\sigma, \pi(\sigma))] \\
     & = \Pr_{\sigma}\left[\textsc{Fail}(\sigma, \pi(\sigma) ) \middle| \pi(\sigma) = \pi^*\right]\cdot \Pr_{\sigma}\left[\pi(\sigma) = \pi^*\right] + \Pr_{\sigma}\left[\textsc{Fail}(\sigma, \pi(\sigma) ) \middle| \pi(\sigma) \neq \pi^*\right]\cdot\Pr_{\sigma}[\pi(\sigma) \neq \pi^*] \\
    &\leq \Pr_{\sigma}\left[\textsc{Fail}(\sigma, \pi(\sigma) )\middle| \pi(\sigma) = \pi^*\right] + \Pr_{\sigma}[\pi(\sigma) \neq \pi^*] \\
    & \leq 2\Pr_{\sigma}[\textsc{Fail}(\sigma, \pi(\sigma) ) \land \pi(\sigma) = \pi^*] + \Pr_{\sigma}[\pi(\sigma) \neq \pi^*] \\
    & =2\Pr_{\sigma}[\textsc{Fail}(\sigma, \pi^*) \land \pi(\sigma) = \pi^*] + \Pr_{\sigma}[\pi(\sigma) \neq \pi^*] \\
    & \leq 2\Pr_{\sigma}[\textsc{Fail}(\sigma, \pi^*)] + \Pr_{\sigma}[\pi(\sigma) \neq \pi^*] \\
    & \leq 2\Pr_{\sigma}[\textsc{Fail}(\sigma, \pi^*)] +  10^{-3} \quad  \quad \text{by the fixed internal randomness \Cref{alg:app_centers} and \Cref{alg:permutation}.}
\end{aligned}
\end{equation}
 Observe that $\Pr_{\sigma}[\textsc{Fail}(\sigma, \pi^*)]$ is precisely the probability of failure of \Cref{alg:polytime}. To bound it, we now define four groups of vertices $W_1,W_2,W_3,W_4$ that together cover the set of vertices $u \in V$ such that $\alpha(u)\neq \iota(u)$.
    \begin{itemize}
        \item  Let $W_1$ be the set of vertices $u \in V$ such that $\tau(u) \in \{*,\sigma(u)\}$ and $\alpha(u)\neq \iota(u)$. For every $u \in W_1$, we know that $\alpha(u)=\sigma(u)$ by definition of \Cref{alg:polytime} (see the return statements in lines~\eqref{line:core_p} and~\eqref{line:far_p}). Then, each $u \in W_1$ satisfies $\tau(u) \in \{*,\sigma(u)\}$ and $\sigma(u)\neq \iota(u)$, so by virtue of \Cref{lemma:spec_or_cross} we know that the expected size of $W_1$ is at most $2 \cdot 10^4 \cdot \eta^2\phi^{-4} \cdot \epsilon \delta \cdot n$.
        \item Let $W_2$ be the set of vertices $u \in V$ such that  $\tau(u) \notin \{*,\sigma(u)\}$, $ u \sim_{G_{\tau(u),\sigma(u)}}\M $ and $\alpha(u)\neq \iota(u)$. For every $u \in W_2$, we know that $\alpha(u)=\sigma(u)$ by definition of \Cref{alg:polytime} (see the return statement in line~\eqref{line:impostor_p}). Then, each $u \in W_2$ satisfies $\sigma(u) \notin \{\iota(u), \tau(u)\}, \tau(u) \neq *$ and \smash{$ u \sim_{G_{\tau(u),\sigma(u)}}\M$}, so by \Cref{lem:false_pos_p} we know that the expected size of $W_2$ is at most \smash{$6 \cdot 10^4 \cdot \eta^2\phi^{-4} \cdot d\epsilon \delta \cdot n$}.
        \item Let $W_3$ be the set of vertices $u \in V$ such that $\tau(u) \notin \{*,\sigma(u)\}$, \smash{$ u \slashed{\sim}_{G_{\tau(u),\sigma(u)}}\M $}, $\sigma(u)=\iota(u)$, and $\alpha(u)\neq \iota(u)$. Then, each $u \in W_3$ satisfies the three conditions $\tau(u) \notin \{*,\sigma(u)\}$, $\sigma(u)=\iota(u)$,  and \smash{$ u \slashed{\sim}_{G_{\tau(u),\sigma(u)}}\M $}, so by \Cref{lem:false_neg_p} we know that the expected size of $W_3$ is at most \smash{$4 \cdot 10^4 \cdot \eta^2\phi^{-5} \cdot \epsilon \delta \cdot n$}.
        \item Let $W_4$ be the set of vertices $u \in V$ such that $\tau(u) \notin \{*,\sigma(u)\}$, \smash{$ u \slashed{\sim}_{G_{\tau(u),\sigma(u)}}\M $}, $\sigma(u)\neq\iota(u)$, and $\alpha(u)\neq \iota(u)$. For every $u \in W_4$, we know that $\alpha(u)=\tau(u)$ by definition of \Cref{alg:polytime} (see the return statement in line~\eqref{line:non-impostor_p}). Then, each vertex $u \in W_4$ satisfies $\tau(u)\neq \iota(u)$ and $\sigma(u)\neq\iota(u)$, which implies that $W_4 \subseteq \{u \in \im: \, \sigma(u) \neq \iota(u)\}$. By linearity of expectation and \Cref{lemma:impostor_n_cross_size}, the expected size of $W_4$ is at most $2 \cdot 10^4 \cdot \eta^2\phi^{-4} \cdot \epsilon \delta \cdot n$.
    \end{itemize}
    Thus, by a simple application of Markov's inequality we have an algorithm that misclassifies at most a \smash{$O({d\epsilon\delta}/{\phi^5})$} fraction of vertices with probability
    \[\Pr_{\sigma}[\textsc{Fail}(\sigma, \pi^*)] \leq 10^{-3} \, .\]
    By~\eqref{eq:pr_justification_poly}, we arrive at
    \[  \Pr_{\sigma}[\textsc{Fail}(\sigma, \pi(\sigma))] \leq 2\cdot10^{-3} +  10^{-3} \leq 0.99 \, .\]
    Since by \Cref{lemma:pi_computation} the probability that the internal randomness of \Cref{alg:app_centers} and \Cref{alg:permutation} meets our requirement is $0.97$, we get an overall success probability of $0.96$ as required.
\paragraph{Running time.} The running time of the algorithm is dominated by by the time needed to approximate $\langle f_u, f_v\rangle$ for all pairs $u, v \in V$ and the time needed to check whether the set $\M$ is reachable from a vertex $u$ for all $u \in V$. We approximate the spectral inner products with the polynomial-time algorithm in \Cref{rem:euclidian_oracle}, and we check reachability with the BFS algorithm.
\end{proof}

\section{Sublinear-time classifier with near-optimal rate}\label{sec:sub_alg}
In this section we outline and analyze a sublinear time algorithm that given the graph, the approximate cluster means, and the perturbed label, outputs a cluster id for each vertex. We will show that the number of vertices such that the output id is different from $\iota(u)$ is at most a ${\epsilon\delta/\phi^6 \cdot\log(1/\delta)}$ fraction.

\begin{algorithm}
\caption{Sublinear-time classifier}\label{alg:sketch+labels}
\begin{algorithmic}[1]
\State \textbf{Input:} $G, \phi, \eta, (\widetilde{\mu}_i)_{i \in [k]}$ as per \Cref{fig:setting}, $\sigma$ as per \Cref{fig:labels}, and a vertex $u \in V$
\State \textbf{Output:} a cluster id in $[k]$

\If{$\tau(u) = \sigma(u)$}
    \State \Return $\sigma(u)$ \Comment{Label agrees with spectral information} \label{line:core}
\ElsIf{$\tau(u) = *$}
    \State \Return $\sigma(u)$ \Comment{Spectral information is ambiguous} \label{line:far}
\Else
    \Comment{Label disagrees with spectral information}

    \If{\Call{\robcon}{$G, \phi, \eta, (\widetilde{\mu}_i)_{i \in [k]}, \sigma, u$} returns ``yes''} \Comment{See \Cref{alg:rob_con}}\label{line:if_impostor}
        \State \Return $\sigma(u)$ \Comment{Trust the label, $u$ is likely a spectral impostor} \label{line:impostor}
    \Else \label{line:if_non-impostor}
        \State \Return $\tau(u)$ \Comment{Trust spectral information, the label is likely wrong} \label{line:non-impostor}
    \EndIf
\EndIf

\end{algorithmic}
\end{algorithm}

\noindent
Note that \Cref{alg:sketch+labels} is very similar to \Cref{alg:polytime}. For a vertex $u$, both \Cref{alg:sketch+labels} and \Cref{alg:polytime} first check if the label and the spectral information about $u$ indicate the same cluster and whether $u$ is a cross vertex (as per \Cref{def:spec}). After this, \Cref{alg:polytime} simply checks whether there exists a walk in the graph $G_{\tau(u), \sigma(u)}$ connecting $u$ to some vertex in $\nei_G(\M)$. In the analysis of \Cref{alg:polytime}, \Cref{lem:false_pos_p} and \Cref{lem:false_neg_p} informally assert the following: for all except for $\approx d\epsilon\delta/\poly(\phi) \cdot n$ many vertices $u$ that are not cross vertices and for which the spectral information $\tau(u)$ disagrees with the label $\sigma(u)$, there exists a walk in the graph $G_{\tau(u), \sigma(u)}$ connecting $u$ and some vertex in $\nei_G(\M)$ if and only if $\sigma(u) = \iota(u)$. As it turns out, a stronger version of this statement is true: 

\begin{lemma*}[Informal -- see \Cref{claim:false_pos} and \Cref{cor:false_neg_updated}]
    There exists  a set $\NM \subseteq \nei_G(\M)$ of a much smaller size than $\nei_G(\M)$ such that the following is true: for all except $\approx \epsilon\delta\log(1/\delta)n$ many vertices $u$ that are not cross vertices and for which the spectral information $\tau(u)$ disagrees with the label $\sigma(u)$, one has that
\begin{itemize}
    \item if $\sigma(u) = \iota(u)$, then \emph{almost all} of bounded-length walks in $G_{\tau(u), \sigma(u)}$ starting at $u$ pass through $\NM$;
    \item if $\sigma(u) \neq \iota(u)$ then \emph{almost no} bounded-length walk in $G_{\tau(u), \sigma(u)}$ starting at $u$ passes through $\NM$.
\end{itemize}
\end{lemma*}

\noindent
We formalize and prove this statement in \Cref{subsec:analysis}.
Intuitively, we think of $\nei_G(\M)$ as a set of ``problematic'' vertices: the vertices in the set $\M$ do not bear any spectral information about which cluster they belong to, and they ``confuse'' their immediate neighbors too. In \Cref{subsec:regnei} we show that all of the actually ``problematic'' vertices of $\nei_G(\M)$ are contained is a subset of much smaller size, referred to as $\NM$ in the above. This observation allows us to improve  the upper bound on the misclassification rate from $\approx d\epsilon\delta$ to $\approx \epsilon\delta\log(1/\delta)$.  Therefore, we replace the reachability test in \Cref{alg:polytime} with a \emph{robust connectedness vs separatedness} test, which amounts to sampling a few walks starting at $u$ in the graph $G_{\tau(u), \sigma(u)}$ and using them to approximate the total number of walks passing through $\textsc{S}$. We formally define this procedure, referred to as \robcon, and prove its performance guarantees in \Cref{sec:robcon}. This modification allows us to replace the expensive graph connectivity test with sublinear time sampling.
\\~\\
We formally state the performance guarantees of \Cref{alg:sketch+labels} in \Cref{thm:sublinear}.

\begin{thm}[Sublinear-time and -space classifier]
\label{thm:sublinear}
There is an algorithm running in time \smash{$d \cdot 2^{O({\phi^2}/{\epsilon}) k^4 \log k}$} \smash{$n^{1/2+O(\epsilon/\phi^2)}\cdot \poly(1/(\xi \phi))\poly(\log (n/\delta)) $} that, given as input $G,d,n,\eta,\epsilon,\phi,k,\delta$ as per \Cref{fig:setting} and $\sigma$ as per \Cref{fig:labels}, produces a data structure of size $d\cdot n^{1/2+O(\epsilon/\phi^2)}\poly(k/(\phi \xi))\poly(\log (n/\delta))$ with query time \smash{$d \cdot n^{1/2+O(\epsilon/\phi^2)}\cdot\poly(k/(\phi \xi))\poly(\log (n/\delta))$}. With probability $0.95$ over the internal randomness of the algorithm and the draw of $\sigma$, there exists a subset of vertices $B \subseteq V$ of size $|B| \leq O({\epsilon\delta} \cdot {\phi^{-6}} \log(1/\delta)) \cdot n$ such that for every vertex $u \in V \setminus B$, querying the data structure on $u$ returns a cluster id $i \in [k]$ such that $u \in C_i$.

\end{thm}

\begin{remark}[Additional bounds on parameters]\label{rem:additional_bounds_eps} Observe that the above algorithm has truly sublinear time and space under the assumption
\[\frac{\phi^2}{\epsilon}\cdot k^4\log k \leq c \log n \]
for a small enough constant $c>0$, which implies both an upper bound on $k$ and a lower bound on $\epsilon$. In particular, the above is satisfied when $k \leq \log^{c_1} n$ for a sufficiently small constant $c_1$, and $\frac{\epsilon}{\phi^2} \geq 1/\log^{c_2} n$ for a sufficiently large constant $c_2$. We do not put these assumptions in the \Cref{fig:setting} as they are not necessary for anything apart from ensuring the data structure from \cite{GKLMS21} has sublinear time and space.  Furthermore, note that if the above assumptions are not satisfied, we may opt for the Euclidean inner product oracle introduced in \Cref{rem:euclidian_oracle} and used in \Cref{sec:poly} in design of a simpler algorithm. While this makes the running time of \Cref{alg:sketch+labels} polynomial, \Cref{alg:sketch+labels} still has merit as its misclassification rate guarantee is independent of $d$.
\end{remark}

\noindent
The rest of the section is structured as follows: in \Cref{sec:robcon},  we formally define the set $\NM$, present the \robcon{} algorithm, and prove its performance guarantees; in \Cref{subsec:regnei}, we establish the key properties of the set $\NM$; finally, in \Cref{subsec:analysis}, we prove \Cref{thm:sublinear}.

\subsection{\robcon}\label{sec:robcon}
Before we describe \robcon, we introduce several notions we will use in the discussion.

\begin{definition}[Lazy walk]
    Let $H$ be a graph, and let $l \ge 1$ be an integer. A sequence of $l$ vertices $\bf{w}$ $\in V(H)^l$ is a $l$-step lazy walk in $H$ if for all $i \in [l-1]$ one has either $(w_i,w_{i+1}) \in E(H)$ or $w_i=w_{i+1}$.
\end{definition}

\begin{definition}[Lazy random walks]\label{def:lazy_r_walk}
    Let $d\ge 3$ be an integer, let $H$ be a $d$-regular graph, let $A$ be its adjacency matrix, let $M = \frac{1}{2}I+\frac{1}{2d}A$, and let $l\ge 1$ be an integer. We define the $l$-step lazy random walk distribution in $H$, denoted $p^l[H]$, to be the following distribution over $l$-step lazy walks $\bf{w}$ in $H$: 
    \begin{itemize}
        \item $w_1$ is distributed uniformly over $V(H)$;
        \item  for each $i \in [l-1]$, $w_{i+1}$ conditioned on $w_{i}$ is distributed as $M \1_{w_{i}}$.
    \end{itemize} Similarly, for a set $S \subseteq V(H)$, we define the $l$-step lazy random walk distribution in $H$ starting from $S$, denoted $p_S^l[H]$, to be the following distribution over $l$-step lazy walks $\bf{w}$ in $H$: $w_1$ is distributed uniformly over $S$ and $w_{i+1}$ conditioned on $w_{i}$ is distributed as $M \1_{w_{i}}$ for each $i \in [l-1]$.
\end{definition}

\begin{nota}
    For $H, l ,S$ as in the definition above, if $S=\{u\}$ for some $u \in V(H)$, we allow ourselves to abuse notation and write $p^l_{u}[H]$ instead of $p_{\{u\}}^l[H]$.
\end{nota}

\begin{definition}[Walks reaching  a set]\label{def:crosswalk_new}
Let $H$ be a graph, let $l \ge 1$ be an integer, and let $\bf{w}$ be a $l$-step lazy walk in $H$, and let $T \subseteq V(H)$.  We say that $\bf{w}$ reaches $T$ if there exists $i \in [l]$ such that $w_i \in T$. If $\bf{w}$ reaches $T$ we write $\bf{w}$ $\sim T$, otherwise we write $\bf{w}$ $\slashed{\sim} \ T$.
\end{definition}

As we mentioned at the beginning of \Cref{sec:sub_alg}, \robcon \ (\Cref{alg:rob_con}) is based on testing reachability to a certain special set $\textsc{S} \subseteq \nei_G(\M)$. We define $\textsc{S}$ as $\M$ itself and the \emph{regularized neighborhood} of $\M$, denoted $\NM$, which we formally introduce in \Cref{def:NM} below. 

\begin{definition}[Regularized neighborhood of $\M$]\label{def:NM}
Let $\phi$, $(f_u)_{u \in V}$ and $\tmu_1,\dots,\tmu_k$ be as per \Cref{fig:setting}. 
    We define the regularized neighborhood of $\M$, which we denote $\NM $, to be the vertices in $\nei_G(\M)$ that belong to some spectral cluster $\spec(i)$ but have at least $\frac{\phi}{2} d$ neighbors in other spectral clusters, i.e.
    $$\NM \coloneqq \left\{ u \in \nei_G(\M) \setminus \M : |\nei_G(u) \cap \spec(\lnot \tau(u))| \geq \frac{\phi}{2} d\right\}.$$
\end{definition}

\begin{definition}[Crossing walk]\label{def:crosswalk} Let $L$ as per \Cref{fig:setting}. Fix $i$ and $j$, $i \neq j$. We say that a walk $\bf{w}$ in $G_{i, j}$ (as per \Cref{def:crossgraph}) starting from a vertex $u \in \specapx(i)$ of length $L$ is a {\em crossing walk} if $\bf{w}$ $\sim \Mapx \cup \NMapx$, i.e. if $\bf{w}$ reaches $ \Mapx \cup \NMapx$ (as per \Cref{def:crosswalk_new}). 
\end{definition}

\begin{algorithm}
\caption{\textsc{\robcon}}\label{alg:rob_con}
\begin{algorithmic}[1]
\State \textbf{Input:} $G, \phi, \eta, \{\widetilde{\mu}_{i}\}_{[k]}, L$ as per \Cref{fig:setting}, $\sigma, \tau$ as per \Cref{fig:labels}, a vertex $u$
\State \textbf{Output:} ``yes'' or ``no''

\State $\langle \cdot, \cdot \rangle_{\text{apx}} \gets \specdp$ \Comment{Dot product oracle from Thm.~\ref{thm:spec_dot_prod_oracle}} 
\State $r \gets 0$ \Comment{Counter of crossing walks}

\For{$i = 1$ to $450 \log n$}
\State $\mathbf{w}_i \gets$ lazy random walk of length $L$ in $G_{\tau(u), \sigma(u)}$ started at $u$ \label{line:sample_walk_cross}

    \For{every vertex $w$ in $\mathbf{w}_i$} \label{line:sample_from_cross}
        \If{$\|f_w - \widetilde{\mu}_{\tau(u)}\|^2_{\text{apx}} > \frac{\phi^2}{400\eta} \|\widetilde{\mu}_{\tau(u)}\|^2_{\text{apx}}$} \label{line:M}
            \State $r \gets r + 1$\label{line:count_r_1} \Comment{$w \in \M$, so $\mathbf{w}_i$ is a crossing walk }
            \State \textbf{break} \label{line:M_end}
        \EndIf

        \State $q \gets 0$ \Comment{Counter of neighbors in $\spec(\neg \tau(u))$} \label{line:NM_begin}
        \For{every neighbor $v \in \nei_G(w)$}
            \If{for some $\widetilde{\mu}_l \neq \widetilde{\mu}_{\tau(u)}$, it holds that $\|f_v - \widetilde{\mu}_l\|^2_{\text{apx}} \leq \frac{\phi^2}{400\eta} \|\widetilde{\mu}_l\|^2_{\text{apx}}$} \label{line:far_neighbor}
                \State $q \gets q + 1$ \Comment{$v \in \spec(\neg \tau(u))$}
            \EndIf
        \EndFor \label{line:NM:end}

        \If{$q \geq \frac{\phi}{2} d$} \label{line:NM}
            \State $r \gets r + 1$\label{line:count_r_2} \Comment{$w \in \NM$, so $\mathbf{w}_i$ is a crossing walk }
            \State \textbf{break} \label{line:NM_end}
        \EndIf
    \EndFor
\EndFor

\If{$r \geq 0.5 \cdot 450 \log n$} \label{line:yes}
    \State \Return ``yes'' \Comment{$u$ is robustly connected $\M \cup \NM$ }
\Else \label{line:no}
    \State \Return ``no'' \Comment{$u$ is \emph{not} robustly connected $\M \cup \NM$ }
\EndIf

\end{algorithmic}
\end{algorithm}

In \Cref{lemma:cores_sparsely_connected}, we showed that vertices in spectral clusters which are \emph{not} connected to the cross vertices $\M$, have few edges to other spectral clusters. This property is important because it allows us prove that spectral impostors have good expansion properties outside of $\nei_G(\M)$ (see \Cref{lem:imposters_expand}). 

This motivates the definition of $\NM$ as the ``problematic" part of $\nei_G(\M)$: It identifies the vertices in $\nei_G(\M)$ that have too many edges to other spectral clusters. With this definition, we can extend \Cref{lemma:cores_sparsely_connected} to also apply to vertices in $\nei_G(\M) \setminus \NM$, as the following lemma shows.

\begin{lemma}[Few neighbors across spectral clusters, extended to $ \nei_G(\M) \setminus \NM$]\label{claim:NM_neighbors}
Consider the setting of \Cref{fig:setting}.
For every $i \in [k]$ and every $u \in \spec(i) \setminus \NM$, it holds that 
\begin{equation*}
    \left|\nei_G(u) \cap  \spec(\lnot i)  \right|  \leq d\cdot \frac{\phi}{2}. 
\end{equation*}

\end{lemma}

\begin{proof}
    This proof is an extension of \Cref{lemma:cores_sparsely_connected}. Recall that \Cref{lemma:cores_sparsely_connected} proves the above statement for $u \in \spec(i) \setminus  \nei_G(\M)$. Therefore, it remains to consider $u \in  \spec(i) \cap  \nei_G(\M)\setminus  \NM. $ 
    
    By definition of $\spec(i)$ (Definition \ref{def:spec}), $u$ is not in the cross vertices $\M$, so $u \in  \nei_G(\M)\setminus  (\NM \cup \M). $  So from the definition of $\NM$ (Definition \ref{def:NM}), we have \[\left|\nei_G(u) \cap  \spec(\lnot i)  \right|  \leq  d\cdot \frac{\phi}{2}.\] 
\end{proof}
 In \Cref{subsec:regnei} we strengthen the ``impostors expand" lemma (\Cref{lem:imposters_expand}) to also hold for vertices in $\nei_G(\M) \setminus \NM$ (see \Cref{lem:imposters_expand_stronger}), and we also  
 show an upper bound of $\approx \epsilon n$ on the size of $\NM$ (see \Cref{lemma:N(B_far)_reg}), which is stronger than the upper bound of $\approx d\epsilon n$ on the size of $\nei_G(\M)$. 

Recall from \Cref{def:crossgraph} the notion of a cross graph $G_{i, j}$ and $V_{i, j} = (\spec(i) \cap \lab(j))\cup \M$.

\begin{remark}[Sampling from $\crossgapx_{\tau(u), \sigma(u)}$, line~\eqref{line:sample_walk_cross}]\label{rem:sample_from_cross} Given a vertex $u \in \crossgapx_{\tau(u), \sigma(u)}$,  access to the $\langle \cdot, \cdot\rangle_{\apx}$, and a sampler from the graph $G$, we can generate a random walk $\walk = (u, w_1, \ldots, w_L)$ of length $L$ in the graph $\crossgapx_{\tau(u), \sigma(u)}$ with the following simple procedure:
\begin{enumerate}
    \item For every $l \leq L$, given $w_l$, sample a vertex $v \sim \unif(\nei_G(w_l))$;
    \item If $v \notin V_{\tau(u), \sigma(u)}$ set $w_{l+1} \coloneqq w_l$; otherwise, $w_{l+1} \coloneqq v$.
\end{enumerate}
Condition $v \notin V_{\tau(u), \sigma(u)}$ is satisfied either if $\sigma(v) \neq \sigma(u)$ or if $v \in \spec(\neg \tau(u))$. The first condition may be verified immediately. To verify the second condition, we need to check if there exists $i \in [k]$, $i \neq \tau(u)$ such that $\|f_v - \tmu_{i}\|^2_{\apx} \leq \frac{\phi^2}{400\eta}\|\tmu_{i}\|^2_{\apx}$. This can be done by using the Spectral Dot Product Oracle from \Cref{thm:spec_dot_prod_oracle} at most $ k$ times.

A random walk in the graph $\crossgapx_{\tau(u), \sigma(u)}$ starting from a vertex $u$ can thus be generated in time 
\[O(L\cdot k\cdot Q) = O(\poly(k)\cdot n^{1/2+O(\epsilon/\phi^2)}\cdot\poly(\log (n))\cdot\poly(1/\phi)\cdot\log(1/\delta)),\]
where $Q$ denotes the query time of the Spectral Dot Product Oracle. 
\end{remark}

We are finally ready to formally define robust connectedness/separatedness and prove why the procedure \robcon \ successfully determines vertices which satisfy either one of these properties in the corresponding cross graph.

\begin{definition}[Robustly connected/separated]\label{def:rob_con}
   Consider the setting of \Cref{fig:setting}.
    
We say that a vertex $u \in \crossvapx_{i, j}$ is \emph{robustly connected to} $\Mapx\cup\NMapx$ in the graph $\crossgapx_{i, j}$ if
    \[\Pr[ \text{a lazy random walk of length $L$ in $\crossgapx_{i, j}$ started at $u$ is crossing}] \geq 0.9.\]
    We write $u \rightsquigarrow_{\crossgapx_{i, j}} \Mapx\cup\NMapx$ if $u$ is robustly connected to $\Mapx\cup\NMapx$ and  $u \slashed{\rightsquigarrow}_{\crossgapx_{i, j}} \Mapx\cup\NMapx$ otherwise.
    
We say that a vertex $u \in \crossvapx_{i, j}$ is \emph{robustly separated from} $\Mapx\cup\NMapx$ in the graph $\crossgapx_{i, j}$ if
    \[\Pr[ \text{a lazy random walk of length $L$ in $\crossgapx_{i, j}$ started at $u$ is crossing}] < 0.1.\]
    We write $u \perp_{\crossgapx_{i, j}} \Mapx\cup\NMapx$ if  $u$ is robustly separated from $\Mapx\cup\NMapx$ and  $u \slashed{\perp}_{\crossgapx_{i, j}} \Mapx\cup\NMapx$ otherwise.
\end{definition}

\begin{lemma}[\robcon \  detects robust conectedness and separatedness]\label{cor:robcon_works} Conditioned on the success of all of the calls to the Spectral Dot Product Oracle, $\specdp$, we have that
\begin{enumerate}
       \item If a vertex $u$ is robustly connected to $\Mapx\cup\NMapx$, then \robcon \ outputs ``yes'' with probability at least $1 - 2/n^5$.
    \item If a vertex $u$ is robustly separated from $\Mapx\cup\NMapx$, then \robcon \ outputs ``no'' with probability at least  $1 - 2/n^{5}$.
 
\end{enumerate}

\end{lemma}

\begin{proof}
First, we show that if all of the calls to \specdp{} are correct, then the parameter $r$, counted in lines~\eqref{line:count_r_1},~\eqref{line:count_r_2} is exactly equal to the number of crossing walks out of the $450\log n $ sampled in line~\eqref{line:sample_walk_cross}. 
Fix $i \in [450\log n]$. Suppose that $\mathbf{w}_i$ is a crossing walk. Then, by \Cref{def:crosswalk}, there exists a vertex $w \in \mathbf{w}_i$ such that $w \in \M \cup \NM$. Let $w$ be the first such vertex in $\mathbf{w}_i$.  If $w \in \M$, then by \Cref{def:spec}, line~\eqref{line:M} gets triggered and we increment the counter $r$ and finish processing $\mathbf{w}_i$. If instead $w \in \NM \setminus \M$, then $\tau(w) \neq *$, so $w \in \spec(\tau(w))$. Then in lines \eqref{line:NM_begin} - \eqref{line:NM:end}, $q$ counts the number of elements in $\nei_G(w) \cap \spec(\neg \tau(u))$. Since $w \in \NM$, it holds that $|\nei_G(w) \cap \spec(\neg \tau(u))|\geq \frac{\phi}{2}d$ (by \Cref{def:NM}), so we increment $r$ in line \eqref{line:count_r_2} and finish processing $\mathbf{w}_i$. Thus, if $\mathbf{w}_i$ is a crossing walk then the counter $r$ gets incremented by exactly one. 

Suppose instead that $\mathbf{w}_i$ is \emph{not} a crossing walk. Then, by \Cref{def:crosswalk}, for every $w \in \mathbf{w}_i$, it holds that $w \notin \M \cup \NM$. Fix $w \in \mathbf{w}_i$. Since $w \notin \M$, line~\eqref{line:M} does not get triggered (by \Cref{def:spec}). Furthermore, since $w \notin \NM$, by \Cref{def:NM} we either have $ |\nei_G(w) \cap \spec(\neg  \tau(w))| < \frac{\phi}{2}d$ or $w \notin \nei_G(\M)$. In the former case, we have $q < \frac{\phi}{2}d$,  so line \eqref{line:NM} does not get triggered. In the latter case, i.e. $w \notin \nei_G(\M)$,  by \Cref{lemma:cores_sparsely_connected} we have $ |\nei_G(w) \cap \spec(\neg  \tau(w))| < \frac{\phi}{2}d$, so again line \eqref{line:NM} does not get triggered. Since this holds for all $w \in \mathbf{w}_i$, the counter $r$ does not get incremented. 

Having shown that the parameter $r$ correctly counts the number of crossing walks, we now prove each of the two statements in the Lemma statement. 
\paragraph{Statement 1.} By \Cref{def:rob_con}, if $u \rightsquigarrow_{\crossgapx_{\tau(u), \sigma(u)}}\Mapx\cup\NMapx$ then a random lazy walk starting at $u$ in the graph $\crossgapx_{\tau(u), \sigma(u)}$ is a crossing walk with probability $ \geq 0.9$. Let $X$ be the number of crossing walks among the $450\log n$ generated walks. We can think of $X$ as the number of successes in a Bernoulli process with $450\log n$ trials and success probability $\geq 0.9$. Therefore, by Chernoff bounds, 
\[\Pr[X \leq 0.9\cdot 450\log n - 90\log n] \leq 2\exp(5\log n) \leq \frac{2}{n^5}.\]
    In other words, with probability $1 - \frac{2}{n^5}$ at least $0.7\cdot 450\log n$ of the walks are crossing walks. \robcon \ will recognize each of those walks and the line~\eqref{line:yes} will be triggered.

\paragraph{Statement 2.} By \Cref{def:rob_con}, if $u \perp_{\crossgapx_{\tau(u), \sigma(u)}}\Mapx\cup\NMapx$ then a random lazy walk starting at $u$ in the graph $\crossgapx_{\tau(u), \sigma(u)}$ is a crossing walk with probability $ \leq 0.1$. Let $X$ be the number of crossing walks among the $450\log n$ generated walks. We can think of $X$ as the number of successes in a Bernoulli process with $450\log n$ trials and success probability $\leq 0.1$. Therefore, by Chernoff bounds, 
    \[\Pr[X \geq 45\log n + 45\log n] \leq 2\exp(-5\log n) \leq \frac{2}{n^{5}}.\]
    In other words, with probability $1 - \frac{2}{n^5}$ at most $90\log n$ of the walks are crossing walks. \robcon \ will recognize at most $90\log n$ of the walks as crossing walks and the line~\eqref{line:no} will be triggered.

\end{proof}

\begin{lemma}[Runtime and space complexity of \robcon]\label{robcon:time_space} The space complexity of \robcon \ is the same as space complexity of the Spectral Dot Product Oracle, as per \Cref{thm:spec_dot_prod_oracle}, which equals
\[O(\poly(k)\cdot n^{1/2+O(\epsilon/\phi^2)}\cdot\polylog (n)\cdot\poly(1/\phi)),\]
and the runtime is
\[O(d\cdot L\cdot k\cdot Q\cdot\log n) = O(d\cdot \poly(k)\cdot n^{1/2+O(\epsilon/\phi^2)}\cdot\polylog (n)\cdot\poly(1/\phi)\cdot \log(1/\delta)),\]
where $Q$ denotes the query time of the Spectral Dot Product Oracle. 
\end{lemma}
\begin{proof} Besides the data structure $\mathcal{D}$ required by the Spectral Dot Product Oracle, \robcon \ stores at most one random walk of length $L$ and two counters $r$ and $q$ at any point in time. Therefore, the space complexity of \robcon \ is the same as the space complexity of the Spectral Dot Product Oracle.

Now, we show an upper bound on the runtime of \robcon. By \Cref{rem:sample_from_cross}, one call to line~\eqref{line:sample_from_cross} takes time $O(L\cdot k\cdot Q)$. For every vertex $w \in \walk_i$, the procedure in lines~\eqref{line:M}-\eqref{line:M_end} which checks if $w \in \M$ takes time $O(Q)$, and the procedure in lines~\eqref{line:NM_begin}-\eqref{line:NM_end} which checks if $w \in \NM$ takes time $O(d\cdot k\cdot Q)$. Therefore, the total runtime of \robcon \ is $O(d\cdot k\cdot L\cdot Q\cdot \log n)$, as desired.
\end{proof}

\subsection{The regularized neighborhood}
\label{subsec:regnei}
In this section, we prove the two crucial properties of the set $\NM$: first, its size is upper-bounded by $\approx \epsilon n$; second, the ``impostors expand'' Lemma (see \Cref{lem:imposters_expand}), which states that subsets of $\im \setminus \nei_G(\M)$ expand, can be strengthened to guarantee that subsets of $\im \setminus \NM$ expand. These two facts together allow us to reduce the number of misclassified vertices from $ \approx  \delta | \nei_G(\M)|  \approx \delta  d \epsilon n$ to $ \approx \delta |\NM| \approx  \delta \epsilon n$.
\\~\\
We begin by showing the latter of the two properties of $\NM$.

\begin{lemma}[Impostors expand, extended to $\nei_G(\M) \setminus \NM$]\label{lem:imposters_expand_stronger}
Consider the setting of \Cref{fig:setting}.
For every $i ,j \in [k]$ with $i \neq j$ and every $S \subseteq \im(i,j) \setminus \NM$, one has
$$|E_{G[C_j]}(S, (\im(i,j) \cup \M ) \setminus S)|\geq \frac{\phi}{2}d|S| \, .$$ 
\end{lemma}
\begin{proof}
The proof is almost identical to the proof of Lemma \ref{lem:imposters_expand}. 
For simplicity of notation, we write $E_j$ to denote $E(C_j)$ throughout this proof. 
Note that \smash{$|S| \leq |\im| \leq 2 \cdot 10^4 \cdot \eta^2 \cdot\frac{\epsilon}{\phi^4}n \leq \frac{|C_j|}{2}$}, where the second inequality follows by \Cref{lemma:impostor_n_cross_size}, and the third inequality follows by the setting of parameters in \Cref{fig:setting}. Since $S \subseteq C_j$, and since $C_j$ is a $\phi$-expander, we have
\begin{equation}\label{eq:S_expands_reg}
\left|E_j(S, C_j \setminus S)\right| =\left|E(S, C_j \setminus S)\right|\geq \phi\cdot  d \cdot |S|.
\end{equation}
\noindent
We will now show that a large fraction of the edges in $E_j(S, C_j \setminus S)$ go to $\im(i,j) \cup \M$. 
We have
\begin{align*}
 C_j 
 & \subseteq \M \cup \spec(\neg i) \cup  \im(i,j).
 \end{align*}
\noindent
Thus,
\begin{equation}\label{eq:edgecounting_reg}
 \begin{aligned}
 |E_j(S, C_j \setminus S)| &  \leq \left|E_j(S, \M)\right| +\left|E_j(S,   \spec(\neg i))  \right| + \left |E_j(S, \im(i,j) \setminus S)\right| \\
 & \leq \left|E_j(S, \M)\right| + \sum_{u \in S}\left|\nei_G(u) \cap \spec(\neg i)  \right|  + \left|E_j(S, \im(i,j) \setminus S)\right| \\
 & \leq\left|E_j(S, \M)\right| +  |S| \cdot d \cdot \frac{\phi}{2} + \left|E_j(S, \im(i,j) \setminus S)\right| \qquad \text{by \Cref{claim:NM_neighbors}.}
 \end{aligned}
 \end{equation}
 \noindent
If we combine the above bound with~\eqref{eq:S_expands_reg} and rearrange, we conclude 

\begin{align*}
    |E_j(S, (\im(i,j) \cup \M) \setminus S)| & =     |E_j(S, \im(i,j) \setminus S)|  +     |E_j(S,\M)|  &&  \text{since $S \subseteq \im$ and $\im \cap \M = \emptyset$} \\
    & \geq |E_j(S, C_j \setminus S)| - |S| \cdot d \cdot \frac{\phi}{2} &&  \text{by Equation \eqref{eq:edgecounting_reg}} \\
    & \geq  |S| \cdot d \cdot \frac{\phi}{2}  &&  \text{by Equation  \eqref{eq:S_expands_reg}.}
\end{align*}
\end{proof}

\noindent
In the rest of the section, we prove the other property of $\NM$, namely $|\NM| \lesssim \epsilon n$. 
\begin{lemma}\label{lemma:N(B_far)_reg}  Consider the setting of \Cref{fig:setting}. Then,
\[|\NM| \leq 10^8 \eta^5 \cdot \frac{\epsilon}{\phi^4} n\, .\]
\end{lemma}
\noindent
To prove the lemma, we first introduce an auxiliary family of sets $\{ H_i\}_{i \in [k]}$, which are the sets of vertices that have an abnormally high projection onto a cluster mean $\mu_i$.

\begin{definition}[High-projection vertices]\label{def:H}  Consider the setting of \Cref{fig:setting}. 
For $i \in [k]$, we define the $i$-th set of high-projection vertices, denoted $H_i$, as 
$$H_i \coloneqq \left \{ v \in V : \langle f_v, \mu_i\rangle \geq \left( 1 + \frac{\phi}{8}\right)\|\mu_i\|^2_2 \right \} \, .$$
\end{definition}

\noindent
We show that the high projection sets $H_i$ are subsets of the cross-vertices $\M$. Intuitively, this is true because vertices in $H_i$ are essentially parallel to $\mu_i$, so they cannot be in any $\spec(j)$ for $j \neq i$ since $\mu_i$ and $\mu_j$ are roughly orthogonal, and they cannot be in $\spec(i)$ since they are too long in the direction of~$\mu_i$.
\begin{lemma}\label{lemma:H_in_M} Consider the setting of \Cref{fig:setting}. 
Then, $H_i \subseteq \M$ for all $i \in [k].$
\end{lemma}
\begin{proof}
Suppose for contradiction that there exists some $v \in H_i \setminus \M$. Since $v \notin \M$ is not a cross vertex, we must have $v\in \spec(j)$ for some $j \in [k]$ (by Definition \ref{def:spec}). Therefore, by \Cref{claim:spec_exact_mu}, we have 
$$\|f_v - \mu_j\|_2 \leq \frac{\phi}{10\sqrt{\eta}}\|\mu_j\|_2,$$ 
so by the triangle inequality we have 
\begin{equation}\label{eq:H_in_M}
|\langle f_v, \mu_i\rangle | \leq |\langle f_v - \mu_j, \mu_i\rangle| + |\langle \mu_j, \mu_i\rangle| \leq \frac{\phi}{10\sqrt{\eta}}\|\mu_j\|_2\cdot\|\mu_i\|_2 + |\langle \mu_i, \mu_j\rangle|.
\end{equation}
If $i = j$, then by~\eqref{eq:H_in_M}, we get
\[|\langle f_v, \mu_i\rangle | \leq \left(1 + \frac{\phi}{10\sqrt{\eta}}\right)\cdot\|\mu_i\|^2_2 <\left(1 + \frac{\phi}{8}\right)\|\mu_i\|^2_2 \, ,\]
since $\eta \ge 1$ by definition (see \Cref{fig:setting}), which contradicts $v \in H_i$. If instead $i \neq j$, we plug guarantee~\ref{bulletpt:mu_i_dot_mu_j} of \Cref{lemma:clustermeans} into~\eqref{eq:H_in_M} and get
\begin{equation}
\label{eq:caseineqj}
    |\langle f_v, \mu_i\rangle| \leq \frac{\phi}{10\sqrt{\eta}}\cdot\|\mu_i\|_2\|\mu_j\|_2 + \frac{8\sqrt{\epsilon}}{\phi} \cdot \frac{1}{\sqrt{|C_i||C_j|}}
\end{equation}
Now, by guarantee~\ref{bulletpt:mu_norm} of \Cref{lemma:clustermeans}, we also have
\begin{equation}
\label{eq:clustermean-clustersize}
    \left(1-\frac{4\sqrt{\epsilon}}{\phi}\right) \frac{1}{|C_i|} \le \|\mu_i\|^2_2 \le \left(1+\frac{4\sqrt{\epsilon}}{\phi}\right) \frac{1}{|C_i|} \quad \text{and} \quad  \left(1-\frac{4\sqrt{\epsilon}}{\phi}\right) \frac{1}{|C_j|} \le \|\mu_j\|^2_2 \le \left(1+\frac{4\sqrt{\epsilon}}{\phi}\right) \frac{1}{|C_j|}  \, .
\end{equation}
We can plug these into~\eqref{eq:caseineqj} to get
\begin{equation*}
    |\langle f_v, \mu_i\rangle| \leq \frac{\phi}{10\sqrt{\eta}}\cdot\|\mu_i\|_2\|\mu_j\|_2 + \frac{8\sqrt{\epsilon}}{\phi} \cdot \frac{1}{1-\frac{4\sqrt{\epsilon}}{\phi}} \cdot \|\mu_i\|_2\|\mu_j\|_2 \, ,
\end{equation*}
Finally, we combine the definition of $\eta$ (see \Cref{fig:setting}) with~\eqref{eq:clustermean-clustersize} and upper-bound
\begin{equation*}
    \|\mu_j\|_2 \le \sqrt{\left(1+\frac{4\sqrt{\epsilon}}{\phi}\right)\frac{1}{|C_j|}} \le  \sqrt{\left(1+\frac{4\sqrt{\epsilon}}{\phi}\right)\frac{\eta}{|C_i|}} \le \sqrt{\frac{1+\frac{4\sqrt{\epsilon}}{\phi}}{1-\frac{4\sqrt{\epsilon}}{\phi}} \cdot \eta \|\mu_i\|_2^2} \le \frac{1+\frac{4\sqrt{\epsilon}}{\phi}}{1-\frac{4\sqrt{\epsilon}}{\phi}} \sqrt{\eta} \cdot \|\mu_i\|_2 \, ,
\end{equation*}
which ultimately gives
\begin{equation}
\label{eq:bound_L_i}
    |\langle f_v, \mu_i\rangle| \leq \left(\frac{\phi}{10} + \frac{8\sqrt{\epsilon}}{\phi} \cdot \frac{\sqrt{\eta}}{1-\frac{4\sqrt{\epsilon}}{\phi}}\right) \frac{1+\frac{4\sqrt{\epsilon}}{\phi}}{1-\frac{4\sqrt{\epsilon}}{\phi}}  \|\mu_i\|^2_2 < \frac{\phi}{8}\|\mu_i\|_2^2 \, ,
\end{equation}
where the penultimate step uses our assumption that ${\sqrt{\epsilon}}\sqrt{\eta} \le {10^{-5}} \phi$ (see \Cref{fig:setting}). We have thus reached a contradiction with $v \in H_i$, since $\phi/8 \cdot \|\mu_i\|_2^2 < (1+\phi/8)\|\mu_i\|_2^2$.
\end{proof}

\noindent
We now introduce some terminology that will help us motivate the definition of another auxiliary family of sets and facilitate our discussions in the remainder of the section.

\begin{term*}\label{def:pullforce}
    In \Cref{lemma:fx_close_to_avg_neighbors}, we show that a spectral embedding of a vertex roughly equals the average of the spectral embeddings of its neighbors. For a vertex $u$ and one of its neighbors $v$, we say that $v$ \emph{pulls} $f_u$ with \emph{force} $f_u - f_v$. Thus, \Cref{lemma:fx_close_to_avg_neighbors} informally says that the average force with which the neighbors of $v$ pull $f_v$ is roughly $0$.
\end{term*}

\noindent
 In \Cref{fig:pulled_set} we depict a vertex $u \in \spec(i)$ and it's neighbors whose projections on $\mu_i$ is either much bigger than $\|\mu_i\|_2$ (these are connected to $u$ with solid vetors) or much smaller (these are connected to you with dashed vectors). From the above, the sum of projections of the solid vectors on $\mu_i$ should we roughly equal to the sum of projections of dashed vectors. Furthermore, \Cref{fig:pulled_set} also provides an illustration of the set~$H_i$.

\begin{figure}[h]
	\centering
	\tdplotsetmaincoords{100}{160}
\begin{adjustbox}{trim=0 0 0 0,clip}
\begin{tikzpicture}[tdplot_main_coords, scale=5]
\definecolor{mydarkgreen}{RGB}{0,100,20}
  \draw[->, >=stealth, semithick] (0,0,0) -- (1.4,0,0);
  \draw[->, >=stealth, semithick] (0,0,0) -- (0,1.8,0);
  \draw[->, >=stealth, semithick] (0,0,0) -- (0,0,1.7);

  \draw plot [mark=*, mark size=0.2] coordinates{(0,0,1.1)}; 

  \draw plot [mark=*, mark size=0.2] coordinates{(0,0,0.9)}; 
  \node[anchor=west, scale = 0.7] at (0, 0, 0.9) {$\mu_i$};
   
  
  \draw plot [mark=*, mark size=0.2] coordinates{(0,-0.08,0.87)}; 
  \node[anchor=west, scale = 0.7] at (0,-0.3,0.88) {$u$};

  \draw plot [mark=*, mark size=0.2] coordinates{(0, 0.75, 0.4)}; 

  \foreach \x/\y/\z in {
    0.25/1.35/0.8,
    0.6/-0.9/1.5,
    0.4/1.5/1.6,
    1.0/0.4/1.8,
    0.7/0.7/1.9
  }{
    \draw plot [mark=*, mark size=0.2] coordinates {(\x,\y,\z)};
    \draw[->, >=stealth, mydarkgreen, thick] (0,-0.08,0.87) -- (\x,\y,\z);
  }

  \draw[->, >=stealth, red, thick, dashed] (0,-0.08,0.87) -- (0, 0.75, 0.4);

  \foreach \x/\y/\z in {
    0.05/0.85/0.05,
    -0.05/0.95/-0.02,
    0.02/0.88/-0.04
  }{
    \draw plot [mark=*, mark size=0.2] coordinates {(\x,\y,\z)};
    \draw[->, >=stealth, red, thick, dashed] (0,-0.08,0.87) -- (\x,\y,\z);
  }

  \draw plot [mark=*, mark size=0.2] coordinates {(0.0, -1, -0.1)};
  \draw[->, >=stealth, red, thick, dashed] (0,-0.08,0.87) -- (0.0, -1, -0.1);

\foreach \angle in {30, 150, 270} {
    \pgfmathsetmacro\x{0.9 + 0.05*cos(\angle)}
    \pgfmathsetmacro\y{0.0 + 0.05*sin(\angle)}
    \pgfmathsetmacro\z{0.02*sin(\angle)}  
    \draw plot [mark=*, mark size=0.2] coordinates {(\x,\y,\z)};
    \draw[->, >=stealth, red, thick, dashed] (0,-0.08,0.87) -- (\x,\y,\z);
}
  \draw plot [mark=*, mark size=2.5, mark options={fill=blue, opacity=0.2}] coordinates {(0,0.9,0)};
  \draw plot [mark=*, mark size=2.5, mark options={fill=blue, opacity=0.2}] coordinates {(0.9,0, 0)};
  \draw plot [mark=*, mark size=2.5, mark options={fill=blue, opacity=0.2}] coordinates {(0,0,0.9)};

  \node[anchor=west, scale = 0.7] at (0, 0, 1.1) {$\sqrt{1 + \frac{\phi}{8}}\cdot \mu_i$};

  \node[anchor=west, scale = 0.7] at (0, 0.25, 0.82) {$\spec(i)$};
  
\begin{scope}[tdplot_screen_coords]
  \draw[dashed, thick, color = mydarkgreen] (-0.3,1.65) ellipse [x radius=1, y radius=0.5];
\end{scope}

  \node[anchor=west, scale = 0.7, color = mydarkgreen] at (0, -4, 2) {$H_i$};

\end{tikzpicture}
\end{adjustbox}
	  \caption{An illustration of a vertex $u$ in $\spec(i)$. The solid arrows point to the neighbors of $u$ with high projection on $\mu_i$, and dashed arrows point to the neighbors with small projection on $\mu_i$. The solid and dashed arrows balance each other.}
\label{fig:pulled_set}
\end{figure}
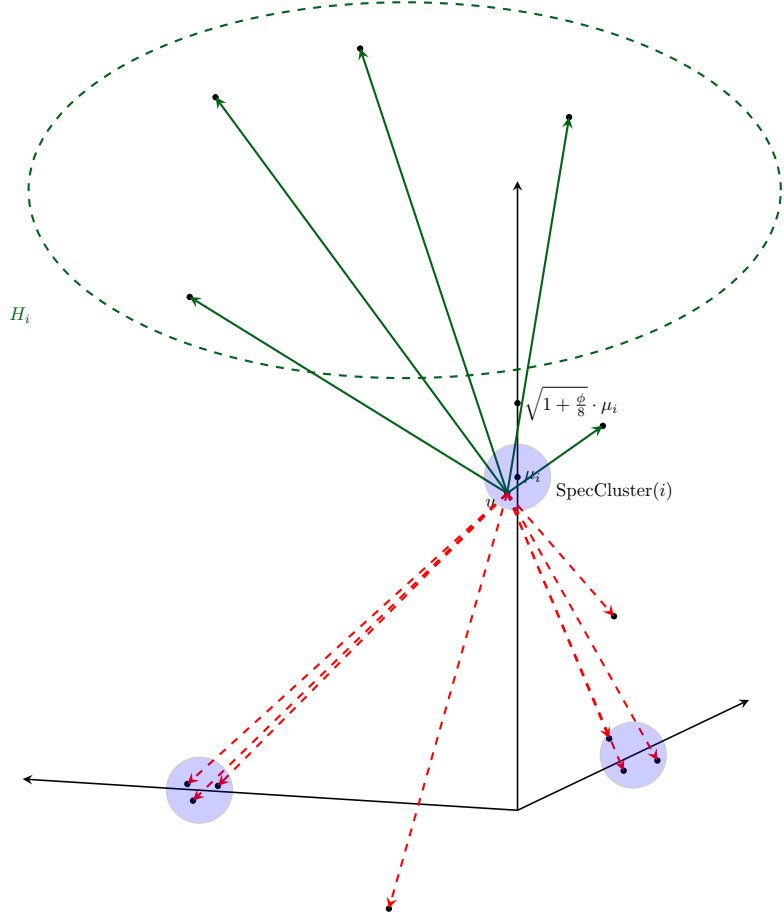

Thus, we define the family of \textit{strongly-pulled vertices} $\{P_i\}_{i \in [k]}$ as follows: $P_i$ consists of the vertices in the $i$-th spectral cluster that experience a strong pulling force in the direction of $\mu_i$ from their neighbors in $H_i$, which have a large projection on $\mu_i$ (by \Cref{def:H}). Another way to think about $P_i$ is the following: it consists of vertices in the $i$-th spectral cluster that have ``many'' neighbors in $H_i$, where ``many'' is weighted by the force $f_v-f_u$. The sets $\{P_i\}_{i \in [k]}$ are formally introduced in \Cref{def:D} below.

\begin{definition}[Strongly-pulled vertices]\label{def:D} Consider the setting of \Cref{fig:setting}. 
For $i \in [k]$, we define the $i$-th set of \textit{strongly-pulled vertices}, denoted $P_i$, as 
$$ P_i \coloneqq \left\{ u \in \spec(i) : \sum_{v \in \nei_G(u) \cap H_i } \langle f_v - f_u, \mu_i \rangle \geq  \frac{\phi}{10}d  \cdot \|\mu_i\|_2^2\right\} \, . $$
\end{definition}

\begin{remark}
    For any $i \in [k]$, according to the definition above, we put $u \in V$ in $P_i$ if it satisfies two conditions:  $u \in \spec(i)$ and $\sum_{v \in \nei_G(u) \cap H_i } \langle f_v - f_u, \mu_i \rangle \geq d\phi/10 \cdot  \|\mu_i\|_2^2$. We remark that the first condition uses approximate cluster means and approximate inner products (see \Cref{def:spec}), while the second condition uses exact cluster means and the standard dot product between spectral embeddings.
\end{remark}

\noindent
Equipped with the notion of strongly-pulled vertices, we have the following two lemmas, from which we can easily conclude a proof of \Cref{lemma:N(B_far)_reg}.

\begin{lemma}\label{lemma:NM_in_D}
 Consider the setting of \Cref{fig:setting}. Then 
    $$\NM \subseteq \bigcup_{i\in [k]} P_i \, .$$
\end{lemma}

\begin{lemma}\label{lemma:D_ub}
Consider the setting of \Cref{fig:setting}. Then 
    $$ \left|  \bigcup_{i\in [k]} P_i \right| \leq  10^8 \eta^5 \cdot \frac{\epsilon}{\phi^4} n \, .$$
\end{lemma}

\begin{proof}[Proof of \Cref{lemma:N(B_far)_reg}]
   Follows immediately from \Cref{lemma:NM_in_D} and \Cref{lemma:D_ub}.  
\end{proof}

\noindent
To conclude, we prove \Cref{lemma:NM_in_D} in \Cref{sec:NM_subset_D} and
\Cref{lemma:D_ub} in \Cref{sec:bounding_D}.

\subsubsection{The regularized neighborhood consists of strongly-pulled vertices}\label{sec:NM_subset_D}
In this section, we prove \Cref{lemma:NM_in_D}, which states that the regularized neighborhood $\NM$ is a subset of the union of $P_1,\dots,P_k$. Before diving into the details, we discuss the intuition for the proof. 
\\~\\
Consider a vertex \smash{$u \in \NM$}. By definition, $u$ is in the $i$-th spectral cluster, for some \smash{$i \in [k]$}, and it has \smash{$\phi d/2$} neighbors in other spectral clusters. Informally, if $u \in \NM$ and $u \in \spec(i)$ then $u$ has many neighbors in other spectral clusters, so looking at \Cref{fig:pulled_set}, there must be many dashed arrows coming out of $u$. Since the dashed arrows have to be balanced by solid arrows, there must be many solid arrows as well. And this, in turn, implies that $u$ has to have many neighbors in $H_i$ which means that $u \in P_i$ by the definition of $P_i$. 

We now provide a more formal sketch of the proof. Let $u \in \NM \cap \spec(i)$ for some $i \in [k]$. 
Suppose that $v \in \nei_G(u) \cap \spec(\neg i)$. By definition of $\NM$ (\Cref{def:NM}), $u$ has at least $\phi d/2$ such neighbors. Note that the amount of force with which any such neighbor $v$ pulls \smash{$f_u$} in the direction of $\mu_i$ is negative, more specifically
$$\langle f_v - f_u , \mu_{i} \rangle \approx  \left \langle \mu_{\tau(v)} - \mu_{i} , \mu_{i} \right \rangle \approx - \|\mu_{i}\|_2^2 \, .$$
On the other hand, we know that $f_u$ actually ends up being very close to $\mu_i$ (since $u \in \spec(i)$ by definition of $P_i$), but by the balance equations (see \Cref{lemma:fx_close_to_avg_neighbors}), we know that
$$ \sum_{v \in \nei_G(u)} \langle f_v - f_u, \mu_{i} \rangle  \approx 0 \, .$$
This implies that the other neighbors of $u$ must compensate by exerting a strong pulling force on $f_u$ back in the direction of $\mu_{i}$, i.e. 
$$ \sum_{\substack{v \in \nei_G(u)\, : \\ v \notin \spec(\neg i)}} \langle f_v - f_u, \mu_{i} \rangle  \gtrsim d \frac{\phi}{2}  \cdot \|\mu_{i}\|^2_2 \, .$$
One can now observe that the sum above is contributed to by the neighbors of $u$ in $\nei_G(u)\cap H_{i}$ (which are those with a large projection onto $\mu_{i}$, see \Cref{def:H}, and note that they are not in $\spec(\neg i)$ by \Cref{lemma:H_in_M}) and by the neighbors of $u$ in $ (\nei_G(u) \setminus \spec(\neg i))\setminus H_i$. However, the latter type of neighbors, can only contribute weakly in this direction,  since for each $v \in  (\nei_G(u) \setminus \spec(\neg i))\setminus H_i$, it holds that 
$\langle f_v, \mu_{i}\rangle < (1+\phi/8) \|\mu_i\|_2^2$. Therefore, 
$$ \sum_{\substack{v \in \nei_G(v)\, : \\ v \notin \spec(\neg i), \\ v\notin H_i}} \langle f_v - f_u, \mu_{i}\rangle \approx   \sum_{\substack{v \in \nei_G(v)\, : \\ v \notin \spec(\neg i), \\ v\notin H_i}}  \langle f_v, \mu_{i}\rangle - 
\|\mu_{i}\|_2^2 \lesssim d  \frac{\phi}{4} \cdot \|\mu_i\|^2_2  \, ,$$
where the last inequality uses that each terms has $\langle f_v, \mu_{i}\rangle - \|\mu_{i}\|_2^2 \leq \frac{\phi}{8}\| \mu_i\|^2_2$. 

Therefore, $f_u$ must be strongly ``pulled'' by neighbors $  v \in H_{i}$ with a high projection onto $\mu_{i}$, which in turn implies that $u$ is in $P_i$:
$$ \sum_{v \in \nei_G(v) \cap H_i} \langle f_v - f_u, \mu_{i}\rangle = \sum_{\substack{v \in \nei_G(v)\, : \\ v \notin \spec(\neg i)}} \langle f_v - f_u, \mu_{i} \rangle - \sum_{\substack{v \in \nei_G(v)\, : \\ v \notin \spec(\neg i), \\ v\notin H_i}} \langle f_v - f_u, \mu_{i}\rangle   \gtrsim d \frac{\phi}{2}  \cdot \|\mu_{i}\|^2_2 \, .$$

\noindent
We now proceed to make the above proof plan formal. We start by showing a version of \Cref{lemma:fx_close_to_avg_neighbors} specialized to vertices $u$ in a spectral cluster $\spec(i)$, in order to bound the average force with which $u$'s neighbor pull $f_u$ in the direction of $\mu_i$.

\begin{lemma}\label{claim:sum_N} Consider the setting of \Cref{fig:setting}.
For all  $i \in [k]$ and  $u \in \spec(i)$, one has
$$\left|  \sum_{v \in \nei_G(u)} \langle f_v - f_u, \mu_i\rangle \right| \leq 4\epsilon d \cdot \|\mu_i\|_2^2 \, .$$
\end{lemma}
\begin{proof}
  By Lemma \ref{lemma:fx_close_to_avg_neighbors}, we have 
  $$ \left \| \sum_{v \in \nei_G(u)} (f_u - f_v) \right\|_2 = d  \left \|f_u - \frac{1}{d} \sum_{v \in \nei_G(u)} f_v \right\|_2 \leq 2d\epsilon \|f_u\|_2 \, .$$
 Then, an application of Cauchy-Schwarz gives 
 \begin{equation}\label{eq:proj_value}
 \left|  \sum_{v \in \nei_G(u)} \langle f_v - f_u, \mu_i\rangle \right| \leq \left \| \sum_{v \in \nei_G(u)} (f_u - f_v) \right\|_2 \|\mu_i\|_2 \leq 2\epsilon d \|f_u\|_2 \|\mu_i\|_2 \, .
 \end{equation}
 Furthermore, since $u \in \spec(i)$, we have by \Cref{claim:spec_exact_mu} that $\|f_u - \mu_i\|_2 \leq \frac{\phi}{10\sqrt \eta}\|\mu_i\|_2$. So by triangle inequality we get 
 $$ \|f_u\|_2 \leq \|f_u - \mu_i\|_2 + \|\mu_i\|_2 \leq \left(1 + \frac{\phi}{10 \sqrt{\eta}}\right)\|\mu_i\| \leq 2 \|\mu_i\|_2 \, ,$$
where the last inequality uses that $\phi \le 1\le \eta$ as per \Cref{fig:setting}. Combining this with~\eqref{eq:proj_value} gives the statement. 
\end{proof}

\noindent
To continue, it is useful to partition the vertices based on their projection onto a cluster mean $\mu_i$: the high projection set $H_i$ (which we already defined in \Cref{def:H}), the medium projection set $M_i$, and the low projection set $L_i$, which we now define.
\begin{definition}[Medium- and low-projection vertices]\label{def:L_n_M}  Consider the setting of \Cref{fig:setting}. 
For $i \in [k]$, we define the $i$-th set of medium-projection vertices, denoted $M_i$, as
$$M_i \coloneqq \left \{ v \in V : \frac{\phi}{8}\|\mu_i\|^2_2  \leq \langle f_v, \mu_i\rangle < \left( 1 + \frac{\phi}{8}\right)\|\mu_i\|^2_2 \right \} \, ,$$
and define the $i$-th set of low-projection vertices, denoted $L_i$, as 
$$L_i \coloneqq \left \{ v \in V : \langle f_v, \mu_i\rangle < \frac{\phi}{8}\|\mu_i\|^2_2 \right \} \, .$$
\end{definition}

\noindent
With this notation, we argue that the average force with which medium-projection neighbors of $u \in \spec(i)$ pull $f_u$ in the direction of $\mu_i$ is positive but mild, while the average force exerted by its low-projection neighbors is negative in the direction of $\mu_i$.

\begin{lemma}\label{claim:sum_M_L}
Consider the setting of \Cref{fig:setting}.
  For all $i \in [k]$ and $u \in \NM \cap \spec(i)$, one has 
 $$\sum_{v \in \nei_G(u) \cap M_i} \langle f_v - f_u, \mu_i\rangle  \leq \frac{1}{4}\phi d \cdot \|\mu_i\|^2_2 \quad \text{and} \quad \sum_{v \in \nei_G(u) \cap L_i} \langle f_v - f_u, \mu_i\rangle  \leq -  \frac{3}{8}\phi d \cdot\|\mu_i\|_2^2 \, .$$
\end{lemma}
\begin{proof}
Since $u \in \spec(i)$, by \Cref{claim:spec_exact_mu}, we have  $\|f_u - \mu_i\|_2 \leq \frac{\phi}{10\sqrt{\eta}}\|\mu_i\|_2$. Therefore, 
\begin{equation}\label{eq:sum_L_eq1}
 \langle f_u, \mu_i\rangle = \langle \mu_i, \mu_i\rangle - \langle \mu_i - f_u, \mu_i\rangle \geq \|\mu_i\|_2^2 - \|f_u - \mu_i\|_2\|\mu_i\|_2 \geq 
 \left( 1- \frac{\phi}{10\sqrt{\eta}}\right)\|\mu_i\|^2_2 \, .
 \end{equation}
For all $v \in \nei_G(u)\cap M_i$, by \Cref{def:L_n_M} we have  
 $ \langle f_v, \mu_i\rangle \leq (1 + {\phi}/{8})\|\mu_i\|^2_2$, which combined with gives \eqref{eq:sum_L_eq1} gives
 \[ \langle f_v, \mu_i\rangle - \langle f_u, \mu_i\rangle \leq \left(1 + \frac{\phi}{8}\right)\|\mu_i\|^2_2 - \left( 1- \frac{\phi}{10\sqrt{\eta}}\right)\|\mu_i\|^2_2 <  \frac{\phi}{4}\|\mu_i\|^2_2 \, ,\]
 since $\eta \ge 1$ by definition (see \Cref{fig:setting}). Summing over all $v \in \nei_G(u)\cap M_i$, we get
\begin{equation*}
\sum_{v \in \nei_G(u) \cap M_i} \langle f_v - f_u, \mu_i\rangle \leq \frac{\phi}{4} d \cdot \|\mu_i\|^2_2 \, .
\end{equation*}
Which proves the first bound in the statement.
Similarly, for all $v \in L_i \cap \nei_G(u)$, it holds that 
 \[ \langle f_v, \mu_i\rangle - \langle f_u, \mu_i\rangle \leq \frac{\phi}{8}\|\mu_i\|^2_2 - \left( 1- \frac{\phi}{10\sqrt{\eta}}\right)\|\mu_i\|^2_2 \le  - \frac{3}{4}\|\mu_i\|_2^2 \, ,\]
 where the first inequality uses \Cref{def:L_n_M} and~\eqref{eq:sum_L_eq1}.

Therefore, we have 
\begin{equation*}
    \sum_{v \in \nei_G(u) \cap L_i} \langle f_v - f_u, \mu_i\rangle  \leq - \frac{3}{4}|\nei_G(u) \cap L_i|\cdot\|\mu_i\|_2^2 \, . 
\end{equation*}
To lower bound $|\nei_G(u) \cap L_i|$, we will show that $\spec(\lnot i) \subseteq L_i$. Then we can lower bound $|\nei_G(u) \cap L_i| $ by $|\nei_G(u) \cap \spec(\lnot i)| $, which in turn we know to be at least $d\phi/2$ since $u \in \NM$ (see \Cref{def:NM}), giving us the second bound in the statement:
\begin{equation*}
    \sum_{v \in \nei_G(u) \cap L_i} \langle f_v - f_u, \mu_i\rangle  \leq - \frac{3}{4}|\nei_G(u) \cap L_i|\cdot\|\mu_i\|_2^2 \le  - \frac{3}{4}|\nei_G(u) \cap \spec(\neg i) |\cdot\|\mu_i\|_2^2 \le \frac{3}{8} d \phi \cdot \|\mu_i\|_2^2 \, . 
\end{equation*}
It remains to show that $\spec(\lnot i) \subseteq L_i$. To this end, let $v \in \spec(\lnot i)$ and let $j \in [k]$ such that $v \in \spec(j)$. One can repeat the calculation from the proof of \Cref{lemma:H_in_M} in the case $j \neq i$ and get the bound $\langle f_v, \mu_i\rangle < \frac{\phi}{8}\|\mu_i\|_2^2$, as shown in~\eqref{eq:bound_L_i}. Then, $v \in L_i$ as per \Cref{def:L_n_M}, proving that $\spec(\lnot i) \subseteq L_i$.
\end{proof}

\noindent
We are now ready to prove \Cref{lemma:NM_in_D}. 
\begin{proof}[Proof of \Cref{lemma:NM_in_D}]
Let $u \in \NM$. By Definition \ref{def:NM}, we have $\NM \cap \M = \emptyset$. Therefore, $u \notin \M$, so from Definition \ref{def:spec}, we have $u \in \spec(i)$ for some $i \in [k]$. Since $L_i$, $M_i$ and $H_i$ partition $V$, we get 
\begin{align*}
    \sum_{v \in \nei_G(u)\cap H_i} \langle f_v - f_u, \mu_i\rangle & =    \sum_{v \in \nei_G(u)} \langle f_v - f_u, \mu_i\rangle -   \sum_{v \in \nei_G(u)\cap M_i} \langle f_v - f_u, \mu_i\rangle -   \sum_{v \in \nei_G(u)\cap L_i} \langle f_v - f_u, \mu_i\rangle \, .
\end{align*}
The first term is lower-bounded using \Cref{claim:sum_N}, and the second and third term are lower-bounded using \Cref{claim:sum_M_L}, which together give
\begin{equation*}
    \sum_{v \in \nei_G(u)\cap H_i} \langle f_v - f_u, \mu_i\rangle  \geq - 4\epsilon d \cdot \|\mu_i\|_2^2- \frac{\phi}{4} d \cdot\|\mu_i\|^2_2 +  \frac{3}{8}\phi d \cdot \|\mu_i\|_2^2 \geq  \frac{\phi}{10} d \cdot\|\mu_i\|_2^2 \, ,
\end{equation*}
where the last inequality used $\epsilon \leq 10^{-5} \phi$ as per Table $\ref{fig:setting}$. Hence, we conclude $x \in \D_i$, as per \Cref{def:D}.
\end{proof}

\subsubsection{Bounding the number of strongly-pulled vertices}\label{sec:bounding_D}

In this section, we prove \Cref{lemma:NM_in_D}, which bounds the size of the storngly-pulled vertices $P_1 \cup \dots \cup P_k$ by $\epsilon n$. Again, before diving into the details, we discuss the intuition for the proof. 
\\~\\
Fix $i \in [k]$, and let us consider the total force exerted in the direction of $\mu_i$ by vertices in $H_i$ on their neighbors across the cut $H_i$, i.e. the sum $\sum_{(u,v) \in E: v \in H_i, u \notin H_i}\langle f_v - f_u, \mu_i\rangle$. We recall that  \Cref{lemma:fx_close_to_avg_neighbors} roughly states that $\sum_{u \in \nei_G(v)} \langle f_v - f_u, \mu_i\rangle \approx 0$ for any vertex $v$, so summing this up over $v \in H_i$ gives
\begin{equation*}
    \sum_{(u,v) \in E: v \in H_i, u \notin H_i}\langle f_v - f_u, \mu_i\rangle = \sum_{v \in H_i} \sum_{u \in \nei_G(v)}\langle f_v - f_u, \mu_i\rangle \approx 0 \, .
\end{equation*}
On the other hand, the force across $H_i$ cannot be too small, since any strongly-pulled vertex $u$ in $P_i$ non-trivially contributes to the sum (see \Cref{def:D}). Combining these two observation, we conclude that there can only be few vertices in $P_i$.
\\~\\
To implement this strategy, we start with an intermediate upper-bound on the size of $\cup_{i\in [k]}P_i$ in terms of the force across the cuts $H_1,\dots,H_k$ in direction $\mu_i$. 

\begin{lemma}\label{lemma:D_ub_step1} Consider the setting of \Cref{fig:setting}. Then 
    $$\left|\bigcup_{i \in [k]} P_i \right| \leq  \frac{n}{k} \cdot  \frac{20\eta}{d\phi}\cdot \sum_{i=1}^k \sum_{\substack{(v,u) \in E :\\v \in H_i,u \in  V \setminus H_i}}\langle f_v - f_u, \mu_i\rangle \, . $$
\end{lemma}
\begin{proof} We prove the statement by fixing $i \in [k]$ and bounding the size of $P_i$. First, we note that $P_i \subseteq V\setminus H_i$, since $P_i \subseteq \spec(i)$ by \Cref{def:D} while $H_i \subseteq \M$ by \Cref{lemma:H_in_M} and $\M \cap \spec(i) = \emptyset$ by \Cref{def:spec}. Also observe that for every $v \in H_i$ and $u \in V\setminus H_i$ one has
$$ \langle f_v - f_u, \mu_i \rangle \geq 0 $$
by \Cref{def:H}. Therefore, we have 
\begin{align*}
  \sum_{\substack{(v,u) \in E :\\v \in H_i,u \in  V \setminus H_i}}\langle f_v - f_u, \mu_i\rangle&  \geq  \sum_{\substack{(v,u) \in E :\\v \in H_i,u \in  P_i}}\langle f_v - f_u, \mu_i\rangle &&  \text{since $P_i \subseteq V \setminus H_i$} \\
  & = \sum_{u \in P_i} \sum_{v \in \nei_G(u) \cap H_i }\langle f_v - f_u, \mu_i\rangle &&  \\
  & \geq |P_i| \cdot d \cdot \frac{\phi}{10}\|\mu_i\|_2^2 && \text{by Definition \ref{def:D}} \\
  & \geq |P_i| \cdot d \cdot \frac{\phi}{20\eta} \cdot \frac{k}{n} &&  \text{by Remark \ref{rem:mu_i_norm}.}
\end{align*}
Rearranging gives the claim.
\end{proof}

\noindent
Next, we give an upper-bound on the sum of the forces across the cuts $H_1, \ldots, H_k$ in terms of $|\M|$ and $\sum_{ u \in \M} \|f_u\|^2_2$. \Cref{lemma:ub_intermediate} formalizes our intuition of the forces across these cuts being roughly zero. 
\begin{lemma}\label{lemma:ub_intermediate}
Consider the setting of \Cref{fig:setting}.  Then
    $$ \frac{1}{d}\sum_{i=1}^k \sum_{\substack{(v,u) \in E :\\v \in H_i,u \in  V \setminus H_i}}\langle f_v - f_u, \mu_i\rangle \leq \epsilon \cdot \frac{4\eta k}{\sqrt n }\sqrt{|\M|\cdot \sum_{v \in \M}\|f_v\|_2^2} \, .$$
\end{lemma}
\begin{proof}
We begin with an outline of the argument. First, we rewrite the left-hand side of the above as \smash{$\sum_{i=1}^k\sum_{v \in H_i} \langle f_v - \frac{1}{d}\sum_{u \in \nei_G(v)}f_u, \mu_i \rangle $}. Next, 
using the fact that the cluster means almost form a basis (see \Cref{lemma:gklmsL9}), we apply Cauchy-Schwarz to relate the quantity \smash{$\sum_{i=1}^k\sum_{v \in H_i} \langle f_v - \frac{1}{d}\sum_{u \in \nei_G(v)}f_u, \mu_i \rangle$} to \smash{$\sum_{v \in \cup_{i}H_i}\|f_v - \frac{1}{d}\sum_{u \in \nei_G(v)} f_u\|^2_2$}.  Finally, we resort to the fact that \smash{$\|f_v - \frac{1}{d}\sum_{u \in \nei_G(v)} f_u\|_2\leq 2 \epsilon \|f_v\|_2$} for all $v \in V$ (see \Cref{lemma:fx_close_to_avg_neighbors}), together with $\cup_i H_i \subseteq \M$ (see \Cref{lemma:H_in_M}), in order to upper-bound \smash{$\sum_{v \in \cup_{i}H_i}\|f_u - \frac{1}{d}\sum_{u \in \nei_G(v)} f_u\|^2_2$} by \smash{$4\epsilon^2 \sum_{v \in \M}\|f_v\|^2_2$}. 

With the outline in place, we now give the detailed proof.
Let us first fix $i \in [k]$ and consider the inner sum in the left-hand side. 
Observe that
\begin{align*}
\frac{1}{d}\sum_{\substack{(v,u) \in E :\\v \in H_i,u \in  V \setminus H_i}}(f_v - f_u) =  \frac{1}{d}\sum_{v \in H_i }\sum_{u \in \nei_G(v)}\left(f_v - f_u\right)  = \sum_{v \in H_i}\left( f_v- \frac{1}{d}\sum_{u \in \nei_G(v)}f_u\right)\, .
\end{align*}
Therefore, we have 
\begin{align*}
    \frac{1}{d}\sum_{\substack{(v,u) \in E :\\v \in H_i,u \in  V \setminus H_i}}\langle f_v - f_u, \mu_i\rangle =  \sum_{v \in H_i}\left \langle f_v- \frac{1}{d}\sum_{u \in \nei_G(v)}f_u, \mu_i \right \rangle \, .
\end{align*}
\noindent
Taking absolute values and using that $H_i \subseteq \M$ (by Lemma \ref{lemma:H_in_M}), we obtain 
\begin{equation}\label{eq:one_term}
     \frac{1}{d}\sum_{\substack{(v,u) \in E :\\v \in H_i,u \in  V \setminus H_i}}\langle f_v - f_u, \mu_i\rangle \leq  \sum_{v \in H_i}\left| \left \langle f_v- \frac{1}{d}\sum_{u \in \nei_G(v)}f_u, \mu_i \right \rangle \right| \leq  \sum_{v \in \M}\left| \left \langle f_v- \frac{1}{d}\sum_{u \in \nei_G(v)}f_u, \mu_i \right \rangle \right| \, .
\end{equation}
Summing~\eqref{eq:one_term} over all $i \in [k]$, we get
\begin{equation}\label{eq:ub1}
 \frac{1}{d}\sum_{i=1}^k \sum_{\substack{(v,u) \in E :\\v \in H_i,u \in  V \setminus H_i}}\langle f_v - f_u, \mu_i\rangle  \leq \sum_{i = 1}^k \sum_{v \in \M}\left| \left \langle f_v- \frac{1}{d}\sum_{u \in \nei_G(v)}f_u, \mu_i \right \rangle \right| \, .
\end{equation}
By Cauchy-Schwarz, we have 
\begin{equation}\label{eq:ub2}
\sum_{i = 1}^k \sum_{v \in \M}\left| \left \langle f_v- \frac{1}{d}\sum_{u \in \nei_G(v)}f_u, \mu_i \right \rangle \right| 
\leq \sqrt{k |\M|}\cdot \sqrt{\sum_{i=1}^k \sum_{v \in \M}\left \langle f_v- \frac{1}{d}\sum_{u \in \nei_G(v)}f_u, \mu_i \right \rangle^2 }  
\end{equation}
It remains to upper-bound the second factor in the right-hand side of~\eqref{eq:ub2}
For simplicity of notation, let $\alpha_v \coloneqq  f_v- \frac{1}{d}\sum_{u \in \nei_G(v)}f_u$ for $v \in \M$. 
We have
\begin{equation}\label{eq:ub3}
\begin{aligned}
     \sum_{i=1}^k \sum_{v \in \M}\left \langle f_v- \frac{1}{d}\sum_{u \in \nei_G(v)}f_u, \mu_i \right \rangle^2 
     & = 
       \sum_{i=1}^k \sum_{v \in \M}\left \langle \alpha_v, \mu_i \right \rangle^2 
 && \\
       &  \leq  \eta\cdot\frac{k}{n}\sum_{v \in \M} \sum_{i=1}^k |C_i|\left \langle \alpha_v, \mu_i \right \rangle^2 && \text{using $\eta = \frac{\max_j|C_j|}{\min_j|C_j|} \geq \frac{n}{k}\cdot \frac{1}{|C_i|}$}\\
     & = \eta\cdot \frac{k}{n}\sum_{v \in \M }\alpha_v^T\left( \sum_{i =1}^k |C_i|\mu_i \mu_i^\top\right) \alpha_v && \\
     & \leq  \eta\cdot \frac{k}{n} \sum_{v \in \M}\left( 1 + \frac{4 \sqrt \epsilon}{\phi}\right) \|\alpha_v\|_2^2 && \text{by Lemma \ref{lemma:gklmsL9}} \\
     & \le \frac{2\eta k}{n} \sum_{v \in \M}\|\alpha_v\|_2^2 && \text{using $\frac{\epsilon}{\phi^2} \leq 10^{-5}$.}
    \end{aligned}
\end{equation}
By Lemma \ref{lemma:fx_close_to_avg_neighbors}, for all $v \in \M$, it holds that 
\begin{equation}\label{eq:ub4}
 \|\alpha_v\|_2^2 = \left\| f_v- \frac{1}{d}\sum_{u \in \nei_G(v)}f_u\right\|_2^2 \leq 4 \epsilon^2 \|f_v\|_2^2 \, ,
\end{equation}
so combining~\eqref{eq:ub3} and~\eqref{eq:ub4} we get 
\begin{equation*}
    \sqrt{k |\M|} \cdot \sqrt{\sum_{i = 1}^k \sum_{v \in \M}\left \langle f_v- \frac{1}{d}\sum_{u \in \nei_G(v)}f_u, \mu_i \right \rangle ^2} \leq  \sqrt{k|\M|} \cdot \sqrt{ \frac{2\eta k}{n}}\cdot \sqrt{\sum_{v \in \M}4 \epsilon^2 \|f_v\|_2^2} \le\epsilon \cdot \frac{4\eta k}{\sqrt n }\sqrt{|\M|\cdot \sum_{v \in \M}\|f_v\|_2^2} \, . 
\end{equation*}
\noindent
Combining the above equation with~\eqref{eq:ub1} and~\eqref{eq:ub2} gives the lemma statement. 
\end{proof}

\noindent
The last thing that remains is to upper-bound $\sum_{v \in \M} \| f_v\|_2^2$.

\begin{lemma}\label{claim:norms_of_bfar} Consider the setting of \Cref{fig:setting}. Then 
$\sum_{v \in \M}\|f_v\|^2_2 \leq  10^5 \eta^3 \cdot \frac{\epsilon}{\phi^4}k$.
\end{lemma}

\begin{proof}
We can split the sum as 
\begin{equation}\label{eq:l2_sum_bound}
\begin{aligned}
    \sum_{v \in \M}\|f_v\|^2_2 & = \sum_{\substack{v \in  \M: \\ \|f_v\|^2_2 \leq 4 \eta \cdot \frac{k}{n}}} \|f_v\|_2^2 +  \sum_{\substack{v \in \M: \\ \|f_v\|^2_2 > 4 \eta \cdot \frac{k}{n}}} \|f_v\|_2^2  \\
    & \leq  |\M| \cdot 4 \eta  \cdot \frac{k}{n} + \sum_{\substack{v \in  \M: \\ \|f_v\|^2_2 > 4 \eta \cdot \frac{k}{n}}} \|f_v\|_2^2   \\
    & \leq 8 \cdot 10^4 \cdot \eta^3 \cdot \frac{\epsilon}{\phi^4 }k  +  \sum_{\substack{v \in  \M: \\ \|f_v\|^2_2 > 4 \eta \cdot \frac{k}{n}}} \|f_v\|_2^2 \qquad \text{by \Cref{lemma:impostor_n_cross_size}.}
\end{aligned}
\end{equation}
It remains to bound the second term. 
Suppose $v \in \M$ is such that $\|f_v\|^2_2 > 4 \eta  \cdot \frac{k}{n}$. 
By Remark \ref{rem:mu_i_norm}, we have $\|\mu_{\iota(v)}\|_2^2 \leq \frac{2\eta k}{n}$, so then $\|\mu_{\iota(v)}\|_2 \leq \frac{1}{2}\|f_v\|_2$. This gives

\begin{equation*}
    \|f_v\|_2^2  = 4\left(\|f_v\|_2 - \frac{1}{2}\|f_v\|_2 \right)^2 \leq 4(\|f_v\|_2 - \|\mu_{\iota(v)}\|_2)^2 \leq4 \|f_v - \mu_{\iota(v)}\|^2_2 \, . 
\end{equation*}
\noindent
Summing over all $v \in \M$ with $\|f_v\|^2_2 > 4\eta \cdot \frac{k}{n}$, and applying \Cref{claim:sum_of_distances}, we get 

\begin{align*}
    \sum_{\substack{v \in  \M: \\ \|f_u\|^2_2 > 4\eta \cdot \frac{k}{n}}} \|f_v\|_2^2   \leq 4 \sum_{\substack{v \in  \M: \\ \|f_v\|^2_2 > 4\eta \cdot \frac{k}{n}}} \|f_v - \mu_{\iota(v)}\|_2^2 \leq 4 \sum_{i = 1}^k \sum_{v \in C_i} \| f_v - \mu_{i}\|^2_2 \leq 16 \frac{\epsilon}{\phi^2}k \, .
\end{align*}
Combining the above inequality with~\eqref{eq:l2_sum_bound} completes the proof. 
\end{proof}

\noindent
We are now ready to prove \Cref{lemma:D_ub}. 
\begin{proof}[Proof of \Cref{lemma:D_ub}]
Combining the results from this subsection, we obtain:
\begin{align*}
    \left|\bigcup_{i \in [k]} P_i \right| & \leq  \frac{n}{k} \cdot  \frac{20\eta}{d\phi}\cdot \sum_{i=1}^k \sum_{\substack{(v,u) \in E :\\v \in H_i,u \in  V \setminus H_i}}\langle f_v - f_u, \mu_i\rangle &&  \text{by \Cref{lemma:D_ub_step1}} \\
    & \leq 80\eta^2 \cdot \frac{\epsilon}{\phi} \sqrt{n} \sqrt{|\M|\cdot \sum_{v \in \M}\|f_v\|_2^2} && \text{by \Cref{lemma:ub_intermediate}} \\ 
    &  \leq 10^5 \eta^3 \cdot \frac{\epsilon^{3/2}}{\phi^3} n \sqrt{ \sum_{v \in \M}\|f_v\|_2^2} && \text{by \Cref{lemma:impostor_n_cross_size}} &&  \\ 
    &  \leq 10^5 \eta^3 \cdot \frac{\epsilon^{3/2}}{\phi^3} n \sqrt{ 10^5 \eta^3 \cdot \frac{\epsilon}{\phi^4}k} && \text{by \Cref{claim:norms_of_bfar}} &&  \\ 
    & \leq 10^8 \eta^5 \cdot \frac{\epsilon}{\phi^4} n  && \text{since $\epsilon k \leq \phi^2$ by Table \ref{fig:setting}.}
\end{align*}
\end{proof}

\subsection{Analysis}\label{subsec:analysis}
In this section, we analyze the misclassification rate and the probability of success of \Cref{alg:sketch+labels}. Just like with the simpler polynomial-time algorithm described in \Cref{sec:poly}, we upper bound the expected number of misclassified vertices over the randomness of $\sigma$ in each of the lines of \Cref{alg:sketch+labels} (Lemmas \ref{claim:false_pos} and \ref{cor:false_neg_updated}), and conclude with a simple applications of Markov's inequality.

Recall from \Cref{def:rob_con} the notion of robustly connected and robustly separated vertices. In \Cref{sec:robcon}, we show that the primitive \robcon \ (\Cref{alg:rob_con}) correctly distinguishes robustly connected vertices (and outputs ``yes'') and robustly separated vertices (and outputs ``no'') with high probability. Conditioned on the correctness of \robcon \, every vertex $u$ that is misclassified in line~\eqref{line:impostor} has to: have a wrong label $\sigma(u)$,  not be robustly separated from $\Mapx\cup\NMapx$, and satisfy $\tau(u) \notin \{*, \sigma(u)\}$. We think of these vertices as \emph{false positives} since the procedure \robcon, which was supposed to identify these vertices as mislabeled through the connectivity test, failed to do so.  We bound the number of such vertices in \Cref{claim:false_pos}.

\begin{lemma}[Misclassification in line~\eqref{line:impostor}]\label{claim:false_pos}   Consider the setting of \Cref{fig:setting}, and let $\iota,\sigma,\tau$ as per \Cref{fig:labels}. The expected fraction of vertices $u \in V$ over the draw of $\sigma$ such that $\sigma(u)$ is wrong and $u$ is not robustly separated (as per \Cref{def:rob_con}) from $\Mapx\cup\NMapx$ in the $(\tau(u),\sigma(u))$-cross graph, is at most $O(\epsilon\delta\log(1/\delta)/\phi^6)$, i.e.

    \begin{equation*}
        \E_\sigma\left[\left|\{u \in V: \, \tau(u) \notin \{\sigma(u), *\}\text{ and }\sigma(u) \neq \iota(u)  \text{ and } u \slashed{\perp}_{G_{\tau(u),\sigma(u)}} \Mapx\cup\NMapx\}\right| \right] \le 2\cdot 10^{10}\eta^5 \frac{\delta\epsilon}{\phi^6 } \log\frac{1}{\delta} \cdot  n\, .
    \end{equation*}
\end{lemma}

\begin{proof}
 Denote the set that we aim to bound as $\fp$ (for ``False Positives"), i.e., 

 $$ \fp \coloneqq \{u \in V: \, \tau(u) \notin \{\sigma(u), *\}\text{ and }\sigma(u) \neq \iota(u)  \text{ and } u \slashed{\perp}_{G_{\tau(u),\sigma(u)}} \Mapx\cup\NMapx\} \, .$$

\noindent
The proof of this lemma relies on one crucial observation: any crossing walk which starts in a mislabeled vertex outside of $\M\cup\NM$ must contain a ``special'' edge, which we will define later. Indeed, any such walk must pass through some vertex in $\M\cup\NM$ so, at the very least, it must contain an edge incident on $\M\cup\NM$. This allows us to bound the number of crossing walks starting in a mislabeled vertex outside of $\M\cup\NM$. At the same time, by the definition of \fp, any vertex in \fp \ must be a starting point of sufficiently many such walks, which allows us to upper bound the size of \fp.

 Now we move on to the formal proof. Recall that we use $\Lambda$ to denote the set of mislabeled vertices (as per \Cref{def:mislabeled}), and note that 
 \[|\fp \cap \NMapx| \leq |\mislabeled \cap \NMapx|,\]
 since, by the definition of $\fp$, any vertex in $\fp$ must be mislabeled. Therefore, the expected contribution of the vertices in $\NMapx$ to $\fp$ is
 \begin{equation}\label{eq:FPcapNM}
 \E_{\sigma}[|\fp\cap \NMapx|] \le \delta\cdot|\NMapx| \leq 10^8\eta^5 \cdot \frac{\delta\epsilon}{\phi^4 }n \, ,
 \end{equation}
 \noindent where the last inequality uses \Cref{lemma:N(B_far)_reg}. 

 In what follows, we bound $|\fp\setminus \NMapx|$.
 \\~\\
 Fix $i, j \in [k]$ with $i \neq j$. We denote the subset of vertices $u \in \fp$  with $(\tau(u), \sigma(u)) = (i, j)$ as $\fp_{i, j}$. Recall from \Cref{def:crossgraph} that $V_{i, j} = (\specapx(i) \cap \lab(j)) \cup \Mapx$ and $E_{i,j} = E(V_{i,j})$. Define $\crossvstarapx_{i, j} = \specapx(i) \cap \lab(j)$. We have that $\fp_{i, j} \subseteq \crossvstarapx_{i, j}$, since for any $u \in \Mapx$ we have that $\tau(u) = *$. Since all of the $\crossvstarapx_{i, j}$ are disjoint (by \Cref{lemma:disjoint_balls}), we have that all $\fp_{i, j}$ are also disjoint. Moreover, we have that $\fp \setminus \NMapx= \cup_{i, j: i\neq j}(\fp_{i, j}\setminus \NMapx)$. We will now upper bound the size of each $\fp_{i,j}\setminus \NMapx$. 
 
 For a fixed $\sigma$, let $p_{i, j}$ denote the probability that a vertex $u \sim V$ sampled uniformly at random belongs to $V_{i,j}^* \setminus \NMapx$, is mislabeled, and that a lazy random walk of length $L$ in the graph $\crossgapx_{i, j}$ started from $u$ is a crossing walk (as per \Cref{def:crosswalk}), i.e. 
\begin{equation}\label{eq:p}
 p_{i,j} \coloneqq \frac{1}{n}\sum_{u \in V_{i,j}^* \cap \mislabeled \setminus \NMapx} \Pr[ \text{a lazy random walk in $\crossgapx_{i, j}$ of length $L$ stated from $u$ is a crossing walk}].
\end{equation}

\noindent
Note that $\fp_{i,j}\setminus \NMapx \subseteq V_{i,j}^* \cap \mislabeled \setminus \NMapx$ by definition of $\fp_{i,j}$. Furthermore, for every vertex $u \in \fp_{i,j}$, we have that $u$ is \emph{not} robustly separated from $\Mapx\cup\NMapx$, and so the probability that a sampled random walk in $\crossgapx_{i,j}$ of length $L$ started from $u$ is a crossing walk,  must be at least $0.1$ (by \Cref{def:rob_con}). Hence, 
 \begin{equation}\label{eq:lower_b_p}
 \begin{aligned}
 p_{i, j} & \geq \frac{1}{n} \sum_{u \in \fp_{i,j} \setminus \NMapx}\Pr[ \text{a lazy random walk in $\crossgapx_{i, j}$ of length $L$ stated from $u$ is a crossing walk}] \\
 & \geq 0.1 \frac{|\fp_{i, j} \setminus \NMapx|}{n}.
 \end{aligned}
 \end{equation}

 \noindent
The above inequality formally shows that any vertex in \fp \ is a starting point of sufficiently many crossing walks starting in a mislabeled vertex outside of $\M\cup\NM$. We now formally prove that any such crossing walk has to contain one of the few ``special'' edges. This will help us obtain an upper bound on $p_{i,j}$, which we will use to upper bound $|\fp_{i,j}|$. 

Recall the spectral impostors $ \im(i,j) \coloneqq \spec(i) \cap C_j $ (\Cref{def:impostor}), and let 
$$X_{i, j} \coloneqq \Mapx \cup \NMapx \cup \imapx(i, j).$$ 
\noindent
We now show that every crossing walk in $\crossgapx_{i, j}$ which starts from a mislabeled vertex $u \notin \Mapx\cup\NMapx$, must have an edge which connects a vertex from $X_{i, j}$ to some vertex outside of $X_{i, j}$. This is our notion of ``special" edges: 
\begin{claim}\label{claim:pass_through}
    Every walk in the graph $\crossgapx_{i, j}$ which starts from a mislabeled vertex $u \notin \Mapx\cup\NMapx$ and passes through $\Mapx \cup \NMapx$, must use at least one edge $e = (v,w)$ such that $v \in V_{i,j} \setminus X_{i,j}$ and $w \in X_{i,j}$. 
\end{claim}
\begin{proof}
Let $\mathbf{w}$ be a walk $\crossgapx_{i, j}$ which starts from a mislabeled vertex $u \notin \Mapx\cup\NMapx$ and passes through $\Mapx \cup \NMapx$. First, we will show that the starting point must be outside of $X_{i, j}$. Indeed, let $u \in V_{i,j}$ be the starting point of $\mathbf{w}$. Then
\begin{itemize}
    \item By assumption, $u \notin \Mapx\cup\NMapx$;
    \item By definition of $\im(i,j)$ and $V_{i,j}^*$, every vertex in $\imapx(i, j)\cap \crossvstarapx_{i,j}$ is correctly labeled. By assumption, $u$ is mislabeled, so $u \notin \imapx(i, j)$. 
\end{itemize} 
From the above two bullet points, we see that the starting vertex $u$ is indeed not in $X_{i,j}$. 
By assumption, $\mathbf{w}$ passes through $\M \cup \NM \subseteq X_{i,j}$, so is must use at lest one edge $e=(v,w)$ with $v \in V_{i,j} \setminus X_{i,j}$ and $w \in X_{i,j}$.

\end{proof}
Using the above claim, we can upper bound $p_{i,j}$ in terms of $|E_{i,j}(X_{i,j}, V_{i,j} \setminus X_{i,j})|$. 

\begin{claim}\label{claim:pij_ub}
    \[p_{i, j}\leq \frac{|E_{{i, j}}(X_{i, j}, \crossvapx_{i, j}\setminus X_{i, j})|L}{n\cdot 2d}.\]
\end{claim}
\begin{proof}
Recall that the lazy random walk on the graph $G_{i,j}$ has transition
matrix $\frac{1}{2d}A_{G_{i,j}} + \frac{1}{2}I.$ Equivalently, we can view this walk as a simple random walk on the graph $\widetilde{G}_{i,j}$
obtained by adding an additional $d$ self-loops to each vertex in $G_{i,j}$. In $\widetilde{G}_{i,j}$ every vertex has degree $2d$, and at each step we pick
one of the $2d$ incident edges uniformly at random. If the chosen edge is
a self-loop, the walk stays at the current vertex.

With this in mind, we can rewrite the definition \eqref{eq:p} of $p_{i,j}$ as 
$$p_{i,j} = \frac{1}{n}\cdot \frac{ | \{ \mathbf{w} : \mathbf{w} \text{ is a crossing walk in $\widetilde{G}_{i,j}$ of length $L$ starting from $ V_{i,j}^* \cap \mislabeled \setminus \NMapx$} \}|}{(2d)^L}.$$
By \Cref{claim:pass_through}, every crossing walk $\mathbf{w}$ in $\crossgapx_{i,j}$ starting from $ V_{i,j}^* \cap \mislabeled \setminus \NMapx$, must use  at least one edge from $E_{{i, j}}(X_{i, j}, \crossvapx_{i, j}\setminus X_{i, j})$. This gives 
$$p_{i,j} \leq \frac{1}{n}\cdot \frac{ | \{ \mathbf{w} : \mathbf{w} \text{ is a walk in $\widetilde{G}_{i,j}$ of length $L$ and uses } E_{i, j}(X_{i, j}, \crossvapx_{i, j}\setminus X_{i, j})  \}|}{(2d)^L}.$$
Finally, for fixed $\sigma$, the number of walks of length $L$ in $\widetilde \crossgapx_{i, j}$ which use an edge in $E_{i, j}(X_{i, j}, \crossvapx_{i, j}\setminus X_{i, j})$ is upper bounded by $|E_{{i, j}}(X_{i, j}, \crossvapx_{i, j}\setminus X_{i, j})|\cdot (2d)^{L-1}\cdot L$. This is because any such walk must have an edge from $E_{i, j}(X_{i, j}, \crossvapx_{i, j}\setminus X_{i, j})$ at place $l = 1, \ldots, L$. Once the edge from $E_{i, j}(X_{i, j}, \crossvapx_{i, j}\setminus X_{i, j})$ and the position $l$ are fixed, there are at most $(2d)^{L-1}$ choices for the remaining steps. 
\end{proof}
It remains to upper bound the quantities $|E_{{i, j}}(X_{i, j}, \crossvapx_{i, j}\setminus X_{i, j})|$. 
\begin{claim}\label{claim:Eij_ub}
$$ \E_{\sigma}\left[\sum_{\substack{i, j \in [k]:\\ i\neq j}}|E_{{i, j}}(X_{i, j}, \crossvapx_{i, j}\setminus X_{i, j})|\right] \leq 2 \cdot 10^{8}\eta^5 \cdot \frac{\delta \epsilon}{\phi^4 }d  \cdot n.$$ 
\end{claim}
\begin{proof}
 By definition of $X_{i,j}$, for every $i \neq j$, we have \newline
\smash{$E_{i,j}(X_{i,j}, V_{i,j}\setminus X_{i,j}) = E_{i,j}(\im(i,j), V_{i,j} \setminus X_{i,j}) \cup E_{i,j}(\M \cup \NMapx, V_{i,j} \setminus X_{i,j})$}. We start by bounding the contribution from $E_{i,j}(\im(i,j), V_{i,j} \setminus X_{i,j})$. 

Fix $i \neq j$. Observe that by definition of $X_{i,j}$, every vertex in $V_{i,j} \setminus X_{i,j}$ has to be mislabeled. This is because 
$V_{i,j} \setminus (\M \cup \im(i,j))$ is exactly the set of vertices that belong to $\spec(i) \cap \left( \lab(j) \setminus C_j \right)$. Therefore, for every edge in $E(\imapx(i, j), V\setminus X_{i, j})$, the probability (over the randomness of $\sigma$) that it belongs to the edge set $E_{i, j}$ of $G_{i,j}$ is at most $\delta$, since its end point in $V\setminus X_{i, j}$ has to be mislabeled. So
\begin{equation}\label{eq:Eij}
    \E_{\sigma}\left[|E_{i, j}(\imapx(i, j), \crossvapx_{i, j}\setminus X_{i, j})|\right] \leq \delta\cdot |E(\imapx(i, j), V\setminus X_{i, j})| \leq \delta\cdot d\cdot|\imapx(i, j)|,
\end{equation}
where the second inequality follows since $G$ is a $d$-regular graph.

Next, we bound the contribution from $E_{i,j}(\M \cup \NMapx, V_{i,j} \setminus X_{i,j})$. Similarly, for every edge in $E(\Mapx\cup\NMapx, V\setminus X_{i, j})$ the probability that it belongs to $E_{i,j}$ is at most $\delta$, since its end point in $V\setminus X_{i, j}$ has to be mislabeled. Therefore, 
\[\E_{\sigma}\left[\sum_{\substack{i, j \in [k]:\\ i\neq j}}|E_{{i, j}}(\Mapx\cup\NMapx, \crossvapx_{i, j}\setminus X_{i, j})|\right] \leq  \E_{\sigma}\left[\sum_{u \in \M \cup \NM} \left|\nei_G(u) \cap \mislabeled\right|\right]  \leq \delta\cdot d\cdot|\Mapx\cup\NMapx|\, .\]

Combining the above with Equation \eqref{eq:Eij}, we get 
\begin{equation*}
    \label{eq:exp_upper_b}
    \E_{\sigma}\left[\sum_{\substack{i, j \in [k]:\\ i\neq j}}|E_{{i, j}}(X_{i, j}, \crossvapx_{i, j}\setminus X_{i, j})|\right] \leq \delta\cdot d\cdot|\Mapx\cup\NMapx| + \sum_{i, j: i\neq j}\delta\cdot d\cdot|\imapx(i, j)| \leq  \delta d |\NMapx| + \delta d |\Mapx\cup \imapx|  \, .
\end{equation*}
where the last inequality uses that $\{ \im(i,j)\}_{i \neq j}$ partitions $\im$ (by \Cref{rem:disjoint_im}), and that $\im$ is disjoint from $\M$ (by \Cref{def:spec} and \Cref{def:impostor}). 
Finally, since \smash{$|\im\cup\Mapx| \leq 2 \cdot 10^4 \eta^2 \cdot \frac{\epsilon}{\phi^4 }n$} by \Cref{lemma:impostor_n_cross_size} and \smash{$|\NMapx| \leq  10^8\eta^5 \cdot \frac{\epsilon}{\phi^4 }n$} by \Cref{lemma:N(B_far)_reg},   we get 
\[ \E_{\sigma}\left[\sum_{\substack{i, j \in [k]:\\ i\neq j}}|E_{{i, j}}(X_{i, j}, \crossvapx_{i, j}\setminus X_{i, j})|\right] \leq 2\cdot 10^8\eta^5 \cdot \frac{\epsilon}{\phi^4 }n \, .\]

\end{proof}

\noindent

Putting everything together, we now obtain an upper bound on $|\fp \setminus \NMapx|$. By Equation~\eqref{eq:lower_b_p} and \Cref{claim:pij_ub}, for every fixed $\sigma$, we have 

$$ |\fp_{i, j} \setminus \NMapx| \leq 10n\cdot  p_{i,j} \cdot \leq \frac{5L}{d}  |E_{i, j}(X_{i, j}, \crossvapx_{i, j}\setminus X_{i, j})| .$$
\noindent
Summing over all $i \neq j$, taking the expectation over $\sigma$, and applying \Cref{claim:Eij_ub}, we get 

\begin{equation*}
    \E_{\sigma}[|\fp\setminus \NMapx|]  =  \sum_{\substack{i, j \in [k]:\\ i\neq j}}\E_{\sigma}\left[|\fp_{i, j}|\right] 
 \leq \frac{5L}{d}\E_{\sigma}\left[\sum_{\substack{i, j \in [k]:\\ i\neq j}}|E_{{i, j}}(X_{i, j}, \crossvapx_{i, j}\setminus X_{i, j})|\right]   \leq  10^{9}\eta^5 \cdot \frac{\delta \epsilon}{\phi^4 } L \cdot n. 
\end{equation*}
Recalling that $L = \lceil 150/\phi^2 \cdot {\log(1/\delta)} \rceil$ (see \Cref{fig:setting}), and combining with Equation \eqref{eq:FPcapNM} we get the desired result.

\end{proof}

\begin{lemma}\label{claim:reachability_updated_new}
Consider the setting of \Cref{fig:setting}, let $\sigma$ as per \Cref{fig:labels}, and let $\crossvstarapx_{i,j}= \spec(i) \cap \lab(j)$. For all $i,j \in [k]$ with $i \neq j$, one has

\begin{equation*}
    \left|\left\{u \in  \imapx(i,j)\setminus \mislabeled: u \slashed{\rightsquigarrow}_{\crossgapx_{i, j}}\Mapx\cup\NMapx\right\}\right|  \le 10\delta |\im(i,j)| +5L|\im(i,j) \cap \mislabeled| + \frac{5L}{d}|E(\im(i,j),\crossvstarapx_{i,j} \setminus \im(i,j))|  \, .
\end{equation*}

\end{lemma}
\begin{remark} We note that, while many of our results related to the labeling $\sigma$ hold in expectation or with high probability, \Cref{claim:reachability_updated_new} holds deterministically for any labeling $\sigma: V \to [k]$. 
\end{remark}

\begin{proof}
Let $T$  ($T$ for ``true") be the set of vertices in $\crossvstarapx_{i,j}$ with the correct label, i.e. 
\[T = \crossvstarapx_{i,j} \setminus \mislabeled = \crossvstarapx_{i,j} \cap \imapx(i,j) = \imapx(i,j) \setminus \mislabeled. \]
Consider a lazy walk $\bf{w}$ starting at $u \in T$ that does not reach $\Mapx \cup \NMapx$ within $L$ steps in $\crossgapx_{i,j}$. Then, $\bf{w}$ must satisfy at least one of this conditions: $\bf{w}$ stays inside $T$ for all $L$ steps, or $\bf{w}$ reaches $\crossvstarapx_{i,j}\setminus T = \crossvstarapx_{i,j}\setminus \imapx(i, j)$. Formally, we have

\begin{equation}
\label{eq:twoterms}
    \Pr_{\walk \sim p_T^L[\crossgapx_{i, j}]}\left[\walk \slashed{\sim} \Mapx \cup \NMapx\right] \le \Pr_{\walk \sim p_T^L[\crossgapx_{i, j}]}\left[\forall\,  r \in [L], \, w_r \in T\right]+ \Pr_{\walk \sim p_T^L[\crossgapx_{i, j}]}\left[\walk \sim \crossvstarapx_{i,j}\setminus T \text{ and } \walk \slashed{\sim} \Mapx \cup \NMapx\right] \,,
\end{equation}
where $\walk \sim p^L_T[\crossgapx_{i, j}]$ describes $\walk$ as an $L$-step lazy random walk in $\crossgapx_{i, j}$ starting in $T$, as per \Cref{def:lazy_r_walk}. 

To bound the first term in~\eqref{eq:twoterms}, we first define a bigger graph $H = (V,F,\ell')$ where $F$ consists of the edges $E_{i,j} \cup E(T,\im(i,j)\setminus T)$ (recall from \Cref{def:crossgraph} that $E_{i,j}$ is the set of edges of the cross graph $G_{i,j}$), and $\ell':V\rightarrow \mathbb{N}$ assigns a number of self-loops $\ell'(u)$ to every $u \in V$ such that $\deg_H(u)=d$ (recall that self-loops are not in the edge set). Note that the cross graph $G_{i,j}$ is a subgraph of $H$. With this notation, we have the following claim, which relates the probability of staying in a set $S$ when one walks in a graph $H_2$, to the probability of staying $S$ when one walks in a subgraph $H_1 \subseteq H_2$. The proof of \Cref{claim:relategraphs} is deferred to  \Cref{subsec:relating}.

\begin{restatable}{claim}{relategraphs}
    \label{claim:relategraphs}
    Let $d \ge 3$ be an integer, let $G=(V,E)$, $H_1=(V,E_1,\ell_1)$, $H_2,=(V,E_2,\ell_2)$ be $d$-regular graphs such that $E_1 \subseteq E_2 \subseteq E$, and there exists $S \subseteq V$ such that $E_2 \setminus E_1 \subseteq E(S,V\setminus S)$. Then, for every $t \geq 1$, it holds that
    \begin{equation*}
        \Pr_{\walk \sim p_S^t[H_1]}\left[\forall\,  r \in [t], \, w_r \in S\right] \le \Pr_{\walk \sim p^t_S[H_2]}\left[\forall\,  r \in [t], \,  \, w_r \in S\right] + t \cdot \frac{|E_2 \setminus E_1|}{2d|S|} \, .
    \end{equation*}
\end{restatable}
\noindent
We want to apply  \Cref{claim:relategraphs} with $H_1 = \crossgapx_{i,j}$, $H_2=H$, $t = L$ and $S = T$. That way, to bound the first term in~\eqref{eq:twoterms} it will suffice to bound the probability of reaching $T$ in the graph $H$. We now verify that our setting meets the conditions of  \Cref{claim:relategraphs}. We have $E_1 = E_{i,j}$ and $E_2 = F = E_{i,j} \cup E(T,\im(i,j)\setminus T)$, so $E_1 \subseteq E_2$ and $E_2 \setminus E_1 = E(T,\im(i,j)\setminus T) \subseteq E(T, V\setminus T)$, as required. Hence, we have
\begin{equation}
\label{eq:applying}
        \Pr_{\walk \sim p^L_T[\crossgapx_{i,j}]}\left[\forall\,  r \in [L], \, w_r \in T\right] \le \Pr_{\walk \sim p^L_T[H]}\left[\forall\,  r \in [L], \, w_r \in T\right] + L \cdot \frac{|E(T,\imapx(i,j)\setminus T)|}{2d|T|} \, .
\end{equation}
To bound the first term, it remains to bound $\Pr_{\walk \sim p^L_T[H]}\left[\forall\,  r \in [L], \, w_r \in T\right]$. We use the following fact.
\begin{restatable}{claim}{explikesets}\label{claim:exp-like_sets}
Let $d\ge 3$ be an integer, let $\psi \in (0,1)$, let $H
=(V,E',\ell')$ be a $d$-regular graph, and let $\emptyset \neq S \subseteq V$ such that $|E'(Q,V\setminus Q)| \ge \psi d |Q|$ for all $Q \subseteq S$. Then, for every $t \ge 1$ one has
\begin{equation*}
    \Pr_{\walk \sim p^t_S[H]}\left[\forall \, r \in [t], \, w_r \in S \right] \le 2 \cdot \exp\left(-\frac{\psi^2}{36}t\right) \, .
\end{equation*}

\end{restatable}
\noindent
We defer the proof of \Cref{claim:exp-like_sets} to \Cref{subsec:explike}, but the main idea is as follows: a graph $H$ as in the lemma statement behaves like an expander locally to $S$, so a lazy random walk has a constant probability of leaving $S$ at every step, resulting in an exponential decay. 

We apply \Cref{claim:exp-like_sets} to the graph $H=(V,F , \ell')$ and the set $T$. 
Since $T \subseteq \im(i,j)$ and since \sloppy  \smash{$E(T, \im(i,j) \setminus T) \subseteq F$} by definition of $F$, we get that for all  $Q \subseteq T$, it holds that $F(Q, V \setminus Q) \subseteq E(Q, \im(i,j) \setminus Q)$, and so by \Cref{lem:imposters_expand}, it holds that 
$ |F(Q, V \setminus Q)| \geq |E(Q, \im(i,j) \setminus Q) | \geq \frac{\phi}{2}d|Q|$. So we can apply \Cref{claim:exp-like_sets} with $\psi = \phi/2$. Combining this with ~\eqref{eq:applying} and using the setting $L = \lceil \frac{150}{\phi^2}\log(1/\delta)\rceil$ (as per \Cref{fig:setting}), we have
\begin{equation}
\label{eq:firstterm}
    \Pr_{\walk \sim p^L_T[\crossgapx_{i,j}]}\left[\forall\,  r \in [L], \, w_r \in T\right] \le \delta + L \cdot \frac{|E(T,\imapx(i,j)\setminus T)|}{2d|T|} \, .
\end{equation}

\noindent
Then, we bound the second term in~\eqref{eq:twoterms} and get
\begin{align}
    \Pr_{\walk \sim p^L_T[\crossgapx_{i,j}]}\left[\walk \sim \crossvstarapx_{i,j}\setminus T \text{ and } \walk \slashed{\sim} \Mapx \cup \NMapx\right] & \le \Pr_{\walk \sim p^L_T[\crossgapx_{i,j}]}\left[\exists r \in [L-1]: \, (w_i,w_{i+1}) \in E(T, \crossvstarapx_{i,j}\setminus T)\right] \\
    & \le L \cdot \frac{|E(T,\crossvstarapx_{i,j}\setminus T)|}{2d|T|} \label{eq:secondterm} \, ,
\end{align}
and plug~\eqref{eq:firstterm} and~\eqref{eq:secondterm} into~\eqref{eq:twoterms} to obtain
\begin{equation}
\label{eq:combined}
     \Pr_{\walk \sim p^L_T[\crossgapx_{i,j}]}\left[\walk \slashed{\sim} \Mapx \cup \NMapx\right] \le \delta + L \cdot \frac{|E(T,\imapx(i,j)\setminus T)|}{2d|T|} + L \cdot \frac{|E(T,\crossvstarapx_{i,j}\setminus T)|}{2d|T|}  \, .
\end{equation}
By Markov's inequality, we have
\begin{align*}
    \Pr_{u \sim \unif(T)} \left[u \slashed{\rightsquigarrow}_{\crossgapx_{i, j}}\Mapx\cup\NMapx\right] & = \Pr_{u \sim \unif(T)} \left[\Pr_{\walk \sim p^L_{u}[\crossgapx_{i, j}]}\left[\walk \sim \Mapx\cup\NMapx\right] <0.9\right] \\
    & \le  10 \cdot \Pr_{\walk \sim p^L_T[\crossgapx_{i,j}]}\left[\walk \slashed{\sim} \Mapx \cup \NMapx\right] \, .
\end{align*}
Recalling $T = \imapx(i,j)\setminus \mislabeled = \crossvstarapx_{i,j} \setminus \mislabeled$, we plug~\eqref{eq:combined} into the above bound to obtain
\begin{align*}
     \left|\left\{u \in  \imapx(i,j)\setminus \mislabeled: u \slashed{\rightsquigarrow}_{\crossgapx_{i, j}}\Mapx\cup\NMapx\right\}\right| & \le 10\delta |T| +\frac{10L}{2d}|E(T,\imapx(i,j)\setminus T)| + \frac{10L}{2d}|E(T,\crossvstarapx_{i,j} \setminus T)|  \, .
\end{align*}

Note that $T \subseteq \imapx(i, j)$, that  $\crossvstarapx_{i,j} \setminus T =  \crossvstarapx_{i,j} \setminus \im(i,j) $ and that $\im(i,j) \setminus T \subseteq \mislabeled$.    Hence, we conclude
\begin{align*}
     \left|\left\{u \in  \imapx(i,j)\setminus \mislabeled: u \slashed{\rightsquigarrow}_{\crossgapx_{i, j}}\Mapx\cup\NMapx\right\}\right| & \le 10\delta |\im(i,j)| +5L|\im(i,j) \cap \mislabeled| + \frac{5L}{d}|E(\im(i,j),\crossvstarapx_{i,j} \setminus \im(i,j))|  \, .
\end{align*}

\end{proof}

\noindent
Every vertex misclassified in line~\eqref{line:non-impostor} has to: be an impostor, be not robustly connected to $\Mapx\cup\NMapx$, and satisfy $\tau(u) \notin \{*, \sigma(u)\}$. We think of these vertices as \emph{false negatives} since the procedure \robcon, which was supposed to identify these vertices as correctly labeled impostors through connectivity test, failed to do so.  We bound the number of such vertices in \Cref{cor:false_neg_updated}.

\begin{lemma}[Misclassification in line~\eqref{line:non-impostor}]\label{cor:false_neg_updated} 
      Consider the setting of \Cref{fig:setting}, and let $\iota,\sigma,\tau$ as per \Cref{fig:labels}. The expected fraction of vertices $u \in V$ over the draw of $\sigma$ such that $\tau(u) \in [k]$ is wrong but $\sigma(u)$ is correct and $u$ is not robustly connected to $\Mapx\cup\NMapx$ in the $(\tau(u),\sigma(u))$-cross graph (as per \Cref{def:rob_con}), is at most $O(\epsilon\delta\log(1/\delta)/\phi^6)$, i.e.

    \begin{equation*}
        \E_\sigma\left[\left|\{u \in V: \, \tau(u) \notin \{*,\sigma(u)\} \text{ and }  \tau(u)\neq \iota(u)  \text{ and } u \slashed{\rightsquigarrow}_{\crossgapx_{\tau(u),\sigma(u)}} \Mapx\cup\NMapx\}\right| \right] \le 5 \cdot 10^7 \eta^2 \cdot \frac{\epsilon\delta\log(1/\delta)}{\phi^6} \cdot n\, .
    \end{equation*}
\end{lemma}

\begin{proof} Denote the set that we aim to bound as $\fn$ (for ``False Negatives'').  Fix $i, j \in [k]$ with
$i \neq j$. We denote the subset of vertices $u \in \fn$  with $(\tau(u), \sigma(u)) = (i, j)$ as $\fn_{i, j}$. Note that $\fn_{i, j} \subseteq \imapx(i, j)$, since for any $u \in \Mapx$ we have that $\tau(u) = *$. From \Cref{rem:disjoint_im}, all $\imapx(i, j)$ are disjoint, and therefore we have that all $\fn_{i, j}$ are also disjoint. Moreover, by the definition of $\fn$, we have that $\fn = \cup_{i, j: i\neq j}\fn_{i, j}$.

 \Cref{claim:reachability_updated_new} bounds the number of impostors which have the correct $\sigma$ label and are not robustly connected to $\Mapx\cup\NMapx$ in the graph $\crossgapx_{i, j}$. In other words, \Cref{claim:reachability_updated_new} bounds precisely the size of $\fn_{i, j}\setminus \mislabeled$:

\[|\fn_{i, j}\setminus \mislabeled|\leq \frac{5L}{d}\cdot|E(\imapx(i, j), \crossvstarapx_{i, j}\setminus\imapx(i, j))| + 5L\cdot|\imapx(i, j)\cap \mislabeled| + 10\delta\cdot|\imapx(i, j)|,\]
where $\crossvstarapx_{i,j}= \spec(i) \cap \lab(j)$.
Therefore,
\begin{align*}
    |\fn\setminus \mislabeled| & \leq \sum_{\substack{i, j \in [k]:\\ i\neq j}}\left(\frac{5L}{d}\cdot|E(\imapx(i, j), \crossvstarapx_{i, j}\setminus\imapx(i, j))| + 5L\cdot|\imapx(i, j)\cap \mislabeled| + 10\delta\cdot|\imapx(i, j)| \right)  \\
    & = \sum_{\substack{i, j \in [k]:\\ i\neq j}}\frac{5L}{d}\cdot|E(\imapx(i, j), \crossvstarapx_{i, j}\setminus\imapx(i, j))|+ 5L\cdot|\imapx\cap\mislabeled| + 10\delta\cdot|\imapx| \, ,
\end{align*}
where the last equality follows from the fact that all $\imapx(i, j)$ are disjoint. Recall that \Cref{lemma:impostor_n_cross_size} ensures $|\im| \le 2 \cdot 10^4 \cdot \eta^2 \cdot \frac{\epsilon}{\phi^4 }n$, so
\[\E_{\sigma}\left[5L\cdot|\imapx\cap\mislabeled| + 10\delta\cdot|\imapx|\right] \leq 2 \cdot 10^5 \cdot \eta^2 \cdot \frac{\epsilon\delta}{\phi^4 }L \cdot n \, .\]
Now note that
\[  \sum_{i, j: i\neq j}\frac{5L}{d}|E(\imapx(i, j), \crossvstarapx_{i, j}\setminus\imapx(i, j))|  \leq \frac{5L}{d}|E(\imapx, \mislabeled\setminus\imapx)| \, ,\]
because any $u \in \crossvstarapx_{i, j}\setminus\imapx(i, j)$, by the definition of $\crossvstarapx_{i, j}$, is a mislabeled vertex, and because the sets $\crossvstarapx_{i, j}$ are disjoint. In the graph $G$, there are at most $d\cdot |\imapx|$ many edges with at least one endpoint in $\imapx$. For any edge, the probability to have at least one mislabeled endpoint  is bounded by $2\delta$. Therefore, 
\[\E_{\sigma}\left[|E(\imapx, \mislabeled\setminus\imapx)|\right] \leq 2\delta\cdot d\cdot |\imapx| \leq 4 \cdot 10^4 \cdot \eta^2 \cdot \frac{d\epsilon\delta}{\phi^4 }n \, ,\]
so, using the setting $L = \lceil 150\frac{\log(1/\delta)}{\phi^2} \rceil$ from \Cref{fig:setting}, we get
\[\E_{\sigma}\left[|\fn\setminus \mislabeled|\right] \leq 3 \cdot 10^5 \cdot \eta^2 \cdot \frac{\epsilon\delta L}{\phi^4 }n \leq 4.5 \cdot 10^7 \cdot \eta^2 \cdot \frac{\epsilon\delta\log(1/\delta)}{\phi^6} \cdot n \, .\]
\noindent
To finish the proof, it remains to bound the expected size of $\fn\cap\mislabeled$. Since $\fn \subseteq \im$, by \Cref{lemma:impostor_n_cross_size} we have
\[\E_{\sigma}[|\fn\cap \mislabeled|] \leq 2 \cdot 10^4 \cdot \eta^2 \cdot \frac{\epsilon\delta}{\phi^4 }n \, .\]
From here, we conclude 
$\E_{\sigma}[|\fn|] \leq \E_{\sigma}[|\fn\setminus\mislabeled|] + \E_{\sigma}[|\fn\cap\mislabeled|] \leq 5 \cdot 10^7 \cdot \eta^2 \cdot \frac{\epsilon\delta\log(1/\delta)}{\phi^6} \cdot n$
as desired.
\end{proof}

\noindent
We are now equipped to prove the main theorem.

\begin{proof}[Proof of \Cref{thm:sublinear}]
    The algorithm that we refer to in the statement of \Cref{thm:sublinear}, prepares a data structure as follows.
    \begin{enumerate}
        \item First, run the preprocessing of \Cref{thm:spec_dot_prod_oracle}  with precision parameter $\xi = \phi^2/(20^4\eta^2 k)$ (so as to meet the requirement of \Cref{fig:setting}). This produces a data structure that we can be queried with pairs $V \times V$ to obtain access to our approximate inner product function $\specdp$.
        \item Next, run \Cref{alg:app_centers} and \Cref{alg:permutation} using $\specdp$ as $\langle f_x, f_y \rangle_{\text{apx}}$ (and implementing access to it via the data structure from the previous step), and obtain vectors $\mu_1,\dots,\mu_k$\footnote{Again, in reality these are vertices $u_1,\dots,u_k$ that indicate the embedding that should be used for the corresponding approximate cluster mean. We write $\mu_1,\dots,\mu_k$ for readability.}.
        \item Additionally, generate a table $r$ of $\lceil 450 \log n \rceil\times \lceil L\rceil$ (recall the definition of $L$ from \Cref{fig:setting}) uniformly random integers in $[2d]$. The table $r$ is then used as seed that we use to generate lazy random walks in  \Cref{alg:rob_con} \robcon. More particularly, whenever we need to generate the $j$-th (out of $\lceil L\rceil $) step of the $i$-th (out of $\lceil 450 \log n \rceil$) random walk starting from a vertex $u$, we take the $r_{i, j}$-th neighbor of the $(j-1)$-st vertex of the $i$-th walk in the graph cross graph $G_{\tau(u), \sigma(u)}$. If the value $r_{i, j}$ is greater than $d$, then the $j$-th step of the $i$-th walk is a self-loop.  
    \end{enumerate}
    With the data structure resulting from these three steps, we answer a query $u \in V$ by running on it \Cref{alg:sketch+labels} with $\specdp$ as $\langle f_x, f_y \rangle_{\text{apx}}$ and feeding it $\mu_1,\dots,\mu_k$.
    \\~\\
    Denote by $\nu$  the randomness involved in the preprocessing, in particular $\nu_1,\nu_2,\nu_3$ is the randomness used in the first, second, and third step above repsectively. Note that for fixed $\sigma$, \Cref{alg:sketch+labels} is deterministic after the preprocessing. Then, for fixed $\nu$ and $\sigma$, let $\alpha:V \rightarrow [k]$ be the mapping of a vertex to cluster id output by \Cref{alg:sketch+labels} with this realization of the processing randomness. We want to show that with high constant probability over the realizations of $\nu$ and $\sigma$, we have an $\alpha$ that is wrong on at most $O(\epsilon \delta/\phi^6 \cdot \log(1/\delta))$ fraction of vertices.

    First, we condition on the randomness $\nu_1$: by \Cref{thm:spec_dot_prod_oracle}, we know that the data structure produced by the first step does give access to a legitimate approximate inner product function $\specdp$ as per \Cref{def:apx} with high probability, so let us assume that the realization of $\nu_1$ meets this condition.

    Then, we condition on the internal randomness $\nu_2$ of \Cref{alg:app_centers} and \Cref{alg:permutation}: by \Cref{lemma:pi_computation}, with probability $0.97$ over $\nu_2$, one has that $\tmu_1,\dots,\tmu_k$ meet the condition of \Cref{def:approxmeans} with probability $0.999$ over $\sigma$, so let us assume that the realization of $\nu_2$ meets this condition. As we noted in the proof of \Cref{thm:polytime}, at this point the vectors $\mu_1,\dots,\mu_k$ are determined up to permutation $\pi$, which is computed by \Cref{alg:permutation} as a function of $\sigma$.

    As in the proof of \Cref{thm:polytime}, we consider an event $\textsc{Fail}(\sigma, \pi(\sigma))$. However, here we do not want to directly relate this event to the misclassification rate of $\alpha$, since the latter also depends on the randomness $\nu_3$ use to generate the seed for \Cref{alg:rob_con}. Hence, we define $\textsc{Fail}(\sigma, \pi(\sigma))$ to be the event such that, permuting the $k$ vectors given by  \Cref{alg:app_centers} with $\pi(\sigma)$, one of the following happens:
    \begin{enumerate}
        \item the number of vertices $u \in V$ such that $\tau(u) \in \{*,\sigma(u)\}$ and $\sigma(u)\neq \iota(u)$ is more than $2 \cdot 10^9 \cdot \eta^2 \cdot {\delta\epsilon}/{\phi^4 }n$;
        \item the number of vertices $u \in V$ such that $\tau(u) \notin \{\sigma(u), *\}$, $\sigma(u) \neq \iota(u)$, and $ u \slashed{\perp}_{G_{\tau(u),\sigma(u)}} \Mapx\cup\NMapx$ is more than $2\cdot 10^{15}\eta^5 {\delta\epsilon}/{\phi^6 } \cdot \log({1}/{\delta}) \cdot  n$;
        \item the number of vertices $u \in V$ such that $\tau(u) \notin \{\sigma(u), *\}$, $\tau(u) \neq \iota(u)$, and $u \slashed{\rightsquigarrow}_{\crossgapx_{\tau(u),\sigma(u)}} \Mapx\cup\NMapx$ is more than $5\cdot 10^{12}\eta^2 {\delta\epsilon}/{\phi^6 } \cdot \log({1}/{\delta}) \cdot  n$.
    \end{enumerate}
    As in the proof of \Cref{thm:polytime}, we can bound
    \begin{equation*}
        \Pr_{\sigma}[\textsc{Fail}(\sigma, \pi(\sigma))] \leq 2\Pr_{\sigma}[\textsc{Fail}(\sigma, \pi^*)] +  10^{-3} \, ,
    \end{equation*}
    where $\pi^*$ denotes the correct permutation for the set of vectors produced by \Cref{alg:app_centers}. 
    Conditioned on the success of  $\nu_1,\nu_2$, the approximate cluster means $(\tmu_i)_{i \in [k]}$, the approximate inner product $\langle \cdot, \cdot\rangle_{\apx}$, and the approximate distance $\| \cdot \|_{\apx}$ satisfy the assumptions in \Cref{fig:setting}, so we are in the setting of \Cref{fig:setting}. Therefore, \Cref{lemma:impostor_n_cross_size}, \Cref{claim:false_pos}, \Cref{cor:false_neg_updated} apply, and, using Markov's inequality, give $\Pr_{\sigma}[\textsc{Fail}(\sigma, \pi^*)] \le 10^{-4}$.

    Finally, we condition on $\sigma$ such that $\textsc{Fail}(\sigma, \pi(\sigma))$ does not occur. Now, for any $\nu_3$, we can make the following observations about the number of vertices misclassified by \Cref{alg:sketch+labels}. 
    \begin{enumerate}
        \item For this choice of $\sigma$, from the first condition in the definition of the $\textsc{Fail}$ event, it follows that
        at most $2 \cdot 10^9 \cdot \eta^2 \cdot {\delta\epsilon}/{\phi^4 }n$ vertices can be misclassified in line~\eqref{line:core} or line~\eqref{line:far}.
        \item  For a vertex to be misclassified in line~\eqref{line:impostor}, we must have $\tau(u) \notin \{\sigma(u), *\}$, $\sigma(u) \neq \iota(u)$ and $\robcon$ must return ``yes''. For there to be more than $2\cdot 10^{15}\eta^5 {\delta\epsilon}/{\phi^6 } \cdot \log({1}/{\delta}) \cdot  n$ such vertices, it must be the case that there is a vertex $ u \perp_{G_{\tau(u),\sigma(u)}} \Mapx\cup\NMapx$ for which $\robcon$ returns ``yes'' (since, by our choice of $\sigma$, the number of vertices $u \in V$ such that $\tau(u) \notin \{\sigma(u), *\}$, $\sigma(u) \neq \iota(u)$, and $ u \slashed{\perp}_{G_{\tau(u),\sigma(u)}} \Mapx\cup\NMapx$ is at most $2\cdot 10^{15}\eta^5 {\delta\epsilon}/{\phi^6 } \cdot \log({1}/{\delta}) \cdot  n$).
        \item  For a vertex to be misclassified in line~\eqref{line:non-impostor}, we must have $\tau(u) \notin \{\sigma(u), *\}$, $\tau(u) \neq \iota(u)$ and $\robcon$ must return ``no''. For there to be more than $5\cdot 10^{12}\eta^2 {\delta\epsilon}/{\phi^6 } \cdot \log({1}/{\delta}) \cdot  n$ such vertices, it must be the case that there is a vertex $ u \rightsquigarrow_{G_{\tau(u),\sigma(u)}} \Mapx\cup\NMapx$ for which $\robcon$ returns ``no'' (since, by our choice of $\sigma$, the number of vertices $u \in V$ such that $\tau(u) \notin \{\sigma(u), *\}$, $\tau(u) \neq \iota(u)$, and $ u \slashed{\rightsquigarrow}_{G_{\tau(u),\sigma(u)}} \Mapx\cup\NMapx$ is at most $5\cdot 10^{12}\eta^2 {\delta\epsilon}/{\phi^6 } \cdot \log({1}/{\delta}) \cdot  n$).
    \end{enumerate}
    By \Cref{cor:robcon_works}, having fixed $\nu_1,\nu_2,\sigma$ as above, we have the following: for any fixed vertex $u$ such that $ u \rightsquigarrow_{G_{\tau(u),\sigma(u)}} \Mapx\cup\NMapx$, the probability over $\nu_3$ that \robcon outputs ``no'' is at most $2/n^5$; for any fixed vertex $u$ such that $ u \perp_{G_{\tau(u),\sigma(u)}} \Mapx\cup\NMapx$, the probability over $\nu_3$ that \robcon  outputs ``yes'' is at most $2/n^5$. By a union bound, we get that the probability over $\nu_3$ that 
    $2\cdot 10^{15}\eta^5 {\delta\epsilon}/{\phi^6 } \cdot \log({1}/{\delta}) \cdot  n$ vertices misclassified in line~\eqref{line:impostor} or more than $5\cdot 10^{12}\eta^2 {\delta\epsilon}/{\phi^6 } \cdot \log({1}/{\delta}) \cdot  n$ vertices misclassified in line~\eqref{line:non-impostor}, is at most $4/n^4$.

    Combining our different rounds of conditioning, we get that with probability $0.95$ over $\nu_1,\nu_2,\nu_3$ and $\sigma$, there exists a set $B$ of size $O(\epsilon \delta/\phi^6 \cdot \log(1/\delta))$ such that our data structure answers correctly all the queries for $u \in V \setminus B$.

    \paragraph{Preproccessing time and space complexity} The preprocessing stage of \Cref{alg:sketch+labels} is the same as that of the Spectral Dot Oracle from \Cref{thm:spec_dot_prod_oracle}. The space complexity of \Cref{alg:sketch+labels} is dominated by the size of the data structure from \Cref{thm:spec_dot_prod_oracle} which is computed in the preprocessing stage.

    \paragraph{Query time.} In lines~\eqref{line:core}, \eqref{line:far} we compute the approximate distance between $f_u$ and each of the $\{\tmu_i\}_{i \in [k]}$ generated by \specdp, as per \Cref{def:dist_apx}. This requires $O(k)$ calls to the  function \specdp. The remainder of the query time comes from running \robcon. By \Cref{robcon:time_space}, \robcon \ requires $O(d\cdot L\cdot k\cdot \log(n))$ many calls to the function \specdp. Letting $Q$ denote the time of one call to \specdp, we get that the total query time of \Cref{alg:sketch+labels} is
\[O(d\cdot L\cdot k\cdot \log(n) \cdot Q) = O(d\cdot \poly(k)\cdot n^{1/2+O(\epsilon/\phi^2)}\cdot\polylog (n)\cdot\poly(1/\phi)\cdot \log(1/\delta)).\]
    
\end{proof}
\pagebreak
\section{Refining communities}\label{sec:reweight}

In this section, we show how to round our classifier into a $(k, \widetilde O( \epsilon \delta),\phi)$-clusterable graph, without significantly increasing the misclassification rate. Specifically, we prove the following.

\begin{thm}[Refining communities]\label{thm:round_random_sigma}
    Let $G,\epsilon,\delta,\phi,k,\eta$ as per \Cref{fig:setting} and assume that these parameters satisfy the relation $10^{16} \eta^5 \epsilon\delta \log(1/\delta)/\phi^9 \le (100\eta k)^{-1}$.
    Then, there is a polynomial-time algorithm that, given as input $G,\eta,\phi,\delta,k$ and $\sigma$ as per \Cref{fig:labels}, computes a weight function $w: E\rightarrow [0,1]$ on the edges of $G$ such that, with probability $0.9$ over the randomness of the algorithm and the draw of $\sigma$, the weighted graph $G'=(V,E,w,\ell)$ (with self-loops $\ell$ to preserve $d$-regularity) admits a \smash{$(k, O({\epsilon\delta/\phi^7 \cdot \log(1/\delta)}), \Omega({\phi^3}/{k}))$}-clustering $C_1', \ldots, C_k'$ with \smash{$\sum_{i=1}^k |C_i' \triangle C_i| \leq O({\epsilon\delta/\phi^7 \cdot\log(1/\delta)}) n$}. 
\end{thm}

\noindent
We show this via a more general result: given a clusterable graph $G$ and a labeling of its vertices $\alpha$, we can improve the clusterability of $G$ all the way up to matching the misclassification rate of $\alpha$. From this, \Cref{thm:round_random_sigma} follows trivially by combining with \Cref{thm:sublinear}. Hence, our goal for the rest of the section is to prove the more general result. A complete list of the objects and parameters that we use and assumptions we make on them is presented in \Cref{fig:setting_reweight}.

\begin{table}[!h]
 \fbox{\begin{minipage}{\textwidth}
\begin{multicols}{2}
\begin{flushleft}
	{\small
\begin{itemize}
	\item $G=(V,E)$ --- regular graph
    \item $d \ge 3$ --- degree of $G$
    \item $n =|V|$ --- number of vertices in $G$
    \item $\eta = O(1)$ --- upper bound on $\max_{i,j \in [k]} \frac{|C_i|}{|C_j|}$
    \item $\epsilon,\phi \in (0, 1)$ --- conductance parameters
    \item $k \ge 2$  --- number of communities
    \item $\{C_i\}_{i \in [k]}  $ --- $(k,\epsilon,\phi)$-clustering of $G$  (Def. \ref{def:clustering})
    \item $\mathcal{L}$ --- normalized Laplacian of $G$
    \item $v_1, \dots v_n $ -- norm. eigenvectors of $\mathcal L$ ordered by ascending eigenvalue
    \item $P=I-\sum_{i=1}^k v_i v_i^\top$ -- projection onto top $n-k$ eigenvectors of $\mathcal L$
    \item $\alpha: V\rightarrow [k]$ --- the given labeling
    \item $\gamma \in (0,1)$ --- misclassification rate of the labeling $\alpha$
    \item  $\theta = \phi^2/5$ -- threshold used in \Cref{alg:sdp_oneshot_k}
    \item $ \frac{\epsilon}{\phi^6}\le \frac{10^{-5}}{\eta^4}, \phi^2\eta < 10^{-3}$
    \item $k \leq \frac{\phi^5}{10^6\epsilon^{1/2}}$
    \item $ \gamma \leq {\phi^3}({100\eta k})$
\end{itemize}}
\end{flushleft}
\end{multicols}	
 \end{minipage}}	
\caption{Parameters and objects used in \Cref{sec:reweight}.}
\label{fig:setting_reweight}
\end{table}

\noindent
In this setting, we prove the following result.

\begin{thm}[Refining communities, general]\label{thm:round_to_clustering}
There is a polynomial-time algorithm (Algorithm \ref{alg:sdp_oneshot_k}) that given $G,\phi$ and $\alpha$ as per \Cref{fig:setting_reweight} computes a weight function $w: E\rightarrow [0,1]$ on the edges of $G$ such that the weighted graph $G'=(V,E,w,\ell)$ (with self-loops $\ell$ to preserve $d$-regularity) admits a $(k, 3\gamma /\phi,  {\phi^3}/({480k}))$-clustering $C_1', \ldots, C_k'$ with $\sum_{i=1}^k |C_i' \triangle C_i| \leq \frac{4 \gamma}{\phi}n$. 
\end{thm}

\noindent
We start by showing how to obtain \Cref{thm:round_random_sigma} from \Cref{thm:round_to_clustering}.

\begin{proof}[Proof of \Cref{thm:round_random_sigma}]
Let $\alpha$ be the classifier given by \Cref{thm:sublinear}, and apply \Cref{thm:round_to_clustering} to the classifier $\alpha$. By \Cref{thm:sublinear},  with probability $0.95$  over the internal randomness of the algorithm and the draw of $\sigma$, the misclassification rate of $\alpha$ is $O(\epsilon\delta\log(1/\delta)\frac{1}{\phi^6}n)$. So applying \Cref{thm:round_to_clustering}, given the classifier $\alpha$, gives the result. 
\end{proof}

\noindent
Having established \Cref{thm:round_random_sigma}, we now focus on proving \Cref{thm:round_to_clustering} for the rest of the section. Ideally, in order to sharpen the communities, we would like to down-weight the cross edges $E_{\cross}$ without touching the internal edges $E_{\inedge}$, as defined below.

\begin{definition}[Internal and cross edges]
    Consider the setting of \Cref{fig:setting_reweight}. We define the set of \emph{internal edges} of $G$, denoted $E_{\inedge}$, to be the edges with both endpoints belonging to the same cluster, that is 
    $$ E_{\inedge}  \coloneqq \bigcup_{i =1}^k E(C_i) \, .$$
    We define the set of \emph{cross-edges} of $G$, denoted $E_{\cross}$, to be the edges whose endpoints belong to different clusters, that is  
    $$E_{\cross} \coloneqq E \setminus E_{\inedge} \, . $$
    Furthermore, we let $\Lin$ and $\Lcross$ denote the normalized Laplacians of the graphs $(V,E_{\inedge},\ell)$ and $(V,E_{\cross},\ell')$ respectively, where the self-loops $\ell,\ell'$ ensure that these graphs are $d$-regular.
 \end{definition}

\noindent
 The challenge is of course that we do not have access to the target clustering $C_1, \ldots C_k$, and in particular, we do not know which edges belong to $E_{\cross}$ and which belong to $E_{\inedge}$. Instead, we use the given classifier $\alpha$ to identify \emph{flagged} edges, which are edges that we believe are likely to be cross-edges. 

\begin{definition}[Flagged edges]
\label{def:flagged}
Consider the setting of \Cref{fig:setting_reweight}. We say that an edge $(u,v) \in E$ is \emph{flagged} if $\alpha(u) \neq \alpha(v)$. Furthermore, we will use $F$ and $E^+$ to denote the set of flagged and non-flagged edges, that is 
    $$F \coloneqq \{ e = \{u, v\} \in E : \alpha(u) \neq \alpha(v) \} \quad \text{and} \quad  E^+ \coloneqq E \setminus F \, . $$
 \end{definition}

 \noindent
The following claim shows that the flagged edges are indeed a good proxy for the cross edges $E_{\cross}$. 
 \begin{lemma}\label{claim:FcapE_cross_k}
 Consider the setting of \Cref{fig:setting_reweight}. Then 
     $$ |F \triangle E_{\cross}| \leq d \gamma n \, .$$
 \end{lemma}

 \begin{proof}
 Let $Q \subseteq V$ denote the set of vertices misclassified by $\alpha$. We have $|Q| \leq \gamma n$, by assumption. Note that the edges in $E_{\inedge}$ can only be flagged if either of their endpoints is misclassified, and analogously the edges in $E^+$ can only be crossing if either of their endpoints is misclassified, i.e. 
    $$  F \triangle E_{\cross}  \subseteq \{ (u,v)  \in E : u \in Q \text{ or } v \in Q \} \subseteq \nei_G(Q) \, .$$
  Hence, we get 
    $$|F \triangle E_{\cross} | \leq d |Q| \leq d \gamma n \, .$$ 
 \end{proof}

\noindent
We are now equipped to present our algorithm. The idea is to use a semidefinite program to reweight the edges so as to minimize their weight while preserving the expansion properties of the clusters. \Cref{claim:FcapE_cross_k} suggests that down-weighting the edges in $F$ should be enough, so this is exactly what we do.

\begin{algorithm}
\caption{Refining communities via SDP}\label{alg:sdp_oneshot_k}
\begin{algorithmic}[1]
    \State \textbf{Input:} $G,\phi$ and $\alpha$ as per \Cref{fig:setting_reweight}
    \State \textbf{Output:} weight function $x^*:E\rightarrow [0,1]$
    
    \State $F \gets \{ e = \{u, v\} \in E : \alpha(u) \neq \alpha(v) \}$ \Comment{flag edges disagreeing endpoints, see \Cref{def:flagged}}
    \State $E^+ \gets E \setminus F$ 
    \State $x^* \gets \arg\min \left\{ \sum_{e \in F} x_e \,:\,
    \begin{array}{l}
        x \in [0,1]^E, \\
        x \geq \mathbbm{1}_{E^+}, \\
        P(\mathcal{L}_x - \theta I)P \succeq 0
    \end{array}
    \right\}$\Comment{$\mathcal L_x $ as per Definition \ref{def:Lx_n_Gx}}
    \State \Return $x^*$ \Comment{down-weight flagged edges subject to PSD constraint} 
\end{algorithmic}
\end{algorithm}
\begin{remark}[Computing the projection matrix $P$]
    Our SDP \ref{alg:sdp_oneshot_k} assumes that we have access to the projection matrix $P = I - \sum_{i = 1}^k v_i v_i^\top$ projecting onto the top $n-k$ eigenvectors of $\mathcal L$. Such a matrix can be computed in polynomial time to precision $1/\poly(n)$, which is sufficient for our purposes. See Appendix \ref{subsec:apxmeans} for more details. 
\end{remark}
\noindent
We now need to analyze the clusterability and spectral properties of a graph obtained by weights such as those produced by \Cref{alg:sdp_oneshot_k}. Hence, we now formally define how we obtain a weighted graph from a weight function $x \in [0,1]^E.$

\begin{definition}[$G_x$ and $\mathcal L_x$]\label{def:Lx_n_Gx}
Consider the setting of \Cref{fig:setting_reweight}, and let $x:E \rightarrow [0,1]$. We let
$G_x=(V,E,w,\ell)$ be the weighted graph with $w_e=x(e)$ for all $e \in E$ and $\ell(v) = d-\sum_{(u,v) \in E} x_e$ (note that $w_v \ge 0$ since $x \in [0,1]^E$). Furthermore, we let $\mathcal{L}_x$ be the normalized Laplacian of $G_x$.
\end{definition}

\noindent
 We now prove that the SDP \ref{alg:sdp_oneshot_k} has a high objective value, i.e., that it down-weights almost all of the flagged edges. This will be useful for showing that the graph $G_{x^*}$ obtained from the optimal solution $x^*$ has few edges across clusters. First, we need the following lemma. 

\begin{restatable}{lemma}{projectionk}
\label{claim:projection_k}
Consider the setting of \Cref{fig:setting_reweight}. Let $v_1, \dots v_k$ and $v_1^*, \dots v_k^*$ denote the bottom $k$ eigenvectors of $\mathcal{L}$ and $\Lin$, respectively. 
 Let $P = I - \sum_{i =1}^k v_i v_i^{\top} $ and $P^* =I - \sum_{i =1}^{k} v_i^*{v_i^*}^\top$. Then 

$$ \| P - P^* \|_{\op} \leq 8k^{1/2}\frac{\epsilon^{1/4}}{\phi^{1/2}}\, .  $$
\end{restatable}
\noindent
We defer the proof of \Cref{claim:projection_k} to \Cref{sec:projection_k}. 

\begin{lemma}[Existence of a good SDP solution]\label{lemma:OPT_SDP_k} Consider the setting of \Cref{fig:setting_reweight}. Let $\opt$ denote the optimum value of the SDP in Algorithm \ref{alg:sdp_oneshot_k}. Then
   $$ \opt \leq d \gamma n.$$
\end{lemma}
\begin{proof}
    Consider the solution given by 
    $$x_e = \begin{cases} 0 \qquad \text{if $e \in F \cap E_{\cross}$} \\
    1 \qquad \text{otherwise} \end{cases}.$$
    To prove the lemma, we show that this is a feasible solution with objective value at most $d \gamma n$. 
    
    \paragraph{Objective value.} The objective value of the solution $x$ is 
    $$ \sum_{e \in F}x_e = |F \setminus E_{\cross}|\leq d \gamma n \, ,$$
    where the last inequality follows by \Cref{claim:FcapE_cross_k}. 
    \paragraph{Feasibility.} We clearly have $x \in [0,1]^{E}$, by definition of $x$. Furthermore, $x_e = 1$ for all $e \in E^+ = E\setminus F$, so the condition $x \geq \mathbbm{1}_{E^+}$  also holds. It remains to check the constraint $P(\mathcal{L}_x-\theta I)P \succeq 0$. Since $x$ only down-weights cross edges and leaves the internal edges $E_{\inedge}$ unchanged, it follows that $\mathcal L_x - \mathcal L_{\inedge} \succeq 0$. This is because $\mathcal L_x - \mathcal L_{\inedge} $ is the degree-$d$ normalized Laplacian of the subgraph graph $(V,E_{\cross},w)$ where the weights are given by $w(e) = x_e$ for all $e \in E_{\cross}$. Hence, we have 
    $$P(\mathcal{L}_x-\theta I)P \succeq P(\mathcal{L}^{\inedge}-\theta I)P \, .$$
   We know that $\calL_{\inedge}$ has $k$ orthonormal eigenvectors $v_1^*,\dots,v_k^*$ with eigenvalue $0$. Let $P^* =  I - \sum_{i=1}^k v_i^*{v_i^*}^\top $, and define
\begin{equation*}
    A = P^*(\mathcal{L}^{\inedge}-\theta I)P^*  \quad \text{and} \quad B  = P(\mathcal{L}^{\inedge}-\theta I)P \, .
\end{equation*}

\noindent
 We conclude the proof by showing that $B \succeq 0$.  Since $B v_i = 0$ for $i =1, \dots, k$, the matrix $B$ has at least $k$ zero eigenvalues. Therefore, to rule out any negative eigenvalues, it suffices to show that $\lambda_{k+1}(B) > 0$. Indeed, if $\lambda_{k+1}(B) > 0$, then the $k$ zero eigenvalues must be $\lambda_1(B) = \lambda_2(B) =  \dots = \lambda_k(B) = 0$ (otherwise  $\lambda_i(B) = 0 < \lambda_{k+1}(B)$ for some $i \geq k+1$, which is impossible), and in particular $B$ only has non-negative eigenvalues. Therefore, our task reduces to proving $\lambda_{k+1}(B) > 0$. 
 
First, we show that $\lambda_{k+1}(A) \geq \phi^2/2  - \theta$. The matrix $P^*$ projects out the $k$ bottom eigenvectors $ v_1^*, \ldots, v_k^k$ of $\Lin$, so for any vector $x \in \spn (  \{ v_i^*\}_{i =1}^k)$, it holds that $Ax = 0$, and for any vector $x \in   \spn( \{ v_i^*\}_{i =1}^k)^\perp$, it holds that $x^\top A x = x^\top ( \mathcal L^{in} - \theta I) x \geq (\lambda_{k+1}(\mathcal L^{in} ) - \theta)\|x\|^2_2 \geq (\phi^2/2 - \theta)\|x\|_2^2$. Here the last inequality is using that $\Lin$ is the Laplacian of a $(k, 0,\phi)$-clusterable graph, so by \Cref{lem:bnd-lambda}, it holds that $\lambda_{k+1}(\mathcal \Lin) \geq \phi^2/2.$
So 
\begin{equation}\label{eq:lambdak_A}
    \lambda_{k+1}(A) \geq \phi^2/2  - \theta \, .
\end{equation}
\noindent
Next, we have 
\begin{align*}
    \|A - B\|_{\op} & =  \| P^*(\mathcal{L}^{\inedge}-\theta I ) P^* - P(\mathcal{L}^{\inedge}-\theta I)P \|_{\op} \\
    & = \| (P^* - P)(\mathcal{L}^{\inedge}-\theta I )P^* +P(\mathcal{L}^{\inedge}-\theta I) (P^* - P) \|_{\op} \\
    & \leq  \| P^* - P \|_{\op} \|\mathcal{L}^{\inedge}-\theta I\|_{\op}\|P^*\|_{\op} + \| P - P^*\|_{\op} \|\mathcal{L}^{\inedge}-\theta I\|_{\op}\|P \|_{\op}\\
    & \leq 4  \| P^* - P\|_{\op}  \\
    & \leq 40 \cdot k^{1/2}\frac{\epsilon^{1/4}}{\phi^{1/2}}, 
\end{align*}
where the second inequality uses $\| \mathcal L - \theta I\|_{\op} = \max_i | \lambda_i(\mathcal L) - \theta| \leq 2$ since $\lambda_i(\mathcal L) \in [0,2]$ for all $i$, and the last inequality follows from \Cref{claim:projection_k}. 
By Weyl's inequality,  we have
\begin{equation*}\label{eq:B'eval}
 \lambda_{k+1}(B) \geq  \lambda_{k+1}(A) - \|A - B\|_{\op}  \geq  \phi^2/2 - \theta-   40\cdot k^{1/2}\frac{\epsilon^{1/4}}{\phi^{1/2}} > 0,
\end{equation*} 
where the second inequality follows by~\eqref{eq:lambdak_A}, and the last inequality follows by the assumption that $k \leq {\phi^5}/({10^6\epsilon^{1/2}})$.
\end{proof}

\noindent
From this, we show that the optimal solution removes almost all of the weight from the cross edges $E_{\cross}$. 

\begin{lemma}[The optimal SDP solution puts little weight on the cross edges]
Consider the setting of \Cref{fig:setting_reweight}, and let $x^* \in [0,1]^E$ be the optimal solution to SDP \eqref{alg:sdp_oneshot_k}. Then, $w(E_{\cross})= \sum_{e \in E_{\cross}}x^*_e \leq 2 d \gamma n$. 
\end{lemma}
\begin{proof}
    We have 
    \begin{align*}
        \sum_{e \in E_{\cross}}x^*_e &= \sum_{e \in E_{\cross}\setminus F}x^*_e + \sum_{e \in E_{\cross} \cap F}x^*_e  && \\
        & \leq |E_{\cross}\setminus F| + \sum_{e \in F}x^*_e &&  \text{since $0 \leq x^*_e \leq 1$ for all $e \in E$} \\
        & \leq d \gamma n + \sum_{e \in F}x^*_e  && \text{by \Cref{claim:FcapE_cross_k}}\\
        & =   d \gamma n + \opt &&  \text{by optimality of $x^*$} \\
        & \leq 2 d \gamma n &&  \text{by \Cref{lemma:OPT_SDP_k}}. 
    \end{align*}
\end{proof}

\noindent
Next, we need to show that the graph $G_{x^*}$ obtained from the optimal solution $x^*$ of SDP \ref{alg:sdp_oneshot_k} can be partitioned into $k$ clusters, each of which has a high internal conductance. We will show a stronger statement, namely that this is true for \emph{any} feasible solution $x$ of the SDP \ref{alg:sdp_oneshot_k}. 
Before we proceed,  we introduce the notion of \emph{$k$-way expansion}.
\begin{definition}[$k$-way expansion]\label{def:rho} Recall the notion of conductance from \Cref{def:conductance}. Define the \emph{$k$-way expansion} of the graph $G$ as
\[\rho_k(G) := \min_{\text{disjoint }A_1, \dots, A_k} \max_{1 \leq i \leq k} \Phi_G(A_i).\]
\end{definition}

\begin{thm}[Theorem 1.1 in \cite{LGT12}]\label{thm:cheeger_higher_order}For a graph $G$, let $\lambda_1 \leq \lambda_2 \leq \ldots \leq \lambda_n$ be the eigenvalues of $\mathcal{L}$, where $\mathcal{L}$ is the normalized Laplacian of $G$. We have that for every $k$, 
    $$ \frac{\lambda_k}{2} \leq \rho_k(G) \, . $$
\end{thm}
\noindent
In order to lower bound the internal conductance of the clusters of $G_x$, we want to bound the $(k+1)$-way expansion $\rho_{k+1}(G_x)$.  Motivated by \Cref{thm:cheeger_higher_order}, we start by showing that any feasible solution $x$ of the SDP satisfies $\lambda_{k+1}(\mathcal L _x)\gtrsim {\phi^2}$. 

\begin{lemma}[$\lambda_{k+1}(\mathcal{L}_x) \geq\theta$]\label{lambda_above_threshold}
Consider the setting of \Cref{fig:setting_reweight}.
For all weight assignments $x: E \rightarrow [0,1]$ satisfying $ x \ge \mathbbm{1}_{E^+} $ and $  P(\mathcal{L}_x-\theta I)P \succeq 0$, it holds that $$\lambda_{k+1}(\mathcal L_x) \ge \theta = \frac{\phi^2}{5}\, .$$ 
\end{lemma}
 
\begin{proof}
Let $v_1, \dots v_n$ denote the eigenvectors of the Laplacian $\mathcal L$ of $G$ corresponding to eigenvalues $\lambda_1(\mathcal L) \leq \dots \leq\lambda_n(\mathcal L) $. Let 
$S \coloneqq \spn \{v_{k+1}, \dots v_n \}$ be the $n-k$ dimensional subspace spanned by the top $n-k$ eigenvectors of $\mathcal L$. 
Since $Pz = z$ for all $z \in S$, we get that for all for all $z \in S$ it holds that 
\begin{align*}
    z^\top \mathcal L_x z & =     z^\top P \mathcal L_x P z && \\
    &= z^\top P(\mathcal L_x - \theta I)Pz  + \theta \| z\|_2^2 &&  \text{by adding and substracting $\theta z^\top PIPz  = \theta \|z\|_2^2$}\\
    & \geq \theta \|z\|_2^2 && \text{by assumption that $P(\mathcal L_x-\theta I)P \succeq 0$.} 
\end{align*}
In particular, 
$$ \min_{\substack{z \in S : \\ \|z\|_2 =1}}z^\top \mathcal L_x z \geq \theta \, . $$
\noindent
Therefore, from the Courant-Fischer Theorem, we get that 
\begin{align*}
    \lambda_{k+1}(\mathcal{L}_x)& = \max_{\substack{S' \subseteq V: \\ \dim(S') = n-k}} \min_{\substack{z \in S': \\ \|z\|_2 =1}}z^\top \mathcal L_x z \\
    & \geq \min_{\substack{z \in S : \\ \|z\|_2 =1}}z^\top \mathcal L_x z \\
    & \geq \theta. 
\end{align*}
\end{proof}
\noindent
From \Cref{lambda_above_threshold} and \Cref{thm:cheeger_higher_order}, we can infer that for any feasible solution $x$, there \emph{exists} a clustering $C_1', \dots C_k'$ of $G_x$. However, a priori it is not clear why this partitioning would be close to the target partitioning $C_1, \dots C_k$. We will argue the existence of a good refined partitioning $C_1', \dots C_k'$ that is close to the target partitioning $C_1, \dots C_k$ as follows: we start by trying the target partitioning $C_1, \dots C_k$. If a cluster $C_i$ has too low expansion in the reweighted graph $G_x$, then we will remove the portion violating the expansion and reassign it to another cluster. The following lemma shows that after removing the largest violating portion $S \subseteq C_i$, the rest of the cluster, $C_i \setminus S$, has a good expansion.   

\begin{lemma}[Lemma 7.2 in \cite{ST11}]\label{lemma:ST_k} Let $G = (V, E,w)$ be a graph and let $\psi \leq 1$. Let
$C \subseteq V$ and let $S \subset C$ be a set maximizing $\vol(S)$ among those satisfying 
\begin{enumerate}
    \item $\vol(S) \leq \vol(C)/2$
    \item $\frac{w(S, C \setminus S)}{\vol(S)} \leq \psi$ 
\end{enumerate}
If $\vol(S) \leq r \cdot  \vol(C)$ for $r \leq 1/3$, then 
$$\Phi(G\{ C \setminus S\}) \geq \psi \cdot \frac{1-3r}{1-r}\, .$$
\end{lemma}
\noindent
We now show that the largest set $S \subseteq C_i$ violating the expansion of a cluster $C_i$ has to be small. This will be useful for arguing that our clustering $C_1', \ldots C_k'$ obtained by reassigning the violating sets will be close to the target partitioning $C_1, \dots, C_k$. 

\begin{lemma}[Sparse cuts are small]\label{lemma:sparse_cut_k} Consider the setting of \Cref{fig:setting_reweight}, and let $x \in [0,1]$ be a feasible solution to SDP \ref{alg:sdp_oneshot_k}. Let $i \in[k]$, and let $\emptyset \neq S \subsetneq C_i$ such that $\Phi_{G_x\{C_i\}}(S) < \phi/2$ and $|S| \le |C_i|/2$. Then, $|S| \le \frac{2}{d\phi}|F(C_i)|$, where $F(C_i)$ denotes the set of edges in $F$ with both endpoints in $C_i$. 
\end{lemma}

\begin{proof}
Any feasible solution $x$ to the SDP \ref{alg:sdp_oneshot_k} must satisfy $x \geq \1_{E^+}$, i.e. it can only down-weight edges in $F$. In particular, we have $ w(S,C_i\setminus S) \ge |E(S,C_i\setminus S)| -  |F(S,C_i\setminus S)| $. Therefore, 
    \begin{align*}
        \frac{\phi}{2} d|S| & > w(S,C_i\setminus S) \\
        & \ge |E(S,C_i\setminus S)| -  |F(S,C_i\setminus S)|  \\
        & \ge |E(S,C_i\setminus S)| -  |F(C_i)| \qquad \text{since $S \subseteq C_i$} \\
        & \ge \phi d|S| -  |F(C_i)|, 
    \end{align*} 
    where the last inequality holds because $C_i$ is a $\phi$-expander.
 Rearranging we get the lemma.
\end{proof}

\noindent
Finally, we need to show that the removed sets can be reassigned to other clusters without breaking their expansion. 
The following (possibly exponential time) algorithm shows how one can achieve that, and obtain the required $(k,O(\gamma/\phi),\poly(\phi)/k)$-clustering  $C_1', \dots, C_k'$ of $G_x$. \Cref{alg:GT_k} below is based on Algorithm 2 in \cite{GT14}. 

\begin{algorithm}[H]
\caption{Construction of $C_1', \dots, C_k'$}\label{alg:GT_k}
\begin{algorithmic}[1]
    \State \textbf{Input:} graph $G_x = (V, E, w)$, partition $C_1, \dots, C_k$
    \State \textbf{Output:} partition $C_1', \dots, C_k'$
    
    \For{$i = 1$ to $k$}
        \State $X_i \gets \arg\max \left\{ |S| : \emptyset \neq S \subsetneq C_i,\ \Phi_{G_x\{C_i\}}(S) < \phi/2,\ |S| \leq |C_i|/2 \right\}$ \label{line:X_i}
        \State $C_i' \gets C_i$ \label{line:C'_i}
    \EndFor

    \While{there exist $i \neq j \in [k]$ and $S \subseteq C_i' \cap \left( \bigcup_{\ell=1}^k X_\ell \right)$ such that $w(S, C_i' \setminus S) < w(S, C_j')$}
        \State Update $C_i' \gets C_i' \setminus S$
        \State Update $C_j' \gets C_j' \cup S$
    \EndWhile
\end{algorithmic}
\end{algorithm}

\noindent
In other words, \Cref{alg:GT_k} does the following: We start by selecting the largest non-expanding subset $X_i$ from each cluster $C_i$. By \Cref{lemma:sparse_cut_k}, the remaining ``core” $C_i \setminus X_i$ is an induced expander. It remains to decide how to reassign the vertices from the sets $X_i$ so that they don’t break the expansion of the cores. We start by trying $C'_1 = C_1, \ldots, C'_k = C_k$, and do the following ``greedy" updates:  while there exists a subset $S \subseteq \cup_{l = 1}^kX_l$ inside of a cluster $C'_i$ that expands more to a different cluster $C'_j$ than to its own, we simply move it to  $C'_j$. 

  Note that the algorithm always terminates in a finite number of steps. This is because for each iteration of the while loop, the total weight $w(\cup_{i\neq j\in [k]}E(C'_i,C'_j))$ across the partitions decreases by at least  $\Delta \coloneqq  \min_{e \in E : w(e) > 0}w(e)$, which is a constant positive amount. Furthermore, we have two important properties.
\begin{claim}[Fact 2.4 in \cite{GT14}]
\label{claim:simpleprop_k}
    The output of Algorithm \ref{alg:GT_k} satisfies the following: 
    \begin{enumerate}
        \item $C_i\setminus X_i \subseteq C_i' $ for $i \in [k]$
        \item For every $i \in [k]$ and every $S \subseteq C_i'\cap(\cup_{\ell=1}^k X_\ell)$, we have 
        $$w(S,C_i'\setminus S) \geq \frac{1}{k}w(S,V\setminus S).$$
    \end{enumerate}
\end{claim}
\noindent
Moreover, the induced subgraphs $G_x\{C_i \setminus X_i\}$  are expanders.
\begin{lemma}\label{lemma:expanders} Consider the setting of \Cref{fig:setting_reweight}, let $x$ be a feasible solution for SDP \ref{alg:sdp_oneshot_k}, and let $G_x$ be obtained from $x$ as per \Cref{def:Lx_n_Gx}. Let $X_i \subseteq C_i$ be the sets obtained in line \ref{line:X_i} of Algorithm \ref{alg:GT_k} on input $G_x$. Then $\sum_{i = 1}^k|X_i| \leq 2 \gamma n/\phi $ and 
$$\forall \, i \in [k], \, \, \, \Phi(G_x\{C_i \setminus X_i\}) \geq \frac{\phi}{6} \, .$$
\end{lemma}
\begin{proof} 
By its the definition, $X_i$ satisfies the conditions of \Cref{lemma:sparse_cut_k} for all $i \in [k]$. Therefore, we conclude that 
\[\sum_{l = 1}^k|X_l| \leq \frac{2}{d\phi}\sum_{l = 1}^k|F(C_l)|.\]
Note that $F(C_l)\cap E_{cross} = \emptyset$ and all $F(C_l)$ are disjoint since all $C_l$ are disjoint. By \Cref{claim:FcapE_cross_k} we have that 
\[\sum_{l = 1}^k|F(C_l)| \leq |F\setminus E_{cross}| \leq d\gamma n,\] which gives $\sum_{l = 1}^k|X_l| \leq 2 \gamma n/\phi$ as desired. Since $G_x$ is a $d$-regular graph, we have that, for any $i$
\[\frac{\vol(X_i)}{\vol(C_i)} = \frac{|X_i|}{|C_i|} \leq \frac{2\eta\gamma k}{\phi},\]
where the last inequality follows from $|X_i| \leq \sum_{l = 1}^k|X_l| \leq  2 \gamma n/\phi$ and $|C_i| \geq \frac{n}{\eta k}$ as per \Cref{fig:setting_reweight}. By \Cref{lemma:ST_k}, applied to $G = G_x$, $C = C_i$, $S = X_i$, $\psi = \phi/2$ and $r = {2\eta\gamma k}/{\phi} \leq {1}/{4} $ we get that 
\[\Phi(G_x\{C_i \setminus X_i\}) \geq \frac{\phi}{2}\cdot \frac{1 - (3/4)}{1 - (1/4)} = \frac{\phi}{6} \, .\]
\end{proof}
\noindent
We are now ready to show that any feasible solution $x$ of the SDP \ref{alg:sdp_oneshot_k} gives rise to a clustering $C_1', \dots C_k'$. Note that the quality of the clustering depends on how much weight $x$ puts on the cross-edges. 

\begin{lemma}[Existence of a good $k$-clustering on $G_x$]\label{lemma:good_clustering}
Consider the setting of \Cref{fig:setting_reweight}, let $x$ be a feasible solution for SDP \ref{alg:sdp_oneshot_k}, let $G_x$ be obtained from $x$ as per \Cref{def:Lx_n_Gx}, and assume that  $\lambda_{k+1}(\mathcal{L}_x) \geq \frac{\phi^2}{10} $. Then, there exists a partitioning $C_1',\dots,C'_k$ of $V$ such that:
    \begin{enumerate}
        \item $\sum_{i =1}^k |C'_i \triangle C_i| \leq \frac{4\gamma  n}{\phi}$; \label{item:symdiff}
        \item $\Phi(G\{C_i'\}) \geq \frac{\phi^3}{480k}$ for $i \in [k]$; \label{item:expansion}
        \item $w\left(\bigcup_{i,  j\in [k]: i \neq j}E(C'_i,C'_j)\right) \leq w(E_\cross)+  2d \gamma n/\phi $. \label{item:cross_weight}
    \end{enumerate}
\end{lemma}
\begin{proof}
Let $C_1', \dots, C_k'$ be the output of \Cref{alg:GT_k} given the input $G_x$ and $C_1, \dots C_k$. We now show that $C_1', \dots, C_k'$ satisfies the three conditions in the lemma statement. \\
\textbf{Condition 1:} 
We have 
   $C_i \triangle C_i' \subseteq \cup_{\ell=1}^k X_\ell$, by construction of $C_1', \dots, C_k'$, and $\left|\bigcup_{\ell=1}^k X_\ell\right| \leq \frac{2 \gamma n}{\phi}, $ by \Cref{lemma:expanders}.
    So 
    \[\sum_{i = 1}^k|C_i \triangle C_i'| = \sum_{i = 1}^k|C_i \setminus C_i'| + \sum_{i = 1}^k|C_i' \setminus C_i| \leq 2\left|\bigcup_{\ell=1}^k X_\ell\right| \leq \frac{4 \gamma n}{\phi},\]
   where the first inequality follows from the fact that the sets in $\{ C_i \setminus C_i'\}_{i \in [k]}$ are pairwise disjoint, and similarly the sets in $\{ C_i' \setminus C_i\}_{i \in [k]}$ are pairwise disjoint. 
   \\~\\
   \textbf{Condition 2:} 
   Let $R_i = C_i' \cap (\cup_{\ell=1}^k X_\ell)$ be the additional vertices that got reassigned to $C_i'$ outside of the ``core" $C_i \setminus X_i$.

Given $S \subseteq C_i'$ with $|S| \leq \frac{|C_i'|}{2}$, let $S_{core} = S \setminus R_i$, and let $S_{out} = S \cap R_i$. Note that since $C_i\setminus X_i =  C'_i\setminus (\cup_{\ell=1}^k X_\ell) = C'_i\setminus R_i$, we have that  $S_{core} = S \setminus R_i= S \cap (C_i \setminus X_i)$ (see \Cref{fig:core} for an illustration).
\vspace{0.05in}
\begin{figure}[H]
	\centering
		\begin{tikzpicture}

  \definecolor{lightblue}{RGB}{200, 230, 255}
  \definecolor{lightred}{RGB}{255, 182, 193}

  \draw[thick, fill=lightblue, draw=blue] (0,0) ellipse (4cm and 2cm);
  \node[text=blue] at (3.5,-1.7) {$C_i'$};

  \draw[thick, fill=lightred, draw=red] (0,0.4) ellipse (2cm and 0.8cm);
 \node[text=red] at (0.7, 1.4) {$S$};

  \node[text=red] at (-1.0,0.4) {$S_{\text{core}}$};
  \node[text=red] at (1.0,0.4) {$S_{\text{out}}$};

  \draw[thick, gray] (0,-2.0) -- (0,2.0);

  \node[text=blue] at (-2.5,-0.5) {$C_i \setminus X_i$};
  \node[text=blue] at (2.3,-0.5) {$R_i$};

\end{tikzpicture}
	  \caption{The cluster $C_i'$ can be partitioned into the ``core" $C_i \setminus X_i$ and the ``reassigned" vertices $R_i$. For $S \subseteq C_i'$, we let $S_{core} = S \setminus R_i $ and $S_{out} = S \cap R_i$. }
\label{fig:core}
\end{figure}
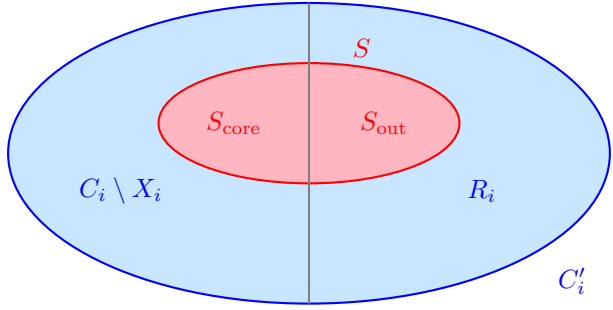

\noindent
We need to show that 
$\frac{w(S,C_i' \setminus S)}{d|S|} \geq \frac{\phi^4}{2000k}.$
   We consider two cases.
   \\~\\
   \textbf{Case 1:} $|S_{core}| \geq |S_{out}|$. We have that $|S| \leq 2|S_{core}|$. From the above derivation, $(C_i\setminus X_i)\setminus S_{core} = C'_i \setminus (R_i \cup S_{core}) \subseteq C'_i \setminus S$, so
   \begin{align*}
       \frac{w(S, C_i'\setminus S)}{d|S|} \geq  \frac{w(S_{core}, C_i' \setminus S)}{2d|S_{core}|} \geq \frac{w(S_{core}, (C_i \setminus X_i) \setminus S_{core})}{2d|S_{core}|}  \geq \frac{\phi}{12},
   \end{align*}
   since $G_x\{C_i\setminus X_i\}$ is a $\frac{\phi}{6}$-expander, by \Cref{lemma:expanders}. Since $k \geq 1$ and $\phi \leq 1$, $\frac{\phi}{12} \geq \frac{\phi^3}{480k}$, which concludes the proof.
   \\~\\
 \textbf{Case 2:} $|S_{core}| < |S_{out}|$. We have that $|S| \leq 2|S_{out}|$. 
 Then 
\begin{equation}\label{eq:SvsSout}
    \begin{aligned}
\frac{w(S, C_i'\setminus S)}{d|S|} & \geq \frac{w(S_{out}, C_i' \setminus S ) + w(S_{core},C_i' \setminus S)}{2d|S_{out}|} && \text{since $|S| \leq 2|S_{out}|$}\\
& \geq \frac{w(S_{out}, C_i' \setminus S ) + w(S_{core}, (C_i \setminus X_i)\setminus S_{core})}{2d|S_{out}|}  &&  \text{since $(C_i\setminus X_i)\setminus S_{core} \subseteq C'_i \setminus S$} \\
& \geq \frac{w(S_{out}, C_i' \setminus S ) + (\phi/12)\cdot d |S_{core}|}{2d|S_{out}|}  && \text{since $G_x\{C_i \setminus X_i\}$ is a $\phi/12$-expander by \Cref{lemma:expanders}} \\
& \geq \frac{w(S_{out}, C_i' \setminus S ) + (\phi/12)\cdot w(S_{core},  S_{out})}{2d|S_{out}|} && \text{since $d |S_{core}| \geq w(S_{core}, V) \geq w(S_{core}, S_{out}) $} \\
& \geq \frac{\phi}{24} \cdot \frac{w(S_{out}, C_i' \setminus S )  
 + w(S_{core},  S_{out})}{d|S_{out}|} && \text{since $\phi \leq 1$} \\ 
& =  \frac{\phi}{24} \frac{w(S_{out}, C_i' \setminus S_{out})}{d|S_{out}|} && \text{since $C_i' \setminus S_{out} = (C_i' \setminus S)\cup S_{core}$} \\
&  \geq \frac{\phi}{24k} \frac{w(S_{out}, V \setminus S_{out})}{d|S_{out}|}  &&   \text{by \Cref{claim:simpleprop_k}.}
\end{aligned}
\end{equation}
\noindent
To continue, we have 
\begin{align*}
    \Phi_{G_x}(C_\ell \setminus X_\ell) & \le \Phi_{G}( C_\ell \setminus X_\ell)  &&  \\
    & = \frac{|E(C_{\ell} \setminus X_{\ell}, V \setminus (C_{\ell} \setminus X_{\ell})) }{d|C_\ell \setminus X_\ell|} && \\
    & \leq \frac{|E(C_{\ell}, V \setminus C_{\ell})| + |E(X_{\ell}, V \setminus X_{\ell})| }{d\left(|C_\ell| - |X_\ell| \right)} && \\
    & \leq \frac{\epsilon dn + d|X_\ell|}{d\left(|C_\ell| - |X_\ell| \right)} && \text{by assumption that $C_1, \dots C_k$ is a $(k, \epsilon, \phi)$-clustering} \\
    & \leq \frac{\epsilon n + 2 \gamma n/\phi}{|C_\ell| - |X_\ell|} && \text{since $|X_\ell| \le 2\gamma n/\phi$ by \Cref{lemma:expanders}} \\
    & =  \frac{\epsilon + 2\gamma/\phi}{1-|X_\ell|/|C_\ell|}\cdot  \frac{n}{|C_\ell|} \\
    & \leq\frac{\epsilon + 2\gamma/\phi}{3/4}\cdot  \frac{n}{|C_\ell|} && \text{since $|X_{\ell}| \leq 2 \gamma n /\phi \leq \frac{n}{4k \eta} \leq \frac{|C_i|}{4}$} \\
  & \leq \frac{\epsilon + 2\gamma/\phi}{3/4}\cdot k \eta && \text{since $\min_i |C_i| \geq \frac{n}{k\eta}$ by \Cref{fig:setting_reweight}}\\
    & \leq \frac{\phi^2}{25}  && \text{since $\epsilon, \gamma/\phi  \le\frac{ \phi^2}{100\eta k}$ by the setting of \Cref{fig:setting_reweight}.}
\end{align*}

On the other hand, by the assumption that  $\lambda_{k+1}(\mathcal L_x) \ge \phi^2/10$, by  \Cref{thm:cheeger_higher_order} we get $\rho_{k+1}(G_x) \ge \phi^2/20$. 

Since $S_{out},C_1 \setminus X_1,\dots, C_k \setminus X_k$ are $k+1$ disjoint sets and since $\Phi_{G_x}(C_{\ell} \setminus X_{\ell}) \leq \phi^2/25 < \rho_{k+1}(G_x)$ for all $\ell \in [k]$, it must hold that $\Phi_{G_x}(S_{out}) \geq \rho_{k+1}(G_x)$, by Definition \ref{def:rho}. Therefore, we have
\begin{align*}
\frac{w(S, C_i'\setminus S)}{d|S|} &\geq \frac{\phi}{24k} \frac{w(S_{out}, V \setminus S_{out})}{d|S_{out}|} \qquad \text{by Equation \eqref{eq:SvsSout}} \\
& \geq \frac{\phi}{24k} \rho_{k+1}(G_x) \\
& \geq \frac{\phi^3}{480k}. 
\end{align*}
\noindent
\textbf{Condition 3:} Again, for $i \in [k]$, let $R_i = C_i' \cap (\cup_{\ell=1}^k X_\ell)$ be the additional vertices that got reassigned to $C_i'$ outside of the ``core" $C_i \setminus X_i$. In other words, $R_i = C'_i \setminus (C_i \setminus X_i)$. Then, we have

\begin{align*}
       w(\cup_{i\neq j\in [k]}E(C'_i,C'_j)) & \leq w\left(\cup_{i\neq j\in [k]}E(C_i\setminus X_i,C_j\setminus X_j)\right) +  w\left(\{e \in E: e \cap (\cup_{i\in [k]}R_i) \neq \emptyset\}\right) \\
       & \leq  w(E_{\cross})  + d\left|\cup_{\ell \in [k]} X_\ell\right| \qquad \qquad\qquad \qquad \qquad \qquad  \text{since $\cup_{\ell \in [k]}R_{\ell} \subseteq \cup_{\ell \in [k]}X_{\ell}$} \\
       & \leq  w(E_{\cross}) + 2d\gamma n/\phi \qquad \qquad\qquad \qquad\qquad \qquad\qquad \text{by \Cref{lemma:expanders}.}
\end{align*}

\end{proof}

\noindent
Finally, we can combine \Cref{lemma:good_clustering} above together with the fact that the optimal solution $x^*$ to the SDP in \Cref{alg:sdp_oneshot_k} puts little weight on the cross edges (see \Cref{lemma:sparse_cut_k}), so show that $x^*$ gives rise to the required refined clustering. 
\begin{proof}[Proof of Theorem \ref{thm:round_to_clustering}]
Let $x$ be the SDP solution returned by Algorithm \ref{alg:sdp_oneshot_k}, and let $G_x$ be the weighted graph obtained from it. By \Cref{lambda_above_threshold}, we have $\lambda_{k+1}(\mathcal{L}_x) \geq  {\phi^2}/{10}$, and so we can apply Lemma \ref{lemma:good_clustering} to $G_x$. 
By Lemma \ref{lemma:good_clustering} item~\eqref{item:cross_weight}, the obtained partitioning $C_1', \ldots C_k'$ satisfies 
\begin{align*}
    w(\cup_{i\neq j\in [k]}E(C'_i,C'_j)) & \leq w(E_\cross)+  2d \gamma n/\phi && \text{by Lemma \ref{lemma:good_clustering},  item~\eqref{item:cross_weight}} \\
    & = \opt +  2d \gamma n/\phi  && \text{by the assumption that $x$ is an optimal SDP solution} \\
    & \leq d \gamma n +  2d \gamma n /\phi && \text{by \Cref{lemma:OPT_SDP_k}}\\
    & \leq 3d \gamma n /\phi. &&
\end{align*}
Furthermore, by Lemma \eqref{lemma:good_clustering} item~\eqref{item:expansion}, the partitioning $C_1', \ldots C_k'$ satisfies  $\Phi(G\{C_i'\}) \geq {\phi^4}/({2000k})$ for $i \in [k]$.
So $C_1', \ldots C_k'$ is a $(k, {3\gamma n}/\phi,{\phi^3}/({480k}))$-clustering. Finally, by \Cref{lemma:good_clustering} item~ \eqref{item:symdiff}, we have $\sum_{i =1}^k |C'_i \triangle C_i| \leq {4\gamma  n}/{\phi}$. 
\end{proof}

\addcontentsline{toc}{section}{References}
\bibliographystyle{alpha}
\bibliography{references}

\appendix
\section{Obtaining approximate cluster means}
\label{subsec:apxmeans}
In this section, we prove two results stated in \Cref{subsec:apx_spec_oracles} without proof: namely, \Cref{rem:proj_matrix_apx} and \Cref{lemma:pi_computation}. 

\projmatrix*

\begin{proof}[Proof of \Cref{rem:proj_matrix_apx}]
For matrices with orthonormal columns $Q, \widehat{Q} \in \R^{n \times k}$ define 
\[\sin\Theta(Q, \widehat{Q}) \coloneqq (I - QQ^T)\widehat{Q}\widehat{Q}^T.\]
From \cite[Theorem 2.5.1]{matrixcomputations} and \cite[Section 3]{bjorck} we have
\begin{equation*}
    \|\widehat{Q}\widehat{Q}^T - QQ^T \|_2 =  \|\sin\Theta(Q, \widehat{Q})\|_2;
\end{equation*} 
see also \cite[Lemma 1.5]{drineas}. 

Let $M = I + \frac{1}{2d}A$ be the lazy random walk matrix in $G$.  Observe that $M$ is a symmetric matrix with eigenvalues $\lambda^M_1 \geq \ldots \geq \lambda^M_n \geq 0$ and, by \Cref{lem:bnd-lambda}, $\alpha:= \frac{\lambda^M_{k+1}}{\lambda^M_k} \leq \frac{1 - \phi^2/4}{1 - \epsilon}$. Let $Q \in \mathbb{R}^{n \times k}$ be a matrix whose columns form an orthonormal basis for the top-$k$ dimensional eigenspace $\mathcal{Q} \subseteq \mathbb{R}^n$ of $M$. The matrix $U_{[k]}$ from \Cref{def:emb} is an example of such a matrix $Q$. Note that the orthogonal projection onto $\mathcal{Q}$ is the Gram matrix of vectors in $Q^\top$:
\[P_{\mathcal{Q}} = QQ^T.\]

Let $\Omega \in \mathbb{R}^{n \times k}$ be a standard Gaussian matrix and let $q \in \N$. Let $\widehat{Q}$ be an orthonormal basis for the range of $M^q \Omega$. Results on subspace iteration \cite{parlett1998symmetric,pmlr-v37-boutsidis15,saibaba_subspace} show that there exist two matrices $\Omega_2 \in \mathbb{R}^{n-k \times k}$ and $\Omega_1 \in \mathbb{R}^{k \times k}$ which deterministically depend on $\Omega$ and 
\begin{equation*}
    \|\sin\Theta(Q, \widehat{Q})\|_2 \leq \alpha^q \|\Omega_2 \Omega_1^{-1}\|_2
\end{equation*} for all possible realizations of $\Omega$. Under the randomness of $\Omega$, $\Omega_1$ and $\Omega_2$ are independent standard Gaussian matrices. Consequently, with probability $1 - 1/\poly(n)$ over the sampling of $\Omega$, we have that
\[\|\Omega_1^{-1}\|_2 \leq O\left(\poly(n)\cdot\sqrt{k}\right) \text{ and } \|\Omega_2\|_2 \leq \sqrt{k} + \sqrt{n-k} + \sqrt{O(\log n)}.\]
where the first is shown in  \cite[Theorem 3.3]{sankar2006smoothed} and the second in shown in \cite[Proposition 10.1 and 10.3]{rsvd}.  Hence, if $q \geq C \cdot \frac{\log(n)}{\log(1/\alpha)}$ for a sufficiently large constant $C$, then  
\[\|\widehat{Q}\widehat{Q}^T - QQ^T \|_2 \leq 1/\poly(n).\]
It remains to note that $1/\alpha \geq \frac{1 - \epsilon}{1 - \phi^2/4} \geq 1 + \phi^2/4$ by the setting of parameters in \Cref{fig:setting}, and therefore it suffices to take $q = O\left(\frac{\log(n)}{\log(1 + \phi^2/4)}\right) \leq O\left(\frac{(1 + \phi^2/4)\log(n)}{\phi^2/4}\right)$. The algorithm for finding $\widehat{Q}$ consists of three steps:
\begin{itemize}
    \item Sample $\Omega$ (takes $O(n\cdot k)$ operations);
    \item Compute $M^q\Omega$ (takes $O(n^2\cdot k\cdot q)$ operations);
    \item Find an orthonormal basis of $M^q\Omega$ (takes $O(n\cdot k^2)$ operations).
\end{itemize}
\end{proof}

\Cref{lemma:pi_computation}, restated below for the convenience of the reader, states that we can compute a set of approximate cluster means \(\left(\widetilde{\mu}_i\right)_{i \in [k]}\) that satisfy \Cref{def:approxmeans}. 
\approxmu*
The algorithm for computing the approximate cluster means \(\left(\widetilde{\mu}_i\right)_{i \in [k]}\) consists of two parts. First, in \Cref{alg:app_centers}, we obtain a set $\{ \tmu_i\}_{i \in [k]}$ of approximate cluster means, where each approximate cluster mean \(\widetilde{\mu}_i\) is close to some true cluster mean \(\mu_j\) in the sense that \(\| \widetilde{\mu}_i - \mu_j \|_{\apx}^2 \le \frac{\phi^2}{1600\eta} \| \mu_j \|_{\apx}^2\). This part of the algorithm is very similar to Algorithm 1 in \cite{SP23}. For completeness, we include our own analysis of it. 

Although the cluster means output by \Cref{alg:app_centers} are indexed as \((\widetilde{\mu}_i)_{i \in [k]}\), the ordering is arbitrary and does not necessarily correspond to the true cluster labels.
To resolve this, we use a second algorithm (described in \Cref{alg:permutation}) to recover the correct ordering on the approximate cluster means. That is, we find the permutation \(\pi\) such that each \(\widetilde{\mu}_i\) is close to the cluster mean \(\mu_{\pi(i)}\). That way, we can relabel the  cluster means, so that they satisfy \Cref{def:approxmeans}.

We begin by presenting \Cref{alg:app_centers}, which outputs the set of approximate cluster means (up to a possible permutation). 
\begin{algorithm}
\caption{Approximate cluster means}\label{alg:app_centers}
\begin{algorithmic}[1]
\State \textbf{Input:} $d$-regular graph $G = (V, E)$
\State \textbf{Output:} vertices $u_1, u_2, \ldots$ \Comment{$u_i$ defines $i$-th approximate cluster center: $\widetilde{\mu}_i \coloneqq f_{u_i}$}
\State $T \gets$ $10\eta \cdot k \log k$ vertices sampled i.i.d. from $\unif(V)$
\State $H \gets \emptyset$\label{line:init_H} \Comment{$H$ is the same-cluster graph induced by $T$}
\ForAll{$x, y \in T$}
    \If{$\langle f_x, f_y \rangle_{\text{apx}} \geq 0.9 \cdot \|f_x\|_{\text{apx}}^2$}
        \State $H \gets H \cup \{(x, y)\}$
    \EndIf
\EndFor\label{line:end_init_H}
\State $i \gets 1$
\While{$H \neq \emptyset$}\label{line:begin_centers}
    \State Select arbitrary vertex $x$ from $H$\label{line:select_center} \Comment{Select a representative of a new spectral cluster}
    \State $u_i \gets x$
    \State $H \gets H \setminus \{(x, y) : y \in N_H(x)\}$ \Comment{Remove all vertices of this spectral cluster}
    \State $i \gets i + 1$
\EndWhile\label{line:end_centers}
\State \Return $u_1, u_2, \ldots $
\end{algorithmic}
\end{algorithm}

We now outline the intuition behind \Cref{alg:app_centers}. The algorithm relies on two key properties of the spectral embeddings \(\{f_v\}_{v \in V}\): Firstly, for almost of all vertices, \(f_v\) is well-concentrated around the corresponding cluster mean \(\mu_i\), and secondly, the cluster means \(\{\mu_i\}_{i \in [k]}\) are nearly orthogonal. Therefore, the inner product \(\langle f_u, f_v \rangle_{\apx}\) is large when \(u\) and \(v\) belong to the same cluster, and small when they belong to different clusters.

Using this observation, the algorithm samples a set of vertices uniformly at random and constructs a similarity graph \(H\) by connecting pairs \(u, v\) if \(\langle f_u, f_v \rangle_{\apx}\) is large. We will show that with a good probability, this graph consists of \(k\) disjoint cliques, each corresponding to a single cluster.

To select the approximate cluster means, we simply select one representative vertex \(u_i\) from each clique and use \(f_{u_i}\) as a proxy for \(\mu_i\). To ensure we select only one representative per cluster, the algorithm iteratively picks an arbitrary vertex from the remaining graph and removes all its neighbors before repeating.

We now analyze the algorithm more formally. 

Note that our definitions of $\M$, $\im$ and $\spec$ from \Cref{sec:setting} assume that we already have the cluster means. Therefore, to analyze the algorithm, we introduce notation $\Mapx'$, $\imapx'$, $\specapx'$ to refer to the objects defined with exact cluster means instead of approximate, with different parameters and with $\langle \cdot, \cdot\rangle$ taken as $\langle \cdot, \cdot \rangle_{\apx}$. Namely, 
\begin{definition}\label{def:primes}
    \[\specapx'(i) \coloneqq \left \{x \in V \colon  \|f_x - \mu_i\|^2_{2} < \frac{\phi^2}{1600 \eta}\|\mu_i\|_{2}^2 \right \},\]
\[ \Mapx' \coloneqq V  \setminus \left(\bigcup_{i \in [k]}\specapx'(i)\right), \]
\[\imapx'(i,j) \coloneqq \specapx'(i) \cap C_j.\]
\end{definition}

At least intuitively, these objects have the same properties as $\specapx, \Mapx, \imapx$ since the exact cluster means satisfy the properties of approximate cluster means, and changing the radius in the definition of $\specapx$ by a factor of constant does not influence the properties of $\specapx$. We will still provide proofs whenever we use any particular properties.

With the definitions of $\M'$, $\im'$ and $\spec'$ in place, we can now show that the number of sampled vertices in \Cref{alg:app_centers} is sufficient to hit all of the clusters while avoiding the small set of ``bad'' vertices that do not concentrate around their cluster means.

\begin{claim}\label{claim:center_samples}
Let $T$ be a set of $10\eta \cdot k\log k$ vertices from $V$ sampled uniformly at random. Then, 
\begin{enumerate}
    \item With probability 0.99, $T$ does not contain any vertices from $\Mapx'\cup\imapx'$, i.e.  $\Pr[ T \cap (\Mapx'\cup\imapx') = \emptyset ] \geq 0.99$.
    \label{item:approxmu_event1}
    \item With probability 0.99, $T$ contains representatives from every cluster, i.e. $\Pr[ T \cap C_i \neq \emptyset \ \forall i] \geq 0.99$.\label{item:approxmu_event2}
\end{enumerate}
\end{claim}
\begin{proof} 
We first show the statement \ref{item:approxmu_event1}, and then we show the statement \ref{item:approxmu_event2}. 
\paragraph{1.} 

We start by bound the size $|\Mapx'\cup\imapx'|$. The argument is similar to \Cref{lemma:impostor_n_cross_size}. From \Cref{def:primes}, for all $u \in \Mapx'\cup\imapx'$, it holds that $\| f_u - \mu_{\iota(u)} \|_2^2 \geq \frac{\phi^2}{1600 \eta } \|\mu_{\iota(u)}\|^2_2 \geq \frac{\phi^2}{3200\eta^2}\cdot \frac{k}{n}$, where the second inequality follows by \Cref{rem:mu_i_norm}. Summing over all $u \in \Mapx'\cup\imapx' $, we get 
$$\sum_{u \in \Mapx'\cup\imapx'}\|f_u-  \mu_{\iota(u)}\|_2^2 \geq  |\im' \cup \M'| \cdot \frac{\phi^2}{3200\eta^2} \cdot \frac{k}{n}. 
 $$
 On the other hand, by \Cref{claim:sum_of_distances}, we have 
$$
\sum_{u \in \im' \cup \M'}\|f_u-  \mu_{\iota(u)}\|_2^2 \leq \sum_{u \in V} \|f_u - \mu_{\iota(u)}\|_2^2 \leq \frac{4 \e k}{\phi^2}.
$$
Combining the above two equations, we get $$|\Mapx'\cup\imapx' | \leq 10^5\eta^2\frac{\epsilon}{\phi^4}n.$$

We now have 
    \[\mathbb{E}_{T}[|(\Mapx'\cup\imapx')\cap T|] = 10\eta\cdot k\log k \cdot \frac{|\Mapx'\cup\imapx'|}{n} \leq 10^6\eta^3 \cdot k\log k\cdot \frac{\epsilon}{\phi^4} \leq 10^{-7},\]
where the second inequality follows from $\eta^3 k\log k\cdot\frac{\epsilon}{\phi^6}\leq 10^{-9}$ as per \Cref{fig:setting}. Therefore, with probability 0.99 by Markov's inequality $(\Mapx'\cup\imapx')\cap T$ is empty.
\paragraph{2.} Fix $i$. A vertex sampled from $G$ uniformly at random is not from $C_i$ with probability $1 - \frac{|C_i|}{n} \leq 1 - \frac{1}{\eta k}.$ 

Therefore, since $T$ is a set of $10\eta\cdot k\log k$ vertices sampled independently uniformly at random, we get 
    \[ \Pr[T \cap C_i = \emptyset]\leq  \left(1 - \frac{1}{\eta k}\right)^{|T|}  =  \left(1 - \frac{1}{\eta k}\right)^{10\eta\cdot k\log k} \leq \frac{1}{k^{10}}. \] 

    By the union bound argument, we get 
 \[ \Pr[T \cap C_i = \emptyset \quad \forall i ]\leq  k\cdot \frac{1}{k^{10}} = \frac{1}{k^9} \leq 0.01,\]
    setting $k \geq 2$ as per \Cref{fig:setting}.
\end{proof}

\begin{definition}[Same-cluster graph] Given a set of vertices $T$, we say that $H = (T, E_H)$ is the same-cluster graph induced by $T$ if the set of edges $E_H$ is exactly the set of pairs of vertices in $T$ belonging to the same cluster, i.e.  
$$E_H = \{ (x,y) \in T\times T : \iota(x) = \iota(y)\}. $$
\end{definition}

\begin{lemma}[Correctness of \Cref{alg:app_centers}]\label{lemma:app_centers} 
    With probability $0.98$, Algorithm \ref{alg:app_centers} stops at $i = k+1$ and returns $k$ vertices $u_1, \ldots u_k$ such that, for some permutation $\pi': [k] \to [k]$, for every $i \in [k]$ it holds that
    \[\|\mu_{\pi'(i)} - f_{u_i}\|_{2} \leq \frac{\phi}{40\sqrt{\eta}}\|\mu_{\pi'(i)}\|_{2}.\] 
We use notation $\tmu_i \coloneqq f_{u_i}$ and refer to $\{\tmu_i\}_{i \in [k]}$ as approximate cluster centers.
\end{lemma}

\begin{proof} We first show in \Cref{claim:H_cliques} that conditioned on the event $(\Mapx' \cup \imapx')\cap T = \emptyset$, the graph $H$ computed in \Cref{alg:app_centers} is the same-cluster graph. 
\begin{claim}[$H$ is the same-cluster graph induced by $T$]\label{claim:H_cliques}
If $(\Mapx' \cup \imapx')\cap T = \emptyset$, then $H = \{ (x,y) \in T \times T : \iota(x) = \iota(y)\}. $
\end{claim}
\begin{proof}
The main idea is to use that for all $x , y \notin \Mapx' \cup \imapx' $,  $\langle f_x, f_y\rangle_{\apx}$ is very close to $\|f_x\|^2_{\apx}$ if and only if $x$ and $y$ belong to the same cluster.

We start by proving upper bounds on $|\langle f_x, f_y\rangle_{\apx} - \langle\mu_{\iota(x)}, \mu_{\iota(y)}\rangle|$ and $|\|f_x\|^2_{\apx} - \|\mu_{\iota(x)}\|^2_{\apx}|$. 

We have that
\[\langle f_x, f_y\rangle  =  \langle f_x - \mu_{\iota(x)}, f_y - \mu_{\iota(y)}\rangle + \langle f_x - \mu_{\iota(x)}, \mu_{\iota(y)}\rangle + \langle \mu_{\iota(x)}, f_y - \mu_{\iota(y)}\rangle + \langle\mu_{\iota(x)}, \mu_{\iota(y)}\rangle.\]

From here, by \Cref{thm:spec_dot_prod_oracle}, 
\begin{equation}\label{eq:langlefxfyrangle}
\begin{aligned}
     \left|\langle f_x, f_y\rangle_{\apx} - \langle\mu_{\iota(x)}, \mu_{\iota(y)}\rangle\right|&\leq \left|\langle f_x, f_y\rangle - \langle f_x, f_y\rangle_{\apx}\right| +  |\langle f_x - \mu_{\iota(x)}, f_y - \mu_{\iota(y)}\rangle| \\
    &+  |\langle f_x - \mu_{\iota(x)}, \mu_{\iota(y)}\rangle| + |\langle f_y - \mu_{\iota(y)}, \mu_{\iota(x)}\rangle| \\
     &\leq \frac{\xi}{n} +  |\langle f_x - \mu_{\iota(x)}, f_y - \mu_{\iota(y)}\rangle| +  |\langle f_x - \mu_{\iota(x)}, \mu_{\iota(y)}\rangle|+ \ |\langle f_y - \mu_{\iota(y)}, \mu_{\iota(x)}\rangle|.
\end{aligned}
\end{equation}

Since $x \notin \Mapx'$, by \Cref{def:primes}, $x$ is spectrally close to some cluster center, and since $x \notin \imapx'$, by \Cref{def:primes} again, $x$ can only be spectrally close to its own cluster center. More formally, we have 
\[\|f_x - \mu_{\iota(x)}\|^2_{2} \leq \frac{\phi^2}{1600\eta}\|\mu_{\iota(x)}\|^2_{2}  \leq \frac{\phi^2}{800}\cdot\frac{k}{n},\]
where the last inequality follows since  $\|\mu_{\iota(x)}\|^2_{2} \leq \frac{2\eta k}{n}$, by \Cref{rem:mu_i_norm}.  Similarly, $\|f_y - \mu_{\iota(y)}\|_{2} \leq \frac{\phi}{\sqrt{800}}\cdot\sqrt{\frac{k}{n}}$. Combining this with a stronger bound $\|f_x - \mu_{\iota(x)}\|^2_{2} \leq \frac{\phi^2}{1600\eta}\|\mu_{\iota(x)}\|^2_{2} $ above we get
\begin{equation}\label{eq:11}
    |\langle f_x - \mu_{\iota(x)}, f_y - \mu_{\iota(y)}\rangle| \leq \| f_x - \mu_{\iota(x)}\|_2 \cdot \| f_y - \mu_{\iota(y)}\|_2 \leq  \frac{\phi^2}{800}\cdot\sqrt{\frac{k}{n}}\cdot\|\mu_{\iota(x)}\|_2,
\end{equation}
and 
\begin{equation}\label{eq:12}
    |\langle f_x - \mu_{\iota(x)}, \mu_{\iota(y)}\rangle|, \ |\langle f_y - \mu_{\iota(y)}, \mu_{\iota(x)}\rangle| \leq \frac{\phi}{40\sqrt{\eta}}\|\mu_{\iota(x)}\|_2\cdot\|\mu_{\iota(y)}\|_2 \leq \frac{\phi}{20}\cdot\sqrt{\frac{k}{n}}\cdot\|\mu_{\iota(x)}\|_2,
\end{equation}
where the last inequality follows from $\|\mu_{\iota(y)}\|^2_2 \leq \frac{2\eta k}{n}$, by \Cref{rem:mu_i_norm}.
Substituting \Cref{eq:11} and \Cref{eq:12} into  \Cref{eq:langlefxfyrangle}, we get
\begin{equation}\label{eq:bound_on}
\begin{aligned}
    \left|\langle f_x, f_y\rangle_{\apx} -  \langle\mu_{\iota(x)}, \mu_{\iota(y)}\rangle \right|\leq\frac{3\phi}{20}\cdot\sqrt{\frac{k}{n}}\cdot\|\mu_{\iota(x)}\|_2 + \frac{\xi}{n} \leq\frac{1}{2} \frac{\phi}{\|\mu_{\iota(x)}\|_2}\cdot\sqrt{\frac{k}{n}}\cdot\|\mu_{\iota(x)}\|^2_2 \leq \phi\cdot\sqrt{\eta}\cdot\|\mu_{\iota(x)}\|^2_2,
\end{aligned}
\end{equation} 
where the second inequality follows by setting $\frac{\xi}{n} \leq \frac{\phi^2}{800\eta}\|\mu_{\iota(x)}\|_2\cdot\sqrt{\frac{k}{n}}$ as per \Cref{fig:setting}, and the last inequality follows since $\|\mu_{\iota(x)}\|_2 \geq \frac{1}{2 \sqrt \eta}$, by \Cref{rem:mu_i_norm}. 

Next, we have
\begin{equation}\label{eq:bound_on_2}
    \left|\|f_x\|^2_{\apx}- \|\mu_{\iota(x)}\|^2_2\right|\leq \left|\|f_x\|^2_{\apx}- \|f_x\|^2_2\right| + \left|\|f_x\|^2_2- \|\mu_{\iota(x)}\|^2_2\right| \leq  \frac{3\phi}{10\sqrt{\eta}}\cdot\|\mu_{\iota(x)}\|^2_2 + \frac{\xi}{n} \leq  \frac{\phi}{2\sqrt{\eta}}\cdot\|\mu_{\iota(x)}\|^2_2. 
\end{equation}
Here, the second inequality follows from \Cref{def:dist_apx} and
\[\|f_x\|_2 \leq \|\mu_{\iota(x)}\|_2 + \|f_x - \mu_{\iota(x)}\|_2 \leq \left(1 + \frac{\phi}{10\sqrt{\eta}}\right)\cdot\|\mu_{\iota(x)}\|_2\] by \Cref{claim:spec_exact_mu}, and the last inequality follows by setting $\frac{\xi}{n} \leq \frac{\phi^2}{800\eta}\|\mu_{\iota(x)}\|^2_2$ as per \Cref{fig:setting}. 
    
If $\iota(x) = \iota(y)$ then
\[\langle f_x, f_y\rangle_{\apx} \geq \|\mu_{\iota(x)}\|^2_2 -  \left|\langle f_x, f_y\rangle_{\apx} -  \|\mu_{\iota(x)}\|^2_2 \right|\geq   (1 - \phi\sqrt{\eta})\cdot\|\mu_{\iota(x)}\|^2_2 \geq \frac{1 - \phi\sqrt{\eta}}{1 + \phi/(2\sqrt{\eta})}\|f_x\|^2_{\apx} \geq 0.9\cdot\|f_x\|^2_{\apx},\]
where the second inequality follows from  \Cref{eq:bound_on}, the third inequality follows from \Cref{eq:bound_on_2}, and the last inequality follows from $\phi^2\eta \leq 10^{-3}$ as per \Cref{fig:setting}. 

If $\iota(x) \neq \iota(y)$ then
\[\langle f_x, f_y\rangle_{\apx} \leq \langle \mu_{\iota(x)}, \mu_{\iota(y)}\rangle + \left|\langle f_x, f_y\rangle_{\apx} -   \langle \mu_{\iota(x)}, \mu_{\iota(y)}\rangle  \right|\leq   2\phi\sqrt{\eta}\cdot\|\mu_{\iota(x)}\|^2_2 \leq \frac{2\phi\sqrt{\eta}}{1 - \phi/(2\sqrt{\eta})}\|f_x\|^2_{\apx} \leq 0.8\cdot\|f_x\|^2_{\apx},\]
where the second inequality follows from  \Cref{eq:bound_on} and 
\[|\langle \mu_{\iota(x)}, \mu_{\iota(y)}\rangle| \leq \frac{8\sqrt{\epsilon}}{\phi}\cdot\frac{\eta k}{n} \leq \frac{10^{-1}\phi}{\eta}\cdot \frac{k}{n}  \leq 10^{-1}\phi\|\mu_{\iota(x)}\|^2_2,\] 
the third inequality follows from \Cref{eq:bound_on_2}, and the last inequality follows from $\phi^2\eta \leq 10^{-3}$ as per \Cref{fig:setting}.

\end{proof}

By \Cref{claim:H_cliques} $H$ is a spectral proximity graph, and by the second item of \Cref{claim:center_samples} it has exactly $k$ cliques. Every time the line~\eqref{line:select_center} is triggered, we pick a vertex $x$, declare it an approximate cluster center by setting $\tmu_i \coloneqq f_x$, and remove the rest of the vertices from $C_{\iota(x)}$ from $H$.

Because $x \notin \Mapx'\cup\imapx'$, $x$ is spectrally close to the center of its own cluster, which by the definition means:
\[\|\tmu_i - \mu_{\iota(x)}\|_{2}  = \|f_x - \mu_{\iota(x)}\|_{2} \leq \frac{\phi}{40\sqrt{\eta}}\|\mu_{\iota(x)}\|_{2},\] as desired.

\end{proof}

\Cref{lemma:app_centers} shows that we can compute cluster means $\{ \tmu_i\}_{i \in [k]}$ up to permutation. We use Algorithm \ref{alg:permutation} to recover the correct permutation. 

\begin{algorithm}
\caption{Obtaining the permutation}\label{alg:permutation}
\begin{algorithmic}[1]
\State \textbf{Input:} $d$-regular graph $G = (V, E)$, labeling $\sigma$ as per \Cref{fig:setting}, 
\State \qquad \qquad $u_1, \ldots, u_k  \coloneqq $ output of \Cref{alg:app_centers}  \Comment{$\exists \pi: [k] \to [k]$: $C_{\pi(i)}$: $\widetilde{\mu}_{\pi(i)} \coloneqq f_{u_i}$}
\State \textbf{Output:} permutation $\pi: [k] \to [k]$ 
\State $S \gets 50\eta \cdot k \log k$ vertices sampled independently $\sim \unif(V)$
\ForAll{$u$}
    \State $S_u \gets \{ v \in S : \| f_v - f_u \|_{\apx}^2 \leq \frac{\phi^2}{400\eta} \|f_u\|_{\apx}^2 \}$\label{line:S_u} \Comment{$S_u$ are the samples that are spectrally close to $u$}
    \State $\pi(u) \gets \arg\max_{j \in [k]} \left| \{ v \in S_u : \sigma(v) = j \} \right|$\label{line:pi} \Comment{Predict the label $\pi(u)$ by majority voting}
\EndFor
\State \Return $\pi$
\end{algorithmic}
\end{algorithm}
Given the cluster means $\{\tmu_i\}_{i \in [k]}$, 
Algorithm \ref{alg:permutation} simply samples $\widetilde O(k)$ vertices at random and uses the labels $\sigma$ of the sampled vertices to determine the permutation. Specifically, for each cluster mean $\tmu_i$, we identify the sampled vertices that are spectrally close to $\tmu_i$, and do a majority vote over these vertices in to determine the label we assign to $\tmu_i$. 

Finally, we prove \Cref{lemma:pi_computation}, which proves the guarantee of \Cref{alg:permutation}. 

\begin{remark}
In the proof of \Cref{lemma:pi_computation}, we will refer to the spectral objects $\spec(i)$, $\M$ and $\im(i, j)$, as defined in \Cref{sec:setting}. Note that these objects are defined in terms of the approximate cluster means $\{ \tmu_i\}_{i \in [k]}$. Conditioned on the success of  \Cref{alg:app_centers}, the spectral objects $\spec(i)$, $\M$ and $\im(i, j)$ exist and satisfy all of the properties from \Cref{sec:setting}. Note that without knowing the correct labeling on $\{ \tmu_i\}_{i \in [k]}$, we cannot verify if a given vertex belongs to a particular $\spec(i)$ or $\im(i, j)$. However, we remark that our analysis below only uses the existence of these objects, but not access to them. 
\end{remark}

\begin{proof}[Proof of \Cref{lemma:pi_computation}] Note that all of the vertices in $G$, except for $\M\cup \im$, are spectrally close to their own cluster mean. First, we would like to show that with high constant probability none of the sampled vertices in $S$ belong to $\M \cup \im$, and that $S$ contains sufficiently many vertices from each of the clusters. In that case, it is possible to partition $S \cup \{u_1, \ldots, u_k\}$ into $k$ groups of spectrally close vertices, and all of the vertices in every group belong to the same cluster. Since every group is sufficiently large, and $\sigma$ errs on every vertex with bounded probability $\delta$, it suffices to declare $\pi(u)$ for every vertex $u \in \{u_1, \ldots, u_k\}$ to be the most common $\sigma$ label in the group which $u$ belongs to.

Now we state the proof more formally. Let $\pi^*$ be the correct permutation, which exists by \Cref{lemma:app_centers}. The goal of the analysis, thus, is to show that  $\pi$ constructed by \Cref{alg:permutation} coincides with $\pi^*$.

Let $\mathcal E_1$ be the event that $S \cap  (\M \cup \im) = \emptyset$. From line~\eqref{line:S_u}, we have that for all $u \in \{u_1, \ldots, u_k\}$, the set $S_u$ is defined by $S_u \gets \{ v \in S : \| f_v - f_u \|_{\apx}^2 \leq \frac{\phi^2}{400\eta} \|f_u\|_{\apx}^2 \}$. 
In other words, $S_{u_i}$ is precisely the set of vertices sampled from $\spec(\pi^*(i))$, as per \Cref{def:spec}.

Let  $\mathcal E_2$ be the event that $S_{u_i} \setminus \im \geq 20 \log k$ simultaneously for all $i \in [k]$. As mentioned in the proof outline above, we would want to show that event $\bar{ \mathcal{E}_1} \cup \bar{\mathcal{E}_2 }$ takes place with small constant probability. The randomness of each of these events is over the choice of $S$ but not over $\sigma$. To bound $\Pr_S[\bar{\mathcal{E}_1}]$, note that by Markov's inequality
\begin{align*}
\Pr_S[\bar{\mathcal{E}_1}] &= \Pr_S[ |S \cap (\M \cup \im)| \geq 1] \\
& \leq \E_S[|S \cap (\M \cup \im)|] \\
& = |S| \cdot \frac{|\M \cap \im|}{n} \\
& \leq 50 \eta k \log k  \cdot 2 \cdot 10^4 \cdot \eta^2 \cdot \frac{\epsilon}{\phi^4 }\frac{n}{n} \qquad \qquad &&\text{by Lemma \ref{lemma:impostor_n_cross_size}} \\
& \leq 10^{-3} \qquad &&\text{by the assumptions in  \Cref{fig:setting}. }
\end{align*}
Next, we bound $\Pr_S[\bar{\mathcal{E}_2}].$ Fix a $i \in [k].$ We have 
\begin{align*}
\E_S\left [| S_{u_i} \setminus \im|\right] & = |S|\cdot  \frac{|\speccore(\pi^*(i))|}{n} \qquad && \text{since $\spec(\pi^*(i))\setminus \im  = \speccore(\pi^*(i))$}\\
& \geq 50\eta k \log k \cdot \frac{|C_{\pi^*(i)}| - |\M \cup \im|}{n} \qquad && \text{since $C_{\pi^*(i)} \setminus (\M\cup \im) \subseteq \speccore(\pi^*(i))$}\\
& \geq \frac{50 \eta k \log k}{n} \left( \frac{1}{\eta}\frac{n}{k} - 2 \cdot 10^4 \cdot \eta^2 \cdot \frac{\epsilon}{\phi^4 }n\right)  \qquad && \text{by definition of $\eta$ and \Cref{lemma:impostor_n_cross_size}} \\
& \geq 40 \log k \qquad && \text{by the assumptions in \Cref{fig:setting}.}
\end{align*}
So by Chernoff bounds, we get
$$\Pr_S[ |S_{u_i} \setminus \im| \leq 20 \log k] \leq \exp(- 50 \log k/12) \leq \exp(-3 \log k) \leq \frac{10^{-3}}{k}.$$
By a union bound over all $i \in [k]$, we get $\Pr_S[\bar{\mathcal{E}_2}] \leq 10^{-3}$. 

Now, condition on $\mathcal{E}_1$ and $\mathcal{E}_2$ as well as on the success of \Cref{alg:app_centers}, and note that these events are independent of the labels $\sigma$. The probability of these events over internal randomness of \Cref{alg:app_centers} and \Cref{alg:permutation} is bounded by 
\[0.02 + \Pr_{S}[\mathcal{E}_1] + \Pr_{S}[\mathcal{E}_2] \leq 0.03.\] 

Since we are conditioning on the event $\mathcal{E}_1$, we have that 
\[S_{u_i} =  S_{u_i} \setminus \im  = S \cap (\spec(\pi^*(i)) \setminus \im) = S \cap \speccore(\pi^*(i)). \]

By the definition of $\pi$ in line~\eqref{line:pi}, we have that $\pi(i)$ can be different from $\pi^*(i)$ only if at least half of the vertices in $S_{u_i}$ are mislabeled. Fix $i \in [k]$, and recall from \Cref{def:mislabeled} the set $\mislabeled$ of mislabeled vertices. 
Note that $|S_{u_i} \cap \mislabeled|$ is the number of successes in the Bernoulli process with $|S_{u_i}|$ many trials and success probability $\delta$. Therefore, 
\begin{align*}
 \Pr_{\sigma}\left[ |S_{u_i} \cap \mislabeled| \geq \frac{1}{2}|S_{u_i}| \right]\leq \exp(-|S_{u_i}|/(20\delta)), 
\end{align*} since $\delta \leq 1/100$ as per \Cref{fig:setting}.
Since we are conditioning on the event $\mathcal{E}_2$ we have that $|S_{u_i}| \geq 20\log k$, and so
\begin{align*}
 \Pr_{\sigma}\left[ |S_{u_i} \cap \mislabeled| \geq \frac{1}{2}|S_{u_i}| \right]\leq \exp(-|S_{u_i}|/(20\delta)) \leq \exp(\log k/\delta) \leq \frac{10^{-3}}{k},
\end{align*}since $\delta \leq 1/100$.

By the union bound argument, with probability at east $1 - 10^{-3}$ it holds that $ |S_{u_i} \cap \mislabeled| < \frac{1}{2}|S_{u_i}| $ for all $i \in [k]$ , which implies 
\[\Pr_{\sigma}\left [ \exists i \in [k] : \pi(i) \neq \pi^*(i)\right ] \leq 10^{-3}.\]
From here, we get 
$$ \Pr_{\sigma}\left[ \forall i \in [k]: \pi(i) = \pi^*(i) \right]  =  1 - \Pr_{\sigma}\left [ \exists i \in [k] : \pi(\tmu_i) \neq i\right ] \geq  1- 10^{-3}.$$
Finally, if $\pi = \pi^*$ then, by \Cref{lemma:app_centers}
\[\|f_{u_i} - \mu_{\pi(i)}\|^2_2 =  \|f_{u_i} - \mu_{\pi^*(i)}\|^2_2 \leq \frac{\phi^2}{1600\eta}\|\mu_{\pi^*(i)}\|^2_2 = \frac{\phi^2}{1600\eta}\|\mu_{\pi(i)}\|^2_2,\]
as desired.
\end{proof}

\section{Simple classifier for $k=2$ and $\delta \le \epsilon$}
\label{sec:majvoting++}

We now describe a more nuanced variant of the vanilla majority voting scheme introduced in \Cref{sec:techoverview_attemps} which we shall call ``majority-voting++'', still in the simple case of $k=2$. The key idea is to make the decision rule less brittle for vertices whose neighborhoods are dominated by the opposite cluster. The algorithm classifies a vertex $u \in V$ as follows.
\begin{algorithm}
\caption{Majority-voting++}\label{alg:approach3}
\begin{algorithmic}[1]
\State \textbf{Input:} $G$, $\sigma$, and a vertex $u \in V$
\State \textbf{Output:} a cluster id in $\{1,2\}$
\If{$|\{v \in \nei_G(u) : \, \sigma(v)=1 \}| \le \frac{2}{3}\phi d $}  \Comment{$u$ is probably in $C_2$} 
\State \label{case:1} \Return $2$
\ElsIf{$|\{v \in \nei_G(u) : \, \sigma(v)=2 \}| \le \frac{2}{3}\phi d $} \Comment{$u$ is probably in $C_1$} 
\State  \label{case:2} \Return $1$
\Else \Comment{for both $i=1,2$, $u$ has  $> \frac{2}{3}\phi d$ neighbors $v$ with $\sigma(v)=i$}
\State \label{case:3} \Return $\sigma(u)$
\EndIf
\end{algorithmic}
\end{algorithm}

\noindent
For simplicity, we only consider the setting where $d$ is bigger than a large enough constant (as a function of $\phi$) and $\delta \ll \phi^2$. To analyze the misclassification rate of this algorithm, we bound how many vertices in $C_1$ can be misclassified. By symmetry, the same reasoning will apply to $C_2$, giving the overall misclassification rate. We consider two types of vertices in $C_1$: those with at least $(1-\phi/10) d$ neighbors in $C_1$, and those with more than  $\phi/10 \cdot d$ neighbors in $C_2$. We first analyze the vertices of the first kind.

\begin{lemma}\label{lem:maj++}
Any vertex of $C_1$  with at least $\left(1-\frac{\phi}{10}\right) d$ neighbors in $C_1$ is misclassified with probability $\approx \delta^2$. 
\end{lemma}
\begin{proof}[Proof sketch] Any such vertex $u$ is misclassified only in one of these two situations:
\begin{itemize}
    \item in line~\eqref{case:1}, if at least $\left(1-2/3 \cdot \phi - \phi/10 \right)d =\Omega(d)$ of its neighbors are wrongly labeled as $\sigma(v)=2$, which happens with probability at most $\approx \Omega(\delta)^{(1-O(\phi))d} \le \delta^2$. This follows by an application of Chernoff bounds, as we only expect to see $\delta$ fraction of the neighbors mislacssified and, by above, at least $(1 - O(\phi))$ has to be mislassified;
    \item in line~\eqref{case:3}, if $u$ is wrongly labeled as $\sigma(u)=2$, which happens with probability $\delta$, and  more than $(2/3 - 1/10)\cdot \phi \cdot d$ of its neighbors in $C_1$ are also wrongly labeled, which happens with probability $\approx \delta^{\Omega(\phi d)} \le\delta$,  by application of Chernoff bounds. Since the two events are independent, they occur simultaneously with probability $\leq \delta^2$. 
\end{itemize}
Overall, any vertex of this type is misclassified with probability $\delta^2$.
\end{proof}

\noindent
Considering now the second type of vertices, we have the following.

\begin{lemma}
Any vertex of $C_1$  with more than  $\phi/10 \cdot d$ neighbors in $C_2$, of which there at most $\approx \epsilon/\phi\cdot n$, is misclassified with probability $\approx \delta$.
\end{lemma}
\begin{proof}
    There can only be $\approx \epsilon/\phi\cdot n$ of them, otherwise there would be too many edges between $C_1,C_2$; second, every vertex in $C_1$ must still have at least $\phi d$ neighbors in $C_1$, since the singleton cuts must $\phi$-expand. Therefore, a vertex $u$ of the this type is misclassified only in one of these two situations:
    \begin{enumerate}
        \item in line~\ref{case:1}, if all but $2\phi/3 \cdot d$ of its neighbors are wrongly labeled as $\sigma(v)=2$, which happens with probability at most $\delta^{\Omega(\phi d)}\le \delta$. This is similar to the proof of \Cref{lem:maj++}~-- note that $u$ has at least $\phi d$ neighbors in its own cluster, which means that at least a $1/3$ fraction of its neighbors have to be mislassified, while we only expect to see $\delta$ fraction of the neighbors mislassified;
        \item in line~\ref{case:3}, if $u$ is wrongly labeled as $\sigma(u)=2$, which happens with probability $\delta$.
    \end{enumerate}
    Overall, any vertex of this type is misclassified with probability $\delta$.
\end{proof}
\noindent
Combining the failure probabilities of the two types of vertices with the fact that there are only $\approx \epsilon/\phi \cdot n$ vertices of the second type, we get a misclassification rate of $\approx \delta^2+\epsilon \delta/\phi$.

This simple algorithm, gives the sought-after rate of $\approx \epsilon \delta$ whenever $\delta \le \epsilon$, or more generally $\delta^{\Omega(\phi d)} \le \epsilon \delta$. Even if this already provides a partial answer to our initial question, we do not want to make any such assumptions on $\delta$ and $d$. A natural attempt to artificially increase the degree, in order to bring $\delta^{\Omega(\phi d)}$ below $\epsilon \delta$, is the following: power the graph $G$ for $t$ times, until \smash{$\delta^{\Omega(\phi d^t)} \le \epsilon \delta$}, and apply the algorithm above to the powered graph $G^t$. However, there is an issue with this plan: we do not know how to argue that $G^t$ remains $(2,\epsilon,\phi)$-clusterable with respect to $C_1,C_2$. To illustrate this, consider the following example graph, where $d \ge 1/\phi$: the cluster $C_1$ contains two sets $S,T$ of size $|S|=\epsilon/10 \cdot n$ and $|T|=\phi d \cdot|S|$. Connect the sets $S$ and $T$ so that $\nei_{C_1}(S) =T$, so each vertex in $S$ has $\phi d$ neighbors in $C_1$. Additionally, connect the set $T$ to the rest of $C_1$ so that every vertex in $T$ has $2 \phi d$ neighbors in $ C_1\setminus (S \cup T)$.  Every other vertex in $C_1$ has all but $1$ of its neighbors in $C_1 \setminus (S \cup T)$, and $G[C_1\setminus (S \cup T)]$ induces a $2\phi$-expander. One can see that this construction does not violate the requirement that $G$ is clusterable. On the other hand, the number of vertices in $C_1$ in the two-hop neighborhood of $u \in S$ is only $\phi d+2\phi d \cdot \phi d \le 3 \phi^2 d^2$. This means that we would need to change the thresholds in the algorithm to be on the order of $\phi^2 d^2$, while we would have liked to use $\phi d^2$. Hence, the misclassification rate would be at least suffer the term \smash{$\delta^{\Omega((\phi d)^2)}$}: this is much worse than the loss \smash{$\delta^{\Omega(\phi d^2)}$} we were wishing for.

\section{Deferred proofs}
\label{sec:apndx}

\subsection{Proof of \Cref{claim:exp_drop}}
\label{subsec:expdrop}
Recall that \Cref{claim:exp_drop} is a technical tool for proving \Cref{lem:false_pos_p}. Before we revisit the statement of the Claim, recall notation used in the proof of \Cref{lem:false_pos_p}.  Let $S_1$ be the set of neighbors of cross and impostor vertices, i.e.

    \begin{equation*}
        S_1 = \nei_G(\M \cup \im) \, .
    \end{equation*}
    For a fixed labeling $\alpha$, for all $i \in [n]$ we let $S_{i+1}(\alpha)$ be the set of neighbors in $G$ of the vertices mislabeled vertices in $S_{i}(\alpha)$ that are not in any previously defined sets, i.e.
    \begin{equation*}
        S_{i+1}(\alpha) = \nei_G(\{u \in S_{i}(\alpha) : \alpha(u)\neq \iota(u)\})\setminus \left(\bigcup_{j =1}^{i} S_j(\alpha)\right) \, .
    \end{equation*}
    Note that the definition of the set $S_1$ does not depend on the labeling $\alpha$ but, for the sake of uniformity of notation, we will still write $S_1(\alpha)$. For convenience, we let $\widehat{S}_i(\alpha)$ be the set of  vertices in $S_i(\alpha)$ with a wrong label, i.e.
    \begin{equation*}
        \widehat{S}_i(\alpha)= \{u \in S_i(\alpha) : \, \alpha(u)\neq \iota(u)\}\, ,
    \end{equation*}
    and let ${S}_{\le i}(\alpha)$ be the set of vertices in any of $S_1(\alpha),\dots,S_i(\alpha)$, i.e.
    \begin{equation*}
        S_{\le i}(\alpha) = \bigcup_{j =1}^i S_j(\alpha) \, ,
    \end{equation*}
    so we can rewrite
    \begin{equation*}
        S_{i+1}(\alpha)=\nei_G(\widehat{S}_i(\alpha))\setminus S_{\le i}(\alpha) \, .
    \end{equation*}
    Similarly, we also let $\widehat{S}_{\le i}(\alpha)$ be the set of vertices in any of $\widehat{S}_1(\alpha),\dots,\widehat{S}_i(\alpha)$, i.e.
    \begin{equation*}
        \widehat{S}_{\le i}(\alpha) = \bigcup_{j =1}^i \widehat{S}_j(\alpha) \, ,
    \end{equation*}
\expdrop*
\begin{proof}
        Consider $i$ fixed everywhere in this proof. Let \[\mathcal{T} = \{\mathbf{T}= (T,\widehat{T}) \in ((2^V)^{i})^2: \forall\, j \in [i], \, \widehat{T}_j \subseteq T_j \text{ and }  \forall\, j \in [i-1], \, T_{j+1} = \nei_G(\widehat{T}_{j})\setminus (\cup_{h=1}^j T_h)\}.\] In other words, any $\mathbf{T}$ is a pair of two sequences $T = T_1, \ldots, T_i$, $\widehat{T} = \widehat{T}_1, \ldots, \widehat{T}_i$, where for each $j \in [i]$ both $T_j$ and $\widehat{T}_j$ are subsets of $V$. The sequence $T$ is defined by its first element $T_1$ and the sequence $\widehat{T}$. $\mathcal{T}$ is the set of all possible pairs of sequences $\mathbf{T}$. Note that for any $\alpha$, $((S_j(\alpha))_{j=1}^i,(\widehat{S}_j(\alpha))_{j=1}^i) \in \mathcal{T}$. $\mathcal{T}$ is thus designed to contain all $((S_j(\alpha))_{j=1}^i,(\widehat{S}_j(\alpha))_{j=1}^i)$ for all possible realizations of $\sigma$.
        
        Define $\text{nxt}(\mathbf{T}) = \nei_G(\widehat{T}_{i})\setminus \cup_{j=1}^i T_j$ for all $\mathbf{T} \in \mathcal{T}$. Note that $\text{nxt}(\mathbf{T})$ is defined by $T = T_1, \ldots, T_i$ and $\widehat{T} = \widehat{T}_1, \ldots, \widehat{T}_i$ in the same way as $T_{j+1}$ is defined by $T_1, \ldots, T_j$ and $\widehat{T}_1, \ldots, \widehat{T}_j$ for any $j < i$, so it's meant to be the \emph{next} object in the sequence $T$. In particular, if $\mathbf{T} = ((S_j(\alpha))_{j=1}^i,(\widehat{S}_j(\alpha))_{j=1}^i)$ then $\text{nxt}(\mathbf{T}) = S_{i+1}(\alpha)$.   
        
        For any $\alpha:V\rightarrow [k]$, let $\text{seq}(\alpha,i)=((S_j(\alpha))_{j=1}^i,(\widehat{S}_j(\alpha))_{j=1}^i)$. Also, for any $Q \subseteq V$, and labelings $\gamma: Q\rightarrow[k]$, $\beta: V\setminus Q\rightarrow[k]$, we denote by $\gamma \| \beta : V\rightarrow [k]$ the labeling that assigns $\gamma(u)$ to each $u \in Q$ and assigns $\beta(u)$ to each $u \in V\setminus Q$. With this notation in place, we rewrite
        \begin{align*}
            \E_\sigma\left[\left|\widehat{S}_{i+1}(\sigma)\right|\right] & = \sum_{\alpha: V \rightarrow [k]} \left|\widehat{S}_{i+1}(\alpha)\right|\cdot \Pr_\sigma[\sigma=\alpha] \\
            & = \sum_{\mathbf{T} \in \mathcal{T}} \sum_{\substack{\alpha: V \rightarrow [k] \text{ s.t.}\\ \text{seq}(\alpha,i)=\mathbf{T}}}  |\widehat{S}_{i+1}(\alpha)|\cdot \Pr_\sigma[\sigma=\alpha] \\
            &= \sum_{\mathbf{T} \in \mathcal{T}} \sum_{\substack{\gamma: \text{nxt}(\mathbf{T}) \rightarrow [k], \\ \beta: V\setminus\text{nxt}(\mathbf{T}) \rightarrow [k] \\ \text{s.t. } \text{seq}(\gamma \| \beta,i)=\mathbf{T}}}  \left|\{u \in \text{nxt}(\mathbf{T}): \gamma(u) \neq \iota(u)\}\right|\cdot  \Pr_\sigma[\sigma=\gamma \| \beta ] \, .
        \end{align*}
        In more details, in the last line we decouple the labeling $\alpha$ into $\gamma$ and $\beta$. As mentioned before, if $\text{seq}(\alpha,i)=\mathbf{T}$ then $S_{i+1}(\alpha) =\text{nxt}(\mathbf{T})$. By the definition of $\widehat{S}_{i+1}(\alpha)$ we have that $\widehat{S}_{i+1}(\alpha)\subseteq S_{i+1}(\alpha)$ is the set of all vertices in $S_{i+1}(\alpha)$ misclassified by $\alpha$. Since $\text{nxt}(\mathbf{T})$ is the domain of $\gamma$, we have that $\widehat{S}_{i+1}(\alpha) = \{u \in \text{nxt}(\mathbf{T}): \gamma(u) \neq \iota(u)\}$.

        Note that for any $\mathbf{T} \in \mathcal{T}$ and $\gamma: \text{nxt}(\mathbf{T}) \rightarrow [k]$, $\beta: V\setminus\text{nxt}(\mathbf{T}) \rightarrow [k]$ one has
        \begin{equation}
        \label{eq:split}
             \Pr_\sigma[\sigma=\gamma \| \beta ] =  \Pr_{\sigma(\text{nxt}(\mathbf{T}))}[\sigma( \text{nxt}(\mathbf{T}))=\gamma] \cdot \Pr_{\sigma(V\setminus \text{nxt}(\mathbf{T}))}[\sigma(V\setminus \text{nxt}(\mathbf{T}))=\beta] \,,
        \end{equation} since random variables $\sigma(u)$ and $\sigma(v)$ are independent for any $u \neq v$. Next, we argue that for any $\mathbf{T} \in \mathcal{T}$ and $ \beta: V\setminus\text{nxt}(\mathbf{T}) \rightarrow [k]$, one has that
        \begin{equation}
        \label{eq:eitheror}
            \text{either} \quad  \forall \, \gamma: \text{nxt}(\mathbf{T}) \rightarrow [k], \,  \text{seq}(\gamma \| \beta,i)=\mathbf{T} \quad \text{or} \quad \forall \, \gamma: \text{nxt}(\mathbf{T}) \rightarrow [k], \,  \text{seq}(\gamma \| \beta,i)\neq \mathbf{T} \, .
        \end{equation}
        To show this, we prove that if there exists $\gamma: \text{nxt}(\mathbf{T}) \rightarrow [k]$ such that $ \text{seq}(\gamma \| \beta,i)=\mathbf{T}$, then one has $ \text{seq}(\gamma' \| \beta,i)=\mathbf{T}$ for all $\gamma': \text{nxt}(\mathbf{T}) \rightarrow [k]$. To this end, we fix $\gamma': \text{nxt}(\mathbf{T}) \rightarrow [k]$ and argue by induction on $j \in [i]$ that for all $h \in [j]$ one has \smash{$S_h(\gamma'\| \beta) = S_h(\gamma\| \beta)$ and $\widehat{S}_h(\gamma'\| \beta) = \widehat{S}_h(\gamma\| \beta)$}. Let $\alpha=\gamma\| \beta$ and $\alpha'=\gamma'\| \beta$ for brevity. For the base case $j=1$, we have
        $$S_1(\alpha')=S_1=S_1(\alpha)$$
        and
        $$\widehat{S}_1(\alpha') = \{u \in S_1: \alpha'(u)\neq \iota(u)\}  = \{u \in T_1: \alpha(u)\neq \iota(u)\}= \{u \in S_1: \alpha(u)\neq \iota(u)\} = \widehat{S}_1(\alpha) \, ,$$
        since $T_1=S_1$, $T_1 \cap \text{nxt}(\mathbf{T})=\emptyset$, and therefore $T_1$ lies entirely in the domain of $\beta$. For the inductive step, let $j \in [i-1]$ and assume that for all $h \in [j]$ one has \smash{$S_h(\alpha') = S_h(\alpha)$ and $\widehat{S}_h(\alpha') = \widehat{S}_h(\alpha)$}. Then, we have
        $$ S_{j+1}(\alpha') = \nei_G(\widehat{S}_j(\alpha')) \setminus  \left(\bigcup_{h=1}^j {S}_h(\alpha')\right) = \nei_G(\widehat{S}_j(\alpha)) \setminus  \left(\bigcup_{h=1}^j {S}_h(\alpha)\right) = S_{j+1}(\alpha) \, ,$$
        and
        $$\widehat{S}_{j+1}(\alpha') = \{u \in {S}_{j+1}(\alpha'): \alpha'(u)\neq \iota(u)\}  = \{u \in T_{j+1}: \alpha(u)\neq \iota(u)\}= \{u \in {S}_{j+1}(\alpha): \alpha(u)\neq \iota(u)\} = \widehat{S}_{j+1}(\alpha) \, ,$$
        since ${S}_{j+1}(\alpha)={S}_{j+1}(\alpha)=T_{j+1}$, $T_{j+1} \cap \text{nxt}(\mathbf{T})=\emptyset$, and therefore $T_{j+1}$ lies entirely in the domain of $\beta$. The inductive proof is thus concluded. 
        
        For a fixed $\mathbf{T} \in \mathcal{T}$, we  use~\eqref{eq:eitheror} to write
        \begin{align*}
             & \sum_{\substack{\gamma: \text{nxt}(\mathbf{T}) \rightarrow [k], \\ \beta: V\setminus\text{nxt}(\mathbf{T}) \rightarrow [k] \\ \text{s.t.}\\ \text{seq}(\gamma \| \beta,i)=\mathbf{T}}}  |\{u \in \text{nxt}(\mathbf{T}): \gamma(u) \neq \iota(u)\}|\cdot  \Pr_\sigma[\sigma=\gamma \| \beta ] \\
             = & \sum_{\substack{ \beta: V\setminus\text{nxt}(\mathbf{T}) \rightarrow [k] \\ \text{s.t.}\\ \forall \, \gamma': \text{nxt}(\mathbf{T}) \rightarrow [k], \, \text{seq}(\gamma' \| \beta,i)=\mathbf{T}}} \sum_{\gamma: \text{nxt}(\mathbf{T}) \rightarrow [k]}  |\{u \in \text{nxt}(\mathbf{T}): \gamma(u) \neq \iota(u)\}|\cdot  \Pr_\sigma[\sigma=\gamma \| \beta ] \, .
        \end{align*}
        Now, for any $\beta: V\setminus\text{nxt}(\mathbf{T}) \rightarrow [k] $ such that $\text{seq}(\gamma' \| \beta,i)=\mathbf{T}$ for all $\gamma': \text{nxt}(\mathbf{T}) \rightarrow [k]$, we use~\eqref{eq:split} to write
        \begin{equation}\label{eq:one_T}
        \begin{aligned}
            & \sum_{\gamma: \text{nxt}(\mathbf{T}) \rightarrow [k]}  |\{u \in \text{nxt}(\mathbf{T}): \gamma(u) \neq \iota(u)\}|\cdot  \Pr_\sigma[\sigma=\gamma \| \beta ]\\
            = & \Pr_{\sigma(V\setminus \text{nxt}(\mathbf{T}))}[\sigma(V\setminus \text{nxt}(\mathbf{T}))=\beta] \sum_{\gamma: \text{nxt}(\mathbf{T}) \rightarrow [k]}  |\{u \in \text{nxt}(\mathbf{T}): \gamma(u) \neq \iota(u)\}|\cdot\Pr_{\sigma(\text{nxt}(\mathbf{T}))}[\sigma( \text{nxt}(\mathbf{T}))=\gamma] \\
            = & \Pr_{\sigma(V\setminus \text{nxt}(\mathbf{T}))}[\sigma(V\setminus \text{nxt}(\mathbf{T}))=\beta] \cdot \E_{\sigma(\text{nxt}(\mathbf{T}))}\left[ |\{u \in \text{nxt}(\mathbf{T}): \sigma(u) \neq \iota(u)\}|\right] \\
            = & \Pr_{\sigma(V\setminus \text{nxt}(\mathbf{T}))}[\sigma(V\setminus \text{nxt}(\mathbf{T}))=\beta] \cdot \delta \cdot |\text{nxt}(\mathbf{T})| \\
            \le & \delta d \cdot |\widehat{T}_i| \cdot \Pr_{\sigma(V\setminus \text{nxt}(\mathbf{T}))}[\sigma(V\setminus \text{nxt}(\mathbf{T}))=\beta], 
             \end{aligned}
        \end{equation}
            
         where the last inequality follows from the fact that $|\text{nxt}(\mathbf{T})| \leq |\nei_G(\widehat{T}_i)| \leq d\cdot|\widehat{T}_i|$. From \Cref{eq:split} and \Cref{eq:eitheror} we have that for all $\mathbf{T}$
         \begin{equation}\label{eq:probs}
             \begin{aligned}
               \sum_{\substack{\beta: V\setminus\text{nxt}(\mathbf{T}) \rightarrow [k] \\ \text{s.t.}\forall \gamma\\ \text{seq}(\gamma \| \beta,i)=\mathbf{T}}} \Pr_{\sigma(V\setminus \text{nxt}(\mathbf{T}))}[\sigma(V\setminus \text{nxt}(\mathbf{T}))=\beta] = \sum_{\substack{\alpha: V \rightarrow [k] \text{ s.t.}\\ \text{seq}(\alpha,i)=\mathbf{T}}}  \Pr_\sigma[\sigma=\alpha ] 
             \end{aligned}
         \end{equation}
         since
         \begin{align*}
         &\sum_{\substack{\beta: V\setminus\text{nxt}(\mathbf{T}) \rightarrow [k] \\ \text{s.t.}\forall \gamma\\ \text{seq}(\gamma \| \beta,i)=\mathbf{T}}} \Pr_{\sigma(V\setminus \text{nxt}(\mathbf{T}))}[\sigma(V\setminus \text{nxt}(\mathbf{T}))=\beta] = \qquad \qquad \text{since} \sum_{\gamma: \text{nxt}(\mathbf{T}) \rightarrow [k]}\Pr_{\sigma(\text{nxt}(\mathbf{T}))}[\sigma( \text{nxt}(\mathbf{T}))=\gamma]  = 1 \\
         &\sum_{\substack{\beta: V\setminus\text{nxt}(\mathbf{T}) \rightarrow [k] \\ \text{s.t.}\forall \gamma\\ \text{seq}(\gamma \| \beta,i)=\mathbf{T}}} \sum_{\gamma: \text{nxt}(\mathbf{T}) \rightarrow [k]}\Pr_{\sigma(V\setminus \text{nxt}(\mathbf{T}))}[\sigma(V\setminus \text{nxt}(\mathbf{T}))=\beta] \cdot \Pr_{\sigma(\text{nxt}(\mathbf{T}))}[\sigma( \text{nxt}(\mathbf{T}))=\gamma] = \sum_{\substack{\alpha: V \rightarrow [k] \text{ s.t.}\\ \text{seq}(\alpha,i)=\mathbf{T}}}  \Pr_\sigma[\sigma=\alpha ] 
         \end{align*}

        By summing the inequality \ref{eq:one_T} over all $\beta$ such that for all $\gamma$ $\text{seq}(\gamma \| \beta,i)=\mathbf{T}$, and applying \Cref{eq:probs} we get that for any $\mathbf{T} \in \mathcal{T}$
        \begin{align*}
            \sum_{\substack{\gamma: \text{nxt}(\mathbf{T}) \rightarrow [k], \\ \beta: V\setminus\text{nxt}(\mathbf{T}) \rightarrow [k] \\ \text{s.t.}\\ \text{seq}(\gamma \| \beta,i)=\mathbf{T}}}  |\{u \in \text{nxt}(\mathbf{T}): \gamma(u) \neq \iota(u)\}|\cdot  \Pr_\sigma[\sigma=\gamma \| \beta ]  \le \delta d \cdot \sum_{\substack{\alpha: V \rightarrow [k] \text{ s.t.}\\ \text{seq}(\alpha,i)=\mathbf{T}}}   |\widehat{T}_i| \cdot \Pr_\sigma[\sigma=\alpha] \, ,
        \end{align*}
        which implies
        \begin{align*}
            \E_\sigma\left[|\widehat{S}_{i+1}(\sigma)|\right] & = \sum_{\mathbf{T} \in \mathcal{T}} \sum_{\substack{\gamma: \text{nxt}(\mathbf{T}) \rightarrow [k], \\ \beta: V\setminus\text{nxt}(\mathbf{T}) \rightarrow [k] \\ \text{s.t.}\\ \text{seq}(\gamma \| \beta,i)=\mathbf{T}}}  |\{u \in \text{nxt}(\mathbf{T}): \gamma(u) \neq \iota(u)\}|\cdot  \Pr_\sigma[\sigma=\gamma\| \beta ] \\
            & \le  \delta d\sum_{\mathbf{T} \in \mathcal{T}} \sum_{\substack{\alpha: V \rightarrow [k] \text{ s.t.}\\ \text{seq}(\alpha,i)=\mathbf{T}}}  |\widehat{T}_i| \cdot \Pr_\sigma[\sigma=\alpha] \\
            & =  \delta d\sum_{\mathbf{T} \in \mathcal{T}} \sum_{\substack{\alpha: V \rightarrow [k] \text{ s.t.}\\ \text{seq}(\alpha,i)=\mathbf{T}}}  |\widehat{S}_i(\alpha)| \cdot \Pr_\sigma[\sigma=\alpha] \\
            & = \delta d \sum_{\alpha: V \rightarrow [k]}  |\widehat{S}_i(\alpha)| \cdot \Pr_\sigma[\sigma=\alpha ] \\
            & = \delta d\cdot \E_\sigma\left[|\widehat{S}_{i}(\sigma)|\right] \, .
        \end{align*}
    \end{proof}

\subsection{Proof of \Cref{claim:exp-like_sets}}
\label{subsec:explike}
\explikesets*

In the above, $\walk \sim p_S^t[H]$ describes $\walk$ as a $t$-step lazy random walk in $H$ starting in $S$, as per \Cref{def:lazy_r_walk}. We first describe the high-level idea. Imagine that $H[S]$ is a subgraph of a sufficiently big expander $\widehat{H}$. Then, for $t$ large enough, the endpoint of a lazy random walk of length $t$ starting in some vertex $u \in S$, is equally likely to end at any vertex in $\widehat{H}$, since random walks mix rapidly on expanders. In particular, the probability that the endpoint is in $S$ is small.

\begin{proof} 
Recall that $H$ is a graph with self-loops, whose edge set excluding the self-loops is denoted by $E'$, and whose self-loops are reflected by the function $\ell': V \to \mathbb{N}$.

    We construct a weighted graph $\widehat{H}=(U,\widehat{E}, \widehat{\ell},\omega)$ defined as follows: its vertex set is $U=S \cup T$, where $T$ is a set of dummy vertices disjoint from $S$ of size $|T|=20 \lceil \exp(\psi^2 t /36)\rceil \cdot |S|$; its edge set $\widehat{E}$ contains all the edges in $E'(S)$ and an edge $(u,v)$ for every $u,v \in U$ such that at least one of $u,v$ is in $T$; its self-loops are copied over $S$ and set to $0$ everywhere else, i.e. $\widehat{\ell}(u)=\ell(u)$ for all $u \in S$ and $\widehat{\ell}(u)=0$ for all $u \in T$; its weight function $\omega:F\rightarrow \mathbb{R}_{\ge 0}$ is
\begin{equation*}
    \omega(u,v)=\begin{cases}
        1 \, ,&\quad \text{if } (u,v) \in E'(S)\\
        \frac{1}{|T|}|E'(\{u\},V\setminus S)|\, ,&\quad \text{if } u \in S, v\in T\\
        \frac{1}{|T|-1}\left(d-\frac{1}{|T|}|E'(S,V\setminus S)|\right)\, ,&\quad \text{if } u,v\in T\\
    \end{cases} \, .
\end{equation*}
First, one can observe the following.
\begin{claim}\label{claim:H-regular}
    The graph  $\widehat{H}$ is $d$-regular, i.e $\widehat{\ell}(u)+\sum_{v: u \neq v \in U}\omega(u,v) = d$ for all $u \in U$.  
\end{claim}
\begin{proof}
For every $u \in S$, we have
\begin{align*}
    \widehat{\ell}(u)+\sum_{v: u \neq v \in U}\omega(u, v) & =\widehat{\ell}(u)+ \sum_{v: u \neq v \in S }\omega(u, v)  + \sum_{v \in T }\omega(u, v)\\
    & = \widehat{\ell}(u)+\left(|E'(\{u\},S\setminus \{u\})| + |E'(\{u\},V\setminus S)|\right) \\
    & = \widehat{\ell}(u)+ \left(d -\ell(u)\right) \qquad \qquad \qquad \qquad \qquad \qquad \qquad \qquad \text{since $H$ is $d$-regular}\\
    & = d\, .
\end{align*}
For every dummy vertex $u \in T$, we have
\begin{align*}
\widehat{\ell}(u)+\sum_{v: u \neq v \in U}\omega(u, v) & =\sum_{v \in S }\omega(u, v)  + \sum_{v: u \neq v \in T }\omega(u, v)\\
& = \sum_{v \in S }\frac{1}{|T|}|E'(\{v\},V\setminus S)| + d-\frac{1}{|T|}|E'(S,V\setminus S)| \\
& = \frac{1}{|T|}|E'(S,V\setminus S)| + d-\frac{1}{|T|}|E'(S,V\setminus S)| \\
& = d \, .
\end{align*}
\end{proof}
\noindent
\Cref{claim:H-regular} allows us to conveniently talk about the distribution of random walks in $\widehat{H}$, and in particular allows us to show the following.
\begin{claim}
      \label{claim:relate-rw}
      One has
      \begin{equation*}
           \Pr_{\walk \sim p^t_S[H]}\left[\forall \, r \in [t], \, w_r \in S\right] =  \Pr_{\walk \sim p^t_S[\widehat{H}]}\left[\forall \, r \in [t], \, w_r \in S\right] \, .
      \end{equation*}
  \end{claim}
\begin{proof}
     Let $A \in \R^{V\times V}$ and $\widehat{A} \in \R^{U\times U}$ be the adjacency matrix of $H$ and $\widehat{H}$ respectively, and define the corresponding lazy random walk transition matrices as $M = \frac{1}{2}I+\frac{1}{2d}A$ and \smash{$\widehat{M} = \frac{1}{2}I+\frac{1}{2d}\widehat{A}$} (where we used that \smash{$\widehat{H}$} is $d$-regular, by \Cref{claim:H-regular}). Observe that $M_{u,v} = \widehat M_{u,v}$ for all $u,v \in S$. We also let $x$ and $\widehat{x}$ be $\1_S$ in $\R^V$ and $\R^U$ respectively, and define $P=\diag(x)$, $\widehat{P}=\diag(\widehat{x})$. Then, we have
  \begin{equation*}
      \Pr_{\walk \sim p^t_S[H]}\left[\forall \, r \in [t], \, w_r \in S\right] =  x^\top (P M)^t x\frac{1}{|S|} \, ,
  \end{equation*}
  and
  \begin{equation*}
      \Pr_{\walk \sim p^t_S[\widehat{H}]}\left[\forall \, r \in [t], \, w_r \in S\right] = \widehat{x}^\top (\widehat{P} \widehat{M})^t \widehat{x}\frac{1}{|S|} \, .
  \end{equation*}
  To prove the statement, it then suffices to show
\begin{equation}
\label{eq:ind-statement}
       \supp((P M)^t x), \ \supp((\widehat{P}\widehat{M})^t \widehat{x}) \subseteq S \quad \text{ and } \quad  \forall \, u \in S, \, \left( (P M)^t x\right)_u = \left( (\widehat{P} \widehat{M})^t \widehat{x} \right)_u \, .
\end{equation}
We argue this by induction on $t$. For the base case $t=1$, one has
\begin{equation*}
    \left( P M x\right)_u = \sum_{v \in S} M_{u,v} =  \sum_{v \in S} \widehat{M}_{u,v} = \left( \widehat{P} \widehat{M} \widehat{x}\right)_u
\end{equation*}
for all $u \in S$. Moreover, $\supp(P z) \subseteq S$ for all $z \in \R^V$ and $\supp(\widehat{P} \widehat{z}) \subseteq S$ for all $\widehat{z} \in \R^U$, so one in particular has \smash{$\supp(P M x), \ \supp(\widehat{P} \widehat{M} \widehat{x} ) \subseteq S$}. For the inductive step, let us assume that~\eqref{eq:ind-statement} holds for $t$. Again, since \smash{$\supp(P z), \ \supp(\widehat{P} \widehat{z})\subseteq S$} for all $z \in \R^V,\widehat{z} \in \R^U$, we have \smash{$\supp((P M)^{t+1} x)$}, \  \smash{$\supp((\widehat{P} \widehat{M} )^{t+1}\widehat{x} ) \subseteq S$}. Then, we are left with task of showing the second condition, so we fix $u \in S$ and write
\begin{equation*}
     \left( (P M)^{t+1} x\right)_u = \left( PM(P M)^{t} x\right)_u  = \sum_{v\in S} M_{u,v} \left((P M)^{t} x\right)_v\, ,
\end{equation*}
where the last equality uses that $\supp((PM)^t x) \subseteq S$ by the inductive hypothesis. Since by the inductive hypothesis we also know $\left((P M)^{t} x\right)_v =  \left((\widehat{P} \widehat{M} )^{t} \widehat{x} \right)_v$ for all $v \in S$, we continue from above and  get
\begin{equation*}
     \left( (P M)^{t+1} x\right)_u = \sum_{v\in S} M_{u,v} \left((P M)^{t} x\right)_v = \sum_{v\in S} \widehat{M} _{u,v} \left((\widehat{P}  \widehat{M} )^{t} \widehat{x} \right)_v\, .
\end{equation*}
Now, from the inductive hypothesis we also know \smash{$\supp((\widehat{P}  \widehat{M} )^{t} \widehat{x} ) \subseteq S$}, so we can continue from above to conclude
  \begin{align*}
      \left( (P M)^{t+1} x\right)_u = \sum_{v\in S} \widehat{M} _{u,v} \left((\widehat{P}  \widehat{M} )^{t} \widehat{x} \right)_v = \left( \widehat{P}  \widehat{M} (\widehat{P}  \widehat{M} )^{t}  \widehat{x}\right)_u  = \left( (\widehat{P}  \widehat{M} )^{t+1} \widehat{x} \right)_u \, .
  \end{align*}
\end{proof}
\noindent
By \Cref{claim:relate-rw}, we have
\begin{equation}
\label{eq:colprob-ub}
    \Pr_{\walk \sim p^t_S[H]}\left[\forall \, r \in [t], \, w_r \in S\right] =  \Pr_{\walk \sim  p^t_S[\widehat{H}]}\left[\forall \, r \in [t], \, w_r \in S\right] \le \Pr_{\walk \sim p^t_S[\widehat{H}]}\left[\walk \text{ ends in } S\right] \, .
\end{equation}
To bound the right-hand side above, we  want to resort to the mixing properties of expander. One can indeed argue that $\widehat{H}$ is an expander graph.
\begin{claim}\label{claim:H-expander}
    The graph $\widehat{H}$ is a $\psi/3$-expander, i.e $\omega(Q,U\setminus Q) \ge \psi/3 \cdot d |Q| $ for all $Q\subseteq U$ with $|Q| \le |U|/2$.  
\end{claim}
\noindent
Before proving this claim, we show that the statement of \Cref{claim:exp-like_sets} now easily follows from~\eqref{eq:colprob-ub}. Let \smash{$\widehat{A},\widehat{\mathcal{L}},\widehat{M}$} be the adjacency matrix, normalized Laplacian, and lazy random walk transition matrix of \smash{$\widehat{H}$} respectively, i.e. $\widehat{\mathcal{L}}=I-\frac{1}{d}\widehat{A}$ and $\widehat{M} = \frac{1}{2}I+\frac{1}{2d}\widehat{A} = I-\frac{1}{2}\widehat{\mathcal{L}}'$. Also, let $0=\lambda_1 \le \lambda_2 \le \dots \le \lambda_{|U|}$ be the eigenvalues of $\mathcal{L}'$, and let $x_1,\dots,x_{|U|}$ be the corresponding eigenvectors, where $x_1=\frac{1}{\sqrt{|U|}} \1$. Then, one has
\begin{equation*}
    \Pr_{\walk \sim p^t_S[\widehat{H}]}\left[\walk \text{ ends in } S\right] = \1_S^\top \widehat{M}^t\1_S\frac{1}{|S|} = \frac{1}{|S|}\sum_{i = 1}^{|U|} \left(1-\frac{\lambda_i}{2}\right)^t \langle \1_S, x_i\rangle^2 = \frac{|S|}{|U|} + \frac{1}{|S|}\sum_{i = 2}^{|U|} \left(1-\frac{\lambda_i}{2}\right)^t \langle \1_S, x_i\rangle^2 \, .
\end{equation*}
Here the last equality follows from the fact that $\lambda_1 = 0$ and $x_1 = \1_U/\sqrt{|U|}$. By Cheeger's inequality, we know that $\lambda_2 \ge (\psi/3)^2/2$, so we continue from above and conclude
\begin{equation*}
    \Pr_{\walk \sim p^t_S[H]}\left[\forall \, r \in [t], \, w_r \in S\right] \le \Pr_{\walk \sim p^t_S[\widehat{H}]}\left[\walk \text{ ends in } S\right] \le \frac{|S|}{|U|}+ \frac{1}{|S|}\sum_{i = 1}^{|U|} \exp\left(-\frac{\psi^2}{36}t\right) \langle \1_S, x_i\rangle^2 =  \frac{|S|}{|U|}+ \exp\left(-\frac{\psi^2}{36}t\right) \, ,
\end{equation*}
where the first inequality uses~\eqref{eq:colprob-ub}. Recalling that $|U| \ge \lceil \exp(\psi^2 t /36)\rceil \cdot |S|$, we get the statement of \Cref{claim:exp-like_sets}. We finish by showing \Cref{claim:H-expander}.

\begin{proof}[Proof of \Cref{claim:H-expander}]
Given a set $Q \subseteq U$ with  $Q \leq |U|/2$, split $Q$ into $Q_S = Q \cap S$ and $Q_T \coloneqq Q \cap T$. We note that
\begin{equation}\label{eq:E(T, V(H)setminus T)}
\begin{aligned}
\omega(Q, U \setminus Q) & =  \omega(Q_S, U \setminus Q_S) - \omega(Q_S, Q_T) + \omega(Q_T, U\setminus Q)\\
&\geq \omega(Q_S, U \setminus Q_S) - \omega(Q_S, Q_T) + \omega(Q_T,T \setminus Q_T) \, .
\end{aligned}
\end{equation}
\noindent
We now bound the three terms, starting with $\omega(Q_T,T \setminus Q_T)$. 
By construction of $H$, we have
\begin{equation}\label{eq:TD_cut}
\omega(Q_T,T \setminus Q_T) = |Q_T|\cdot |T \setminus Q_T|\cdot \frac{1}{|T| - 1}\left(d-\frac{1}{|T|}|E'(S,V\setminus S)|\right) \, .
\end{equation}
By the assumption that $t,\psi \geq 0$, we get $|T| \ge 20(\exp(\psi^2t/36))|S| \geq 20|S|$, which rearranges to 
\begin{equation}\label{eq:S_size}
 |S| \leq \frac{|T|}{20}.
 \end{equation}
Thus, 
$|E'(S, V\setminus S)| \leq d\cdot |S| \leq d \cdot \frac{|T|}{20}$, which together with Equation~\eqref{eq:TD_cut} gives
\[\omega(Q_T,T \setminus Q_T) \geq 0.94\cdot|Q_T|\cdot |T \setminus Q_T|\cdot \frac{d}{|T|} \, .\]
Using Equation~\eqref{eq:S_size} again, and noting that $|Q_T| \leq |Q| \leq |U|/2 = (|S| + |T|)/2$, we have 
$|T \setminus Q_T|\geq  |T| -\frac{1}{2}(|S|+ |T|) =  \frac{1}{2}|T| - \frac{1}{2}|S| \geq\frac{19}{40}|T|$. Hence, 
\begin{equation}\label{eq:E(T_D, V(H)...)}
    \omega(Q_T,T \setminus Q_T) \geq 0.94 \cdot \frac{19}{40}d \cdot |Q_T|  > \frac{\psi}{3}d \cdot |Q_T|
\end{equation}
where the last inequality used the assumption that $\psi \leq 1$. Next, we bound the two terms $  \omega(Q_S, U \setminus Q_S)  - \omega(Q_S, Q_T)$ together. By construction of $\widehat{H}$, we get 

\begin{equation}\label{eq:E(T_S, T_D)}
     \omega(Q_S, Q_T)  = \sum_{u \in Q_S}\frac{1}{|T|}|E'(\{u\}, V\setminus S)|\cdot|Q_T|  = |E'(Q_S, V\setminus S)|\cdot \frac{|Q_T|}{|T|}  \leq  |E'(Q_S, V\setminus Q_S)| \cdot \frac{2}{3}\, ,  
\end{equation}
where the last inequality uses $Q_S \subseteq S$ and $\frac{|Q_T|}{|T|} \leq \frac{|S| + |T|}{2|T|} \leq \frac{21}{40}|T|$ as per \Cref{eq:S_size}. The definition of $\omega$ for edges incident on $S$, in conjunction with the expansion property of $S$ in $H$, gives
\begin{equation}
    \label{eq:eedges_leaving_S'}
      \omega(Q_S, U \setminus Q_S)  = |E'(Q_S,V\setminus Q_S)| \ge \psi d |Q_S| \, .
\end{equation}
Combining~\eqref{eq:eedges_leaving_S'} with~\eqref{eq:E(T_S, T_D)}, we get 
\begin{equation}\label{eq:(1 - T_D^2/n^2)}
 \omega(Q_S, U \setminus Q_S)  - \omega(Q_S, Q_T)\ge  |E'(Q_S,V\setminus Q_S)| -|E'(Q_S, V\setminus Q_S)| \cdot \frac{2}{3} \ge \frac{\psi}{3}d |Q_S|\, .
\end{equation}

Now we plug~\eqref{eq:E(T_D, V(H)...)} and~\eqref{eq:(1 - T_D^2/n^2)} into~\eqref{eq:E(T, V(H)setminus T)} and get
\[\omega(Q, U \setminus Q)\geq \frac{\psi}{3}d\cdot |Q_S| + \frac{\psi}{3}d\cdot|Q_T| = \frac{\psi}{3}d\cdot|Q| \, ,\]
which concludes the proof.
\end{proof}

\end{proof}

\subsection{Proof of \Cref{claim:relategraphs}}
\label{subsec:relating}
\relategraphs*
\begin{proof}[Proof of \Cref{claim:relategraphs}]
  Let $A_1$ and $A_2$ be the adjacency matrix of $H_1$ and $H_2$ respectively, and define the corresponding lazy random walk transition matrices as $M_i = \frac{1}{2}I+\frac{1}{2d}A_i$ for $i =1,2$. Then, letting $P=\diag(\1_S)$ we have
  \begin{equation*}
      \Pr_{\walk \sim p^t_S[H_i]}\left[\forall \, r \in [t], \, w_r \in S\right] = \1_S^\top (P M_i)^t\1_S\frac{1}{|S|} \, .
  \end{equation*}
  By virtue of this rewriting, it suffices to show
  \begin{equation}
  \label{eq:matrixgoal}
      \1_S^\top (P M_1)^t\1_S\frac{1}{|S|}  \le \1_S^\top (P M_2)^t\1_S\frac{1}{|S|} + t\frac{|E_2\setminus E_1|}{2d|S|} \, .
  \end{equation}
  
  Consider the matrix $M_1 - M_2$. Since both $H_1$ and $H_2$ are $d$-regular, the sum of values in any row of $M_1 - M_2$ is 0. Moreover, $(M_1 - M_2)_{u, v} = -1/(2d)$ for all $u \neq v$ if $(u, v) \in E_2\setminus E_1$ and 0 otherwise, which means that $(M_1 - M_2)_{u, u} = (l_1(u) - l_2(u))/(2d) = |\{v \in V: (u, v) \in E_2\setminus E_1\}|/(2d)$.
  
  Let  $H=(V,E_2\setminus E_1)$ be the graph induced by the edges $E_2\setminus E_1$. Observe that \[M_1-M_2 = \frac{1}{2d} L_H,\] where $L_H$ is the unnormalised Laplacian of $H$.  We apply \Cref{fact:frechetderivative} with $X=PM_1$ and $Y = PM_2$, and get
  \begin{equation*}
      (PM_1)^t = (PM_2)^t + \frac{1}{2d} \sum_{i=0}^{t-1} (PM_2)^i P  L_H (PM_1)^{t-i-1} \, .
  \end{equation*}
  In particular, this gives
  \begin{equation*}
      \1_S^\top (P M_1)^t\1_S\frac{1}{|S|}  = \1_S^\top (P M_2)^t\1_S\frac{1}{|S|} + \frac{1}{2d|S|}  \sum_{i=0}^{t-1} \1_S^\top(PM_2)^i P  L_H (PM_1)^{t-i-1}\1_S \, .
  \end{equation*}
    We use the following claim to bound each term in the sum above.
  \begin{claim}
  \label{claim:boundterms}
      For  all $i = 0,\dots, t-1$, one has
    \begin{equation*}
         \1_S^\top(PM_2)^i P  L_H (PM_1)^{t-i-1}\1_S \le |E_2\setminus E_1| \, .
    \end{equation*}
  \end{claim}
\begin{proof}[Proof of \Cref{claim:boundterms}]
  Let $i = 0,\dots, t-1$, and define $x = (PM_1)^{t-i-1}\1_S$ and $y = P (M_2 P)^i \1_S = (PM_2)^i \1_S$. We note that $x,y$ are vectors with support contained in $S$. Observe that $L_H = B_H^TB_H$, where $B_H \in \R^{|E_2\setminus E_1|\times |V|}$ is the incidence matrix of $H$, where, just to formally define $B_H$, edges in $H$ are directed arbitrarily. Since $E_2\setminus E_1 \subseteq E(S,V\setminus S)$, one has
  \begin{equation*}
  \label{eq:quadform}
      y^\top L_H x = \sum_{\substack{(u,v) \in E_2\setminus E_1; \\ u \in S, \  v\in V\setminus S}} (y_u-y_v)(x_u-x_v) = \sum_{\substack{(u,v) \in E_2\setminus E_1; \\ u \in S, \  v\in V\setminus S}} y_u x_u  \, ,
  \end{equation*}where the last equation follows from the fact that $x$ and $y$ are both 0 at any coordinate $v$ such that $v \notin S$. 
  To continue, let us fix $u \in S$ and expand out the definition of $x,y$ to write 
  \begin{align*}
      y_u x_u & = \1_u^\top (PM_2)^i \1_S \1_u^\top  (PM_1)^{t-i-1} \1_S \\
      & = \sum_{z \in S} \1_u^\top (PM_2)^i \1_z \1_u^\top  (PM_1)^{t-i-1} \1_S \\
      & = \left(\1_u^\top (PM_2)^i \1_u\right) \left(\1_u^\top  (PM_1)^{t-i-1} \1_S \right) \, .
  \end{align*}
  We now bound both factors above by $1$: if $i = 0$, then trivially $\1_u^\top (PM_2)^i \1_u = 1$, otherwise we can write
  \begin{equation*}
      \1_u^\top (PM_2)^i \1_u = \Pr_{\walk \sim p^i_u[H_2]}\left[w_i=u \text{ and } \forall \, r \in [i], \, \walk \in S \right]  \le 1 \, ,
  \end{equation*}
  and analogously if $i = t-1$ we trivially have $\1_u^\top  (PM_1)^{t-i-1} \1_S = 1$, otherwise we bound
  \begin{equation*}
      \1_u^\top  (PM_1)^{t-i-1} \1_S = \Pr_{\walk \sim p^{t-i-1}_u[H_1]}\left[ \forall \, r \in [i], \, \walk \in S \right]  \le 1 \, .
  \end{equation*}
  Hence, every term in~\eqref{eq:quadform} is bounded by $1$ and we can conclude.
  \end{proof}
  \noindent
  One can see that \Cref{claim:boundterms} gives~\eqref{eq:matrixgoal}, and the statement is proven.
\end{proof}

\begin{fact}
    \label{fact:frechetderivative}
    Let $X,\Delta$ be real square matrices, let $Y = X-\Delta$, and let $t \ge 1$ be an integer. Then, one has
    \begin{equation*}
        X^t = Y^t + \sum_{i=0}^{t-1} Y^i \Delta X^{t-i-1} \, .
    \end{equation*}
\end{fact}
\begin{proof}
    We prove the claim by induction on $t$. The base case for $t=1$ is trivial as $X = Y+\Delta$. For the inductive step, let us assume that
    \begin{equation*}
        X^t = Y^t + \sum_{i=0}^{t-1} Y^i \Delta X^{t-i-1} \, .
    \end{equation*}
    Then, we have
    \begin{equation*}
        X^{t+1} = X^t \cdot X = \left(Y^t + \sum_{i=0}^{t-1} Y^i \Delta X^{t-i-1}\right)X
    \end{equation*}
    by the inductive hypothesis. We distribute $X$ and obtain
    \begin{equation*}
        X^{t+1} = Y^t X + \sum_{i=0}^{t-1} Y^i \Delta X^{t-i-1} X = Y^{t+1} +Y^t\Delta + \sum_{i=0}^{t-1} Y^i \Delta X^{(t+1)-i-1} = Y^{t+1} + \sum_{i=0}^{t} Y^i \Delta X^{(t+1)-i-1} \, ,
    \end{equation*}
    which proves the identity for $t+1$.
\end{proof}

\subsection{Proof of Lemma \ref{claim:projection_k}}\label{sec:projection_k}
\projectionk*
\begin{proof}[Proof of Lemma \ref{claim:projection_k}]
For every real symmetric matrix $A$, it holds that $\|A\|_{\op} = \sqrt{\lambda_{\max}(A^\top A)} = \max_{i \in [n]} \{ |\lambda_i(A)|\}$. Therefore, we have \begin{equation}\label{eq:P_op_main}
    \begin{aligned}
       \| P - P^* \|_{\op}  & = \max_{\|x\|_2 = 1}\|(P - P^*)x\|_2 \\
       & = \max_{i \in [n]} \{|\lambda_i(P-P^*)| \} \qquad  \text{ since $P - P^*$ is symmetric}\\
       & = \max_{\|x\|_2 = 1} | x^\top(P-P^*)x| \\
       & = \max_{\|x\|_2 = 1} \left|x^\top\left(\sum_{i = 1}^k {v_i^* v^*_i}^{\top}- \sum_{i = 1}^{k}v_i v_i^{\top}\right)x\right| \, . \\
    \end{aligned}
    \end{equation}
\begin{claim}\label{claim:inspan}
    For all $x \in \spn(\{v_i\}_{i=1}^k)$, it holds that 
    $$\left|x^\top\left(\sum_{i = 1}^k {v_i^* v^*_i}^{\top}- \sum_{i = 1}^{k}v_i v_i^{\top}\right)x\right|\leq \frac{4\sqrt{\epsilon}}{\phi}\|x\|_2^2 \, .$$
\end{claim}
\begin{proof}
    Suppose $x \in \spn(\{v_i\}_{i=1}^k)$. Then $x = U_{[k]}\alpha$ for some $\alpha \in \R^k$. 
    We will now show that 
$$x^\top\left(\sum_{i = 1}^{k}v_i v_i^{\top} - \sum_{i = 1}^k {v_i^* v^*_i}^{\top} \right)x =\alpha^\top \left(I -  \sum_{i =1}^k |C_i| \mu_i \mu_i ^\top \right)\alpha \, . $$

\noindent
    First, note that 
    \begin{equation}\label{eq:sumvi}
        \sum_{i=1}^k v_i v_i^\top = U_{[k]}U_{[k]}^\top \, .
    \end{equation}
\noindent
    Furthermore, note that $\{(1/\sqrt{|C_i|})\cdot\mathbbm{1}_{C_i}\}_{i=1}^{k}$ is a set of $k$ orthonormal eigenvectors of $\Lin$ with eigenvalue $0$,  and therefore they span the eigenspace of $\Lin$ corresponding to eigenvalue $0$ (otherwise, the orthogonal component gives an additional $(k+1)$-th $0$-eigenvalue, which is impossible as $\lambda_{k+1}\geq \phi^2/2$ by \Cref{lem:bnd-lambda}). Consequently, we have 
    $$\sum_{i = 1}^kv_i^*{v_i^*}^\top = \sum_{i = 1}^k \frac{1}{|C_i|}\mathbbm{1}_{C_i}\mathbbm{1}_{C_i}^\top \, .$$
We also have that 
    \begin{align*}
     \mu_i &=\frac{1}{|C_i|}\sum_{u\in C_i }f_u = \frac{1}{|C_i|}\sum_{u \in C_i }U_{[k]}^\top \mathbbm{1}_u 
     = \frac{1}{|C_i|}U_{[k]}^\top \sum_{u \in C_i}\mathbbm{1}_u  = \frac{1}{|C_i|}U_{[k]}^\top\mathbbm{1}_{C_i}\, ,
    \end{align*}
which gives
\begin{equation}\label{eq:sumvi*}
    \sum_{i=1}^k |C_i|\mu_i \mu_i^\top  =   \sum_{i=1}^k \frac{1}{|C_i|}U_{[k]}^\top \mathbbm{1}_{C_i}\mathbbm{1}_{C_i}^\top U_{[k]}=  U_{[k]}^\top \left(\sum_{i=1}^k \frac{1}{|C_i|} \mathbbm{1}_{C_i} \mathbbm{1}_{C_i}^\top \right) U_{[k]}  =  U_{[k]}^\top \left(\sum_{i=1}^k v_i^*{v_i^*}^\top \right) U_{[k]}\, . 
\end{equation}
Combining~\eqref{eq:sumvi} and~\eqref{eq:sumvi*}, and recalling that $x = U_{[k]}\alpha$ for some $\alpha \in \R^k$,  we obtain
\begin{align*}
x^\top\left( \sum_{i = 1}^{k}v_i v_i^{\top} - \sum_{i = 1}^k v_i^* {v^*_i}^{\top}\right)x &=\alpha^\top U_{[k]}^\top \left(\sum_{i = 1}^{k}v_i v_i^{\top} - \sum_{i = 1}^k {v_i^* v^*_i}^{\top}\right)U_{[k]}\alpha \\
& = \alpha^\top  U_{[k]}^\top \left(\sum_{i=1}^k v_i v_i^{\top} \right)U_{[k]}\alpha -  \alpha^\top  U_{[k]}^\top \left(\sum_{i=1}^k v_i^* {v_i^*}^{\top} \right)U_{[k]}\alpha \\
& = \alpha^\top  U_{[k]}^\top U_{[k]}U_{[k]}^\top U_{[k]} \alpha - \alpha^\top \left(  \sum_{i=1}^k |C_i|\mu_i \mu_i^\top \right)\alpha \\
& = \alpha^\top\left(I - \sum_{i=1}^k |C_i|\mu_i \mu_i^\top \right)\alpha\, ,
\end{align*}
where the last equality follows by the fact that $U_{[k]}^\top U_{[k]} = I_k.$
Thus, applying Lemma \ref{lemma:gklmsL9}, we obtain 
\begin{equation*}
    \left|x^\top\left(\sum_{i = 1}^k {v_i^* v^*_i}^{\top}- \sum_{i = 1}^{k}v_i v_i^{\top}\right)x\right|  =  \left| \alpha^\top\left(I - \sum_{i=1}^k |C_i|\mu_i \mu_i^\top \right)\alpha \right| \leq \frac{4\sqrt{\epsilon}}\phi\|\alpha\|^2_2 = \frac{4\sqrt{\epsilon}}\phi\|U_{[k]}^\top x\|^2_2  = \frac{4\sqrt{\epsilon}}\phi\|x\|_2^2\, . 
\end{equation*}
\end{proof}
\begin{claim}\label{claim:inspanperp}
    For all $x \in \spn\left(\{ v_i\}_{i = 1}^k\right) ^{\perp}$, it holds that 
     $$\left|x^\top\left(\sum_{i = 1}^k {v_i^* v^*_i}^{\top}- \sum_{i = 1}^{k}v_i v_i^{\top}\right)x\right|\leq k \frac{4\sqrt{\epsilon}}{\phi}\|x\|_2^2 \, .$$
\end{claim}
\begin{proof}
    Suppose not, and let $x \in \spn(\{ v_i\}_{i = 1}^k)^{\perp}$ be a unit norm vector with $|x^\top(\sum_{i = 1}^k {v_i^* v^*_i}^{\top}- \sum_{i = 1}^{k}v_i v_i^{\top})x| > k {4\sqrt{\epsilon}}/{\phi} $. Since $x \perp v_i$ for $i =1, \ldots, k$, we have 

    \begin{equation}\label{eq:xvi*}
     k \frac{4\sqrt{\epsilon}}{\phi} < \left|x^\top\left(\sum_{i = 1}^k {v_i^* v^*_i}^{\top}- \sum_{i = 1}^{k}v_i v_i^{\top}\right)x\right|= \sum_{i = 1}^k \langle x, v_i^* \rangle ^2.
    \end{equation}

    \noindent
    Furthermore, $v_1, \dots, v_k,x$ are $(k+1)$ orthonormal vectors, so we can extend them to an orthonormal basis $ v_1, \ldots, v_k,x,  u_{k+2}, \ldots u_{n} $ of $\R^n$ for some vectors $u_i$. 
   Expanding $v_1^*, \ldots, v_k^*$ with respect to this basis, we obtain 
   \begin{align*}
       k &= \sum_{i=1}^k \|v_i^*\|^2_2 \\
       & = \sum_{i = 1}^{k}\left( \sum_{j =1}^{k}\langle v_i^*, v_j\rangle ^2 + \langle v_i^*, x\rangle^2 + \sum_{j = k+2}^n\langle v_i^*, u_j\rangle ^2\right) \\
       & \geq \sum_{i = 1}^k \sum_{j = 1}^{k}\langle v_i^*, v_j\rangle ^2 + \sum_{i = 1}^k\langle v_i^*, x\rangle^2 \\
       & > \sum_{j =1}^k v_j^\top \left(\sum_{i = 1}^k v_i^*{v_i^*}^\top \right)v_j + k \frac{4\sqrt{\epsilon}}{\phi}  \qquad \text{by~\eqref{eq:xvi*}} \\
       & \geq \sum_{j =1}^k v_j^\top \left(\sum_{i = 1}^k v_i{v_i}^\top \right)v_j - \frac{4\sqrt{\epsilon}}{\phi}\sum_{i = 1}^k\|v_i\|^2_2 + k \frac{4\sqrt{\epsilon}}{\phi}   \qquad \text{by applying \Cref{claim:inspan} to $v_1, \dots, v_k$} \\
       & = k \, ,
   \end{align*}
   which is a contradiction. 
\end{proof}
\noindent
To complete the proof of the lemma, we need to show that $ \max_{\|x\|_2 = 1} |x^\top(\sum_{i = 1}^k {v_i^* v^*_i}^{\top}- \sum_{i = 1}^{k}v_i v_i^{\top})x| \leq 8   \sqrt{k {\sqrt{\epsilon}}/{\phi}}$. Given any $x \in \R^n$ with $\|x\|_2 = 1$, decompose $x = u +  w$ for vectors $u \in \spn(\{v_i \}_{i=1}^k)$,  $w \in \spn(\{v_i \}_{i=1}^k)^{\perp}$. Applying \Cref{claim:inspan} and \Cref{claim:inspanperp} to $u$ and $w$, respectively, we get 
\begin{equation}\label{eq:xmixed}
\begin{aligned}
    \left|x^\top\left(\sum_{i = 1}^k {v_i^* v^*_i}^{\top}- \sum_{i = 1}^{k}v_i v_i^{\top}\right)x\right|& =\left| ( u + w)^\top \left(\sum_{i = 1}^k {v_i^* v^*_i}^{\top}- \sum_{i = 1}^{k}v_i v_i^{\top}\right)(u + w) \right| \\
    & \leq \left| u^\top  \left(\sum_{i = 1}^k {v_i^* v^*_i}^{\top}- \sum_{i = 1}^{k}v_i v_i^{\top}\right)u \right| + \left|w^\top \left(\sum_{i = 1}^k {v_i^* v^*_i}^{\top}- \sum_{i = 1}^{k}v_i v_i^{\top}\right)w \right| \\
    & + 2\left| u^\top  \left(\sum_{i = 1}^k {v_i^* v^*_i}^{\top}- \sum_{i = 1}^{k}v_i v_i^{\top}\right)w\right| \\
    & \leq 4 \frac{\sqrt{\epsilon}}{\phi}\|u\|^2_2 + 4k \frac{\sqrt{\epsilon}}{\phi}\|w\|^2_2 + 2\left| u^\top  \left(\sum_{i = 1}^k {v_i^* v^*_i}^{\top}- \sum_{i = 1}^{k}v_i v_i^{\top}\right)w\right| \, .
\end{aligned}
\end{equation}
To bound the term $2| u^\top  (\sum_{i = 1}^k {v_i^* v^*_i}^{\top}- \sum_{i = 1}^{k}v_i v_i^{\top})w|$, note that 
\begin{align*}
     \left| u^\top  \left(\sum_{i = 1}^k {v_i^* v^*_i}^{\top}- \sum_{i = 1}^{k}v_i v_i^{\top}\right)w \right| & = \left| u^\top  \left(\sum_{i = 1}^k {v_i^* v^*_i}^{\top}\right)w \right| \qquad \text{since $w \perp v_i$ for all $i \in [k]$} \\
     & = \left|\sum_{i=1}^k \langle u, v_i^*\rangle \langle w, v_i^*\rangle\right| \\
     & \leq \sqrt{ \left(\sum_{i=1}^k \langle u,v_i^*\rangle ^2\right)\left(\sum_{i=1}^k \langle w,v_i^*\rangle ^2 \right)} \qquad \text{by Cauchy-Schwarz} \\
     & \leq \|u\|_2\sqrt{\left(\sum_{i=1}^k \langle w,v_i^*\rangle ^2 \right)} \\
     & =  \|u\|_2 \sqrt{ w^\top\left(\sum_{i = 1}^k {v_i^* v^*_i}^{\top}- \sum_{i = 1}^{k}v_i v_i^{\top}\right) w} \qquad \text{since $w \perp v_i$ for all $i \in [k]$} \\
     & \leq \|u\|_2 \|w\|_2 \cdot 2 k^{1/2} \frac{\epsilon^{1/4}}{\phi^{1/2}} \qquad \text{by Claim \ref{claim:inspanperp}.}
\end{align*}
Substituting back into~\eqref{eq:xmixed}, we obtain 
\begin{align*}
     \left|x^\top\left(\sum_{i = 1}^k {v_i^* v^*_i}^{\top}- \sum_{i = 1}^{k}v_i v_i^{\top}\right)x\right| &\leq 4 \frac{\sqrt{\epsilon}}{\phi}\|u\|^2_2 + 4k \frac{\sqrt{\epsilon}}{\phi}\|w\|^2_2 + 2\left| u^\top  \left(\sum_{i = 1}^k {v_i^* v^*_i}^{\top}- \sum_{i = 1}^{k}v_i v_i^{\top}\right)w\right|   \\
     & \leq  4 \frac{\sqrt{\epsilon}}{\phi}\|u\|^2_2 + 4k \frac{\sqrt{\epsilon}}{\phi}\|w\|^2_2 + 4 \|u\|_2 \|w\|_2 \cdot k^{1/2} \frac{\epsilon^{1/4}}{\phi^{1/2}} \\
     & \leq 8 \max \left \{ k \frac{\sqrt{\epsilon}}{\phi}, \sqrt{k \frac{\sqrt{\epsilon}}{\phi}}  \right\} \|x\|_2^2 \qquad \text{using the fact that $\|u\|^2_2, \|w\|^2_2 \leq \|x\|^2_2$}   \\
     & =  8   \sqrt{k \frac{\sqrt{\epsilon}}{\phi}}\|x\|_2^2 \qquad \text{since $k \leq \frac{\phi}{\sqrt{\epsilon}}$, by the lemma assumption.}
     \end{align*}
Finally, combining with~\eqref{eq:P_op_main}, we get 
$ \|P - P^*\|_{\op} \leq \max_{\|x\|_2=1}  |x^\top(\sum_{i = 1}^k {v_i^* v^*_i}^{\top}- \sum_{i = 1}^{k}v_i v_i^{\top})x|  \leq  8   \sqrt{k \frac{\sqrt{\epsilon}}{\phi}}$.
\end{proof}

\end{document}